\documentclass[aos,preprint]{imsart}

\usepackage{amsmath,amssymb,amsthm}
\usepackage[numbers]{natbib}

\usepackage{xr}
\usepackage[pdfstartview=FitB]{hyperref}
\usepackage{graphicx,epsfig,epstopdf}
\usepackage{url,mathtools}
\usepackage{colortbl}

\renewcommand{\qed}{\hfill{\tiny \ensuremath{\blacksquare} }}%

\newcommand{\mS}{\mathcal{S}}

\newcommand{\mF}{\mathcal{F}}

\newcommand{\mG}{\mathcal{G}}

\newcommand{\mT}{\mathcal{T}}

\newcommand{\supp}{\mathrm{supp}}

\newcommand{\bG}{\mathbb{G}}

\newcommand{\G}{{\Lambda}}

\renewcommand{\qed}{\hfill {\tiny {\ensuremath{\blacksquare}}}}

\newtheorem{theorem}{Theorem}[section]
\newtheorem{corollary}{Corollary}[section]
\newtheorem{lemma}{Lemma}[section]

\newtheorem{assumption}{Assumption}[section]
\newtheorem{definition}{Definition}[section]

\theoremstyle{definition}

\newtheorem{algorithm}{Algorithm}

\newtheorem{remark}{Comment}[section]
\numberwithin{remark}{section}

\numberwithin{equation}{section}
\numberwithin{theorem}{section}

\usepackage{html}
\usepackage{url}

\newcommand{\Gn}{\mathbb{G}_n}

\newcommand{\Pn}{\mathbb{P}_n}

\newcommand{\F}{\mathcal{F}}

\newcommand{\Ep}{{\mathrm{E}}}

\renewcommand{\Pr}{{\mathrm{P}}}

\renewcommand{\hat}{\widehat}
\newcommand{\En}{{\mathbb{E}_n}}

\renewcommand{\Pr}{{\mathrm{P}}}
\newcommand{\RR}{\mathbb{R}}

\newcommand{\pp}{{\tilde p}}
\newcommand{\uu}{{u}}
\newcommand{\UU}{\mathcal{U}}
\newcommand{\UUU}{\widetilde{\mathcal{U}}}
\newcommand{\tu}{{\tilde u}}

\newcommand{\ceil}[1]{\left\lceil #1 \right\rceil}
\newcommand{\semin}[1]{\phi_{{\rm min}}(#1)}
\newcommand{\semax}[1]{\phi_{{\rm max}}(#1)}
\renewcommand{\hat}{\widehat}
\renewcommand{\leq}{\leqslant}
\renewcommand{\geq}{\geqslant}

\newcommand{\diag}{{\rm diag}}

\renewcommand{\[}{\left[}
\renewcommand{\]}{\right]}

\begin{document}
\begin{frontmatter}

\title{Uniformly Valid Post-Regularization Confidence Regions for Many Functional Parameters in Z-Estimation Framework\thanksref{T1}}
\runtitle{Post-Regularization Confidence Regions}
\thankstext{T1}{First ArXiv version: December 23, 2015. Revised version: \today.}

\begin{aug}

\author{\fnms{Alexandre} \snm{Belloni}\thanksref{m1}\ead[label=e1]{abn5@duke.edu}}
\author{\fnms{Victor} \snm{Chernozhukov}\thanksref{m2}\ead[label=e2]{vchern@mit.edu}},\\
\author{\fnms{Denis} \snm{Chetverikov}\thanksref{m3}\ead[label=e3]{chetverikov@econ.ucla.edu}}
\and
\author{\fnms{Ying} \snm{Wei}\thanksref{m4}\ead[label=e4]{yw2148@cumc.columbia.edu}}

\runauthor{Belloni Chernozhukov Chetverikov Wei}

\affiliation{Duke University\thanksmark{m1}, MIT\thanksmark{m2}, UCLA\thanksmark{m3}, and Columbia University\thanksmark{m4}}

\address{Fuqua School of Business\\
Duke University\\
100 Fuqua Drive\\
Durham, NC 27708, USA.\\
\printead{e1}}

\address{Department of Economics and\\
Operations Research Center, MIT \\
50 Memorial Drive \\
Cambridge, MA 02142, USA.\\
\printead{e2}}

\address{Department of Economics, UCLA\\
Bunche Hall, Rm 8283 \\
315 Portola Plaza \\
Los Angeles, CA 90095, USA.\\
\printead{e3}}

\address{Department of Biostatistics\\
Columbia University\\
722 West 168th St, Rm 633 \\
New York, NY 10032, USA.\\
\printead{e4}}

\end{aug}

\begin{abstract}
In this paper we develop procedures to construct simultaneous confidence bands for $\pp$ potentially infinite-dimensional parameters after model selection for general moment condition models where $\pp$ is potentially much larger than the sample size of available data, $n$. This allows us to cover settings with functional response data where each of the $\pp$ parameters is a function. The procedure is based on the construction of score functions that satisfy certain orthogonality condition.
The proposed simultaneous confidence bands rely on uniform central limit theorems for high-dimensional vectors (and not on Donsker arguments as we allow for $\pp \gg n$). To construct the bands, we employ a multiplier bootstrap procedure which is computationally efficient as it only involves resampling the estimated score functions (and does not require resolving the high-dimensional optimization problems). We formally apply the general theory to inference on regression coefficient process in the distribution regression model with a logistic link, where two implementations are analyzed in detail. Simulations and an application to real data are provided to help illustrate the applicability of the results.
\end{abstract}


\end{frontmatter}

\section{Introduction}
High-dimensional models have become increasingly popular in the last two decades. Much research has been conducted on estimation of these models. However, inference about parameters in these models is much less understood, although the literature on inference is growing quickly; see the list of references below. In particular, despite its practical relevance, there are few papers trying to solve the problem of construction of simultaneous confidence bands on many target parameters in these models. (One exception is \cite{BCK-LAD}. Simultaneous confidence bands have been constructed by Bonferroni adjustments but the resulting confidence bands are often conservative, see, e.g., simulation results in \cite{vandeGeerBuhlmannRitov2013}.) In this paper we provide a solution to this problem by constructing simultaneous confidence bands for parameters in a very general framework of moment condition models, allowing for many functional parameters, where each parameter itself can be an infinite-dimensional object, and the number of parameters can be much larger than the sample size of available data.

As a substantive application, we apply our general results to provide simultaneous confidence bands for parameters in a logistic regression model with functional response data
 \begin{equation}\label{Def:FunctionalLogistic}
 \Ep_P[ Y_u \mid D, X] = \G(D'\theta_u+X'\beta_u), \ \ u\in \UU,
 \end{equation}
 where $D = (D_1,\dots,D_{\pp})'$ is a $\pp$-vector of covariates whose effects are of interest, $X = (X_1,\dots,X_p)'$ is a $p$-vector of controls, $\Lambda\colon\mathbb R\to\mathbb R$ is the logistic link function, $\mathcal U = [0,1]$ is a set of indices, and for each $u\in \mathcal U$, $Y_u = 1\{Y\leq (1-u) \underline y + u\bar y\}$ for some constants $\underline y\leq \bar y$ and the response variable $Y$,  $\theta_u = (\theta_{u1},\dots,\theta_{u\pp})'$ is a vector of target parameters, and $\beta_u = (\beta_{u 1},\dots,\beta_{u p})'$ is a vector of nuisance parameters. Here, both $\pp$ and $p$ are allowed to be potentially much larger than the sample size $n$, and we have $\pp$ functional target parameters $(\theta_{u j})_{u\in\mathcal U}$ and $p$ functional nuisance parameters $(\beta_{u j})_{u\in\mathcal U}$. This example is important because it demonstrates that our methods can be used for inference about the whole distribution of the response variable $Y$ given $D$ and $X$ in a high-dimensional setting, and not only about some particular features of it such as mean or median. This model is called a  distribution regression model in \cite{CFM13} and a conditional transformation model in \cite{HKB14}, who argue that the model provides a rich class of models for conditional distributions, and offers a useful generalization of traditional proportional hazard models as well as a useful alternative to quantile regression. We develop inference methods to construct simultaneous confidence bands for many functional parameters of this model in Section \ref{Sec:Application}. Toward this goal, our contributions include to effectively estimate a continuum of high-dimensional nuisance parameters, allow for approximately sparse models, control sparse eigenvalues of a continuum of random matrices, and establish the validity of a multiplier bootstrap for the construction of confidence bands for the many functional parameters of interest. In particular, these contributions go much beyond \cite{BCK-LAD}, which considers the special case of many scalar parameters.

Our general results refer to the problem of estimating the set of parameters $(\theta_{u j})_{u\in\UU,j\in[\pp]}$ in the moment condition model,
\begin{equation}\label{Def:Moment1}
 \Ep_P[ \psi_{u j}(W,\theta_{u j},\eta_{u j}) ] = 0, \quad \uu \in \UU, \ j \in [\pp],
\end{equation}
where $W$ is a random element that takes values in a measurable space $(\mathcal{W},\mathcal{A}_{\mathcal W})$ according to a probability measure $P$, $\UU\subset \mathbb R^{d_u}$ and $[\pp]:=\{1,\dots,\pp\}$ are sets of indices, and for each $u\in\UU$ and $j\in[\pp]$, $\psi_{u j}$ is a known score function, $\theta_{u j}$ is a scalar parameter of interest, and $\eta_{uj}$ is a potentially high-dimensional (or infinite-dimensional) nuisance parameter. Assuming that a random sample of size $n$, $(W_i)_{i=1}^n$, from the distribution of $W$ is available together with suitable estimators $\widehat\eta_{u j}$ of $\eta_{u j}$, we aim to construct simultaneous confidence bands for $(\theta_{u j})_{u\in\UU,j\in[\pp]}$ that are valid uniformly over a large class of probability measures $P$, say $\mathcal P_n$. Specifically, for each $u\in\UU$ and $j\in[\pp]$, we construct an appropriate estimator $\check\theta_{u j}$ of $\theta_{u j}$ along with an estimator of the standard deviation of $\sqrt{n}(\check\theta_{u j} - \theta_{u j})$, $\widehat \sigma_{uj}$, such that

\begin{equation}\label{UnifCoverage}
\Pr_P\left ( \check\theta_{\uu j} -  \frac{c_\alpha\hat\sigma_{\uu j}}{\sqrt n} \leq \theta_{uj} \leq \check\theta_{\uu j} + \frac{c_\alpha\hat\sigma_{\uu j}}{\sqrt n}, \forall \uu \in \UU,  j\in[\pp] \right) \to 1-\alpha,
\end{equation}
uniformly over $P \in \mathcal{P}_n$, where $\alpha\in(0,1)$ and $c_\alpha$ is an appropriate critical value, which we choose to construct using a multiplier bootstrap method. The left- and the right-hand sides of the inequalities inside the probability statement \eqref{UnifCoverage} then can be used as bounds in simultaneous confidence bands for $\theta_{u j}$'s.
In this paper, we are particularly interested in the case when $\pp$ is potentially much larger than $n$ and $\UU$ is an uncountable subset of $\mathbb R^{d_u}$, so that for each $j\in[\pp]$, $(\theta_{uj})_{u\in\UU}$ is an infinite-dimensional (that is, functional) parameter.

In the presence of high-dimensional nuisance parameters, construction of valid confidence bands is delicate. Dealing with high-dimensional parameters requires relying upon regularization that leads to lack of asymptotic linearization of the estimators of target parameters since regularized estimators of nuisance parameters suffer from a substantial bias and this bias spreads into the estimators of the target parameters. This lack of asymptotic linearization in turn typically translates into severe distortions in coverage probability of the confidence bands constructed by traditional techniques that are based on perfect model selection; see \cite{leeb:potscher:pms}, \cite{leeb:potscher:review}, \cite{leeb:potscher:hodges},  \cite{potscher:leeb:dpe}.
To deal with this problem, we assume that the score functions $\psi_{u j}$ are constructed to satisfy a near-orthogonality condition that makes them immune to first-order changes in the value of the nuisance parameter, namely
\begin{equation}\label{DefOrtho}
\left.\partial_r\bigg\{\Ep_P\Big[\psi_{\uu j}(W,\theta_{\uu j},\eta_{uj}+r\tilde \eta\})\Big]\bigg\}\right|_{r=0}\approx 0, \quad \ u \in \UU, \ j\in[\pp],
\end{equation}
for all $\tilde\eta$ in an appropriate set where $\partial_r$ denotes the derivative with respect to $r$.  We shall often refer to  this condition as \textit{Neyman orthogonality}, since in low-dimensional parametric settings the orthogonality property originates in the work of Neyman on the $C(\alpha)$ test in the 50s. In Section \ref{sec: general} below, we describe a few general methods for constructing the score functions $\psi_{u j}$ obeying the Neyman orthogonality condition.

The Neyman orthogonality condition \eqref{DefOrtho} is important because it helps to make sure that the bias from the estimators of the high-dimensional nuisance parameters does not spread into the estimators of the target parameters. In particular, under \eqref{DefOrtho}, it follows that
$$
\Ep_{P,W}\Big[\psi_{u j}(W,\theta_{u j},\widehat\eta_{u j})\Big]\approx 0,\quad u\in\UU, \ j\in[\pp],
$$
at least up to the first order, where the index $W$ in $\Ep_{P,W}[\cdot]$ means that the expectation is taken over $W$ only. This makes the estimators of the target parameters $\theta_{u j}$ immune to the bias in the estimators $\widehat \eta_{u j}$, which in turn improves their statistical properties and opens up the possibilities for valid inference.

As the framework \eqref{Def:Moment1} covers a broad variety of applications, it is instructive to revisit the logistic regression model with functional response data (\ref{Def:FunctionalLogistic}). To construct score functions $\psi_{u j}$ that satisfy both the moment conditions \eqref{Def:Moment1} and the Neyman orthogonality condition \eqref{DefOrtho} in this example, for $u\in\UU$ and $j\in[\pp]$, define a $(\pp + p - 1)$-vector of additional nuisance parameters
\begin{equation}\label{eq: gamma definition}
\gamma_u^j = \arg\min_{\gamma \in \mathbb R^{\pp + p - 1}} \Ep_P[f_u^2\{D_j - X^j\gamma\}^2],
\end{equation}
where $X^j = (D_{[\pp]\setminus j}', X')$, $D_{[\pp]\setminus j} = (D_1,\dots,D_{j - 1}, D_{j + 1}, \dots, D_{\pp})'$, and
\begin{equation}\label{eq: logit weight function}
f_u^2 = f_u^2(D,X) = {\rm Var}_P(Y_u\mid D,X).
\end{equation}
Then, denoting $W = (Y,D,X)$ and splitting $\theta_u$ into $\theta_{u j}$ and $\theta_{u[\pp]\setminus j} = (\theta_{u 1},\dots,\theta_{u j-1},\theta_{u j+1},\dots,\theta_{u [\pp]})'$, we set
$$
\psi_{u j}(W,\theta_{u j},\eta_{u j})=\Big\{ Y_u - \Lambda(D_j\theta_{u j} +X^j \beta_{u}^j) \Big\} ( D_j - X^j\gamma_{u}^j),
$$
where $\eta_{u j} = (\beta_{u}^j, \gamma_{u}^j)$ and $\beta_u^j = (\theta_{u[\pp]\setminus j}', \beta_u')'$.
It is straightforward to see that these score functions $\psi_{u j}$ satisfy the moment conditions \eqref{Def:Moment1} and to see that they also satisfy the Neyman orthogonality condition \eqref{DefOrtho}, observe that
{\small
\begin{align*}
&\left.\partial_\beta\Big\{\Ep_P[\psi_{u j}(W,\theta_{u j},\beta,\gamma_u^j)]\Big\}\right|_{\beta=\beta_u^j}=-\Ep_P\Big[f^2_u\{ D_j - X^j\gamma_u^j\}(X^j)'\Big] = 0,\\
&\left.\partial_\gamma\Big\{\Ep_P[\psi_{u j}(W,\theta_{u j},\beta_u^j,\gamma)]\Big\}\right|_{\gamma=\gamma_u^j}=-\Ep_P\Big[\{Y_u - \Lambda(D'\theta_u+X'\beta_u) \}(X^j)'\Big] = 0,
\end{align*}}\!where the first line holds by definition of $f_u^2$ and $\gamma_u^j$ since ${\rm Var}_P(Y_u\mid D,X) = \Lambda'(D'\theta_u + X'\beta_u)$, and the second by (\ref{Def:FunctionalLogistic}). Because of this orthogonality condition, we can exploit the moment conditions \eqref{Def:Moment1} to construct regular, $\sqrt{n}$-consistent, estimators of $\theta_{u j}$ even if non-regular, regularized or post-regularized, estimators of $\eta_{u j} = (\beta_u^j, \gamma_u^j)$ are used to cope with high-dimensionality. Using these regular estimators of $\theta_{u j}$, we then can construct valid confidence bands \eqref{UnifCoverage}.

% To construct score functions $\psi_{u j}$ that satisfy both the moment conditions \eqref{Def:Moment1} and the Neyman orthogonality condition \eqref{DefOrtho} in this example, for $u\in\UU$ and $j\in[\pp]$, define a $(\pp + p - 1)$-vector of additional nuisance parameters

%Then, denoting $W = (Y,D,X)$ and splitting $\theta_u$ into $\theta_{u j}$ and $\theta_{u[\pp]\setminus j} = (\theta_{u 1},\dots,\theta_{u j-1},\theta_{u j+1},\dots,\theta_{u [\pp]})'$, we set
%$$
%\psi_{u j}(W,\theta_{u j},\eta_{u j})=\Big\{ Y_u - \Lambda(D_j\theta_{u j} +X^j \beta_{u}^j) \Big\} ( D_j - X^j\gamma_{u}^j),
%$$
%where $X^j = (D_{[\pp]\setminus j}', X')$, $D_{[\pp]\setminus j} = (D_1,\dots,D_{j - 1}, D_{j + 1}, \dots, D_{\pp})'$, $\eta_{u j} = (\beta_{u}^j, \gamma_{u}^j)$, and $\beta_u^j = (\theta_{u[\pp]\setminus j}', \beta_u')'$.

% Because of this orthogonality condition, we can exploit the moment conditions \eqref{Def:Moment1} to construct regular, $\sqrt{n}$-consistent, estimators of $\theta_{u j}$ even if non-regular, regularized or post-regularized, estimators of $\eta_{u j} = (\beta_u^j, \gamma_u^j)$ are used to cope with high-dimensionality. Using these regular estimators of $\theta_{u j}$, we then can construct valid confidence bands \eqref{UnifCoverage}.

Our general approach to construct simultaneous confidence bands, which is developed in Section \ref{sec: general}, can be described as follows. First, we construct the moment conditions \eqref{Def:Moment1} that satisfy the Neyman orthogonality condition \eqref{DefOrtho}, and use these moment conditions to construct estimators $\check\theta_{u j}$ of $\theta_{u j}$ for all $u\in\UU$ and $j\in[\pp]$. Second, under appropriate regularity conditions, we establish a Bahadur representation for $\check\theta_{u j}$'s. Third, employing the Bahadur representation, we are able to derive a suitable Gaussian approximation for the distribution of $\check\theta_{u j}$'s. Importantly, the Gaussian approximation is possible even if both $\pp$ and the dimension of the index set $\UU$, $d_u$, are allowed to grow with $n$, and $\pp$ asymptotically remains much larger than $n$. Finally, from the Gaussian approximation, we construct simultaneous confidence bands using a multiplier bootstrap method. Here, the Gaussian and bootstrap approximations are constructed by applying the results on high-dimensional central limit and bootstrap theorems established in \cite{chernozhukov2013gaussian}, \cite{chernozhukov2014clt}, \cite{chernozhukov2012gaussian}, \cite{chernozhukov2012comparison}, and \cite{chernozhukov2015noncenteredprocesses} by verifying the conditions there.

Although regularity conditions underlying our approach can be verified for many models defined by moment conditions, for illustration purposes, we explicitly verify these conditions for the logistic regression model with functional response data (\ref{Def:FunctionalLogistic}) in Section \ref{Sec:Application}. We also note that the regularity conditions, in particular those related to the entropy of the nuisance parameter estimators, can be substantially relaxed if we use sample splitting, so that the nuisance parameters and parameters of interest are estimated on separate samples; see \cite{CCDDHNR16}. In addition, we examine the performance of the proposed procedures in a Monte Carlo simulation study and provide an example based on real data in Section \ref{sec: simulations}. Moreover, in the Supplementary Material, we discuss the construction of simultaneous confidence bands based on a double-selection estimator. This estimator does not require to explicitly construct the score functions satisfying the Neyman orthogonality condition but nonetheless is first-order equivalent to the estimator based on such functions.

% Furthermore we note that the theory developed here allow us to study not only the optimal score estimator (\ref{OptimalScoreEstimator}) but also different estimators that are first order equivalent but might have different finite sample performance. We also use the logistic model with functional data to illustrate one estimator based on the optimal score, and, due to space constrains, a second estimator based on the ideas of double selection is discussed in detail in the Supplementary Material.

We also develop new results for $\ell_1$-penalized $M$-estimators in Section \ref{FunctionalLassoSection} to handle functional data and criterion functions that depend on nuisance functions for which only estimates are available (for brevity of the paper, generic results are deferred to Appendix \ref{sec: generic results} of the Supplementary Material, and Section \ref{FunctionalLassoSection} only contains results that are relevant for the logistic regression model studied in Section \ref{Sec:Application}). Specifically, we develop a method to select penalty parameters for these estimators and extend the existing theory to cover functional data to achieve rates of convergence and sparsity guarantees that hold uniformly over $u\in \UU$. The ability to allow both for functional data and for nuisance functions is crucial in the implementation and in theoretical analysis of the methods proposed in this paper.

Orthogonality conditions like that in \eqref{DefOrtho} have played an important role in statistics and econometrics.  In low-dimensional settings, a similar condition was used by Neyman in \cite{N59} and \cite{Neyman1979} while in semiparametric models the orthogonality conditions were used in \cite{newey90}, \cite{andrews94}, \cite{newey94}, \cite{robins:dr} and \cite{linton96}. In high-dimensional settings, \cite{BellChernHans:Gauss} and \cite{BellChenChernHans:nonGauss} were the first to use the orthogonality condition (\ref{DefOrtho}) in a linear instrumental variables model with many instruments. Related ideas have also been used in the literature to construct confidence bands in high-dimensional linear models, generalized linear models, and other non-linear models; see \cite{BCH2011:InferenceGauss}, \cite{c.h.zhang:s.zhang}, \cite{BelloniChernozhukovHansen2011}, \cite{vandeGeerBuhlmannRitov2013}, \cite{BCW-rootLasso}, \cite{javanmard2014confidence}, \cite{javanmard2013confidence}, \cite{BCK-LAD}, \cite{BCK-SparseQRinference}, \cite{BCFH2013program}, \cite{ZhaoKolarLiuQR2014}, and \cite{NingLiuGT2014}. We contribute to this quickly growing literature by providing procedures to construct {\em simultaneous} confidence bands for {\em many infinite-dimensional} parameters identified by moment conditions.

Throughout the paper, we use the standard notation from the empirical process theory. In particular, we use $\mathbb E_n$ to denote the expectation with respect to the empirical measure associated with the data $(W_i)_{i=1}^n$, and we use $\mathbb G_n$ to denote the empirical process $\sqrt n(\mathbb E_n - \Ep_P)$. More details about the notation are given in Appendix \ref{subsec:notation} of the Supplementary Material.

 \section{Confidence Regions for Function-Valued Parameters Based on Moment Conditions}\label{sec: general}

\subsection{Generic Construction of Confidence Regions}
 In this section, we state our results under high-level conditions. In the next section, we will apply these results to construct simultaneous confidence bands for many infinite-dimensional parameters in the logistic regression model with functional response data.

Recall that we are interested in constructing simultaneous confidence bands for a set of target parameters $(\theta_{u j})_{u\in\UU, j\in[\pp]}$ where for each $u\in\UU\subset \mathbb{R}^{d_u}$ and $j\in[\pp]=\{1,\dots,\pp\}$, the parameter $\theta_{u j}$ satisfies the moment condition \eqref{Def:Moment1} with $\eta_{u j}$ being a potentially high-dimensional (or infinite-dimensional) nuisance parameter. Assume that $\theta_{\uu j} \in \Theta_{\uu j}$, a finite or infinite interval in $\mathbb{R}$, and that $\eta_{uj} \in T_{uj}$, a convex set in a normed space equipped with a norm $\|\cdot\|_e$. We allow $\UU$ to be a possibly uncountable set of indices, and $\pp$ to be potentially large.

%We assume that for each $u \in \mathcal{U}$ and $j\in [\pp]$, the parameter $\theta_{uj}$ satisfies the moment condition
% \begin{equation}\label{eq:ivequation}
% \Ep_P[ \psi_{\uu j}(W, \theta_{uj}, \eta_{uj} )] = 0,
% \end{equation}
%where   $W$ is a random element that takes values in a measurable space $(\mathcal{W}, \mathcal{A}_\mathcal{W})$, with law determined by a probability measure $P \in \mP_n$, $\eta_{uj}$ is a nuisance parameter with $\eta_{uj} \in T_{uj}$, a convex set equipped with a norm $\|\cdot\|_e$, and the score function
%$
%\psi_{\uu j}\colon  \mathcal{W} \times \Theta_{uj} \times T_{uj} \to \mathbb{R}
%$
% is a measurable map (where we equip $\Theta_{u j}$ and $T_{u j}$ with their Borel $\sigma$-fields).   Here $\mP_n$ is some set of probability measures on $(\mathcal{W}, \mathcal{A}_\mathcal{W})$.  Note that our formulation allows the nuisance parameters $\eta_{uj}$ to be infinite-dimensional (i.e. functions).

We assume that a random sample $(W_i)_{i=1}^n$ from the distribution of $W$ is available for constructing the confidence bands. We also assume that for each $u\in\UU$ and $j\in[\pp]$, the nuisance parameter $\eta_{u j}$ can be estimated by $\hat\eta_{uj}$ using the same data $(W_i)_{i=1}^n$. In the next section, we discuss examples where $\widehat \eta_{u j}$'s are based on Lasso or Post-Lasso methods (although other modern regularization and post-regularization methods can be applied). Our confidence bands will be based on the estimators $\check \theta_{\uu j}$ of $\theta_{\uu j}$ that are for each $u\in\UU$ and $j\in[\pp]$ defined as approximate $\epsilon_n$-solutions in $\Theta_{uj}$ to sample analogs of the moment conditions (\ref{Def:Moment1}), that is,
\begin{equation}\label{eq:analog}
\sup_{\uu \in \UU, j\in[\pp]}  \left\{ \Big|\En[ \psi_{uj}(W, \check \theta_{uj}, \hat \eta_{uj} ) ] \Big| - \inf_{\theta \in \Theta_{uj}}\Big|\En[ \psi_{\uu j}(W, \theta, \hat \eta_{uj} ) ] \Big| \right\} \leq \epsilon_n,
\end{equation}
where $\epsilon_n = o(\delta_n n^{-1/2})$ for all $n\geq 1$ and some sequence $(\delta_n)_{n\geq 1}$ of positive constants converging to zero.

To motivate the construction of the confidence bands based on the estimators $\check \theta_{\uu j}$, we first study distributional properties of these estimators. To do that, we will employ the following regularity conditions. Let $C_0$ be a strictly positive (and finite) constant, and for each $u\in\UU$ and $j\in[\pp]$, let $\mT_{u j}$ be some subset of $T_{u j}$, whose properties are specified below in assumptions. In particular, we will choose the sets $\mT_{u j}$ so that, on the one hand, their complexity does not grow too fast with $n$ but, on the other hand,  for each $u\in\UU$ and $j\in [\pp]$, the estimator $\hat \eta_{u j}$ takes values in $\mT_{u j}$ with high probability. As discussed before, we rely on the following near-orthogonality condition:
\begin{definition}[Near-orthogonality condition] For each $u \in \mathcal{U}$ and $j\in[\pp]$, we say that $\psi_{uj}$ obeys the near-orthogonality condition with respect to $\mT_{uj}\subset T_{u j}$ if the following conditions hold: The Gateaux derivative map
 $$
  \mathrm{D}_{\uu,j,\bar r}[ \eta - \eta_{uj}]:=  \left. \partial_r  \bigg\{\Ep_P \Big [  \psi_{\uu j} (W, \theta_{\uu j}, \eta_{uj}+ r ( \eta - \eta_{uj}))   \Big ]\bigg\}\right|_{r=\bar r}
  $$
  exists for all $\bar r \in [0,1)$ and $ \eta \in \mT_{uj}$ and (nearly) vanishes at $\bar r=0$, namely,
 \begin{equation}\label{eq:cont}
 \Big|\mathrm{D}_{\uu,j,0}[\eta - \eta_{uj}]\Big|\leq C_0\delta_n n^{-1/2},
\end{equation}
for all $\eta \in \mT_{u j}$.\qed
\end{definition}

%If the original score functions $m_{uj}$ do not satisfy this near orthogonality condition, we have to transform them into score functions $\psi_{u j}$ that satisfy this condition.
At the end of this section, we describe several methods to obtain score functions $\psi_{u j}$ that obey the near-orthogonality condition. Together these methods cover a wide variety of applications.

%$$
%\gamma = \left(\right)^{-1}\frac{\partial}{\partial\beta}%\Ep_P\left[m_{u j}(W,\theta_{u j},\beta)\frac{\partial}{\partial\beta'}\Ep_P[m_{u j}(W,\theta_{u j},\beta)]\Big_{\beta = \beta_{u j}}\right]\bigg_{\beta = \beta_{u j}}
%$$
%\qed
% If we  have a moment function $m_{uj}$ as in (\ref{Def:Moment1}) that identifies  $\theta_{uj}$ but does not obey this near orthogonality condition,  we can construct a moment function $\psi_{uj}$ that still identifies $\theta_{uj}$ but satisfies this condition by projecting the original moment function $m_{uj}$ onto the orthocomplement of the tangent space induced by the nuisance parameter $\eta_{uj}$; see, for example, \cite{vdV-W}, \cite[Chap. 25]{vdV}, \cite{kosorok:book}, \cite{BelloniChernozhukovHansen2011}, and \cite{BCK-LAD}. \qed
%\end{remark}
Let $\omega$ and $c_0$ be some strictly positive (and finite) constants, and let $n_0\geq 3$ be some positive integer. Also, let $(B_{1n})_{n\geq 1}$ and $(B_{2n})_{n\geq 1}$ be some sequences of positive constants, possibly growing to infinity,  where $B_{1 n}\geq 1$ for all $n\geq 1$. In addition, denote
\begin{equation}\label{eq: un and juj}
\begin{array}{cc}
&\mS_n:= \Ep_P\left[\sup_{\uu \in \UU, j\in [\pp]}\Big|\sqrt{n}\En[\psi_{uj}(W, \theta_{uj}, \eta_{uj} )]\Big|\right],\\
& J_{\uu j} :=  \left.\partial_\theta\Big\{ \Ep_P[ \psi_{\uu j} (W, \theta, \eta_{uj})]\Big\}\right|_{\theta=\theta_{uj}}.
\end{array}
\end{equation}
The quantity $\mS_n$ measures how rich the process $\{\psi_{u j}(\cdot,\theta_{u j},\eta_{u j})\colon u\in\UU, j\in[\pp]\}$ is. %In many applications, it satisfies $\mS_n \leq C(1+d_u + \log \pp)^{1/2}$ for some constant $C$.
The quantity $J_{u j}$ measures the degree of identifiability of $\theta_{u j}$ by the moment condition \eqref{Def:Moment1}. In many applications, it is bounded in absolute value from above and away from zero. Finally, let $\mathcal P_n$ be a set of probability measures $P$ of possible distributions of $W$ on the measurable space $(\mathcal{W},\mathcal{A}_{\mathcal W})$.

We collect our main conditions on the score functions $\psi_{u j}$ and the true values of the target parameters $\theta_{u j}$ in the following assumption.

\begin{assumption}[Moment condition problem]\label{ass: S1}
For all $n \geq n_0$, $P \in \mathcal{P}_n$, $u\in\UU$, and $j\in[\pp]$, the following conditions hold: (i) The true parameter value $\theta_{\uu j}$ obeys (\ref{Def:Moment1}), and $\Theta_{\uu j}$ contains a ball of radius $C_0 n^{-1/2}\mS_n \log n $ centered at $\theta_{\uu j}$.   (ii) The map $ (\theta,\eta)  \mapsto \Ep_P[\psi_{uj}(W, \theta,\eta )]$ is twice continuously Gateaux-differentiable on $\Theta_{uj} \times \mT_{uj}$. (iii) The score function $\psi_{uj}$ obeys the near-orthogonality condition given in Definition 2.1 for the set $\mT_{uj}\subset T_{uj}$. (iv) For all $\theta \in \Theta_{uj}$, $|\Ep_P[\psi_{\uu j}(W, \theta, \eta_{uj})]| \geq 2^{-1}|J_{uj} (\theta- \theta_{u j})| \wedge c_0$, where $J_{\uu j}$ satisfies $c_0\leq |J_{uj}| \leq C_0 $. (v) For all $r\in[0,1)$, $\theta\in\Theta_{u j}$, and $\eta \in \mT_{uj}$,
%{\em (a)} $\Ep_P[  \{ \psi_{uj}(W, \theta,\eta) - \psi_{uj}(W, \bar\theta,\bar\eta)\}^2]\leq C_0  ( |\theta-\bar\theta|\vee \| \eta - \bar \eta\|_e)^{\alpha}$,\\
%{\em (b)} $ \sup_{(\theta, \eta) \in \Theta_{uj}\times \mT_{ujn}}   | \partial_{r} \Ep_P \left [ \psi_{uj}(W, \theta,\eta_{uj}+r(\eta-\eta_{uj})) \right ]| \leq \bar B_{1n} \|\eta-\eta_{uj}\|_e $,\\
%{\em (c)} $\sup_{(\theta,\eta) \in \Theta_{uj}\times \mT_{ujn}} |\partial_{r}^2 \Ep_P[\psi_{uj}(W, \theta_{uj}+r(\theta-\theta_{uj}),\eta_{uj}+r(\eta - \eta_{uj}))]| \leq  \bar B_{2n} (|\theta-\theta_{uj}|^2 \vee \|\eta-\eta_{uj}\|_{e}^2)$.

\begin{itemize}
\item[(a)]
$\Ep_P[  ( \psi_{uj}(W, \theta,\eta) - \psi_{uj}(W, \theta_{u j},\eta_{u j}))^2]\leq C_0  ( |\theta-\theta_{u j}|\vee \| \eta - \eta_{u j}\|_e)^{\omega}$,

\item [(b)]
$| \partial_{r} \Ep_P \left [ \psi_{uj}(W, \theta,\eta_{uj}+r(\eta-\eta_{uj})) \right ]| \leq B_{1n} \|\eta-\eta_{uj}\|_e $,

\item[(c)]
$ |\partial_{r}^2 \Ep_P[\psi_{uj}(W, \theta_{uj}+r(\theta-\theta_{uj}),\eta_{uj}+r(\eta - \eta_{uj}))]| \leq   B_{2n} (|\theta-\theta_{uj}|^2 \vee \|\eta-\eta_{uj}\|_{e}^2)$.
\end{itemize}

%\begin{align*}
%&\sup_{(\theta, \bar \theta) \in \Theta_{uj}^2, (\eta,\bar\eta) \in \mT_{ujn}^2} \frac{\Ep_P[  \{ \psi_{uj}(W, \theta,\eta) - \psi_{uj}(W, \bar\theta,\bar\eta)\}^2]}{\{ |\theta-\bar\theta|\vee \| \eta - \bar \eta\|_e\}^{\alpha}}\leq C_0,\\
%&\sup_{(\theta, \eta) \in \Theta_{uj}\times \mT_{ujn}, r\in[0,1)} \frac{ | \partial_{r} \Ep_P \left [ \psi_{uj}(W, \theta,\eta_{uj}+r(\eta-\eta_{uj})) \right ]|}{\|\eta-\eta_{uj}\|_e} \leq \bar B_{1n},\\
%&\sup_{(\theta,\eta) \in \Theta_{uj}\times \mT_{ujn}, r\in[0,1)} \frac{|\partial_{r}^2 \Ep_P[\psi_{uj}(W, \theta_{uj}+r(\theta-\theta_{uj}),\eta_{uj}+r(\eta - \eta_{uj}))]|}{|\theta-\theta_{uj}|^2 \vee \|\eta-\eta_{uj}\|_{e}^2} \leq  \bar B_{2n}.
%\end{align*}
\end{assumption}

Assumption \ref{ass: S1} is mild and standard in moment condition problems. Assumption \ref{ass: S1}(i) requires $\theta_{u j}$ to be sufficiently separated from the boundary of $\Theta_{u j}$. Assumption \ref{ass: S1}(ii) requires that the functions $(\theta,\eta)\mapsto \Ep_P[\psi_{u j}(W,\theta,\eta)]$ are smooth. It is a mild condition because it does not require smoothness of the functions $(\theta,\eta)\mapsto \psi_{u j}(W,\theta,\eta)$. Assumption \ref{ass: S1}(iii) is our key condition and is discussed above. Assumption \ref{ass: S1}(iv) implies sufficient identifiability of $\theta_{u j}$.  In particular, it implies that the equation $\Ep_P[\psi_{u j}(W,\theta,\eta_{u j})] = 0$  has only one solution $\theta = \theta_{u j}$. If this equation has multiple solutions, Assumption \ref{ass: S1}(iv) implies that the set $\Theta_{u j}$ is restricted enough so that there is only one solution in $\Theta_{u j}$. Assumption \ref{ass: S1}(v-a) means that the functions $(\theta,\eta)\mapsto \psi_{u j}(W, \theta,\eta)$ mapping $\Theta_{u j}\times T_{u j}$ into $L^2(P)$ are Lipschitz-continuous at $(\theta,\eta) = (\theta_{u j},\eta_{u j})$ with Lipschitz order $\omega/2$. In most applications, we can set $\omega = 2$. Assumptions \ref{ass: S1}(v-b,v-c) impose smoothness bounds on the functions $(\theta,\eta)\mapsto \Ep_P[\psi_{u j}(W,\theta,\eta)]$. %The suitably measurable condition, defined in Appendix \ref{subsec:notation}, is a mild condition satisfied in most practical cases.

Next, we state our conditions related to the estimators $\hat{\eta}_{u j}$. Let $(\Delta_n)_{n\geq 1}$ and $(\tau_n)_{n\geq 1}$ be some sequences of positive constants converging to zero. Also, let $(a_n)_{n\geq 1}$, $(v_n)_{n\geq 1}$, and $(K_n)_{n\geq 1}$ be some sequences of positive constants, possibly growing to infinity, where $a_n\geq n\vee K_n$ and $v_n\geq 1$ for all $n\geq 1$. Finally, let $q\geq 2$ be some constant.

\begin{assumption}[Estimation of nuisance parameters]\label{ass: AS}
For all $n \geq n_0$ and $P \in  \mathcal{P}_n$, the following conditions hold: (i) With probability at least $1- \Delta_n$, we have $\hat \eta_{uj} \in \mT_{uj}$ for all $u\in\UU$ and $j\in[\pp]$. (ii) For all $u\in\UU$, $j\in[\pp]$, and $\eta\in \mT_{u j}$,
$
\| \eta - \eta_{uj}\|_{e} \leq \tau_n.
$
(iii) For all $u\in\UU$ and $j\in[\pp]$, we have $\eta_{u j}\in \mT_{u j}$. (iv) The function class
$\mathcal{F}_1 = \{  \psi_{uj}(\cdot, \theta, \eta)\colon  u \in \mathcal{U}, j  \in [\pp],  \theta \in \Theta_{uj}, \eta \in \mT_{uj} \}$
is suitably measurable and its uniform entropy numbers obey
\begin{equation}\label{eq: F1 entropy bound}
\sup_Q  \log N(\epsilon \|F_1\|_{Q,2}, \mathcal{F}_1, \| \cdot \|_{Q,2}) \leq v_{n}  \log (a_n/\epsilon), \quad \text{for all } 0<\epsilon\leq 1
\end{equation}
where $F_1$ is a measurable envelope for $\mathcal{F}_1$ that satisfies $\|F_1\|_{P,q}\leq K_{n}$. (v) For all $f\in\mathcal{F}_1$, we have $c_0\leq \|f\|_{P,2}\leq C_0$. (vi) The complexity characteristics $a_n$ and $v_{n}$ satisfy
\begin{itemize}
\item[(a)]
$(v_{n} \log a_n/n)^{1/2} \leq C_0\tau_n$,

\item[(b)]
$(B_{1n}\tau_n + \mS_n\log n/\sqrt n)^{\omega/2} (v_{n} \log a_n)^{1/2}  + n^{-1/2+1/q} v_{n} K_n\log a_n \leq C_0\delta_n$,

\item[(c)]
$n^{1/2} B_{1n}^2 B_{2n}  \tau_n^2 \leq C_0\delta_n$.
\end{itemize}
%where $\bar B_{1n}$, $\bar B_{2n}$, $q$ and $\alpha$ are defined in Assumption \ref{ass: S1}.
\end{assumption}

Assumption \ref{ass: AS} provides sufficient conditions for the estimation of the nuisance parameters $(\eta_{u j})_{u\in\UU,j\in[\pp]}$. Assumption \ref{ass: AS}(i) requires that the set $\mT_{u j}$ is large enough so that $\hat\eta_{u j}\in\mT_{u j}$ with high probability. Assumptions \ref{ass: AS}(i,ii) together require that the estimator $\hat\eta_{u j}$ converges to $\eta_{u j}$ with the rate $\tau_n$. This rate should be fast enough so that Assumptions \ref{ass: AS}(vi-b,vi-c) are satisfied. Assumption \ref{ass: AS}(iv) gives a bound on the complexity of the set $\mT_{u j}$ expressed via uniform entropy numbers, and Assumptions \ref{ass: AS}(vi-a,vi-b) require that the set $\mT_{u j}$ is small enough so that its complexity does not grow too fast. Assumption \ref{ass: AS}(v) requires that the functions $(\theta,\eta)\mapsto \psi_{u j}(W,\theta,\eta)$ are scaled properly. Suitable measurability of $\mathcal{F}_1$, required in Assumption \ref{ass: AS}(iv), is a mild condition that is satisfied in most practical cases; see Supplementary Material \ref{subsec:notation} for clarifications. Overall, Assumption \ref{ass: AS} shows the trade-off in the choice of the sets $\mathcal T_{u j}$: setting $\mathcal T_{u j}$ large, on the one hand, makes it easy to satisfy Assumption \ref{ass: AS}(i) but, on the other hand, yields large values of $a_n$ and $v_n$ in Assumption \ref{ass: AS}(iv) making it difficult to satisfy Assumption \ref{ass: AS}(vi).

%Suitable measurability of $\mathcal{F}_1$, required in Assumption \ref{ass: AS}(iv), is a mild condition that is satisfied in most practical cases; see Appendix \ref{subsec:notation} for clarifications. The index $v_{n}$ captures the complexity of the class of functions $\mathcal{F}_1$ and typically grows with $d_u$, $\pp$, and $\mathcal {T}_{u j}$'s. In particular, in the case of approximately sparse models, like a logistic regression model with functional response data studied in Section \ref{Sec:Application}, $v_{n}$ can be typically set up to a constant as the sum of the dimension of the approximating model, the dimension of the selected model, and the dimension of $\UU$. However, we note that our conditions potentially cover other frameworks, where assumptions other than approximate sparsity are used to make the estimation problem  manageable.

We stress that the class $\mathcal{F}_1$ does not need to be Donsker because its uniform entropy numbers are allowed to increase with $n$. This is important because allowing for non-Donsker classes is necessary to deal with high-dimensional nuisance parameters. Note also that our conditions are very different from the conditions imposed in various settings with nonparametrically estimated nuisance functions; see, e.g., \cite{vdV-W}, \cite{vdV}, and \cite{kosorok:book}.

In addition, we emphasize that the conditions stated in Assumption \ref{ass: AS} are sufficient for our results for the general model \eqref{Def:Moment1} but can often be relaxed if the structure of the functions $\psi_{u j}(W, \theta, \eta)$ is known. For example, it is possible to relax Assumption \ref{ass: AS}(vi) if the functions $\psi_{u j}(W, \theta, \eta)$ are linear in $\theta$, which happens in the linear regression model with $\theta$ being the coefficient on the covariate of interest; see \cite{BelloniChernozhukovHansen2011}. Moreover, it is possible to relax the entropy condition \eqref{eq: F1 entropy bound} of Assumption \ref{ass: AS} by relying upon sample splitting, where a part of the data is used to estimate $\eta_{u j}$, and the other part is used to estimate $\theta_{u j}$ given an estimate $\hat\eta_{u j}$ of $\eta_{u j}$; see \cite{BellChenChernHans:nonGauss}.

The following theorem is our first main result in this paper:

\begin{theorem}[{Uniform Bahadur representation}] \label{theorem:semiparametric} Under  Assumptions \ref{ass: S1} and \ref{ass: AS}, for an estimator $(\check \theta_{uj})_{u \in \mathcal{U},j\in[\pp]}$ that obeys (\ref{eq:analog}), we have
\begin{equation}\label{eq: bahadur representation main}
\sqrt{n}\sigma_{uj}^{-1}(\check\theta_{uj} - \theta_{uj}) =   \Gn    \bar \psi_{uj}   + O_P(\delta_n) \text{ in } \ell^\infty(\mathcal{U}\times[\pp])
\end{equation}
uniformly over $P \in  \mathcal{P}_n$, where  $\bar \psi_{uj}(\cdot):= - \sigma_{uj}^{-1}J^{-1}_{uj}  \psi_{uj}(\cdot, \theta_{uj}, \eta_{uj})$ and  $\sigma_{uj}^2 :=J^{-2}_{uj} \Ep_P[  \psi_{uj}^2(W, \theta_{uj}, \eta_{uj})]$.
\end{theorem}
\begin{remark}[On the proof of Theorem \ref{theorem:semiparametric}]
To prove this theorem, we use the following identity:
{\small
\begin{align}
&\sqrt n \Ep_{P,W}[\psi_{u j}(W,\check\theta_{u j},\hat\eta_{u j}) - \psi_{u j}(W,\theta_{u j},\eta_{u j})]  = - \sqrt n\En[\psi_{u j}(W,\theta_{u j},\eta_{u j})] \label{eq: identity line 1}\\
&\qquad + \sqrt n\En[\psi_{u j}(W,\check\theta_{u j},\hat\eta_{u j})] + \mathbb G_n \psi_{u j}(W,\theta_{u j},\eta_{u j}) - \mathbb G_n \psi_{u j}(W,\check\theta_{u j},\hat\eta_{u j}).\label{eq: identity line 2}
\end{align}}\!Here, the term on the right-hand side of \eqref{eq: identity line 1} is the main term on the right-hand side of \eqref{eq: bahadur representation main}, up to a normalization $(\sigma_{u j} J_{u j})^{-1}$. Also, we show that the first term in \eqref{eq: identity line 2} is $O_P(\delta_n)$ since $\check \theta_{u j}$ satisfies \eqref{eq:analog}. Moreover, using a rather standard theory of $Z$-estimators, we show that $\check\theta_{u j} - \theta_{u j} = O_P(B_{1n}\tau_n)$. This in turn allows us to show with the help of empirical process arguments that the difference of the last two terms in \eqref{eq: identity line 2} is $O_P(\delta_n)$ as well. (In \cite{CCDDHNR16}, we also point out that this difference is $O_P(\delta_n)$ under much weaker entropy conditions than those in Assumption \ref{ass: AS} if $\hat\eta_{u j}$ and $\check\theta_{u j}$ are obtained using separate samples.) Thus, it remains to show that the left-hand side of \eqref{eq: identity line 1} is equal to the left-hand side of \eqref{eq: bahadur representation main} up to an approximation error $O_P(\delta_n)$ and up to a normalization $(\sigma_{u j} J_{u j})^{-1}$. To do so, we use second-order Taylor's expansion of the function
$$
f(r) = \sqrt n \Ep_{P,W}[\psi_{u j}(W,\theta_{u j} + r(\check\theta_{u j} - \theta_{u j}),\eta_{u j} + r(\hat\eta_{u j} - \eta_{u j}))]
$$
at $r = 1$ around $r = 0$. This gives
\begin{align*}
&\sqrt n \Ep_{P,W}[\psi_{u j}(W,\check\theta_{u j},\hat\eta_{u j}) - \psi_{u j}(W,\theta_{u j},\eta_{u j})] = f(1) - f(0)\\
&\qquad = \sqrt n J_{u j}(\check \theta_{u j} - \theta_{u j}) + \sqrt nD_{u,j,0}[\hat\eta_{u j} - \eta_{u j}] + \sqrt n f''(\bar r)/2
\end{align*}
for some $\bar r \in (0,1)$. Here, $\sqrt nf''(\bar r) = O_P(\delta_n)$ follows from Assumptions \ref{ass: S1} and \ref{ass: AS} and the key near-orthogonality condition also allows us to show that $\sqrt n D_{u,j,0}[\hat\eta_{u j} - \eta_{u j}] = O_P(\delta_n)$. Without this condition, the term $\sqrt n D_{u ,j,0}[\hat\eta_{u j} - \eta_{u j}]$ would give first-order bias and lead to slower-than-$\sqrt n$ rate of convergence of the estimator $\check\theta_{u j}$. Finally, again using the empirical process arguments, we can show that all the bounds including the term $O_P(\delta_n)$ hold uniformly over $u\in\UU$ and $j\in[\pp]$.
\qed
\end{remark}

\begin{remark}[On uniformity in $u$ in Theorem \ref{theorem:semiparametric}]
When the functions $u\mapsto \sqrt n \sigma_{u j}^{-1}(\check\theta_{u j} - \theta_{u j}) - \mathbb G_n\bar\psi_{u j}$ are Lipschitz-continuous, one can use a simple discretization argument to conclude that the approximation in \eqref{eq: bahadur representation main} holds uniformly over $(u,j)\in\mathcal U\times[\pp]$ as long as we can show that it holds for each $(u,j)\in\mathcal U\times[\pp]$. However, in many applications, including the logistic regression model with functional response data discussed in Section \ref{Sec:Application}, this function is actually not continuous, and the location of jumps depends on the data. Therefore, we have to rely on a more complicated argument to establish uniformity in $u$ in the bound \eqref{eq: bahadur representation main}.
\qed
\end{remark}

%\begin{remark} It is important to mention here that the associated class of functions is not Donsker in many settings we are interested. The Donsker assumption was key in prior approaches dealing with continuums of moment conditions with infinite-dimensional nuisance parameters, see e.g. \cite{CH06}, \cite{EZ2013}, \cite{andrews:emp}.\qed \end{remark}

The uniform Bahadur representation derived in Theorem \ref{theorem:semiparametric} is useful for the construction of simultaneous confidence bands for $(\theta_{uj})_{u\in\UU, j\in[\pp]}$ as in (\ref{UnifCoverage}). For this purpose, we apply new high-dimensional central limit and bootstrap theorems that have been recently developed in a sequence of papers \cite{chernozhukov2013gaussian}, \cite{chernozhukov2014clt}, \cite{chernozhukov2012gaussian}, \cite{chernozhukov2012comparison}, and \cite{chernozhukov2015noncenteredprocesses}. To apply these theorems, we make use of the following regularity condition.

Let $(\bar \delta_n)_{n\geq 1}$ be a sequence of positive constants converging to zero. Also, let $(\varrho_n)_{n\geq 1}$, $(\bar \varrho_n)_{n\geq 1}$, $(A_n)_{n\geq 1}$, $(\bar A_n)_{n\geq 1}$, and $(L_n)_{n\geq 1}$ be some sequences of positive constants, possibly growing to infinity, where $\varrho_n\geq 1$, $A_n\geq n$, and $\bar A_n\geq n$ for all $n\geq 1$. In addition, from now on, we assume that $q> 4$. Denote by  $\hat \psi_{uj}(\cdot):= - \hat\sigma_{uj}^{-1}\hat J_{uj}^{-1} \psi_{uj}(\cdot, \check\theta_{uj}, \hat \eta_{uj})$ an estimator of $\bar \psi_{uj}(\cdot)$, with $\hat J_{uj}$ and $\hat\sigma_{uj}$ being suitable estimators of $J_{uj}$ and $\sigma_{uj}$.

\begin{assumption}[Additional score regularity]\label{ass: OSR}
For all $n \geq n_0$ and $P \in  \mathcal{P}_n$, the following conditions hold: (i) The function class $\mathcal{F}_0 = \{ \bar \psi_{uj}(\cdot)\colon u \in \mathcal{U},  j  \in [\pp]  \}$
is suitably measurable and its uniform entropy numbers obey
 $$
\sup_Q  \log N(\epsilon \|F_0\|_{Q,2}, \mathcal{F}_0, \| \cdot \|_{Q,2}) \leq \varrho_{n}  \log (A_n/\epsilon),\quad \text{for all }0<\epsilon\leq 1,
$$
where $F_0$ is a measurable envelope for $\mathcal{F}_0$ that satisfies $\|F_0\|_{P,q}\leq L_{n}$. (ii) For all $f\in\mathcal{F}_0$ and $k=3,4$, we have $\Ep_P[ |f(W)|^k] \leq C_0L_n^{k-2}$. (iii) The function class $\widehat{\mathcal{F}}_0 = \{  \bar \psi_{uj}(\cdot)-\hat\psi_{uj}(\cdot)\colon u \in \mathcal{U}, j  \in [\pp] \}$ satisfies with probability $1-\Delta_n$:
$
\log N(\epsilon, \widehat{\mathcal{F}}_0, \|\cdot\|_{\Pn,2}) \leq \bar \varrho_{n}  \log (\bar A_n/\epsilon)$ for all $0<\epsilon\leq 1$ and $\|f\|_{\mathbb{P}_n,2} \leq \bar\delta_n$ for all $f\in\widehat{\mathcal{F}}_0$.
\end{assumption}

%This assumption is technical, and its verification in applications is rather standard.
%Assumption \ref{ass: OSR} imposes condition on the class of functions induces by $\bar\psi_{uj}$ and on its estimators $\hat\psi_{uj}$. Typically the bound $L_n$ on the moment of the envelope is smaller than $K_n$, and in many settings $\bar\rho_n=\rho_n \lesssim \dn$, the dimension of $\UU$.

This assumption is technical, and its verification in applications is rather standard. For the Gaussian approximation result below, we actually only need the first and the second part of this assumption. The third part will be needed for establishing validity of the simultaneous confidence bands based on the multiplier bootstrap procedure. As a side note, observe that Assumption \ref{ass: OSR} allows to bound $\mathcal S_n$ defined in \eqref{eq: un and juj} and used in Assumptions \ref{ass: S1} and \ref{ass: AS}; see Appendix \ref{sec:BoundSn} of the Supplementary Material.

Next, let $(\mathcal{N}_{u j})_{u\in\mathcal{U},j\in[\tilde{p}]}$ denote a tight zero-mean Gaussian process indexed by $\UU \times [\pp]$ with covariance operator given by $\Ep_P[ \bar\psi_{uj}(W)\bar\psi_{u'j'}(W)]$ for $u,u'\in\UU$ and $j,j' \in [\pp]$. %Because of the high-dimensionality (e.g. $\pp$ can be larger than the sample size $n$), we cannot invoke results based on Donsker classes nor standard central limit theorems. Instead, we will exploit the results of Chernozhukov, Chetverikov, and Kato cited above.
We have the following corollary of Theorem \ref{theorem:semiparametric}, which is our second main result in this paper.

\begin{corollary}[Gaussian approximation]\label{cor: general CLT}
Suppose that  Assumptions \ref{ass: S1}, \ref{ass: AS}, and  \ref{ass: OSR}(i,ii) hold. In addition, suppose that the following growth conditions are satisfied: $\delta_n^2  \varrho_n\log A_n = o(1)$, $L_n^{2/7} \varrho_n \log A_n =o(n^{1/7})$, and $L_n^{2/3} \varrho_n\log A_n =o(n^{1/3-2/(3q)})$. Then
{\small
$$
\sup_{t\in \mathbb{R}} \left| \Pr_P\left(  \sup_{u\in\UU,j\in[\tilde{p}]}|\sqrt n\sigma_{uj}^{-1}(\check\theta_{uj}-\theta_{uj})| \leq t \right) - \Pr_P\left(\sup_{u\in\mathcal{U},j\in[\tilde{p}]}|\mathcal{N}_{u j}| \leq t\right) \right| = o(1)
$$}
uniformly over $P\in\mathcal{P}_n$.
\end{corollary}

%In many applications we have $\dn$ fixed, $\rho_n \lesssim 1 $, $A_n = \tilde p$, and $L_n^2 \lesssim 1 + \log \pp$ so that $\log^8\pp = o(n)$ suffices for the additional growth requirements of $\pp$ and $n$.

Based on Corollary \ref{cor: general CLT}, we are now able to construct simultaneous confidence bands for $\theta_{u j}$'s as in (\ref{UnifCoverage}). In particular, we will use the Gaussian multiplier bootstrap method employing the estimates $\hat\psi_{u j}$ of $\bar\psi_{u j}$. To describe the method, define the process
\begin{equation}\label{eq: process g}
 \widehat{\mathcal{G}} = (\widehat{\mathcal{G}}_{uj})_{u\in\UU, j\in [\pp]} = \left( \frac{1}{\sqrt{n}}\sum_{i=1}^n\xi_i\hat\psi_{uj}(W_i)\right)_{u\in\UU, j\in [\pp]}
\end{equation}
where $(\xi_i)_{i=1}^n$ are independent standard normal random variables which are independent from the data $(W_i)_{i=1}^n$. Then the multiplier bootstrap critical value
$c_\alpha$ is defined as the $(1-\alpha)$ quantile of the conditional distribution of $\sup_{u\in\mathcal{U},j\in[\tilde{p}]}|\widehat{\mathcal{G}}_{u j}|$ given the data $(W_i)_{i=1}^n$. To prove validity of this critical value for the construction of simultaneous confidence bands of the form (\ref{UnifCoverage}), we will impose the following additional assumption. Let $(\varepsilon_n)_{n\geq 1}$ be a sequence of positive constants converging to zero.
\begin{assumption}[Variation estimation]\label{ass: variance}
For all $n\geq n_0$ and $P\in\mathcal{P}_n$,
$$
\Pr_P\left(\sup_{u\in\UU,j\in[\pp]}\left|\frac{\hat{\sigma}_{u j}}{\sigma_{u j}} - 1\right|>\varepsilon_n\right)\leq \Delta_n.
$$
\end{assumption}

The following corollary establishing validity of the multiplier bootstrap critical value $c_\alpha$ for the simultaneous confidence bands construction is our third main result in this paper.

\begin{corollary}[Simultaneous confidence bands]\label{theorem: general bs}
Suppose that Assumptions \ref{ass: S1} -- \ref{ass: variance} hold. In addition, suppose that the growth conditions of Corollary \ref{cor: general CLT} hold. Finally, suppose that $\varepsilon_n  \varrho_n\log A_n = o(1)$, and $\bar{\delta}_n^2\bar{\varrho}_n\varrho_n (\log \bar{A}_n)\cdot (\log A_n) =o(1)$. Then \eqref{UnifCoverage} holds uniformly over $P\in\mathcal{P}_n$.
\end{corollary}
\begin{remark}[Confidence bands based on other bootstrap schemes]
Results in \cite{DZ17} suggest that the conditions of Corollary \ref{theorem: general bs} can be somewhat relaxed if, instead of using the Gaussian weights in the multiplier bootstrap method, we use Mammen's weights as in \cite{M93} or if we use the empirical bootstrap instead of the multiplier bootstrap. Since the results in \cite{DZ17} apply only to high-dimensional random vectors and do not apply to infinite-dimensional random processes, we leave a formal discussion of the results under these alternative bootstrap schemes to future work.\qed
\end{remark}

\subsection{Construction of score functions satisfying the orthogonality condition}\label{rem: obtaining orthogonality condition}

Here we discuss several methods for generating orthogonal scores in a wide variety of settings,
including the classical Neyman's construction. In what follows,  since the argument applies
to each $u$ and $j$, it is convenient to omit the indices $u$ and $j$ and also to use the subscript $0$ to indicate the true values of the parameters.   For simplicity we also focus
the discussion on the exactly orthogonal case. With these simplifications, we can restate the orthogonality condition as follows:  we say that the score $\psi$ obeys the  Neyman orthogonality condition with respect to $\eta_0 \in \mT$ if the following conditions hold: The Gateaux derivative map
 $$
  \mathrm{D}_{\bar r}[ \eta - \eta_{0}]:=  \left. \partial_r  \bigg\{\Ep_P \Big [  \psi (W, \theta_0, \eta_0+ r ( \eta - \eta_0)   \Big ]\bigg\}\right|_{r=\bar r}
  $$
  exists for all $\bar r \in [0,1)$ and $ \eta \in \mT$ and vanishes at $\bar r=0$, namely,
 \begin{equation}\label{eq:cont2}
\mathrm{D}_{\uu,j,0}[\eta - \eta_{0}] = 0,
\end{equation}
for all $\eta \in \mT$.

\subsection*{1) Orthogonal Scores for Likelihood Problems
with Finite-Dimensional Nuisance Parameters}
 In likelihood settings with finite-dimensional parameters, the construction of orthogonal equations was proposed by Neyman \cite{N59} who used them in construction of his celebrated $C(\alpha)$-statistic.\footnote{Note that the $C(\alpha)$-statistic, or the orthogonal score statistic, had been explicitly used for testing (and also for setting up estimation) in high-dimensional sparse models in \cite{BCK-LAD} and in \cite{NingLiuGT2014}, where it is referred to as the decorrelated score statistic. The discussion of Neyman's construction here draws on \cite{CHS15}. Note also that our results cover other types of orthogonal score statistics as well, which allows us to cover much broader classes of models; see for example the discussion of conditional moment models with infinite-dimensional nuisance parameters below.}

To describe the construction, suppose that the log-likelihood function associated to observation
$W$ is $(\theta,\beta)\mapsto \ell(W, \theta, \beta)$, where $\theta \in \Theta \subset \mathbb{R}^{d}$
is the target parameter and $\beta \in T \subset \mathbb{R}^{p_0}$ is the nuisance parameter. Under regularity conditions, the true parameter values $\theta_0$ and $\beta_0$ obey
\begin{equation}\label{eq: foc lik}\Ep_P [\partial_\theta \ell (W, \theta_0, \beta_0)]=0, \quad \Ep_P[\partial_{\beta} \ell (W, \theta_0, \beta_0)] =0.
\end{equation}
Now consider the new score function
\begin{equation}
\psi(W, \theta, \eta) =  \partial_\theta \ell (W, \theta, \beta) -   \mu \partial_{\beta} \ell (W, \theta, \beta),
\end{equation}
where the nuisance parameter is $$  \eta= (\beta', \textrm{vec}(\mu)')'  \in T \times \mathcal{D} \subset \mathbb{R}^{p}, \quad p=p_0 + dp_0,$$
$\mu$ is the $d \times p_0$ \textit{orthogonalization} parameter matrix whose
true value  $\mu_0$ solves the equation
$$
J_{\theta\beta}  - \mu J_{\beta \beta} =0  \ (\text{i.e., }  \mu_0 =  J_{\theta\beta}J^{-1}_{\beta \beta}),
$$
and
$$
J =
\left(
                                                                                   \begin{array}{cc}
                                                                                     J_{\theta \theta} & J_{\theta \beta} \\
                                                                                     J_{\beta \theta} &  J_{\beta \beta} \\
                                                                                   \end{array}
                                                                                 \right)
 = \partial_{(\theta',\beta')}\Ep_P\Big[\partial_{(\theta',\beta')'}\ell(W,\theta,\beta)\Big]\Big\vert_{\theta = \theta_0; \ \beta = \beta_0}.
$$
Provided that $\mu_0$ is well-defined, we have by (\ref{eq: foc lik}) that
$
\Ep_P [\psi(W, \theta_0, \eta_0)] = 0,
$
where $\eta_0 = (\beta_0',\textrm{vec}(\mu_0)')'$.
Moreover, it is trivial to verify that under standard regularity conditions the score function $\psi$ obeys the near orthogonality condition (\ref{eq:cont}) exactly (i.e., with $C_0 = 0$), i.e.
$$
\partial_{\eta} \Ep_P [\psi(W, \theta_0, \eta)] \Big \vert_{\eta  = \eta_0} =0.
$$
Note that in this example, $\mu_0$ not only creates the necessary orthogonality but also creates
the \textit{efficient score} for inference on the main parameter $\theta$, as emphasized by Neyman.

\subsection*{2) Orthogonal Scores for Likelihood Problems
with Infinite-Dimensional Nuisance Parameters}
The Neyman's construction can be extended to semi-parametric models, where the nuisance parameter $\beta$ is a function. In this case, the original score functions $(\theta,\beta)\mapsto \partial_\theta\ell(W,\theta,\beta)$ corresponding to the log-likelihood function $(\theta,\beta)\mapsto \ell(W,\theta,\beta)$ associated to observation $W$ can be transformed into efficient score functions $\psi$ that obey the exact orthogonality condition (\ref{eq:cont2})  by projecting the original score functions onto the orthocomplement of the tangent space induced by the nuisance parameter $\beta$; see Chapter 25 of \cite{vdV} for a detailed description of this construction.   Note that the projection may create additional nuisance parameters, so that the new nuisance parameter $\eta$
could be of larger dimension than $\beta$. Other relevant references include \cite{vdV-W}, \cite{kosorok:book}, \cite{BelloniChernozhukovHansen2011}, and \cite{BCK-LAD}.  The approach is related to Neyman's construction in the sense that the score $\psi$ arising in this model is actually the Neyman's score arising in a one-dimensional least favorable parametric subfamily, \cite{S56}; see Chapter 25 of \cite{vdV} for details.

\subsection*{3) Orthogonal Scores for Conditional Moment Problems
with Infinite-Dimensional Nuisance Parameters}
Next, consider a conditional moment restrictions framework studied by Chamberlain \cite{C92}. To define the framework, let $W$, $D$, and $V$ be random vectors in $\mathbb R^{d_W}$, $\mathbb R^{d_D}$, and $\mathbb R^{d_V}$, respectively, with $D$ and $V$ being sub-vectors of $W$, so that $d_D + d_V \leq d_W$. Also, let $\theta\in\Theta\subset\mathbb R^{d_\theta}$ be a finite-dimensional parameter whose true value $\theta_0$ is of interest, and let $h\colon \mathbb R^{d_{V}}\to\mathbb R^{d_h}$ be a vector-valued functional nuisance parameter, with the true value being $h_0\colon \mathbb R^{d_V}\to\mathbb R^{d_h}$. The conditional moment restrictions framework assumes that $\theta_0$ and $h_0$ satisfy the following equation:
\begin{equation}\label{eq: conditional moment restrictions}
\Ep_P [ m(W, \theta_{0}, h_{0}(V)) \mid D,V] = 0,
\end{equation}
where $m\colon \mathbb R^{d_W} \times \mathbb R^{d_\theta} \times \mathbb R^{d_h}\to \mathbb R^{d_m}$ is some known function. This framework is of interest because it covers an extremely rich variety of models, without having to explicitly rely on the likelihood formulation. For example, it covers the partial linear model,
\begin{equation}\label{eq: partial linear model}
Y = D\theta_0 + h_0(V) + U,\quad \Ep_P[U\mid D,V] = 0,
\end{equation}
where $Y$ is a scalar dependent random variable, $D$ is a scalar independent treatment random variable, $V$ is a vector of control random variables, and $U$ is a scalar unobservable noise random variable. Indeed, \eqref{eq: partial linear model} implies \eqref{eq: conditional moment restrictions} by setting $W = (Y,D,V')'$ and $m(W,\theta,h) = Y - D\theta - h$.

Here we would like to build a (generalized) score function $(\theta,\eta)\mapsto \psi (W, \theta,  \eta)$ for estimating $\theta_{0}$, the true value of parameter $\theta$, where $\eta$ is a new nuisance parameter with true value $\eta_{0}$, that obeys the near orthogonality condition \eqref{eq:cont}. There are many ways to do so but one particularly useful way is the following. Consider the functional parameters $\Sigma\colon \mathbb R^{d_D + d_V}\to \mathbb R^{d_m\times d_m}$ and $\varphi\colon \mathbb R^{d_D + d_V}\to \mathbb R^{d_\theta\times d_m}$ whose true values are given by
$$
\Sigma_0(D,V) = \Ep_P[m(W,\theta_0,h_0(V))m(W,\theta_0,h_0(V))'\mid D,V],
$$
$$
\varphi_0(D,V) = \Big(A_0(D,V) - \Gamma_0(D,V)G_0(V)\Big)',
$$
where
\begin{align*}
A_0(D,V) &= \partial_{\theta'} \Ep_P[m(W,\theta,h_0(V))\mid D,V]\Big|_{\theta = \theta_0},\\
\Gamma_0(D,V) &= \partial_{h'}\Ep_P[m(W,\theta_0,h)\mid D,V]\Big|_{h = h_0(V)},\\
G_0(V) &= \Big(\Ep_P[\Gamma_0(D,V)'\Sigma_0(D,V)^{-1}\Gamma_0(D,V)\mid V]\Big)^{-1}\\
&\qquad \times  \Ep_P[\Gamma_0(D,V)'\Sigma_0(D,V)^{-1}A_0(D,V)\mid V].
\end{align*}
Then set
$
\eta = (h,\varphi,\Sigma)$ and $\eta_0 = (h_0,\varphi_0,\Sigma_0),
$
and define the score function
$$
\psi (W, \theta, \eta) = \underbracket[.5pt][.5pt]{\varphi(D,V)}_{\text{``instrument"}}      \underbracket[.5pt][.5pt]{ \Sigma(D,V)^{-1}}_{\text{weight}} \underbracket[.5pt][.5pt]{m(W, \theta, h(V))}_{\text{residual}}.
$$
It is rather straightforward to verify that under mild regularity conditions, the score function $\psi$ satisfies the moment condition,
$
\Ep_P[\psi(W,\theta_0,\eta_0)] = 0,
$
and in addition, the orthogonality condition,
$$
\partial_{\eta}\Ep_P[\psi(W,\theta_0,\eta)]\Big|_{\eta = \eta_0} = 0.
$$
Note that this construction gives the efficient score function $\psi$ that yields an estimator of $\theta_0$ achieving the semi-parametric efficiency bound, as calculated by Chamberlain \cite{C92}.

\section{Application to Logistic Regression Model with Functional Response Data}\label{Sec:Application}

In this section we apply our main results to a logistic regression model with functional response data described in the Introduction.

\subsection{Model}

We consider a response variable $Y\in\RR$ that induces a functional response $(Y_{u})_{u\in \UU}$ by $Y_u = 1\{Y\leq (1-u)\underline{y} + u\bar y\}$ for a set of indices $\UU=[0,1]$ and some constants $\underline{y} \leq \bar{y}$. We are interested in the dependence of this functional response on a $\pp$-vector of covariates, $D=(D_1,\ldots,D_\pp)'\in \mathbb{R}^\pp$, controlling for a $p$-vector of additional covariates $X = (X_1,\dots, X_p)'\in\mathbb{R}^p$. We allow both $\pp$ and $p$ to be (much) larger than the sample size of available data, $n$.

For each $u\in \UU$, we assume that $Y_u$ satisfies the generalized linear model with the logistic link function
\begin{equation}\label{Eq:MainLogisticModel}
\Ep_P[ Y_u \mid D,X] = \G(D'\theta_u + X'\beta_u)+r_u,
\end{equation}
where $\theta_u=(\theta_{u 1},\dots,\theta_{u \pp})'$ is a vector of parameters that are of interest, $\beta_u = (\beta_{u 1},\dots,\beta_{u p})'$ is a vector of nuisance parameters, $r_u = r_u(D,X)$ is an approximation error, $\Lambda\colon\mathbb{R}\to\mathbb{R}$ is the logistic link function defined by
$$
\Lambda(t)=\frac{\exp(t)}{1+\exp(t)},\quad t\in\mathbb R,
$$
and $P\in\mathcal P_n$ is the distribution of the triple $W = (Y,D,X)$. As in the previous section, we construct simultaneous confidence bands for the parameters $(\theta_{u j})_{u\in\UU,j\in[\pp]}$ based on a random sample $(W_i)_{i=1}^n = (Y_i,D_i,X_i)_{i=1}^n$ from the distribution of $W = (Y,D,X)$.

\subsection{Orthogonal Score Functions}
Before setting up score functions that satisfy both the moment conditions \eqref{Def:Moment1} and the orthogonality condition \eqref{DefOrtho}, observe that ``naive'' score functions that follow directly from the model \eqref{Eq:MainLogisticModel},
$$
m_{u j}(W,\theta_{u j}, \theta_{u[\pp]\setminus j}, \beta_u, r_u)=\Big\{Y_u-\G\Big(D_j\theta+X^j(\theta_{u[\pp]\setminus j}', \beta'_u)'\Big)-r_u\Big\}D_j,
$$
where $X^j = (D_{[\pp]\setminus j}',X')$, satisfy the moment conditions $\Ep_P[m_{uj}(W,\theta_{u j})] = 0$ but violate the orthogonality condition \eqref{DefOrtho} (with $m_{u j}$ replacing $\psi_{u j}$ and $\eta_{u j} = (\theta_{u[\pp]\setminus j},\beta_{u}, r_u)$). To satisfy the orthogonality condition \eqref{DefOrtho}, we proceed using an approach from Section \ref{rem: obtaining orthogonality condition} as in the Introduction. Specifically, for each $u\in\UU$ and $j\in[\pp]$, define the $(\pp + p - 1)$-vector of additional nuisance parameters $\gamma_u^j$ by (\ref{eq: gamma definition}) where $f_u^2 = f_u^2(D,X)$ is defined in (\ref{eq: logit weight function}). Thus, by the first order condition of (\ref{eq: gamma definition}),  the nuisance parameters $\gamma_u^j$ satisfy
\begin{equation}\label{Eq:Decomposition}
f_{u} D_j = f_{u}X^j\gamma_u^j  + v_{u}^j, \quad \Ep_P[ f_{u} X^jv_{u}^j]=0.
\end{equation}
Also, denote $\beta_u^j = (\theta_{u[\pp]\setminus j}',\beta_u')'$. Then we set
$$
\psi_{uj}(W,\theta_{u j},\eta_{uj}) = \Big\{ Y_{u} - \Lambda\Big(D_j \theta_{u j} + X^j\beta_u^j\Big) - r_u \Big\}(D_j - X^j \gamma_u^j),
$$
where the nuisance parameter is $\eta_{uj}=(r_u,\beta_u^j,\gamma_u^j)$. As we formally demonstrate in the proof of Theorem \ref{theorem:inferenceAlg1} below, this function satisfies the near-orthogonality condition \eqref{DefOrtho}.

\subsection{Estimation Using Orthogonal Score Functions}

Next, we discuss estimation of $\eta_{u j}$'s and $\theta_{u j}$'s. First, we assume that the approximation error $r_u = r_u(D,X)$ is asymptotically negligible, so that it can be estimated by $\mathcal{O} = \mathcal{O}(D,X)$, the identically zero function of $D$ and $X$. Second, for $\gamma_u^j$, we consider an estimator $\widetilde \gamma_u^j$ defined as a post-regularization weighted least squares estimator corresponding to the problem (\ref{eq: gamma definition}). Third, for $\beta_u^j$, we consider a plug-in estimator
$\hat{\beta}_u^j = (\widetilde \theta_{u[\pp]\setminus j}',\widetilde \beta_u')',$
where $\widetilde \theta_u$ and $\widetilde\beta_u$ are suitable estimators of $\theta_u$ and $\beta_u$. In particular, we assume that $\widetilde \theta_u$ and $\widetilde \beta_u$ are post-regularization maximum likelihood estimators corresponding to the log-likelihood function $(\theta,\beta)\mapsto -M_u(W,\theta,\beta)$ where
\begin{align}
M_u(W,\theta,\beta)
&= -\Big(1\{Y_u=1\}\log \G( D'\theta+X'\beta ) \notag \\
&\qquad\qquad  + 1\{Y_u=0\}\log(1-\G( D'\theta+X'\beta ))\Big).\label{eq: log-likelihood function}
\end{align}
The details of the estimators $\widetilde \theta_u$, $\widetilde\beta_u$, and $\widetilde\gamma_u^j$ are given in Algorithm 1 below. The results in this paper can also be easily extended to the case where $\widetilde \theta_u$, $\widetilde\beta_u$, and $\widetilde\gamma_u^j$ are replaced by penalized maximum likelihood estimators $\hat\theta_u$ and $\hat\beta_u$ and penalized weighted least squares estimator $\hat \gamma_u^j$, respectively.

Then our estimator of $\eta_{u j}$ is $\hat\eta_{u j} = (\mathcal{O}, \hat\beta_u^j, \widetilde\gamma_u^j)$. Substituting this estimator into the score function $\psi_{u j}$ gives
\begin{equation}\label{OSpsi}
\psi_{uj}(W,\theta_{u j},\hat\eta_{u j}) = \Big\{ Y_{u} - \Lambda\Big( D_j\theta_{u j}+X^j\hat\beta_u^j \Big)  \Big\}(D_j-X^j\widetilde\gamma_u^j),
\end{equation}
which, using the sample analog (\ref{eq:analog}) of the moment conditions \eqref{Def:Moment1}, gives the following estimator of $\theta_{u j}$:
\begin{equation}\label{OS:Step3}
\check \theta_{uj} \in \arg \inf_{\theta \in \Theta_{uj}}  \Big| \En  \Big[ \psi_{uj}(W,\theta,\hat\eta_{uj})  \Big] \Big|.
\end{equation}
The algorithm is summarized as follows.
%Next, we consider the estimators $\hat\theta_{u[\pp]\setminus j}$, $\hat\beta_u$, and $\hat\gamma_u^j$.

%In the next algorithm, we define $M$ as the negative of the log-likelihood of a linear model with a logit link function, namely
%$$
%M(Y_u,D,X,\theta,\beta) = -\Big(1\{Y_u=1\}\log \G( D'\theta+X'\beta ) + 1\{Y_u=0\}\log(1-\G( D'\theta+X'\beta ))\Big).
%$$

\begin{algorithm}
For each $u\in \UU$ and $j\in[\pp]$: \\
\enspace \emph{Step 1}. Run post-$\ell_1$-penalized logistic estimator (\ref{PostL1Logistic})  of $Y_{u}$ on  $D$ and $X$ to compute $(\widetilde\theta_u,\widetilde\beta_u)$.\\
\enspace \emph{Step 2}. Define the weights $\hat f_{u}^2 = \hat f_u^2(D,X) = \G'(D'\widetilde\theta_u+X'\widetilde\beta_u)$.\\
\enspace \emph{Step 3}. Run the post-lasso estimator (\ref{EstPostLasso}) of $\hat f_u D_j$ on $\hat f_u X^j$ to compute $\widetilde\gamma_u^j$.\\
\enspace \emph{Step 4}. Compute $\hat\beta_u^j = (\widetilde \theta_{u[\pp]\setminus j}',\widetilde \beta_u')'$.\\
\enspace \emph{Step 5}. Solve (\ref{OS:Step3}) with $\psi_{u j}(W,\theta,\hat\eta_{u j})$ defined in \eqref{OSpsi} to compute $\check \theta_{uj}$.
%\indent\indent\indent\indent Report $\check \theta_{uj}$ for $u\in \UU$ and $j\in[\pp]$.
\end{algorithm}

\subsection{Regularity Conditions}

Next, we specify our regularity conditions.
For all $u\in\UU$ and $j\in[\pp]$, denote $Z_u^j = D_j - X^j \gamma_u^j$. Also, denote $a_n = p\vee \pp\vee n$. Let $q$, $c_1$, and $C_1$ be some strictly positive (and finite) constants where $q>4$. Moreover, let $(\delta_n)_{n\geq 1}$ and $(\bar \Delta_n)_{n\geq 1}$ be some sequences of positive constants converging to zero. Finally, let $(M_{n,1})_{n\geq 1}$ and $(M_{n,2})_{n\geq 1}$ be some sequences of positive constants, possibly growing to infinity, where $M_{n,1}\geq 1$ and $M_{n,2}\geq 1$ for all $n$.

%Fix some sequences of constants, $\delta_n\to 0$,  $\ell_n \to \infty$, and $\Delta_n \to 0$, and constants $0 < c < C <\infty$. We assume $p\wedge n\wedge \pp \geq 1$.

\begin{assumption}[Parameters]\label{ass: parameters}
For all $u\in\UU$, we have $\|\theta_u\|+\|\beta_u\|+\max_{j\in [\pp]}\|\gamma_u^j\|\leq C_1$ and $\max_{j\in[\pp]}\sup_{\theta\in\Theta_{u j}}|\theta|\leq C_1$. In addition, for all $u_1,u_2\in\UU$, we have $(\|\theta_{u_2} - \theta_{u_1}\| + \|\beta_{u_2} - \beta_{u_1}\|)\leq C_1|u_2 - u_1|$. Finally, for all $u\in\UU$ and $j\in[\pp]$, $\Theta_{u j}$ contains a ball of radius $(\log\log n)(\log a_n)^{3/2}/n^{1/2}$ centered at $\theta_{u j}$.
\end{assumption}

\begin{assumption}[Sparsity]\label{ass: sparsity}
There exist $s=s_n$ and $\bar\gamma_u^j$, $u\in\UU$ and $j\in[\pp]$, such that  for all $u\in\UU$, $\|\beta_u\|_0 + \|\theta_u \|_0 + \max_{j\in[\pp]}\|\bar \gamma_u^j\|_0\leq s_n$ and $\max_{j\in[\pp]}(\|\bar \gamma_u^j-\gamma_u^j\|+s_n^{-1/2}\|\bar \gamma_u^j-\gamma_u^j\|_1)\leq C_1 (s_n\log a_n/n)^{1/2}$.
\end{assumption}

\begin{assumption}[Distribution of $Y$]\label{ass: density}
The conditional pdf of $Y$ given $(D,X)$ is bounded by $C_1$. %In addition, $\Pr_P(Y>\bar{y}\mid D,X)\geq c_1$ and $\Pr_P(Y<\underline{y}\mid D,X)\geq c_1$.
\end{assumption}

Assumptions \ref{ass: parameters}-\ref{ass: density} are mild and standard in the literature. In particular, Assumption \ref{ass: parameters} requires the parameter spaces $\Theta_{u j}$ to be bounded, and also requires that for each $u\in\UU$ and $j\in[\pp]$, the parameter $\theta_{u j}$ to be sufficiently separated from the boundaries of  the parameter space $\Theta_{u j}$. Assumption \ref{ass: sparsity} requires approximate sparsity of the model \eqref{Eq:MainLogisticModel}. Note that in Assumption \ref{ass: sparsity}, given that $\bar\gamma_u^j$'s exist, we can and will assume without loss of generality that $\bar\gamma_u^j = \gamma_{uT}^j$ for some $T\subset\{1,\dots,p+\pp-1\}$ with $|T| \leq s_n$, where $T = T_u^j$ is allowed to depend on $u$ and $j$. Here the $(p + \pp - 1)$-vector $\gamma_{u T}^j$ is defined from $\gamma_u^j$ by keeping all components of $\gamma_u^j$ that are in $T$ and setting all other components to be zero. Assumption \ref{ass: density} can be relaxed at the expense of more technicalities.

\begin{assumption}[Covariates]\label{ass: covariates}
For all $u\in\UU$, the following inequalities hold: (i) $\inf_{\|\xi\|=1}\Ep_P[ f^2_u\{(D',X')\xi\}^2]\geq c_1$, (ii) $\min_{j,k} (\Ep_P[|f_u^2 Z_u^j X_k^j|^2]\wedge \Ep_P[|f_u^2 D_j X_k^j|^2])\geq c_1$, and (iii) $\max_{j,k}\Ep_P[|Z_u^jX^j_k|^3]^{1/3}\log^{1/2} a_n \leq \delta_n n^{1/6}$. In addition, we have that (iv) $\sup_{\|\xi\|=1}\Ep_P[ \{(D',X')\xi\}^4]\leq C_1$, (v) $M_{n,1} \geq \Ep_P[ \sup_{u\in\UU, j\in[\pp]}|Z_u^j|^{2q}]^{1/(2q)}$, (vi) $M_{n,1}^2 s_n\log a_n \leq \delta_n n^{1/2 - 1/q}$, (vii) $M_{n,2} \geq \{\Ep_P[(\|D\|_\infty\vee \|X\|_\infty)^{2q}]\}^{1/(2q)}$, (viii) $M_{n,2}^2 s_n \log^{1/2} a_n \leq \delta_n n^{1/2-1/q}$, and (ix) $M_{n,1}^{2}M_{n,2}^4 s_n \leq \delta_n n^{1 - 3/q}$.
\end{assumption}

This assumption requires that there is no multicollinearity between covariates in vectors $D$ and $X$. In addition, it requires that the constants $\underline y$ and $\bar y$ are chosen so that the probabilities of $Y < \bar y$ and $Y > \underline y$ are both non-vanishing since otherwise we would have $\Ep[f_u^2] = \Ep[\text{Var}_P(Y_u\mid D,X)]$ vanishing either for $u=0$ or $u=1$ violating Assumption \ref{ass: covariates}(i). Intuitively, sending $\underline y$ and $\bar y$ to the left and to the right tails of the distribution of $Y$, respectively, would blow up the variance of the estimators $\check \theta_{u j}$, given by $\sigma^2_{u j}$ in Theorem \ref{theorem:semiparametric}, and leading eventually to the estimators with slower-than-$\sqrt n$ rate of convergence. Although our results could be extended to allow for the case where $\underline y$ and $\bar y$ are sent to the tails of the distribution of $Y$ slowly, we skip this extension for the sake of clarity.
Moreover, Assumption \ref{ass: covariates} imposes constraints on various moments of covariates. Since these constraints might be difficult to grasp, at the end of this section, in Corollary \ref{cor: simple example}, we provide an example for which these constraints simplify into easily interpretable conditions.

\begin{assumption}[Approximation error]\label{ass: approximation error}
For all $u\in\UU$, we have (i) $\sup_{\|\xi\|=1} \Ep_P[r_u^2\{(D',X')\xi\}^2]\leq C_1 \Ep_P[r_u^2]$, (ii)  $\Ep_P[r_u^2]\leq C_1 s_n \log a_n / n$, (iii) $\max_{j\in[\pp]}|\Ep_P[r_uZ_u^j]|\leq \delta_n n^{-1/2}$, and (iv) $|r_{u}(D,X)|\leq f_{u}^2(D,X)/4$ almost surely. In addition, with probability $1-\bar \Delta_n$, (v) $\sup_{u\in \UU, j\in[\pp]}(\En[(r_u Z_u^j/f_u)^2] + \En[r_u^2/f_u^6])\leq C_1 s_n\log a_n/n$.
\end{assumption}

This assumption requires the approximation error $r_u = r_u(D,X)$ to be sufficiently small. Under Assumption \ref{ass: covariates}, the first condition of Assumption \ref{ass: approximation error} holds if the approximation error is such that $r_u^2\leq C\Ep_P[r_u^2]$ almost surely for some constant $C$.

\subsection{Formal Results}

Under specified assumptions, our estimators $\check\theta_{u j}$ satisfy the following uniform Bahadur representation theorem.
\begin{theorem}[Uniform Bahadur representation for logistic model]\label{theorem:inferenceAlg1}  Suppose that Assumptions \ref{ass: parameters} -- \ref{ass: approximation error} hold for all $P\in\mathcal P_n$. In addition, suppose that the following growth condition holds: $\delta_n^2\log a_n = o(1)$. Then for the estimators  $\check \theta_{uj}$ satisfying \eqref{OS:Step3}, we have
\begin{equation}
\sqrt{n}\sigma_{u j}^{-1} (\check \theta_{uj} -  \theta_{uj}) = \Gn    \bar \psi_{uj} + O_P(\delta_n)\text{ in } \ell^\infty(\mathcal{U}\times[\pp]),
\end{equation}
uniformly over $P\in\mathcal P_n$, where $\bar \psi_{uj}(W):= - \sigma_{uj}^{-1}J^{-1}_{uj}  \psi_{uj}(W, \theta_{uj}, \eta_{uj})$, $\sigma_{uj}^2 := \Ep_P[J^{-2}_{uj}  \psi_{uj}^2(W, \theta_{uj}, \eta_{uj})]$, and $J_{u j}$ is defined in \eqref{eq: un and juj}.
\end{theorem}

%Theorem \ref{theorem:inferenceAlg1} relies on estimation methods based on $\ell_1$-penalization adapted for functional response data. Indeed, the rates of convergence of the relevant estimators should be uniform over $\UU\times [\pp]$.

This theorem allows us to establish a Gaussian approximation result for the supremum of the process $\{\sqrt n \sigma_{u j}^{-1}(\check \theta_{u j} - \theta_{u j})\colon u\in\UU, j\in[\pp]\}$:

\begin{corollary}[Gaussian approximation for logistic model]\label{cor: logistic clt}
Suppose that Assumptions \ref{ass: parameters} -- \ref{ass: approximation error} hold for all $P\in\mathcal P_n$. In addition, suppose that the following growth conditions hold: $\delta_n^2\log a_n = o(1)$, $M_{n,1}^{2/7}\log a_n = o(n^{1/7})$, and $M_{n,1}^{2/3}\log a_n = o(n^{1/3 - 2/(3q)})$. Then %$s_n\log^3 a_n = o(n)$. Then
{\small
$$
\sup_{t\in \mathbb{R}} \left| \Pr_P\left(  \sup_{u\in\UU,j\in[\tilde{p}]}|\sqrt n\sigma_{uj}^{-1}(\check\theta_{uj}-\theta_{uj})| \leq t \right) - \Pr_P\left(\sup_{u\in\mathcal{U},j\in[\tilde{p}]}|\mathcal{N}_{u j}| \leq t\right) \right| = o(1)
$$}\!uniformly over $P\in\mathcal P_n$, where $(\mathcal N_{u j})_{u\in\UU,j\in[\pp]}$ is a tight zero-mean Gaussian process indexed by $\UU\times[\pp]$ with the covariance given by $\Ep_P[\bar\psi_{u j}(W)\bar\psi_{u'j'}(W)]$ for $u,u'\in\UU$ and $j,j'\in[\pp]$.
\end{corollary}

Based on this corollary, we are now able to construct simultaneous confidence bands for the parameters $\theta_{u j}$. Observe that
$$
J_{u j} = -\Ep_P\Big[\Lambda'\Big(D_j \theta_{u j} + X^j \beta_u^j\Big)D_j(D_j - X^j \gamma_u^j)\Big],\quad u\in\UU, \ j\in[\pp],
$$
and so it can be estimated by
$$
\hat J_{u j} = -\En\Big[\Lambda'\Big(D_j \widetilde\theta_{u j} + X^j \hat\beta_u^j\Big)D_j (D_j - X^j \widetilde\gamma_u^j)\Big],\quad u\in\UU, \ j\in[\pp],
$$
In addition, $\sigma_{u j}^2 = \Ep_P[J_{u j}^{-2}\psi_{u j}^2(W,\theta_{u j},\eta_{u j})]$, and so it can be estimated by
$$
\hat\sigma_{u j}^2 = \En[\hat J_{u j}^{-2}\psi_{u j}^2(W,\widetilde\theta_{u j},\hat\eta_{u j})],\quad u\in\UU, \ j\in[\pp],
$$
Moreover, as in Section \ref{sec: general}, define  $\hat\psi_{u j}(W) = -\hat\sigma_{u j}^{-1}\hat J_{u j}^{-1}\psi_{u j}(W,\check\theta_{u j},\hat\eta_{u j})$, and let $c_\alpha$ be the $(1-\alpha)$ quantile of the conditional distribution of $\sup_{u\in\UU,j\in[\pp]}|\hat{\mathcal G}_{u j}|$ given the data $(W_i)_{i=1}^n$ where the process $\hat{\mathcal G} = (\hat{\mathcal G}_{u j})_{u\in\UU,j\in[\pp]}$ is defined in \eqref{eq: process g}. Then we have

\begin{corollary}[Simultaneous confidence bands for logistic model]\label{cor: logistic bands}
Suppose that Assumptions \ref{ass: parameters} -- \ref{ass: approximation error} hold for all $P\in\mathcal P_n$. In addition, suppose that the following growth conditions hold: $\delta_n^2\log a_n = o(1)$, $M_{n,1}^{2/7}\log a_n = o(n^{1/7})$, $M_{n,1}^{2/3}\log a_n = o(n^{1/3 - 2/(3q)})$, and $s_n\log^3 a_n = o(n)$. Then \eqref{UnifCoverage} holds uniformly over $P\in\mathcal P_n$.
\end{corollary}

To conclude this section, we provide an example for which conditions of Corollary \ref{cor: logistic bands} are easy to interpret. Recall that $a_n = n\vee p\vee \pp$.
\begin{corollary}[Uniform confidence bands for logistic regression model under simple conditions]\label{cor: simple example}
Suppose that Assumptions \ref{ass: parameters} -- \ref{ass: density}, \ref{ass: covariates}(i,ii,iv), and \ref{ass: approximation error}(i,ii,iv,v) hold for $q>4$ for all $P\in\mathcal P_n$. In addition, suppose that $\{\Ep_P[(\|D\|_\infty\vee \|X\|_\infty)^{2q}]\}^{1/(2q)}\leq C_1$ and $\sup_{u\in \UU, j\in[\pp]}\|\gamma_u^j\|_1\leq C_1$. Moreover, suppose that the following growth conditions hold:
$\log^7 a_n / n = o(1)$, $s_n^2 \log ^3 a_n/n^{1 - 2/q} = o(1)$, and $\sup_{u\in\UU,j\in[\pp]}|\Ep_P[r_u Z_u^j]| =o( (n\log a_n)^{-1/2})$. Then \eqref{UnifCoverage} holds uniformly over $P\in\mathcal P_n$.
\end{corollary}

%\begin{remark}\label{rem: low level conditions}
%LOW-LEVEL CONDITIONS.
%\end{remark}
%\section{Estimation Methods for Functional Response Data}

\begin{remark}[Estimation of variance] When constructing the confidence bands based on \eqref{UnifCoverage}, we find in simulations that it is beneficial to replace the estimators $\hat\sigma_{u j}^2$ of $\sigma_{u j}^2$ by $\max\{\hat \sigma_{uj}^2, \widehat \Sigma_{uj}^2\}$ where $\widehat\Sigma_{u j}^2 = \En[\hat f_u^2(D-X^j\widetilde\gamma_u^j)^2 ]$ is an alternative consistent estimator of $\sigma_{u j}^2$. \qed
\end{remark}

\begin{remark}[Alternative implementations, double selection]
We note that the theory developed here is applicable for different estimators that construct the new score function with the desired orthogonality condition %optimal score
implicitly. For example, the double selection idea yields an implementation of an estimator that is first-order equivalent to the estimator based on the score function. %optimal score function.
The algorithm yielding the double selection estimator is as follows.

\begin{algorithm}
For each $u\in \UU$ and $j\in[\pp]$: \\
\enspace \emph{Step $1'$}. Run post-$\ell_1$-penalized logistic estimator (\ref{PostL1Logistic})  of $Y_{u}$ on  $D$ and $X$ to compute $(\widetilde\theta_u,\widetilde\beta_u)$. \\
\enspace \emph{Step $2'$}. Define the weights $\hat f_{u}^2 = \hat f_u^2(D,X) = \G'(D_i'\widetilde\theta_u+X_i'\widetilde\beta_u)$.\\
\enspace \emph{Step $3'$}. Run  the lasso estimator (\ref{EstLasso}) of $\hat f_u D_j$ on $\hat f_u X$ to compute $\hat\gamma_u^j$. \\
\enspace \emph{Step $4'$}. Run logistic regression of $Y_{u}$ on  $D_j$ and all the selected variables in Steps $1'$ and $3'$ to compute $\check \theta_{u j}$.
\end{algorithm}

As mentioned by a referee, it is surprising that the double selection procedure has uniform validity. The use of the additional variables selected in Step 3', through the first order conditions of the optimization problem, induces the necessary near-orthogonality condition. We refer to the Supplementary Material for a more detailed discussion. \qed
\end{remark}

\begin{remark}[Alternative implementations, one-step correction]
Another implementation for which the theory developed here applies is to replace Step 5 in Algorithm 1 with a one-step procedure. This relates to the debiasing procedure proposed in \cite{vandeGeerBuhlmannRitov2013} to the case when the set $\UU$ is a singleton. In this case instead of minimizing the criterion (\ref{OS:Step3}) in Step 5, the method makes a full Newton step from the initial estimate,

\noindent
{\it Step $5''$}. Compute $ \bar \theta_{uj} = \hat\theta_{uj} - \hat J_{uj}^{-1}\En[\psi_{uj}(W,\hat\theta_{uj},\hat\eta_{uj})].$
%{\it Step (3'')} Define  $ \bar \theta_{uj} = \hat\theta_{uj} - \En[\G'(\hat z_{u}^j\hat\theta_{uj}+\hat h_{u}^j)(\hat z_{u}^j)^2 ]^{-1}\En[\{Y_{u} - \G(\hat z_{u}^j\hat\theta_{uj}+\hat h_{u}^j)\}\hat z_{u}^j].$
%Alternatively, Step 3 can be replaced by a double selection procedure similar to the proposed procedure in \cite{BelloniChernozhukovHansen2011} and \cite{BCH2011:InferenceGauss} for partial linear mean regression models. We define the estimator as
%\begin{enumerate}
%\item[]{\bf Step 3''} Run Logistic Regression of $Y_{ui}$ on  $D_{ij}$ and the covariates selected in Step 1 and 2:
%$$\begin{array}{rl}
%\ \ \ \ \ \ \ \ \ \ \  \ \ \ \ \ \ (\check\theta_u,\check\beta_u)\in & \arg{\displaystyle\min_{\theta,\beta}} \ \ \En[M(Y_{u},D,X,\theta,\beta)] \ : \ \supp(\theta_{[\pp]\setminus j},\beta)\subseteq \supp(\hat\theta_u,\hat\beta_u) \cup \supp(\hat\gamma_u^j)\end{array}$$
%\end{enumerate}

\noindent
The theory developed here directly apply to those estimators as well.\qed
\end{remark}

\begin{remark}[Extension to other approximately sparse generalized linear models]
Inspecting the proofs of Theorem \ref{theorem:inferenceAlg1} and Corollaries \ref{cor: logistic clt}--\ref{cor: simple example} reveals that these results can be extended with minor modifications to cover other approximately sparse generalized linear models. For example, the results can be extended to cover the model  \eqref{Eq:MainLogisticModel} where we use the probit link function instead of the logit link function $\Lambda$.\qed
\end{remark}

\section{$\ell_1$-Penalized M-Estimators: Nuisance Functions and Functional Data}\label{FunctionalLassoSection}

In this section, we define the estimators $\widetilde \theta_{u}$, $\widetilde \beta_u$, and $\widetilde \gamma_u^j$, which were used in the previous section, and study their properties. We consider the same setting as that in the previous section. The results in this section rely upon a set of new results for $\ell_1$-penalized $M$-estimators with functional data presented in Appendix \ref{sec: generic results} of the Supplementary Material.

\subsection{$\ell_1$-Penalized Logistic Regression for Functional Response Data: Asymptotic Properties}

Here we consider the generalized linear model with the logistic link function and functional response data \eqref{Eq:MainLogisticModel}. As explained in the previous section, we assume that $\widetilde \theta_u$ and $\widetilde \beta_u$ are post-regularization maximum likelihood estimators of $\theta_u$ and $\beta_u$ corresponding to the log-likelihood function $M_u(W,\theta,\beta) = M_u(Y_u,D,X,\theta,\beta)$ defined in \eqref{eq: log-likelihood function}. To define these estimators, let $\widehat \theta_u$ and $\widehat\beta_u$ be $\ell_1$-penalized maximum likelihood (logistic regression) estimators
 \begin{equation}\label{L1Logistic}
(\hat\theta_u,\hat\beta_u) \in \arg\min_{\theta,\beta}\left( \En[M_u(Y_u,D,X,\theta,\beta)] + \frac{\lambda}{n}\|\hat\Psi_u(\theta',\beta')'\|_1\right)
\end{equation}
where $\lambda$ is a penalty level and $\widehat\Psi_u$ a diagonal matrix of penalty loadings. We choose parameters $\lambda$ and $\widehat\Psi_u$ according to Algorithm 3 described below. Using the $\ell_1$-penalized estimators $\widehat \theta_u$ and $\widehat\beta_u$, we then define post-regularization estimators $\widetilde\theta_u$ and $\widetilde\beta_u$ by
 \begin{equation}\label{PostL1Logistic}
(\widetilde\theta_u, \widetilde\beta_u) \in \arg\min_{\theta} \En[M_u(Y_u,D,X,\theta,\beta)]   : \supp(\theta,\beta)\subseteq \supp (\hat\theta_u,\hat\beta_u).
\end{equation}
We derive the rate of convergence and sparsity properties of $\widetilde\theta_u$ and $\widetilde\beta_u$ as well as of $\widehat\theta_u$ and $\widehat\beta_u$ in Theorem \ref{Thm:RateEstimatedLassoLogistic} below. Recall that $a_n = n\vee p \vee \pp$.

\begin{algorithm}[Penalty Level and Loadings for Logistic Regression]\label{AlgFunc} Choose $\gamma \in [1/n, 1/\log n]$ and $c > 1$ (in practice, we set $c=1.1$ and $\gamma = .1/\log n$). Define $\lambda = c\sqrt n\Phi^{-1}(1-\gamma/(2(p + \pp) N_n))$ with $N_n=n$. To select $\widehat\Psi_u$, choose a constant $\bar m \geq 0$ as an upper bound on the number of loops and proceed as follows: (0) Let $\widetilde X = (D',X')'$, $m=0$, and initialize $\hat l_{uk,0} = \frac{1}{2}\{\En[\widetilde X_k^2]\}^{1/2}$ for $k\in[p+\pp]$. (1)  Compute $(\widehat\theta_u,\widehat\beta_u)$ and $(\widetilde\theta_u,\widetilde\beta_u)$ based on $\widehat\Psi_u = \diag(\{\hat l_{uk,m}, k\in[p+\pp]\})$. (2) Set $\hat l_{uk,m+1} := \{\En[\widetilde X_k^2(Y_{u}- \G(D'\widetilde\theta_u+X'\widetilde \beta_u))^2]\}^{1/2}.$ (3) If $m \geq \bar m$,  report the current value of $\widehat\Psi_u$ and stop; otherwise set $m \leftarrow m+1 $ and go to step (1). \end{algorithm}

\begin{theorem}[Rates and Sparsity for Functional Response under Logistic Link]\label{Thm:RateEstimatedLassoLogistic}
 Suppose that Assumptions \ref{ass: parameters} -- \ref{ass: approximation error} hold for all $P\in\mathcal P_n$. In addition, suppose that the penalty level $\lambda$ and the matrices of penalty loadings $\widehat\Psi_u$ are chosen according to Algorithm \ref{AlgFunc}. Moreover, suppose that the following growth condition holds: $\delta_n^2\log a_n = o(1)$. Then there exists a constant $\bar C$ such that uniformly over all $P \in \mathcal{P}_n$ with probability $1-o(1)$,
\begin{align*}
&\sup_{u\in\mathcal{U}} \Big(\| \hat\theta_u - \theta_{u}\|+\|\hat\beta_u-\beta_u\|\Big) \leq \bar C \sqrt{\frac{s_n \log a_n}{n}},\\
&\sup_{u\in\mathcal{U}}\Big(\|\hat\theta_u-\theta_{u}\|_1 +\|\hat\beta_u-\beta_{u}\|_1 \Big)\leq \bar C   \sqrt{\frac{s_n^2\log a_n}{n}},
\end{align*}
and the estimators $\widehat\theta_u$ and $\widehat\beta_u$ are uniformly sparse: $\sup_{u\in \mathcal{U}}\|\hat \theta_u \|_0  +\|\hat \beta_u \|_0\leq \bar C  s_n$. Also, uniformly over all $P \in \mathcal{P}_n$, with probability $1-o(1)$,
%$$ \sup_{u\in \mathcal{U}} \| D'(\widetilde \theta_u -\theta_u)+X(\widetilde\beta_u-\beta_u)\|_{\Pn,2} \leq \bar C  \sqrt{\frac{s \log (p \pp n)}{n}},  \
% \displaystyle \sup_{u\in\mathcal{U}}\|\widetilde \theta_u-\theta_{u}\|_1 + \|\widetilde \beta_u-\beta_{u}\|_1 \leq \bar C \sqrt{\frac{s^2\log (p\pp n)}{n}}. $$
\begin{align*}
&\sup_{u\in \mathcal{U}} \Big(\|\widetilde \theta_u -\theta_u\|+\|\widetilde\beta_u-\beta_u\|\Big) \leq \bar C  \sqrt{\frac{s_n \log a_n}{n}},  \\
& \sup_{u\in\mathcal{U}}\Big(\|\widetilde \theta_u-\theta_{u}\|_1 + \|\widetilde \beta_u-\beta_{u}\|_1\Big) \leq \bar C \sqrt{\frac{s_n^2\log a_n}{n}}.
\end{align*}
 \end{theorem}

\subsection{Lasso with Estimated Weights: Asymptotic Properties}\label{FunctionalestimatedLassoSection}

Here we consider the weighted linear model \eqref{Eq:Decomposition} for $u\in\UU$ and $j\in[\pp]$. Using the parameter $\bar\gamma_u^j$ appearing in Assumption \ref{ass: sparsity}, it will be convenient to rewrite this model as
\begin{equation}\label{Def:ModelEstLasso}
f_{u}D_j = f_{u} X^j\bar \gamma_u^j + f_{u}\bar r_{u j} + v_{u}^j, \quad \Ep_P[ f_{u}X^jv_{u}^j ] = 0
\end{equation}
where $\bar r_{u j} = X^j(\gamma_u^j - \bar\gamma_u^j)$ is an approximation error, which is asymptotically negligible under Assumption \ref{ass: sparsity}. As explained in the previous section, we assume that $\widetilde\gamma_u^j$ is a post-regularization weighted least squares estimator of $\gamma_u^j$ (or $\bar\gamma_u^j$). To define this estimator, let $\widehat\gamma_u^j$ be an $\ell_1$-penalized (weighted Lasso) estimator
\begin{equation}\label{EstLasso}
\hat \gamma_u^j \in \arg \min_{\gamma}\left( \mbox{$\frac{1}{2}$}\En[ \hat f_{u}^2 ( D_{j} - X^j\gamma)^2] + \frac{\lambda}{n}\|\hat \Psi_{u j}\gamma\|_1 \right)
\end{equation}
where $\lambda$ and $\hat \Psi_{u j}$ are the associated penalty level and the diagonal matrix of penalty loadings specified below in Algorithm 4 and where $\widehat f_u^2$'s are estimated weights. As in Algorithm 1 in the previous section, we set $\widehat f_u^2 = \widehat f_u^2(D,X) = \Lambda'(D'\widetilde \theta_u + X'\widetilde\beta_u)$. Using $\widehat\gamma_u^j$, we define a post-regularized weighted least squares estimator
\begin{equation}\label{EstPostLasso}
\widetilde \gamma_u^j \in \arg \min_{\gamma} \mbox{$\frac{1}{2}$}\En[ \hat f_{u}^2 ( D_{j} - X^j\gamma)^2]  :  \supp(\gamma)\subseteq \supp(\hat \gamma_u^j).
\end{equation}
 We derive the rate of convergence and sparsity properties of $\widetilde\gamma_u^j$ as well as of $\widehat\gamma_u^j$ in Theorem \ref{Thm:RateEstimatedLassoLinear} below. %One of the difficulties we have to deal with in the proof The new difficulty is to account for the impact of estimated weights $\hat f_{u}$.

\begin{algorithm}[Penalty Level and Loadings for Weighted Lasso]\label{AlgFunc2} Choose $\gamma \in [1/n, 1/\log n]$ and $c > 1$ (in practice, we set $c=1.1$ and $\gamma = .1/\log n$). Define $\lambda = c\sqrt n\Phi^{-1}(1-\gamma/(2(p + \pp) N_n))$ with $N_n = p\pp^2 n^2$. To select $\widehat\Psi_{u j}$, choose a constant $\bar m \geq 1$ as an upper bound on the number of loops and proceed as follows: (0) Set $m=0$ and $\hat l_{ujk,0} = \max_{1\leq i\leq n}\|\hat f_{ui}X_i^j\|_\infty\{\En[\hat f_u^2D_j^2]\}^{1/2}$. (1)  Compute $\hat \gamma_u^j$ and $\widetilde \gamma_u^j$ based on $\widehat\Psi_{u j} = \diag(\{\hat l_{ujk,m}, k\in[p+\pp-1]\})$. (2) Set $\hat l_{ujk,m+1} := \{\En[\hat f_u^4(D_j- X^j\widetilde\gamma_u^j)^2(X_k^j)^2]\}^{1/2}.$ (3) If $m\geq \bar m$, report the current value of $\widehat\Psi_u^j$ and stop; otherwise set $m \leftarrow m+1 $ and go to step (1). \end{algorithm}

%\begin{remark}[Initial Penalty Loadings]\label{IPL}Regarding the initial choice of penalty loadings, we note that it is constant across components, namely $\hat l_{ujk,0} =\hat l_{ujk',0}$. Therefore, it yields the same solution if we set $\bar \lambda = \lambda \{\max_{i\leq n}\|\hat f_{ui}X_i^j\|_\infty\En[\hat f_u^2D_j^2]\}^{1/2} / \|\widehat \Psi_{uj0}\|_\infty\}$ and $\bar l_{ujk,0}= \|\widehat \Psi_{u0}\|_\infty$. This implies that with high probability Assumption \ref{ass: M}(b) holds with $\ell = 1$ and $L=\|\widehat \Psi_{uj0}\|_\infty\|\widehat \Psi_{uj0}^{-1}\|_\infty$ which is bounded with probability $1-o(1)$ under mild conditions.\qed\end{remark}

\begin{theorem}[Rates and Sparsity for Lasso with Estimated Weights]\label{Thm:RateEstimatedLassoLinear}
 Suppose that Assumptions \ref{ass: parameters} -- \ref{ass: approximation error} hold for all $P\in\mathcal P_n$. In addition, suppose that the penalty level $\lambda$ and the matrices of penalty loadings $\widehat\Psi_{u j}$ are chosen according to Algorithm \ref{AlgFunc2}. Moreover, suppose that the following growth condition holds: $\delta_n^2\log a_n = o(1)$. Then there exists a constant $\bar C$ such that uniformly over all $P \in \mathcal{P}_n$ with probability $1-o(1)$,
\begin{align*}
&\max_{j\in[\pp]}\sup_{u\in\mathcal{U}} \| \hat\gamma_u^j - \bar\gamma_{u}^j\| \leq \bar C \sqrt{\frac{s_n\log a_n}{n}},\ \  \max_{j\in[\pp]} \sup_{u\in\mathcal{U}}\|\hat\gamma_u^j-\bar\gamma_{u}^j\|_1  \leq \bar C   \sqrt{\frac{s_n^2 \log a_n}{n}},
 \end{align*}
and the estimator $\hat\gamma_u^j$ is uniformly sparse, $\max_{j\in[\pp]}\sup_{u\in \mathcal{U}}\|\hat \gamma_u^j \|_0  \leq \bar C  s_n$.
Also, uniformly over all $P \in \mathcal{P}_n$, with probability $1-o(1)$,
\begin{align*}
&\max_{j\in[\pp]}\sup_{u\in \mathcal{U}} \| \widetilde \gamma_u^j -\bar\gamma_u^j\| \leq \bar C  \sqrt{\frac{s_n \log a_n}{n}},\ \ \max_{j\in[\pp]}\sup_{u\in\mathcal{U}}\|\widetilde \gamma_u^j-\bar\gamma_{u}^j\|_1  \leq \bar C \sqrt{\frac{s_n^2\log a_n}{n}}.
\end{align*}
 \end{theorem}

\section{Monte Carlo Simulations}\label{sec: simulations}\label{sec: monte carlo}

In this section we provide a simulation study to investigate the finite sample properties of the proposed estimators and the associated confidence regions. We report only the performance of the estimator based on the double selection procedure due to space constraints and note that it is very similar to the performance of the estimator based on score functions with near-orthogonality property. We will compare the proposed procedure with the traditional estimator that refits the model selected by the corresponding $\ell_1$-penalized M-estimator (naive post-selection estimator).

We consider a logistic regression model with functional response data where the response $Y_u = 1\{ y \leq u\}$ for $u\in \mathcal{U}$ a compact set. We specify two different designs: (1) a location model, $y = x'\beta_0 + \xi$, where $\xi$ is distributed as a logistic random variable, the first component of $x$ is the intercept and the other $p-1$ components are distributed as $N(0,\Sigma)$ with $\Sigma_{k,j}=|0.5|^{|k-j|}$; (2) a location-shift model, $y = \{ (x'\beta_0 + \xi ) / x'\vartheta_0 \}^3$, where $\xi$ is distributed as a logistic random variable, $x_j=|w_j|$ where $w$ is a $p$-vector distributed as $N(0,\Sigma)$ with $\Sigma_{k,j}=|0.5|^{|k-j|}$, and $\vartheta_0$ has non-negative components. Such specification implies that for each $u\in \mathcal{U}$
$$ \mbox{Design 1:} \ \ \theta_u = u(1,0,\ldots,0)'-\beta_0 \ \ \ \mbox{and} \ \ \ \mbox{Design 2:} \ \ \theta_u = u^{1/3}\vartheta_0 - \beta_0.$$
In our simulations we will consider $n=500$ and $p=2000$. For the location model (Design 1) we will consider two different choices for $\beta_0$: (i) $\beta_{0j}^{(i)}=2/j^2$ for $j=1,\ldots,p$, and (ii) $\beta_{0j}^{(ii)} = (1/2)/\{ j-3.5\}^2$ for $j>1$ with the intercept coefficient $\beta_{01}^{(ii)}=-10$. (These choices ensure  $\max_{j>1}|\beta_{0j}|=2$ and that $y$ is around zero in Design 2(ii).) We set $\vartheta_0=\frac{1}{8}(1,1,1,1,0,0,\ldots,0,0,1,1,1,1)'$. For Design 1 we have $\mathcal{U} = [1,2.5]$ and for Design 2 we have $\mathcal{U}=[-.5,.5]$. The results are based on 500 replications (the bootstrap procedure is performed 5000 times for each replication).

We report the (empirical) rejection frequencies for confidence regions with 95\% nominal coverage. That is, the fraction of simulations the confidence regions of a particular method did not cover the true value (thus 0.05 rejection frequency is the ideal performance). We report the rejection frequencies for the proposed estimator and the post-naive selection estimator.

Table 1 presents the performance of the methods when applied to construct a confidence interval for a single parameter ($\tilde p =1$ and $\UU$ is a singleton). Since the setting is not symmetric we investigate the performance for different components. Specifically, we consider $\{u\}\times \{j\}$ for $j=1,\ldots,5$. First consider the location model (Design 1). The difference between the performance of the naive estimator for Design 1(i) and 1(ii) highlights its fragile performance which is highly dependent on the unknown parameters. We can see from Table 1 that in Design 1(i) the Naive method achieve (pointwise) rejection frequencies up to 0.162 when the nominal level is 0.05. In Design 1(ii) it can be as high as 0.886. We also note that it is important to look at the performance of each component and avoid averaging across components (large $j$ components are essentially not in the model, indeed for $j>50$ we obtain rejection frequencies very close to 0.05 regardless of the model selection procedure). In contrast the proposed estimator exhibits a much more robust behavior. For Design 1(i) the rejection frequencies are between 0.040 and 0.062 while for Design 1(ii) the rejection frequencies of the proposed estimator are between 0.040 and 0.056.

{\tiny \begin{table}
\begin{tabular}{c|c|cccccc}
\hline
 \multicolumn{2}{c}{$p=2000, n=500$ } & \multicolumn{5}{c}{Rejection Frequencies  for $j \in \{1,\ldots,5\}$ } \\
\hline
Design  & Method  &  $j=1$ & $j=2$  &  $j=3$ & $j=4$  &  $j=5$ \\
\hline
{1}{(i)}  & Proposed    & 0.042 & 0.040 & 0.062 & 0.050 & 0.044   \\
          & Naive & 0.100 & 0.098 & 0.108 & 0.108 & 0.162  \\
\hline
{1}{(ii)} & Proposed     & 0.044  & 0.040 & 0.054 & 0.056 & 0.056  \\
          & Naive   & 0.038 & 0.030 & 0.070 & 0.886 & 0.698  \\
\hline
{2}{(i)}  & Proposed & 0.046 & 0.054 & 0.044 & 0.052 & 0.054 \\
          & Naive & 0.046 & 0.050 & 0.038 & 0.070 & 0.054  \\
\hline
{2}{(ii)} & Proposed & 0.092 & 0.074 & 0.034 & 0.088 & 0.082  \\
          & Naive & 0.034 & 0.972 & 0.182 & 0.312 & 0.916 \\
\hline
\end{tabular}\caption{We report the pointwise rejection frequencies of each method for (pointwise) confidence intervals for each $j\in\{1,\ldots,5\}$. For Design 1 we used $\UU=\{1\}$ and for Design 2 we used $\UU=\{.5\}$. The results are based on 500 replications. }\label{Tab:MC1}
\end{table}}

Table 2 presents the performance for simultaneous confidence bands of the form
$\{ [ \tilde \theta_{uj} - {\rm cv} \tilde\sigma_{uj}, \tilde \theta_{uj} + {\rm cv} \tilde\sigma_{uj}] $ for $u\in \UU\times[\tilde p]\}$ where $\tilde \theta_{uj}$ is a point estimate, $\tilde \sigma_{uj}$ is an estimate of the pointwise standard deviation, and ${\rm cv}$ is a critical value that accounts for the uniform estimation. For the point estimate we consider the proposed estimator and the post-naive selection estimator which have estimates of standard deviation. We consider two critical values: from the multiplier bootstrap (MB) procedure and the Bonferroni (BF) correction (which we expect to be conservative). For each of the four different designs (1(i), 1(ii), 2(i) and 2(ii) described above), we consider four different choices of $\UU\times[\pp]$. Table 2 displays rejection frequencies for confidence regions with 95\% nominal coverage (and again $.05$ would be the ideal performance). The simulation results confirms the differences between the performance of the methods and overall the proposed procedure is closer to the nominal value of $.05$ . The proposed estimator performed within a factor of two to the nominal value in 10 out of the 16 designs considered (and 13 out 16 within a factor of three). The post-naive selection estimator performed within a factor of two only in 3 out of the 16 designs  when using the multiplier bootstrap as critical value (7 out of 16 within a factor of three) and similarly with the Bonferroni correction as the critical value.

{\tiny \begin{table}
\begin{tabular}{l|c|c|c|cccc}
\hline
 \multicolumn{2}{c}{$p=2000, n=500$} & \multicolumn{3}{c}{Uniform over $\UU\times [\tilde p]$} \\
\hline
 Design & Method  &    $[1,2.5]\times \{1\}$ & $\{1\}\times [10] $ & $[1,2.5]\times [10]$ & $\{1\} \times[1000]$\\
\hline
          & Proposed         & 0.054  & 0.036 & 0.048 & 0.040 \\
{1}(i)  & Naive (MB)   & 0.126 & 0.136 & 0.172 & 0.032 \\
        & Naive (BF)   & 0.014 & 0.124 & 0.026 & 0.032 \\
\hline
          & Propose         & 0.270 & 0.036 & 0.032 & 0.142\\
{1}(ii)  & Naive (MB) & 0.014   & 0.802 & 0.934 & 0.404 \\
        & Naive (BF)   & 0.000 & 0.802 & 0.718 & 0.376\\
\hline
  Design & Method  &   $[-.5,.5]\times \{1\} $ & $\{.5\}\times [10]$ & $[-.5,.5]\times [10]$ & $\{.5\}\times[1000]$ \\
\hline
          & Proposed         & 0.364 & 0.038 & 0.052  & 0.062\\
{2}(i)  & Naive (MB)   & 0.116 & 0.040 & 0.022 & 0.048\\
        & Naive (BF)   & 0.018 & 0.038 & 0.000 & 0.046 \\
\hline
          & Proposed         & 0.140  & 0.090 & 0.408 & 0.084 \\
{2}(ii)  & Naive (MB)  & 0.002 & 0.946 & 0.996 & 0.362 \\
        & Naive (BF)   & 0.000 & 0.946 &  0.944  & 0.298 \\
\hline
\end{tabular}\caption{We report the rejection frequencies of each method for the (uniform) confidence bands for $\mathcal{U}\times[\pp]$. The proposed estimator computes the critical value based on the multiplier bootstrap procedure. For the naive  post-selection estimator we report the results for two choices of critical values, one choice based on the multiplier bootstrap (MB), and  another based on Bonferroni (BF) correction. The results are based on 500 replications. }\label{Tab:MC2}
\end{table}
}

\appendix

\section{Proofs for Section 2}

In this appendix, we use $C$ to denote a strictly positive constant that is independent of $n$ and $P\in\mathcal P_n$. The value of $C$ may change at each appearance. Also, the notation $a_n \lesssim b_n$ means that $a_n\leq C b_n$ for all $n$ and some $C$. The notation $a_n\gtrsim b_n$ means that $b_n\lesssim a_n$. Moreover, the notation $a_n = o(1)$ means that there exists a sequence $(b_n)_{n\geq 1}$ of positive numbers such that (i) $|a_n|\leq b_n$ for all $n$, (ii) $b_n$ is independent of $P\in\mathcal P_n$ for all $n$, and (iii) $b_n\to 0$ as $n\to\infty$. Finally, the notation $a_n = O_P(b_n)$ means that for all $\epsilon>0$, there exists $C$ such that $\Pr_P(a_n > C b_n)\leq 1-\epsilon$ for all $n$. Using this notation allows us to avoid repeating ``uniformly over $P\in\mathcal P_n$'' many times in the proofs of Theorem \ref{theorem:semiparametric} and Corollaries \ref{cor: general CLT} and \ref{theorem: general bs}. Throughout this appendix, we assume that $n\geq n_0$.

\subsection*{Proof of Theorem \ref{theorem:semiparametric}}
We split the proof into five steps.

%\medskip
{\bf Step 1.} (Preliminary Rate Result). We claim that  with probability $1- o(1)$, $\sup_{\uu \in \UU, j\in [\pp]}| \check \theta_{uj} - \theta_{uj}| \lesssim B_{1n} \tau_n.$
To show that, note that by definition of $\check\theta_{u j}$, we have for each $\uu \in \UU$ and $j\in[\pp]$,
$$
\Big| \En [\psi_{uj}(W, \check \theta_{u j}, \hat \eta_{uj})]\Big|
\leq
\inf_{\theta \in \Theta_{uj}}\Big| \En[ \psi_{u j}(W, \theta, \hat \eta_{uj})] \Big| + \epsilon_n,
$$ which implies
via the triangle inequality that uniformly over $u \in \mathcal{U}$ and $j\in[\pp]$, with probability $1-o(1)$,
\begin{equation}\label{eq: rate proof}
\Big | \left.  \Ep_P [\psi_{\uu j}(W, \theta, \eta_{uj} )]  \right|_{\theta=\check \theta_{u j}}\Big |\leq \epsilon_n + 2 I_1+ 2 I_2 \lesssim  B_{1 n}\tau_n, \ \ \mbox{where}
\end{equation}
\begin{align*}
I_1 & :=  \sup_{u \in \mathcal{U},j\in[\pp],\theta \in \Theta_{uj}} \Big | \En [\psi_{\uu j}(W, \theta, \hat \eta_{uj}) ]- \En[ \psi_{uj}(W, \theta, \eta_{uj} ) ]\Big | \lesssim  B_{1n}\tau_n , \\
I_2 & :=    \sup_{u \in \mathcal{U},j\in[\pp],\theta \in \Theta_{uj}}  \Big | \En[ \psi_{\uu j}(W, \theta, \eta_{uj})] - \Ep_P[ \psi_{\uu j}(W, \theta, \eta_{uj} ) ]\Big | \lesssim  \tau_n.
\end{align*}
and the bounds on $I_1$ and $I_2$ are derived in Step 2 (note also that $\epsilon_n = o(\tau_n)$ by construction of the estimator and Assumption \ref{ass: AS}(vi)).
 Since by Assumption \ref{ass: S1}(iv), $2^{-1}  | J_{\uu j}(\check \theta_{\uu j}- \theta_{\uu j})| \wedge c_0$ does not exceed the left-hand side of  (\ref{eq: rate proof}), $\inf_{u \in \mathcal{U}, j\in [\pp]}|J_{\uu j}|\gtrsim 1$, and by Assumption \ref{ass: AS}(vi), $B_{1 n}\tau_n = o(1)$, we conclude that
\begin{equation}\label{PreRate}
\sup_{u \in \mathcal{U}, j\in[\pp]}| \check \theta_{uj} - \theta_{uj}|  \lesssim \left(\inf_{u \in \mathcal{U}, j\in[\pp]} |J_{\uu j}|\right)^{-1}  B_{1 n} \tau_n \lesssim B_{1 n} \tau_n,
\end{equation}
with probability $1-o(1)$ yielding the claim of this step.

%\medskip
{\bf Step 2.} (Bounds on $I_1$ and $I_2$)  We claim that with probability $1- o(1)$,
$I_1 \lesssim  B_{1 n} \tau_n$ and $I_2  \lesssim   \tau_n$.
To show these relations, observe that with probability $1-o(1)$, we have $I_1 \leq 2I_{1a} + I_{1b}$ and $I_2 \leq I_{1a}$, where
\begin{align*}
I_{1a} & :=  \sup_{u \in \mathcal{U}, j\in [\pp],  \theta \in \Theta_{uj}, \eta \in \mT_{uj}}  \Big | \En [\psi_{\uu j}(W, \theta, \eta)] - \Ep_P[ \psi_{uj}(W, \theta, \eta )]  \Big |, \\
I_{1b} & :=    \sup_{u \in \mathcal{U}, j \in [\pp], \theta \in \Theta_{uj}, \eta \in \mT_{uj} }  \Big | \Ep_P [\psi_{u j}(W, \theta, \eta)] - \Ep_P[ \psi_{uj}(W, \theta, \eta_{uj} )]  \Big |.
\end{align*}
To bound $I_{1b}$, we employ Taylor's expansion:
\begin{align*}
I_{1b} & \leq   \sup_{ \uu \in \UU, j \in [\pp], \theta \in \Theta_{uj}, \eta \in \mT_{uj}, r \in [0,1) }   \partial_r  \Ep_P \Big [  \psi_{\uu j} ( W, \theta, \eta_{uj}+ r ( \eta - \eta_{uj}) )  \Big ]\\
& \leq B_{1n}  \sup_{\uu \in \UU, j\in [\pp], \eta \in \mT_{uj}} \| \eta - \eta_{uj}\|_e\leq B_{1 n}\tau_n,
%& \leq&   \|B\|_{P,2} \max_{u \in \mU, h \in \mathcal{H}_{ujn}, m \in [d_t]} \| h_m - h_{um}\|_{P,2},
\end{align*}
by  Assumptions \ref{ass: S1}(v) and \ref{ass: AS}(ii).

To bound $I_{1a}$, we apply the maximal inequality of Lemma \ref{lemma:CCK}
to the class $\mathcal{F}_1$ defined in Assumption \ref{ass: AS} to conclude that with probability $1-o(1)$,
\begin{equation}\label{eq: thm2.1-step2}
I_{1a}  \lesssim   n^{-1/2}  \Big (    \sqrt{ v_n \log a_n } + n^{-1/2+1/q} v_n K_n  \log a_n  \Big).
\end{equation}
Here we used: $\log \sup_{Q} N(\epsilon \| F_1 \|_{Q,2}, \mF_1,  \| \cdot \|_{Q,2}) \leq v_n \log (a_n/\epsilon)$ for all $0<\epsilon\leq 1$ with $\|F_1\|_{P, q} \leq K_n$ by Assumption \ref{ass: AS}(iv); $\sup_{f \in \mathcal{F}_1} \| f\|^2_{P,2} \leq C_0$ by Assumption \ref{ass: AS}(v); $a_n \geq n\vee K_n$ and $v_n\geq 1$ by the choice of $a_n$ and $v_n$. In turn, the right-hand side of (\ref{eq: thm2.1-step2}) is bounded from above by $O(\tau_n)$ by Assumption \ref{ass: AS}(vi) since $(v_n\log a_n/n)^{1/2}\lesssim \tau_n$ and
$$
n^{-1/2} n^{-1/2+1/q}v_nK_n\log a_n\lesssim n^{-1/2}\delta_n \lesssim n^{-1/2}\lesssim \tau_n.
$$
Combining presented bounds gives the claim of this step.

%\medskip
{\bf Step 3.} (Linearization) Here we prove the claim of the theorem. Fix $u\in\UU$ and $j\in[\pp]$. By definition of $\check\theta_{u j}$, we have
\begin{equation}\label{eq: thm 2.1 step 3 null}
\sqrt{n} \Big|  \En [\psi_{u j}(W, \check \theta_{uj}, \hat \eta_{uj} ) ]\Big| \leq \inf_{\theta \in \Theta_{uj}} \sqrt{n} \Big|  \En [\psi_{u j}(W, \theta, \hat \eta_{uj} ) ]\Big|+ \epsilon_n\sqrt n.
\end{equation}
Also, for any $\theta\in\Theta_{u j}$ and $\eta\in\mathcal T_{u j}$, we have
{\small
\begin{align}
&\sqrt n \En[\psi_{u j}(W,\theta,\eta)] = \sqrt n\En[\psi_{u j}(W,\theta_{u j},\eta_{u j})] - \mathbb G_n \psi_{u j}(W,\theta_{u j},\eta_{u j})\notag\\
&\qquad  - \sqrt n\Big(\Ep_P[\psi_{u j}(W,\theta_{u j},\eta_{u j})] - \Ep_P[\psi_{u j}(W,\theta,\eta)]\Big) + \mathbb G_n\psi_{u j}(W,\theta,\eta).\label{eq: thm 2.1 step 3 first}
\end{align}}\!Moreover, by Taylor's expansion of the function $r\mapsto \Ep_P[\psi_{u j}(W,\theta_{u j} + r(\theta - \theta_{u j}),\eta_{u j} + r(\eta - \eta_{u j}))]$,
{\small
\begin{align}
&\Ep_P[\psi_{u j}(W,\theta,\eta)] - \Ep_P[\psi_{u j}(W,\theta_{u j},\eta_{u j})]  = J_{u j}(\theta - \theta_{u j}) + \mathrm{D}_{u,j,0}[\eta - \eta_{u j}]\notag  \\
&\qquad \quad + \left.2^{-1}\partial_r^2 \Ep_P[W,\theta_{u j} + r(\theta - \theta_{u j}),\eta_{u j} + r(\eta - \eta_{u j})]\right|_{r = \bar r}\label{eq: thm 2.1 step 3 second}
\end{align}}\!for some $\bar r \in(0,1)$. Substituting this equality into \eqref{eq: thm 2.1 step 3 first}, taking $\theta = \check\theta_{u j}$ and $\eta = \hat\eta_{u j}$, and using \eqref{eq: thm 2.1 step 3 null} gives
{\small
\begin{align}
&\sqrt{n}\Big |  \En [\psi_{uj}(W, \theta_{uj}, \eta_{uj}) ]+ J_{uj}  (\check \theta_{u j} - \theta_{u j}) + \mathrm{D}_{u,j,0}[\hat \eta_{uj} -\eta_{uj} ]\Big | \notag \\
 &\qquad \leq  \epsilon_n \sqrt{n}  +   \inf_{\theta \in \Theta_{uj}} \sqrt{n} |  \En [\psi_{uj}(W, \theta, \hat \eta_{uj} )] | + |II_1(u,j)| + |II_2(u,j)|,\label{eq: theorem 2.1 second line}
\end{align}}\!where
\begin{align*}
 II_{1} (u,j) &:=     \sqrt{n} \sup_{r\in[0,1)}\left| \partial_{r}^2 \Ep_P \Big[ \psi_{uj}(W, \theta_{uj}+r(\theta-\theta_{uj}),  \eta_{uj} + r\{\eta - \eta_{uj} \}) \Big]\right|,\\
 II_{2} (u,j) &:= \Gn\Big(    \psi_{uj}(W, \theta, \eta)- \psi_{uj}(W, \theta_{uj}, \eta_{uj}) \Big)
\end{align*}
with $\theta=\check\theta_{uj}$ and $\eta=\hat\eta_{uj}$.
It will be shown in Step 4 that
\begin{equation}\label{eq: bound on II_1 and II_2}
\sup_{u\in\UU, j\in[\pp]}\Big( |II_1(u,j)| + |II_2(u,j)|\Big) = O_P(\delta_n).
\end{equation}
In addition, it will be shown in Step 5 that
\begin{equation}\label{eq: bound step 5 theorem 2.1}
\sup_{u \in \mathcal {U},j\in[\pp]}  \inf_{\theta \in \Theta_{uj}} \sqrt{n} |  \En [\psi_{uj}(W, \theta, \hat \eta_{uj} )] | = O_P(\delta_n).
\end{equation}
Moreover, $\epsilon_n \sqrt{n} = o(\delta_n)$ by construction of the estimator. Therefore, the expression in \eqref{eq: theorem 2.1 second line} is $O_P(\delta_n)$.
Also,
$
\sup_{u \in \mathcal {U}, j\in[\pp]}|\mathrm{D}_{u,j,0}[\hat \eta_{uj} -\eta_{uj} ]|=  O_P(\delta_n n^{-1/2})
$
by the near-orthogonality condition since $\hat{\eta}_{u j}\in \mT_{u j}$ for all $u\in\UU$ and $j\in[\pp]$ with probability $1-o(1)$ by Assumption \ref{ass: AS}(i). Therefore,
Assumption \ref{ass: S1}(iv) gives
$$
\sup_{u \in \mathcal{U}, j\in[\pp]}\Big | J_{uj}^{-1}  \sqrt{n}  \En [\psi_{uj}(W, \theta_{uj}, \eta_{uj}) ]+  \sqrt{n} (\check \theta_{uj} - \theta_{uj}) \Big | = O_P(\delta_n).
$$
The asserted claim now follows by dividing both parts of the display above by $\sigma_{u j}$ (under the supremum on the left-hand side) and noting that $\sigma_{u j}$ is bounded below from zero uniformly over $u\in\UU$ and $j\in[\pp]$ by Assumptions \ref{ass: AS}(iii) and \ref{ass: AS}(v).

%\medskip
{\bf Step 4.} (Bounds on $II_1(u,j)$ and $II_2(u,j)$). Here we prove \eqref{eq: bound on II_1 and II_2}.
First, with probability $1-o(1)$,
\begin{align*}
 \sup_{u \in \mathcal{U}, j \in [\pp]}|II_{1}(u,j)| & \leq   \sqrt{n}  B_{2n}  \sup_{u \in \mathcal{U}, j \in [\pp]} |\check\theta_{uj}-\theta_{uj}|^2 \vee \|\hat\eta_{uj} - \eta_{uj}\|_e^2\\
 &  \lesssim  \sqrt{n}  B_{1 n}^2 B_{2n}  \tau_n^2 \lesssim \delta_n,
\end{align*}
where the first inequality follows from Assumptions  \ref{ass: S1}(v) and \ref{ass: AS}(i), the second from Step 1 and Assumptions \ref{ass: AS}(ii) and \ref{ass: AS}(vi), and the third from Assumption \ref{ass: AS}(vi).

Second, we have with probability $1-o(1)$ that
$
\sup_{u \in \mathcal{U}, j \in [\pp]}|II_{2}(u,j)|  \lesssim  \sup_{f \in \mathcal{F}_2} | \Gn(f)|,
$
where
{\small
$$
\mathcal{F}_2 =  \Big \{ \psi_{uj}(\cdot, \theta, \eta) - \psi_{uj}(\cdot, \theta_{uj}, \eta_{uj})\colon
 u \in \mathcal{U}, j \in [\pp],
 \eta \in \mT_{uj}, |\theta-\theta_{uj}| \leq C B_{1 n} \tau_n  \Big \}
 $$}\!for sufficiently large constant $C$. To bound $\sup_{f \in \mathcal{F}_2} | \Gn(f)|$, we apply Lemma \ref{lemma:CCK}. Observe that
{\small
\begin{align*}
\sup_{f \in \mathcal{F}_2} \| f\|^2_{P,2} & \leq    \sup_{u \in \mathcal{U}, j \in [\pp], |\theta-\theta_{uj}|\leq C B_{1 n}\tau_n, \eta \in  \mT_{uj}}   \Ep_P \left[   ( \psi_{uj}(W,\theta,\eta) - \psi_{uj}(W, \theta_{uj},\eta_{uj}))^2  \right] \\
& \leq  \sup_{u \in \mathcal{U}, j \in [\pp],  |\theta-\theta_{uj}|\leq C B_{1 n}\tau_n, \eta \in  \mT_{uj}}   C_0 (|\theta-\theta_{uj}| \vee \|\eta-\eta_{uj}\|_e)^\omega  \lesssim ( B_{1 n} \tau_n)^{\omega},
\end{align*}}\!where we used Assumption \ref{ass: S1}(v) and  Assumption \ref{ass: AS}(ii). Also, observe that $( B_{1 n}\tau_n)^{\omega/2}\geq n^{-\omega/4}$ by Assumption \ref{ass: AS}(vi) since $ B_{1 n}\geq 1$. Therefore, an application of Lemma \ref{lemma:CCK} with an envelope $F_2 =  2F_1$ and $\sigma = (C B_{1 n}\tau_n)^{\omega/2}$ for sufficiently large constant $C$ gives with probability $1-o(1)$,
\begin{equation}\label{eq: thm2.1eqx}
\sup_{f \in \mathcal{F}_2} | \Gn(f)| \lesssim  ( B_{1 n} \tau_n)^{\omega/2} \sqrt{ v_n \log a_n}    +  n^{-1/2+1/q} v_n K_n \log a_n ,
\end{equation}
since $ \sup_{f \in \mF_2} |f| \leq 2 \sup_{f \in \mF_1} |f| \leq 2  F_1$ and $\|F_1\|_{P,q} \leq K_n$ by Assumption \ref{ass: AS}(iv) and
$$
\log \sup_Q   N( \epsilon \|F_2\|_{Q,2},  \mathcal{F}_{2}, \|\cdot\|_{Q,2}) \lesssim v_n \log (a_n/\epsilon), \quad \text{for all }0<\epsilon\leq 1
$$
by Lemma \ref{lemma: andrews} because $\mF_2 \subset \mF_1 - \mF_1$ for $\mF_1$ defined in Assumption \ref{ass: AS}(iv). The claim of this step now follows from an application of Assumption \ref{ass: AS}(vi) to bound the right-hand side of (\ref{eq: thm2.1eqx}).

%\medskip
{\bf Step 5.} Here we prove \eqref{eq: bound step 5 theorem 2.1}. For all $u\in\UU$ and $j\in[\pp]$, let $\bar  \theta_{uj} = \theta_{uj} - J^{-1}_{uj} \En[ \psi_{uj}(W, \theta_{uj}, \eta_{uj})]$. Then $\sup_{u \in \mathcal{U}, j\in[\pp]} | \bar  \theta_{u j} - \theta_{u j} | =O_P( \mS_n/\sqrt{n})$ since $\mS_n=\Ep_P[\sup_{u\in\UU,j\in[\pp]}|\sqrt{n}\En[\psi_{uj}(W_{u j},\theta_{uj},\eta_{uj})]|]$ and $J_{u j}$ is bounded in absolute value below from zero uniformly over $u\in\UU$ and $j\in[\pp]$ by Assumption \ref{ass: S1}(iv). Therefore, $\bar \theta_{u j} \in\Theta_{u j}$ for all $u\in\UU$ and $j\in[\pp]$ with probability $1-o(1)$ by Assumption \ref{ass: S1}(i). Hence, with the same probability, for all $u\in\UU$ and $j\in[\pp]$,
$$
\inf_{\theta \in \Theta_{uj}} \sqrt{n} \Big|  \En[ \psi_{uj}(W, \theta, \hat \eta_{uj} )] \Big| \leq \sqrt{n} \Big|  \En [\psi_{uj}(W, \bar  \theta_{uj}, \hat \eta_{uj} )] \Big|,
$$
and so it suffices to show that
\begin{equation}\label{eq: thm 2.1 step 5 first}
\sup_{u \in \mathcal {U},j\in[\pp]} \sqrt{n} \Big|  \En [\psi_{uj}(W, \bar\theta_{u j}, \hat \eta_{uj} )] \Big| = O_P(\delta_n).
\end{equation}
To prove \eqref{eq: thm 2.1 step 5 first}, for given $u\in\UU$ and $j\in[\pp]$, substitute $\theta = \bar\theta_{u j}$ and $\eta = \widehat\eta_{u j}$ into \eqref{eq: thm 2.1 step 3 first} and use Taylor's expansion in \eqref{eq: thm 2.1 step 3 second}. This gives
\begin{align*}
&\sqrt{n} \Big|  \En [\psi_{uj}(W, \bar  \theta_{uj}, \hat \eta_{uj} )] \Big|  \leq |\widetilde{II}_1(u,j)| + |\widetilde{II}_{2}(u,j)|\\
&\qquad  + \sqrt{n} \Big|  \En [\psi_{uj}(W, \theta_{uj}, \eta_{uj} ) ]+ J_{uj} (\bar  \theta_{uj} - \theta_{uj}) + {\mathrm{D}_{u,j,0}[\hat \eta_{uj} - \eta_{uj}]} \Big|,
\end{align*}
where $\widetilde{II}_1(u,j)$ and $\widetilde{II}_2(u,j)$ are defined as $II_1(u,j)$ and $II_2(u,j)$ in Step 3 but with $\check\theta_{u j}$ replaced by $\bar\theta_{u j}$. Then, given that $\sup_{u \in \mathcal{U}, j\in[\pp]} | \bar  \theta_{u j} - \theta_{u j} | \lesssim \mS_n\log n/\sqrt n$ with probability $1-o(1)$, the argument in Step 4 shows that
$$
\sup_{u\in\UU,j\in[\pp]} \Big( |\widetilde{II}_1(u,j)| + |\widetilde{II}_2(u,j)|\Big) = O_P(\delta_n).
$$
In addition,
$
 \En [\psi_{uj}(W, \theta_{uj}, \eta_{uj} ) ]+ J_{uj} (\bar  \theta_{uj} - \theta_{uj}) = 0
$
by the definition of $\bar  \theta_{uj}$ and  $\sup_{u\in\UU, j\in[\pp]}|\mathrm{D}_{u,j,0}[\hat \eta_{uj} - \eta_{uj}]| = O_P(\delta_n n^{-1/2})$ by the near-orthogonality condition. Combining these bounds gives \eqref{eq: thm 2.1 step 5 first}, so that the claim of this step follows, and completes the proof of the theorem.
\qed

\section{Remaining Proofs for Section 2}
See Supplementary Material.

\section{Proofs for Sections 3 and 4}

See Supplementary Material.

\bibliographystyle{plain}

\pagebreak

\begin{center}{Supplementary Material}
\end{center}

\section{Notation}\label{subsec:notation}

Throughout the paper, the symbols $\Pr$ and $\Ep$ denote probability and expectation operators with respect to a generic probability measure. If we need to signify the dependence on a probability measure  $P$, we use $P$ as a subscript in $\Pr_P$ and $\Ep_P$. Note also that we use capital letters such as $W$ to denote random elements and use the corresponding lower case letters such as $w$ to denote fixed values that these random elements can take.  For a positive integer $k$, $[k]$ denotes the set $\{1,\ldots, k\}.$

We denote by $\mathbb{P}_{n}$ the (random) empirical probability measure that assigns probability $n^{-1}$ to each $W_{i} \in (W_{i})_{i=1}^n$. $\En$ denotes the expectation with respect to the empirical measure, and
$\mathbb G_n = \mathbb{G}_{n,P}$ denotes the empirical process $\sqrt{n}(\En - \Ep_P)$, that is,
\begin{align*}
&\mathbb{G}_{n,P}(f) = \mathbb{G}_{n,P}(f(W))  = n^{
-1/2} \sum_{i=1}^n \{f(W_i) - \Ep_P[f(W)]\},  \\
&\Ep_P [f (W)] := \int f(w) dP(w),
\end{align*}
indexed by a class of measurable functions $\mathcal{F}\colon \mathcal{W} \to \mathbb{R}$; see \cite[chap. 2.3]{vdV-W}. In what follows, we use $\|\cdot\|_{P,q}$ to denote the $L^q(P)$ norm; for example, we use $\|f(W)\|_{P,q} = (\int |f(w)|^q d P(w))^{1/q}$ and $\| f(W)\|_{\Pn,q} = (n^{-1} \sum_{i=1}^n |f(W_i)|^q)^{1/q}$. For a vector $v = (v_1,\ldots,v_p)'\in \mathbb{R}^p$, $\|v\|_0$ denotes the $\ell_0$-``norm" of $v$, that is, the number of non-zero components  of $v$,  $\|v\|_1$ denotes the $\ell_1$-norm of $v$, that is, $\|v\|_1 = |v_1| + \cdots + |v_p| $,  and $\|v\|$ denotes the Euclidean norm of $v$, that is, $\|v\| = \sqrt{v'v}$.

We say that a class of functions $ \mathcal{F}= \{f(\cdot,t)\colon t \in T\}$, where $f\colon \mathcal{W} \times T \to \Bbb{R}$, is suitably measurable if it is an image admissible Suslin class, as defined in \cite{Dudley99}, p 186. In particular, $\mathcal{F}$ is suitably measurable if $f\colon \mathcal{W}\times T \to \Bbb{R}$ is measurable and $T$ is a Polish space equipped with its Borel $\sigma$-field, see \cite{Dudley99}, p 186.

\section{Empirical Application: US Presidential Approval Ratings}

In this section we illustrate the applicability of the tools proposed in this work with data on US presidential approval ratings used in \cite{berlemann2014unraveling}. There several economic and political factors that impact presidential approval ratings. In this illustration, we are interested on the impact of unemployment rates and on the impact of time in office on the approval rate of a sitting president. However, the impact of such factors might not be homogeneous and in fact depend on current ratings. For example, a sitting president is likely to have a fraction of voters who would support him regardless of economic factors. Thus, a low unemployment rate might not have an effect when approval ratings are low and have a significant effect when approval ratings are high. This would imply a different effect on different parts of the conditional distribution of the approval rating.

To study the distributional effect of these factors we use a logistic regression model with functional data as described in Section \ref{Sec:Application}. Letting $Y$ denote the approval rating, we define $Y_u = 1\{ Y \leq u \}$ to be the binary variable that indicates if the approval rating is below the threshold $u\in \mathcal{U}=[.45,.65]$. For each level of approval rating $u\in \mathcal{U}$ we estimate the model
$$ \Ep[ Y_u \mid D_{unemp}, D_{time}, X ] = \Lambda( \theta_{u,unemp}D_{unemp} + \theta_{u,time}D_{time}+X'\beta_u) + r_u $$
where $D_{unemp}$ denotes the unemployment rate, $D_{time}$ the number of months the president has been in office, $r_u$ a (small) approximation error, and $X$ denotes several additional control variables. In addition to the linear terms for the variables\footnote{Those include dummy variables for each president, Watergate scandal, causalities in different wars, political shocks, and other variables, see  \cite{berlemann2014unraveling} for a complete description.} used in \cite{berlemann2014unraveling}, we also consider interactions among controls. Therefore, for each $u\in \mathcal{U}$ we have a generalized linear model using logistic link function with 160 variables and 603 observations.

We construct simultaneous confidence bands for both coefficients ($\tilde p=2$) uniformly over $u\in\mathcal{U}$. Although $p<n$ for every $u\in\mathcal{U}$, the full model (applying logistic regression with all regressors) led to numerical instabilities and other numerical failures. We proceed to construct (asymptotically) valid confidence regions based on the double selection procedure for functional logistic regression.

Figure \ref{Fig:EmpFirst} displays the estimation results. Specifically, the figure displays point estimates of each coefficient for every $u\in\mathcal{U}$ (solid line), pointwise confidence intervals (dotted line), and uniform confidence bands (dot-dash lines). Point estimates account for model selection mistakes and are computed based the double selection procedure. For $95\%$ coverage, the pointwise critical value is taken to be $\Phi^{-1}(.975)\approx 1.96$ from the normal approximation and the critical value for uniform confidence bands (uniformly over both coefficients and over $u\in \mathcal{U}$) was calculated to be $3.003$ based on 5000 repetitions of the multiplier bootstrap.

At 95\% confidence level, the uniform confidence band rule out ``no effect" over $\UU$ for both variables. Indeed, the straight line at $0$ is not contained in the uniform confidence bands for either variable. Next consider the process of the unemployment coefficient. The analysis suggests that the unemployment rate has an overall negative effect on the approval rate of a sitting president (as the coefficient is positive increasing the probability to be below a threshold). Regarding time-in-office, the effect also seems to be predominantly negative. However, the impact is not homogeneous across $u\in\UU$. Indeed, the lowest value of the upper confidence band (0.007, $u=0.584$) is smaller than the largest value of the lower confidence band ($0.013$, $u=0.647$), see the circles in the plot of the process of the time-in-office coefficient. The impact seems to be greater for the lower and higher values of $u \in \UU$ while the effect of seems negligible in the range of $u\in[.575,.625]$. In particular, at 95\% level, no effect is ruled out for large values  of $u$.

\begin{figure}
\includegraphics[width=0.49\textwidth]{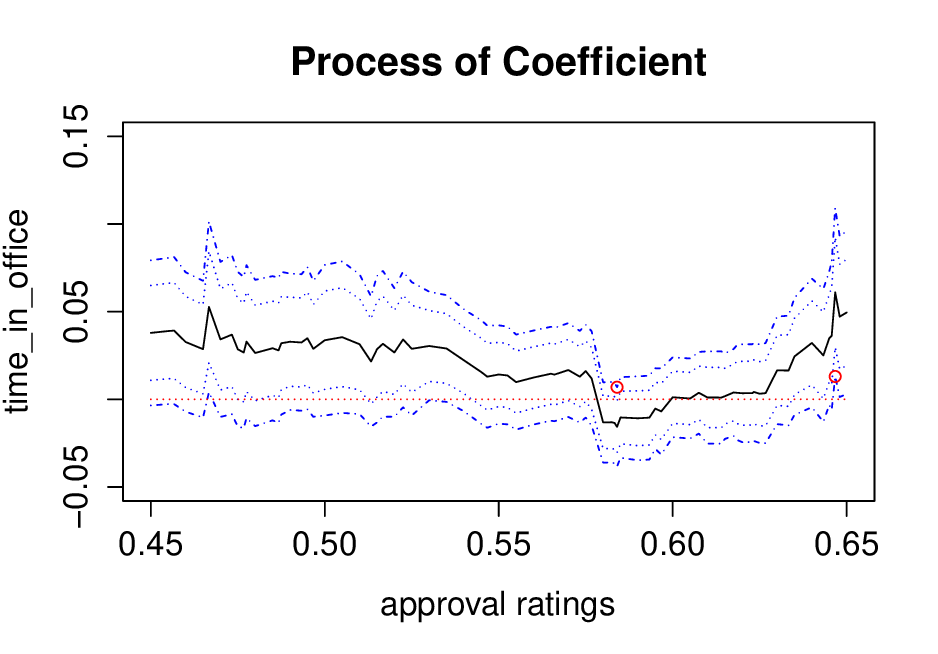}
\includegraphics[width=0.49\textwidth]{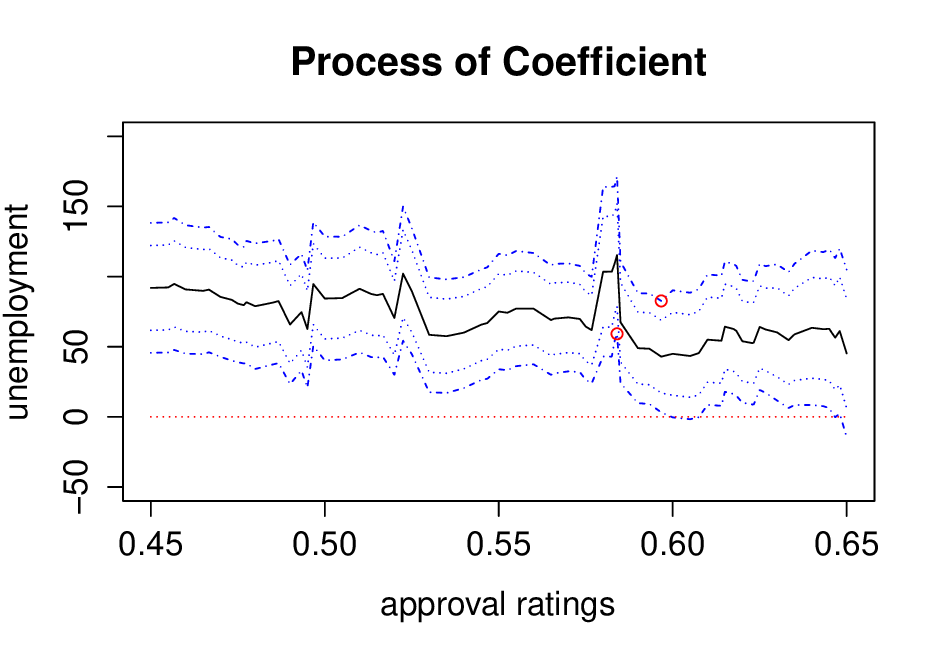}
\caption{\footnotesize The panels display pointwise (dotted) and uniform (dotdash) confidence bands for the unemployment coefficient and the time-in-office coefficient. These confidence regions are set to have 95\% coverage and were constructed based on the double selection algorithm. The critical value for the uniform confidence bands was set to $2.93$ based on the multiplier bootstrap procedure with 5000 repetitions. For both coefficients, the lowest value of the upper confidence band is smaller than the largest value of the lower confidence band.
}\label{Fig:EmpFirst}
\end{figure}

%\section{Additional Simulation Results}
%Here, we present additional simulation results discussed in Section \ref{sec: monte carlo}.
%{\tiny \begin{table}[htb]
%\begin{tabular}{c|c|cccccc}
%\hline
% \multicolumn{2}{c}{$p=2000, n=500$ } & \multicolumn{5}{c}{Rejection Frequencies  for $j \in \{6,\ldots,10\}$ } \\
%\hline
%Design  & Method  &  $j=6$  &  $j=7$ & $j=8$  &  $j=9$ & $j=10$\\
%\hline
%{1}{(i)}  & Proposed    & 0.028 & 0.044 & 0.048 & 0.042 & 0.036  \\
%          & Naive & 0.110 & 0.078 & 0.072 & 0.068 & 0.032 \\
%\hline
%{1}{(ii)} & Proposed     & 0.052 & 0.044 & 0.048 & 0.054 & 0.050 \\
%          & Naive    & 0.172 & 0.062 & 0.030 & 0.052 & 0.018 \\
%\hline
%{2}{(i)}  & Proposed & 0.040 & 0.038 & 0.052&  0.036 & 0.044 \\
%          & Naive & 0.040 & 0.064 & 0.038 & 0.038 & 0.044 \\
%\hline
%{2}{(ii)} & Proposed & 0.046 & 0.068 & 0.058 & 0.068 & 0.070  \\
%          & Naive &  0.096 & 0.028 & 0.018 &  0.016 & 0.024 \\
%\hline
%\end{tabular}\caption{We report the pointwise rejection frequencies of each method for (pointwise) confidence intervals %for each $j\in\{6,\ldots,10\}$. For Design 1 we used $\UU=\{1\}$ and for Design 2 we used $\UU=\{.5\}$. The results %are based on 500 replications. }\label{Tab:MC1}
%\end{table}}

\section{Bound on $\mS_n$ via Assumption 2.3}\label{sec:BoundSn}

%Indeed, by
By definition of $\bar\psi_{u j}(\cdot)$ in Theorem \ref{theorem:semiparametric},
$$
\psi_{u j}(\cdot,\theta_{u j},\eta_{u j}) = \bar\psi_{u j}(\cdot)\sqrt{\Ep_P[\psi_{u j}^2(W,\theta_{u j},\eta_{u j})]},\quad u\in\mathcal U, j\in[\pp].
$$
Hence, by Assumption \ref{ass: AS}(v),
\begin{align*}
\mS_n
& \leq C_0 \Ep_P\Big[ \sup_{u\in\mathcal U, j\in[\pp]}|\sqrt n\En[\bar\psi_{u j}(W)]| \Big] = C_0\Ep_P\Big[ \sup_{f\in\mathcal F_0}|\mathbb G_n f| \Big]\\
&\leq C\Big( \sqrt{\varrho_{n}\log(A_n L_n)} + \frac{\varrho_n L_n}{n^{1/2 - 1/q}}\log(A_n L_n) \Big)
\end{align*}
for some constant $C$ depending on $C_0$ only, where the second line follows from Assumption \ref{ass: OSR} and Lemma \ref{lemma:CCK}.

\section{Remaining Proofs for Section 2}

\subsection*{Proof of Corollary \ref{cor: general CLT}}
To prove the asserted claim, we will apply Lemma 2.4 in \cite{chernozhukov2012gaussian}. Denote
$$
Z_n = \sup_{u\in\mathcal{U},j\in[\tilde{p}]} \Big|n^{1/2} \sigma_{u j}^{-1}(\check{\theta}_{u j} - \theta_{u j})\Big|.
$$
Under our assumptions, $L_n^{2/7}\varrho_n\log A_n=o(n^{1/7})$, and so given that $\varrho_n \geq 1$ and $A_n\geq n$, it follows that $\log L_n \lesssim \log n$. Hence, since $ \Ep_P[\bar{\psi}_{u j}^2(W)] = 1$, Assumption \ref{ass: OSR}(i) and Corollary 2.2.8 in \cite{vdV-W} imply that
\begin{equation}\label{eq: cor2.1-0}
\Ep_P\left[\sup_{u\in\mathcal{U},j\in[\tilde{p}]}|\mathcal{N}_{u j}|\right] \lesssim \sqrt{\varrho_n \log (A_n L_n)} \lesssim \sqrt{\varrho_n \log A_n}.
\end{equation}
Further, Theorem \ref{theorem:semiparametric} shows that
\begin{equation}\label{eq: cor2.1-1}
\left|Z_n - \sup_{u\in\mathcal{U},j\in[\tilde{p}]}|\mathbb{G}_n \bar{\psi}_{u j}|\right| = O_P(\delta_n),
\end{equation}
and Theorem 2.1 in \cite{chernozhukov2015noncenteredprocesses}, together with Assumptions \ref{ass: OSR}(i,ii), shows that one can construct a version $\widetilde{Z}_n$ of $\sup_{u\in\mathcal{U},j\in[\tilde{p}]}|\mathcal{N}_{u j}|$ such that
{\small
\begin{equation}\label{eq: cor2.1-2}
\left|\sup_{u\in\mathcal{U},j\in[\tilde{p}]}|\mathbb{G}_n \bar{\psi}_{u j}| - \widetilde{Z}_n  \right| = O_P\left(\frac{L_n \varrho_n \log A_n }{n^{1/2 - 1/q}} + \frac{L_n^{1/3}(\varrho_n \log A_n)^{2/3}}{n^{1/6}}\right).
\end{equation}}\!Combining (\ref{eq: cor2.1-1}) and (\ref{eq: cor2.1-2}) gives
\begin{equation}\label{eq: cor2.1-3}
|Z_n - \widetilde{Z}_n | = O_P\left(\delta_n + \frac{L_n \varrho_n \log A_n}{n^{1/2 - 1/q}} + \frac{L_n^{1/3}(\varrho_n \log A_n)^{2/3}}{n^{1/6}}\right).
\end{equation}
Therefore, it follows from Lemma 2.4 in \cite{chernozhukov2012gaussian} that (\ref{eq: cor2.1-0}) and (\ref{eq: cor2.1-3}) imply
\begin{equation}\label{eq: cor2.1-4}
\sup_{t\in\mathbb{R}}\Big|\Pr_P(Z_n \leq t) - \Pr_P(\widetilde{Z}_n \leq t) \Big| = o(1)
\end{equation}
under our growth conditions $\delta_n^2  \rho_n\log A_n = o(1)$, $L_n^{2/7} \rho_n \log A_n =o(n^{1/7})$, and $L_n^{2/3} \rho_n\log A_n =o(n^{1/3-2/(3q)})$; note that formally their Lemma 2.4 requires $Z_n$ to be the supremum of an empirical process but this requirement is not used in the proof. The asserted claim now follows by substituting the definitions of $Z_n$ and $\widetilde{Z}_n$.
\qed

\subsection*{Proof of Corollary \ref{theorem: general bs}}
Denote
$
\widetilde{Z}_n^* = \sup_{u\in\mathcal{U},j\in[\tilde{p}]} |\hat{\mathcal{G}}_{u j}|.
$
For all $\vartheta\in(0,1)$, let $c^0_\vartheta$ be the $(1-\vartheta)$ quantile of $\sup_{u\in\mathcal{U},j\in[\tilde{p}]} |\mathcal{N}_{u j}|$.
We proceed in several steps.

%\medskip
{\bf Step 1.} Here we show that $c^0_\vartheta$ satisfies the bound
$$
c^0_\vartheta \leq \Ep_P\left[\sup_{u\in\mathcal{U},j\in[\tilde{p}]} |\mathcal{N}_{u j}|\right] + \sqrt{2\log(1/\vartheta)}
$$
for all $\vartheta\in(0,1)$. Indeed, recall that $\Ep_P[\mathcal{N}_{u j}^2] = 1$ for all $u\in\mathcal{U}$ and $j\in[\tilde{p}]$. Therefore, this bound follows from Borell's inequality; see Proposition A.2.1 in \cite{vdV-W}.

%\medskip
{\bf Step 2.} Here we show that for any $\vartheta\in(0,1)$ and $\beta\in(0,\vartheta)$,
$$
c^0_{\vartheta - \beta} - c^0_\vartheta \geq c\beta \Big. \Big/ \Ep_P\left[\sup_{u\in\mathcal{U},j\in[\tilde{p}]} |\mathcal{N}_{u j}|\right]
$$
for some absolute constant $c>0$. Indeed, given that $\Ep_P[\mathcal{N}_{u j}^2] = 1$ for all $u\in\mathcal{U}$ and $j\in[\tilde{p}]$, this bound follows from Corollary 2.1 in \cite{chernozhukov2014honestbands}.

%\medskip
{\bf Step 3.} Here we show that
\begin{equation}\label{eq: cor2.2-1}
\left| \widetilde{Z}_n^* - \sup_{u\in\mathcal{U},j\in[\tilde{p}]}\Big|\frac{1}{\sqrt{n}}\sum_{i=1}^n \xi_i\bar{\psi}_{u j}(W_i)\Big| \right| = O_P\left( \bar{\delta}_n  \sqrt{\bar{\varrho}_n\log\bar{A}_n} \right).
\end{equation}
Indeed, the left-hand side of (\ref{eq: cor2.2-1}) is bounded from above by
{\small
$$
\sup_{u\in\mathcal{U},j\in[\tilde{p}]}\left| \widehat{\mathcal{G}}_{u j} - \frac{1}{\sqrt{n}}\sum_{i=1}^n \xi_i \bar{\psi}_{u j}(W_i) \right| = \sup_{u\in\mathcal{U},j\in[\tilde{p}]}\left| \frac{1}{\sqrt{n}}\sum_{i=1}^n \xi_i (\hat{\psi}_{u j}(W_i) - \bar{\psi}_{u j}(W_i)) \right|.
$$}\!Conditional on $(W_i)_{i=1}^n$, $n^{-1/2}\sum_{i=1}^n \xi_i (\hat{\psi}_{u j}(W_i) - \bar{\psi}_{u j}(W_i))$ is zero-mean Gaussian with variance $\En[(\hat{\psi}_{u j}(W_i) - \bar{\psi}_{u j}(W_i))^2]\leq \bar{\delta}_n^2$ with probability at least $1-\Delta_n$ by Assumption \ref{ass: OSR}(iii). Thus, with the same probability,
$$
\Ep_P\left[\sup_{u\in\UU, j\in[\pp]}\left| \frac{1}{\sqrt{n}}\sum_{i=1}^n \xi_i (\hat{\psi}_{u j}(W_i) - \bar{\psi}_{u j}(W_i)) \right|\mid (W_i)_{i=1}^n \right] \lesssim  \bar{\delta}_n  \sqrt{\bar{\varrho}_n\log \bar{A}_n}
$$
by Assumption \ref{ass: OSR}(iii) and Corollary 2.2.8 in \cite{vdV-W}. Since $\Delta_n \to 0$, (\ref{eq: cor2.2-1}) follows.

%\medskip
{\bf Step 4.} Here we construct a version $\widetilde{Z}_n$ of $\sup_{u\in\mathcal{U},j\in[\tilde{p}]}|\mathcal{N}_{u j}|$ such that
{\small
$$
\left|\sup_{u\in\mathcal{U},j\in[\tilde{p}]} \Big| \frac{1}{\sqrt{n}} \sum_{i=1}^n \xi_i \bar{\psi}_{u j}(W_i)
\Big| - \widetilde{Z}_n  \right| = O_P\left(\frac{L_n \varrho_n \log A_n}{n^{1/2 - 1/q}} + \frac{L_n^{1/2}(\varrho_n \log A_n )^{3/4}}{n^{1/4}}\right).
$$}\!Indeed, as was discussed in the proof of Corollary \ref{cor: general CLT}, we have $\log L_n  \lesssim \log n$ and $\Ep_P[\bar{\psi}_{u j}^2(W_i)] = 1$. Therefore, the claim follows from Theorem 2.2 in \cite{chernozhukov2015noncenteredprocesses} combined with Assumption \ref{ass: OSR}(i,ii).

%\medskip
{\bf Step 5.} Here we show that there exists a sequence of positive constants $(\vartheta_n)_{n\geq 1}$ such that $\vartheta_n \to 0$ and $
\Pr_P\Big(c_\alpha(1+\varepsilon_n) > c^0_{\alpha - \vartheta_n}\Big) \to 0.
$
Indeed, Steps 3 and 4 imply that
$
| \widetilde{Z}_n^* - \widetilde{Z}_n  | = O_P(r_n)
$
where
$$
r_n \lesssim  \bar{\delta}_n  \sqrt{\bar{\varrho}_n \log \bar{A}} + \frac{L_n \varrho_n \log  A_n}{n^{1/2 - 1/q}} + \frac{L_n^{1/2}(\varrho_n \log A_n )^{3/4}}{n^{1/4}}.
$$
Also, as in the proof of Corollary \ref{cor: general CLT}, $\Ep_P[\widetilde{Z}_n] \lesssim (\varrho_n\log A_n)^{1/2}$. Hence, $r_n \Ep_P[\widetilde{Z}_n] = o(1)$ under our conditions, and so there exists a sequence of positive constants $(\chi_n)_{n\geq 1}$ such that $\chi_n \to \infty$ but $\chi_n r_n \Ep_P[\widetilde{Z}_n] =o(1)$. Further, let
$
\beta_n =\{ \Pr_P( | \widetilde{Z}_n^* - \widetilde{Z}_n | > \chi_n r_n)\}^{1/2}.
$
Observe that $\beta_n = o(1)$ and
$$
\Pr_P\bigg(\Pr_P\Big( | \widetilde{Z}_n^* - \widetilde{Z}_n | > \chi_n r_n \mid (W_i)_{i=1}^n\Big) > \beta_n \bigg) \leq \beta_n.
$$
The last display implies that with probability at least $1-\beta_n$,
$
c_\alpha \leq c^0_{\alpha - \beta_n} + \chi_n r_n.
$
Hence, with the same probability, using the bounds in Steps 1 and 2, we obtain for some sequence of positive constants $(\vartheta_n)_{n\geq 1}$ such that $\vartheta_n = o(1)$,
$$
c_\alpha(1+\varepsilon_n)\leq (c^0_{\alpha - \beta_n} + \chi_n r_n)(1+\varepsilon_n)\leq c^0_{\alpha -\vartheta_n}
$$
since $\chi_n r_n \Ep_P[\widetilde{Z}_n] = o(1)$ and $\varepsilon_n(\Ep_P[\widetilde{Z}_n])^2 = o(1)$ by assumption. The claim of this step follows.

%\medskip
{\bf Step 6.} Here we complete the proof. We have
\begin{align*}
 &\Pr_P\Big( \sup_{u\in\mathcal{U}, j\in[\tilde{p}]} |\sqrt{n} \hat{\sigma}_{u j} ^{-1} (\check{\theta}_{u j} - \theta_{u j})| \leq c_\alpha \Big)\\
 &\qquad = \Pr_P\Big( |\sqrt{n} \sigma_{u j} ^{-1} (\check{\theta}_{u j} - \theta_{u j})| \leq c_\alpha \hat{\sigma}_{u j} / \sigma_{u j} ,\ \forall u\in\UU,j\in[\pp]\Big)\\
& \qquad \leq \Pr_P\Big( \sup_{u\in\mathcal{U}, j\in[\tilde{p}]} |\sqrt{n} \sigma_{u j} ^{-1} (\check{\theta}_{u j} - \theta_{u j})| \leq c_\alpha (1 + \varepsilon_n) \Big) + o(1)\\
& \qquad \leq \Pr_P\Big( \sup_{u\in\mathcal{U}, j\in[\tilde{p}]} |\sqrt{n} \sigma_{u j} ^{-1} (\check{\theta}_{u j} - \theta_{u j})| \leq c^0_{\alpha  - \vartheta_n}  \Big) + o(1)\\
& \qquad = 1- \alpha + \vartheta_n + o(1) = 1-\alpha + o(1),
\end{align*}
where the third line follows by Assumption \ref{ass: variance}, the fourth by Step 5, and the fifth by Corollary \ref{cor: general CLT}. Similar arguments also give the same bound from the other side. Therefore,
\begin{equation}\label{eq: cor2.2-2}
\Pr_P\left( \sup_{u\in\mathcal{U}, j\in[\tilde{p}]} |\sqrt{n} \hat{\sigma}_{u j} ^{-1} (\check{\theta}_{u j} - \theta_{u j})| \leq c_\alpha \right) = 1-\alpha + o(1).
\end{equation}
This completes the proof.
\qed

\section{Proofs for Section 3}
In this appendix, we use $c$ and $C$ to denote strictly positive constants that depend only on $c_1$ and $C_1$ (but do not depend on $n$, $u$, $j$, or $P\in\mathcal P_n$). The values of $c$ and $C$ may change at each appearance. Also, the notation $a_n \lesssim b_n$ means that $a_n\leq C b_n$ for all $n$ and some $C$. The notation $a_n\gtrsim b_n$ means that $b_n\lesssim a_n$. Moreover, the notation $a_n = o(1)$ means that there exists a sequence $(b_n)_{n\geq 1}$ of positive numbers such that (i) $|a_n|\leq b_n$ for all $n$, (ii) $b_n$ is independent of $P\in\mathcal P_n$ for all $n$, and (iii) $b_n\to 0$ as $n\to\infty$. Finally, the notation $a_n = o_P(b_n)$ means that for any $C$, we have $\Pr_P(a_n > C b_n) = o(1)$. Using this notation allows us to avoid repeating ``uniformly over $P\in\mathcal P_n$'' and ``uniformly over $u\in\UU$ and $j\in[\pp]$'' many times in the proofs of Theorem \ref{theorem:inferenceAlg1} and Corollaries \ref{cor: logistic clt} -- \ref{cor: simple example}. %Throughout this appendix, we assume that $n\geq n_0$.

\subsection*{Proof of Theorem \ref{theorem:inferenceAlg1}}
%Below we verify Assumptions \ref{ass: S1} and \ref{ass: AS}, so that the result follows from Theorem \ref{theorem:semiparametric}. In this proof, we use $c$ and $C$ to denote strictly positive constants whose values depend only on $c_1$ and $C_1$ (but do not depend on $u$ or $j$, so that the presented inequalities hold uniformly over $u\in\UU$ and $j\in[\pp]$ where appropriate). The values of $c$ and $C$ may change at each appearance. The symbol $\lesssim$ means that the left-hand side is bounded by the right-hand side up to a constant $C$.

Observe that for all $u\in\UU$ and $j\in[\pp]$, we have $\Ep_P[|Z_u^j|^2]\lesssim 1$ by Assumptions \ref{ass: parameters} and \ref{ass: covariates}. We use this fact several times in the proof without further notice.

For $u\in\UU$ and $j\in[\pp]$, define
{\small
\begin{align*}
&T_{u j} = \Big\{\eta = (\eta^{(1)},\eta^{(2)},\eta^{(3)})\colon \eta^{(1)}\in\ell^\infty(\mathbb{R}^{\pp+p}), \eta^{(2)}\in\mathbb{R}^{\pp -1 + p}, \eta^{(3)} \in \mathbb{R}^{\pp - 1 + p} \Big\},
\end{align*}}
so that $\eta_{u j} = (r_u, \beta_u^j, \gamma_u^j) \in T_{u j}$, and $T_{u j}$ is convex. Endow $T_{u j}$ with a norm $\|\cdot\|_e$ defined by
$$
\|\eta\|_e = \sqrt{\Ep_P[\eta^{(1)}(D,X)^2]} \vee \|\eta^{(2)}\| \vee \|\eta^{(3)}\|,\quad \eta = (\eta^{(1)},\eta^{(2)},\eta^{(3)}) \in T_{u j}.
$$
Further, recall that $a_n = p\vee \pp\vee n$ and define $\tau_n=C(s_n\log a_n/n)^{1/2}$ and
{\small
\begin{multline*}
\mathcal{T}_{u j} = \{\eta_{u j}\} \cup \Big\{\eta =  (\eta^{(1)},\eta^{(2)},\eta^{(3)}) \in T_{u j}\colon \eta^{(1)} = \mathcal{O}, \|\eta^{(2)}\|_0 \vee \|\eta^{(3)}\|_0\leq Cs_n,\\
 \|\eta^{(2)} - \beta_{u}^j\| \vee \|\eta^{(3)} - \gamma_u^j\| \leq \tau_n, \|\eta^{(3)} - \gamma_u^j\|_1\leq C\sqrt{s_n}\tau_n\Big\}
\end{multline*}}
for sufficiently large $C$.

First, we verify Assumption \ref{ass: S1}(i). To bound $\mS_n$, below we establish the following inequality: %uniformly over $u_1,u_2\in\UU$ and $j\in[\pp]$,
\begin{equation}\label{eq: gamma-lipshitz}
\|\gamma_{u_2}^j - \gamma_{u_1}^j\|_1\lesssim \sqrt{p+\pp}|u_2 - u_1|.
\end{equation}
Recall that
$$
f_u^2=f_u^2(D,X) = \text{Var}_P(Y_u\mid D,X)=\Ep_P[Y_u\mid D,X](1-\Ep_P[Y_u\mid D,X]),
$$
and so
\begin{equation}\label{eq: f-lipshitz}
|f_{u_2}^2 - f_{u_1}^2|\leq \Big|\Ep_P[Y_{u_2}\mid D,X] - \Ep_P[Y_{u_1}\mid D,X]\Big|\lesssim |u_2 - u_1|
\end{equation}
by Assumption \ref{ass: density}. In addition,
\begin{align}
&\Ep_P\Big[f_{u_2}^2\Big(X^j(\gamma_{u_1}^j-\gamma_{u_2}^j)\Big)^2\Big]\notag\\
& \qquad =\Ep_P\Big[f_{u_2}^2\Big(X^j(\gamma_{u_1}^j - \gamma_{u_2}^j)\Big)\Big(D_j-X^j\gamma_{u_2}^j-(D_j-X^j\gamma_{u_1}^j)\Big)\Big]\notag\\
& \qquad = -\Ep_P\Big[f_{u_2}^2\Big(X^j(\gamma_{u_1}^j - \gamma_{u_2}^j)\Big)(D_j-X^j\gamma_{u_1}^j)\Big] \notag\\
& \qquad = -\Ep_P\Big[(f_{u_2}^2-f_{u_1}^2)\Big(X^j(\gamma_{u_1}^j - \gamma_{u_2}^j)\Big)(D_j-X^j\gamma_{u_1}^j)\Big]\notag \\
& \qquad =-\Ep_P\Big[(f_{u_2}^2-f_{u_1}^2)\Big(X^j(\gamma_{u_1}^j - \gamma_{u_2}^j)\Big)Z_{u_1}^j\Big],\label{eq: gamma-lipshitz-derivation}
\end{align}
where the first line follows from adding and subtracting $D_j$, the second from the equality $\Ep_P[f_{u_2}^2(D_j-X^j\gamma_{u_2}^j)X^j]=0$, the third from the equality $\Ep_P[f_{u_1}^2(D_j-X^j\gamma_{u_1}^j)X^j]=0$, and the fourth from $D_j - X^j\gamma_{u_1}^j = Z_{u_1}^j$. Now, by the Cauchy-Schwarz inequality, the expression in (\ref{eq: gamma-lipshitz-derivation}) is bounded in absolute value by
\begin{align*}
&\Big(\Ep_P\Big[\Big(X^j(\gamma_{u_1}^j-\gamma_{u_2}^j)\Big)^2\Big]\cdot\Ep_P\Big[(f_{u_1}^2-f_{u_2}^2)^2(Z_{u_1}^j)^2\Big]\Big)^{1/2}\\
 &\qquad \lesssim \Big(\Ep_P\Big[f_{u_2}^2\Big(X^j(\gamma_{u_1}^j-\gamma_{u_2}^j)\Big)^2\Big]\cdot\Ep_P\Big[(f_{u_1}^2-f_{u_2}^2)^2(Z_{u_1}^j)^2\Big]\Big)^{1/2}\\
&\qquad \lesssim |u_2-u_1|\Big(\Ep_P\Big[f_{u_2}^2\Big(X^j(\gamma_{u_1}^j-\gamma_{u_2}^j)\Big)^2\Big]\Big)^{1/2},
\end{align*}
where the second line follows from $\Ep_P[(X^j(\gamma_{u_1}^j-\gamma_{u_2}^j))^2]\lesssim \|\gamma_{u_1}^j - \gamma_{u_2}^j\|^2\lesssim \Ep_P[f_{u_2}^2(X^j(\gamma_{u_1}^j-\gamma_{u_2}^j))^2]$, which holds by Assumption \ref{ass: covariates}, and the third line follows from (\ref{eq: f-lipshitz}). Hence,
$
(\Ep_P[f_{u_2}^2(X^j(\gamma_{u_1}^j-\gamma_{u_2}^j))^2])^{1/2}\lesssim |u_2-u_1|,
$
and so
\begin{equation}\label{eq: gamma-lipshitz2}
\|\gamma_{u_2}^j - \gamma_{u_1}^j\| \lesssim \Big(\Ep_P\Big[f_{u_2}^2\Big(X^j(\gamma_{u_1}^j-\gamma_{u_2}^j)\Big)^2\Big]\Big)^{1/2}\lesssim |u_2 - u_1|.
\end{equation}
Therefore,
$$
\|\gamma_{u_2}^j - \gamma_{u_1}^j\|_1 \leq \sqrt{p+\pp}\|\gamma_{u_2}^j - \gamma_{u_1}^j\|\lesssim \sqrt{p+\pp}|u_2-u_1|,
$$
and so (\ref{eq: gamma-lipshitz}) follows.

%and we obtain $\Ep_P[f_{u'}^2\{X^j(\gamma_u^j-\gamma_{u'}^j)\}^2]^{1/2} \leq C'\Ep_P[(f_u^2-f_{u'}^2)^2(Z_u^j)^2]^{1/2} $. Together with (\ref{Lipsf}), we have that (\ref{LipsGamma}) holds with  $L_\gamma := C' L_Y\sup_{u\in\UU,j\in[\pp]}\Ep_P[(Z_u^j)^2]^{1/2}\leq C''$ by Condition lL.

%By Condition lL we have $\Ep_P[\{X^j(\gamma_u^j-\gamma_{u'}^j)\}^2]^{1/2}\leq (C/c)\Ep_P[f_{u}^2\{X^j(\gamma_u^j-\gamma_{u'}^j)\}^2]^{1/2}$. Moreover

Next, let
\begin{align*}
&\mG_1 = \Big\{(Y,D,X)\mapsto 1\{ Y\leq u\bar y + (1-u)\underline{y}\}\colon u\in \UU\Big\},\\
& \mG_2 = \Big\{(Y,D,X)\mapsto \Ep_P[g(Y,D,X) \mid D,X]\colon g\in \mG_1\Big\},\\
&\mG_{3,j} = \Big\{(Y,D,X)\mapsto D_j - X^j\gamma_u^j\colon u\in \UU\Big\},\quad j\in[\pp].
\end{align*}
Then the function class $\widetilde{\mathcal{F}} = \{\psi_{u j}(\cdot,\theta_{u j},\eta_{u j})\colon u\in\UU, j\in[\pp]\}$ satisfies
$$
\widetilde{\mathcal F}\subset (\mG_1 - \mG_2)\cdot (\cup_{j\in[\pp]}\mG_{3j}).
$$
Observe that $\mG_1$ is a VC-subgraph class with index bounded by $C$, and so by Theorem 2.6.7 in \cite{vdV-W}, its uniform entropy numbers obey
\begin{equation}\label{eq: G1 entropy bound}
\sup_Q  \log N(\epsilon \|\widetilde F_1\|_{Q,2}, \mG_1, \| \cdot \|_{Q,2}) \leq C\log (C/\epsilon), \quad \text{for all } 0<\epsilon\leq 1,
\end{equation}
where $\widetilde F_1 \equiv 1$ is its envelope. In addition, Lemma \ref{Lemma:PartialOutCovering} implies that the uniform entropy numbers of $\mG_2$ obey the same inequalities with the same envelope $\widetilde F_1$ (but possibly different constant $C$). Moreover, for any $u\in\UU$ and $j\in[\pp]$, by Assumptions \ref{ass: parameters} and \ref{ass: sparsity} and the triangle inequality,
\begin{align}
\|\gamma_u^j\|_1& \lesssim \|\bar \gamma_u^j\|_1 + s_n\sqrt{\log a_n/n}\lesssim \sqrt{s_n}\|\bar \gamma_u^j\| +  s_n\sqrt{\log a_n/n} \notag\\
& \lesssim \sqrt{s_n}\|\gamma_u^j\| + s_n\sqrt{\log a_n/n}\lesssim \sqrt{s_n},\label{eq: thm 3.1 gamma bound}
\end{align}
because by Assumption \ref{ass: covariates}, $s_n\log a_n/n = o(1)$.
Therefore, (\ref{eq: gamma-lipshitz}) and Lemma \ref{lem: linear classes} with $k=1$ imply that for all $j\in[\pp]$, the uniform entropy numbers of $\mG_{3,j}$ obey
$$
\sup_Q  \log N(\epsilon \|\widetilde F_3\|_{Q,2}, \mG_{3,j}, \| \cdot \|_{Q,2}) \leq C\log (a_n/\epsilon), \quad \text{for all } 0<\epsilon\leq 1,
$$
where $\widetilde F_3(Y,D,X) = \sup_{u\in\UU,j\in[\pp]}|Z_u^j|+ M_{n,2}^{-1}(\|D\|_\infty\vee \|X\|_\infty)$ is its envelope, and so Lemma \ref{lemma: andrews} gives that the uniform entropy numbers of $\cup_{j\in[\pp]}\mG_{3,j}$ obey the same inequalities with the same envelope $\widetilde F_3$.
Hence, Lemma \ref{lemma: andrews} also shows that the uniform entropy numbers of $\widetilde {\mathcal F}$ obey
\begin{equation}\label{eq: uniform entropy numbers f tilde}
\sup_Q  \log N(\epsilon \|\widetilde F\|_{Q,2}, \widetilde{\mathcal F}, \| \cdot \|_{Q,2}) \leq C\log (a_n/\epsilon), \quad \text{for all } 0<\epsilon\leq 1,
\end{equation}
where $\widetilde F(Y,D,X) = C\{\sup_{u\in\UU,j\in[\pp]}|Z_u^j|+ M_{n,2}^{-1}(\|D\|_\infty \vee \|X\|_\infty)\}$ is its envelope. Now observe that $\|\widetilde F\|_{P,q}\lesssim M_{n,1}$ and that $\|f\|_{P,2}\lesssim 1$ uniformly over $f\in\widetilde{\mathcal F}$ by Assumption \ref{ass: covariates}. Therefore, it follows from Lemma \ref{lemma:CCK} that
\begin{align*}
\mS_n&= \Ep_P\Big[\sup_{\uu \in \UU, j\in [\pp]}\Big|\sqrt{n}\En[\psi_{uj}(W, \theta_{uj}, \eta_{uj} )]\Big|\Big]\\
&\lesssim \log^{1/2}(a_n M_{n,1}) + n^{-1/2+1/q}M_{n,1} \log(a_n M_{n,1})\\
&\lesssim \log^{1/2}(a_n M_{n,1})(1+\delta_n)\lesssim \sqrt{\log a_n},
\end{align*}
where the last two inequalities follow from Assumption \ref{ass: covariates} and the facts that $\delta_n = o(1)$ and that $\log M_{n,1}\lesssim \log n$, which is another consequence of Assumption \ref{ass: covariates}. Hence, Assumption \ref{ass: parameters} implies that for all $u\in\UU$ and $j\in [\pp]$, $\Theta_{u j}$ contains a ball of radius $C_0 n^{-1/2}\mS_n \log n$ centered at $\theta_{u j}$ for all sufficiently large $n$ for any constant $C_0$. Therefore, Assumption \ref{ass: S1}(i) holds.

Next, Assumption \ref{ass: S1}(ii) follows from the observation that for all $u\in\UU$ and $j\in[\pp]$, the map $(\theta,\eta)\mapsto \psi_{u j}(W,\theta,\eta)$ is twice continuously Gateaux-differentiable on $\Theta_{u j}\times T_{u j}$,  and so is the map $(\theta,\eta)\mapsto \Ep_P[\psi_{u j}(W,\theta,\eta)]$.

To verify the orthogonality condition in Assumption \ref{ass: S1}(iii), note that for all $u\in\UU$, $j\in[\pp]$, and $\eta = (\eta^{(1)},\eta^{(2)},\eta^{(3)})\in\mT_{u j}$ with $\eta\neq \eta_{u j}$, we have
$$
\mathrm{D}_{\uu,j,0}[\eta - \eta_{uj}] = \Ep_P\Big[r_{u}Z_u^j - \Lambda'(D_j \theta_{u j} + X^j\beta_u^j)Z_u^jX^j(\eta^{(2)} - \beta_u^j)\Big],
$$
where we used the equality $\Ep_P[\{Y_u-\G(D_j \theta_{u j}+X^j\beta_u^j)-r_u\}X^j]=0$. In addition,
$|\Ep_P[r_u Z_u^j]|\leq \delta_n n^{-1/2}$ by Assumption \ref{ass: approximation error}. Further, recall that $\Ep_P[f_u^2Z_u^j X^j]=0$ and observe that
\begin{align*}
f_u^2
&= f_u^2(D,X) = \text{Var}(Y_u\mid D,X) = \Ep_P[Y_u\mid D,X](1 - \Ep_P[Y_u\mid D,X])\\
&=(\Lambda(D'\theta_u + X'\beta_u) + r_u) (1 - \Lambda(D'\theta_u + X'\beta_u) - r_u)\\
&=\Lambda'(D'\theta_u + X'\beta_u) + r_u - r_u^2 - 2r_u\Lambda(D'\theta_u + X'\beta_u),
\end{align*}
where we used the equality $\Lambda'(t) = \Lambda(t) - \Lambda^2(t)$, which holds for all $t\in \mathbb R$. Hence,
\begin{align*}
&\Big|\Ep_P\Big[\Lambda'(D_j \theta_{u j} + X^j\beta_u^j)Z_u^jX^j(\eta^{(2)} - \beta_u^j)\Big]\Big|\\
&\qquad \lesssim \Big(\Ep_P\Big[(r_u Z_u^j)^2\Big]\cdot \Ep_P\Big[\Big(X^j(\eta^{(2)} - \beta_u^j)\Big)^2\Big]\Big)^{1/2}\\
&\qquad \lesssim \Big(\Ep_P[r_u^2]\Big)^{1/2}\cdot \|\eta^{(2)} - \beta_u^j\| \lesssim s_n\log a_n/n\lesssim \delta_n n^{-1/2},
\end{align*}
where the second line follows from the Cauchy-Schwarz inequality and the observations that $|r_u|\leq 1$ and that $|\Lambda(t)|\leq 1$ for all $t\in\mathbb R$, and the third line from Assumptions \ref{ass: covariates} and \ref{ass: approximation error} (the last inequality holds because $s_n^2\log^2 a_n \leq \delta_n^2 n$ by Assumption \ref{ass: covariates}). Also, when $\eta = \eta_{u j}$, we have $|\mathrm{D}_{\uu,j,0}[\eta - \eta_{uj}]| = 0$, and so Assumption \ref{ass: S1}(iii) holds.

Next, we verify Assumption \ref{ass: S1}(iv). Fix $u\in \UU$ and $j\in[\pp]$. Observe that
\begin{align*}
J_{u j}
& = -\Ep_P\Big[\Lambda'(D_j \theta_{u j} + X^j \beta_u^j)D_j Z_u^j\Big]\\
& = -\Ep_P\Big[f_u^2 D_j Z_u^j \Big] + \Ep_P\Big[(r_u -r_u^2 - 2r_u \Lambda(D'\theta_u + X'\beta_u)) D_j Z_u^j\Big]\\
&= -\Ep_P\Big[f_u^2 |Z_u^j|^2 \Big] + \Ep_P\Big[(r_u -r_u^2 - 2r_u \Lambda(D'\theta_u + X'\beta_u)) D_j Z_u^j\Big].
\end{align*}
Hence, by Assumptions \ref{ass: parameters}, \ref{ass: covariates} and \ref{ass: approximation error} and the Cauchy-Schwarz inequality,
$$
|J_{u j}|\geq c_1 - 4\Big(\Ep_P[r_u^2]\Ep_P[|D_j Z_u^j|^2]\Big)^{1/2}=c_1 + o(1),
$$
 and also $|J_{u j}|\lesssim 1$ uniformly over $u\in\UU$ and $j\in[\pp]$. In addition,
$$ \Ep_P[\psi_{uj}(W,\theta,\eta_{uj})] = J_{uj}(\theta-\theta_{uj})+\frac{1}{2}\partial_\theta^2\Big.\Big\{\Ep_P[\psi_{uj}(W,\theta,\eta_{uj})]\Big\}\Big|_{\theta = \bar\theta}(\theta-\theta_{uj})^2 $$
for some $\bar\theta\in\Theta_{u j}$. Moreover, for all $\theta\in\Theta_{u j}$, we have $|\partial_\theta^2\Ep_P[\psi_{uj}(W,\theta,\eta_{uj})]|\leq \Ep_P[|D_j^2 Z_u^j|] \lesssim 1$ by Assumptions \ref{ass: parameters} and \ref{ass: covariates} since $|\Lambda''(t)|\leq 1$ for all $t\in\mathbb R$. These inequalities together imply Assumption \ref{ass: S1}(iv).

Next, we verify Assumption \ref{ass: S1}(v) with $\omega = 2$ and $ B_{1 n} =  B_{2 n} = C$ for sufficiently large $C$. Fix $u\in\UU$, $j\in[\pp]$, $r\in(0,1]$, $\theta\in\Theta_{u j}$, and $\eta=(\eta^{(1)},\eta^{(2)},\eta^{(3)})\in \mathcal{T}_{u j}$. We consider the case $\eta\neq \eta_{u j}$, and the other case is similar. Denote
\begin{align*}
&I_{1,1} = 2|X^j(\eta^{(3)} - \gamma_u^j)| + |r_u Z_u^j|,\\
&I_{1,2} =  \Big| D_j(\theta - \theta_{u j}) +  X^j (\eta^{(2)} - \beta_u^j)\Big| \cdot |Z_u^j|.
\end{align*}
Then
$$
\Big|\psi_{u j}(W,\theta,\eta) - \psi_{u j}(W,\theta_{u j},\eta_{u j})\Big| \leq I_{1,1} + I_{1,2},
$$
since $|\Lambda'(t)|\leq 1$ for all $t\in\mathbb R$.
In addition,
\begin{align*}
&\Ep_P[(r_u Z_u^j)^2] \lesssim \Ep_P[r_u^2]\leq \|\eta - \eta_{u j}\|_e^2,\\
& \Ep_P[(X^j(\eta^{(3)} - \gamma_u^j))^2]\lesssim \|\eta^{(3)} - \gamma_u^j\|^2\leq \|\eta - \eta_{u j}\|_e^2
\end{align*}
by Assumptions \ref{ass: approximation error} and \ref{ass: covariates}, respectively. Thus, $\Ep_P[I_{1,1}^2]\lesssim \|\eta - \eta_{u j}\|_e^2$. Also,
\begin{align*}
\Ep_P[I_{1,2}^2]
&\leq \Ep_P\Big[(Z_u^j)^2\Big] \cdot \Ep_P \Big[( D_j(\theta - \theta_{u j})  +  X^j (\eta^{(2)} - \beta_u^j))^2\Big]\\
&\lesssim \Ep_P \Big[( D_j(\theta - \theta_{u j})  +  X^j (\eta^{(2)} - \beta_u^j))^2\Big]\\
&\lesssim  |\theta - \theta_{u j}|^2 + \|\eta^{(2)} - \beta_u^j\|^2 \lesssim |\theta - \theta_{u j}|^2 + \|\eta - \eta_{u j}\|_e^2,
\end{align*}
where the first line follows from the Cauchy-Schwarz inequality, the second from $\Ep_P[(Z_u^j)^2]\lesssim 1$, and the third from Assumption \ref{ass: covariates}. Therefore, Assumption \ref{ass: S1}(v-a) holds.

To verify Assumption \ref{ass: S1}(v-b), observe that under our conditions,
$$
\partial_r \Ep_P\Big[\psi_{u j}(W,\theta,\eta_{u j} + r (\eta - \eta_{u j}))\Big] = \Ep_P\Big[\partial_r\psi_{u j}(W,\theta,\eta_{u j} + r (\eta - \eta_{u j}))\Big].
$$
Further, denote
\begin{align*}
x_r &= D_j \theta + X^j \beta_u^j + r X^j(\eta^{(2)} - \beta_u^j),\\
I_{2,1} & = -X^j (\eta^{(3)} - \gamma_u^j)(Y_u - \Lambda(x_r) - (1-r)r_u),\\
I_{2,2} & = r_u(Z_u^j -r X^j(\eta^{(3)} - \gamma_u^j)),\\
I_{2,3} & = - \Lambda'(x_r)(Z_u^j - r X^j(\eta^{(3)} - \gamma_u^j))X^j(\eta^{(2)} - \beta_u^j).
\end{align*}
Then
$
\partial_r \psi_{u j}(W,\theta,\eta_{u j} + r (\eta - \eta_{u j})) = I_{2,1} + I_{2,2} + I_{2,3},
$
and so
$$
\Ep_P\Big[\partial_r\psi_{u j}(W,\theta,\eta_{u j} + r (\eta - \eta_{u j}))\Big] = \Ep_P[I_{2,1}] + \Ep_P[I_{2,2}] + \Ep_P[I_{2,3}].
$$
Now, observe that
\begin{align*}
\Ep_P[|I_{2,1}|]
&\lesssim \Ep_P\Big[|X^j(\eta^{(3)} - \gamma_u^j)|\Big] \\
& \lesssim \Big(\Ep_P\Big[|X^j(\eta^{(3)} - \gamma_u^j)|^2\Big]\Big)^{1/2} \lesssim \|\eta - \eta_{u j}\|_e,
\end{align*}
where the first inequality holds since $|r_u|\leq 1$, the second by Jensen's inequality, and the third by Assumption \ref{ass: covariates}. Also, by the Cauchy-Schwarz inequality,
$$
\Ep_P[|I_{2,2}|]\leq \Big(\Ep_P[r_u^2]\cdot \Ep_P\Big[(Z_{u j} - r X^j(\eta^{(3)} - \gamma_u^j))^2\Big]\Big)^{1/2}\lesssim \|\eta - \eta_{u j}\|_e,
$$
where the second inequality follows from $(\Ep_P[r_u^2])^{1/2}\leq \|\eta - \eta_{u j}\|_e$. Moreover, since $|\Lambda'(t)|\leq 1$ for all $t\in \mathbb{R}$, the Cauchy-Schwarz inequality gives
\begin{align*}
\Ep_P[|I_{2,3}|]
&\lesssim \Big(\Ep_P\Big[(Z_u^j - r X^j(\eta^{(3)} - \gamma_u^j))^2\Big] \\
&\qquad \times \Ep_P\Big[(X^j(\eta^{(2)} - \beta_u^j) )^2\Big]\Big)^{1/2}\lesssim \|\eta - \eta_{u j}\|_e.
\end{align*}
 Therefore, Assumption \ref{ass: S1}(v-b) holds.

To verify Assumption \ref{ass: S1}(v-c), denote
\begin{align*}
I_{3,1} & = -r_u X^j(\eta^{(3)} - \gamma_u^j) + \Lambda'(x_r) X^j(\eta^{(3)} - \gamma_u^j)X^j(\eta^{(2)} - \beta_u^j),\\
I_{3,2} & = -r_u X^j(\eta^{(3)} - \gamma_u^j),\\
I_{3,3} & = -\Lambda''(x_r)(Z_u^j - r X^j(\eta^{(3)} - \gamma_u^j))(X^j(\eta^{(2)} - \beta_u^j))^2\\
&\quad +  \Lambda'(x_r) X^j(\eta^{(3)} - \gamma_u^j)X^j(\eta^{(2)} - \beta_u^j),
\end{align*}
so that $\partial_r I_{2,1} = I_{3,1}$, $\partial_r I_{2,2} = I_{3,2}$, and $\partial_r I_{2,3} = I_{3,3}$. Now, observe that since $|\Lambda'(t)|\leq 1$ for all $t\in\mathbb{R}$,
$$
\Ep_P[|I_{3,1}|]\lesssim \sqrt{\Ep_P[r_u^2]}\|\eta^{(3)} - \gamma_u^j\| + \|\eta^{(3)} - \gamma_u^j\|\|\eta^{(2)} - \beta_u^j\| \lesssim \|\eta - \eta_{u j}\|_e^2
$$
by the Cauchy-Schwarz inequality, the triangle inequality, and Assumptions \ref{ass: parameters} and \ref{ass: covariates}. Similarly, $\Ep_P[|I_{3,2}|]\lesssim \|\eta - \eta_{u j}\|_e^2$. In addition, since $|\Lambda''(t)|\leq 1$ for all $t\in\mathbb R$,
\begin{align*}
\Ep_P[|I_{3,3}|]&\lesssim \|\eta - \eta_{u j}\|_e^2 + \Big(\Ep_P\Big[(X^j(\eta^{(2)} - \beta_u^j))^4\Big]\Big)^{1/2}  \lesssim \|\eta - \eta_{u j}\|_e^2
\end{align*}
by the arguments used above, the Cauchy-Schwarz inequality, and Assumption \ref{ass: covariates}. Also, the terms in $\Ep_P[\partial_r^2\psi_{u j}(W,\theta_{u j} + r(\theta - \theta_{u j}),\eta_{u j} + r(\eta - \eta_{u j}))]$ arising from differentiation of $\theta_{u j} + r(\theta - \theta_{u j})$ can be bounded similarly. Therefore, Assumption \ref{ass: S1}(v-c) holds.

Next, we verify Assumption \ref{ass: AS}(i). Observe that by Theorems \ref{Thm:RateEstimatedLassoLogistic} and \ref{Thm:RateEstimatedLassoLinear}, with probability $1-o(1)$,
$$ \sup_{u\in\UU} \Big(\|\widetilde \theta_{u} - \theta_{u}\| + \|\widetilde \beta_u - \beta_u\|\Big) \lesssim \sqrt{s_n\log a_n/n},  \ \  \sup_{u\in\UU} \Big(\|\widetilde \theta_{u}\|_0 + \|\widetilde \beta_u\|_0\Big) \lesssim s_n, $$
$$\sup_{u\in\UU, j\in[\pp]} \|\widetilde \gamma_{u}^j - \gamma_{u}^j\| \lesssim \sqrt{s_n\log a_n/n}, \ \ \mbox{and} \ \  \sup_{u\in\UU, j\in[\pp]} \|\widetilde \gamma_{u}^j\|_0 \lesssim s_n. $$
In addition, $\hat{\beta}_u^j = (\widetilde \theta_{u[\pp]\setminus j}',\widetilde \beta_u')'$, and so uniformly over $u\in\UU$ and $j\in[\pp]$, with probability $1-o(1)$,
$$
 \|\hat \beta^j_u-\beta^j_u\| \lesssim \sqrt{s_n \log a_n/n}\ \ \mbox{and} \ \ \|\hat\beta_u^j\|_0 \lesssim s_n.
$$
Moreover, uniformly over $u\in\UU$ and $j\in[\pp]$, with probability $1-o(1)$,
\begin{align*}
\|\widetilde \gamma_u^j - \gamma_u^j\|_1
&\leq \|\widetilde\gamma_u^j - \bar\gamma_u^j\|_1 + \|\bar\gamma_u^j - \gamma_u^j\|_1 \lesssim \|\widetilde\gamma_u^j - \bar\gamma_u^j\|_1 + s_n\sqrt{\log a_n /n}\\
&\lesssim \sqrt{s_n}\|\widetilde\gamma_u^j - \bar\gamma_u^j\| + s_n\sqrt{\log a_n /n}\lesssim s_n\sqrt{\log a_n /n}
\end{align*}
by Assumption \ref{ass: sparsity} and the triangle and the Cauchy-Schwarz inequalities.
Therefore, Assumption \ref{ass: AS}(i) holds. In addition, Assumption \ref{ass: AS}(ii) holds by construction of $\mathcal{T}_{u j}$ and since $\Ep_P[r_u^2]\leq C_1 s_n\log a_n /n$, which in turn follows from Assumption \ref{ass: approximation error}. Also, Assumption \ref{ass: AS}(iii) holds by construction of $\mathcal{T}_{u j}$.

Next, we establish the entropy bound of Assumption \ref{ass: AS}(iv) with $v_n = Cs_n$ and $K_n = CM_{n,1}$ for sufficiently large constant $C>0$ (recall that $a_n = p\vee \pp\vee n$). Let
\begin{align*}
%\mG_1 &= \Big\{(Y,D,X)\mapsto 1\{ Y\leq u\bar y + (1-u)\underline{y}\}\colon u\in \UU\Big\},\\
%\mG_2 &= \Big\{(Y,D,X)\mapsto \Ep_P[g(Y,D,X) \mid D,X]\colon g\in \mG_1\Big\},\\
%&\mG_4 = \Big\{(Y,D,X)\mapsto (D',X')\xi\colon \xi\in \mathbb{R}^{\pp + p}, \|\xi\|_0 \leq Cs\Big\},\quad \mG_6 = \Big\{(Y,D,X)\mapsto \xi\colon \xi\in\mathbb R, |\xi|\leq C\Big\}\\
&\mG_4 = \Big\{(Y,D,X)\mapsto (D',X')\xi\colon \xi\in \mathbb{R}^{\pp + p}, \|\xi\|_0 \leq Cs_n, \|\xi\|\leq C \Big\},\\
&\mG_{5,j} = \Big\{(Y,D,X)\mapsto \xi D_j + X^j\beta_u^j\colon u\in \UU, |\xi|\leq C\Big\},\quad j\in[\pp],
\end{align*}
for sufficiently large $C$. Moreover, recall that $W=(Y,D,X)$ and let
\begin{align*}
&\mF_{1,1} = \Big\{W\mapsto \psi_{u j}(W,\theta,\eta)\colon u\in\UU, j\in[\pp],\theta\in\Theta_{u j},\eta\in\mathcal{T}_{u j}\setminus \eta_{u j}\Big\},\\
&\mF_{1,2} = \Big\{W\mapsto \psi_{u j}(W,\theta,\eta_{u j})\colon u\in\UU, j\in[\pp],\theta\in\Theta_{u j}\Big\}.
\end{align*}
Then $\mF_1 = \mF_{1,1}\cup \mF_{1,2}$ and
\begin{align*}
&\mF_{1,1}\subset ( \mG_1 - \G(\mG_4) )\cdot \mG_4,\\
&\mF_{1,2}\subset (\mG_1 - \mG_2 + \Lambda(\cup_{j\in[\pp]}\mG_{5,j}) - \Lambda(\cup_{j\in[\pp]}\mG_{5,j}))\cdot (\cup_{j\in[\pp]}\mG_{3,j}),
\end{align*}
where $\mG_1$, $\mG_2$, and $\mG_{3,j}$, $j\in[\pp]$, are defined above, because $\psi_{u j}(W,\theta,\eta_{u j})=\{ 1\{ Y \leq u \bar y + (1-u)\underline y\} - r_u - \Lambda( D_j\theta+X^j\beta_u^j)\}Z_u^j$. A bound for the uniform entropy numbers of $\mG_1$ is established above in (\ref{eq: G1 entropy bound}). Also, $\mG_4$ is a union over ${{p + \pp} \choose Cs_n}$  VC-subgraph classes with indices $O(s_n)$, and so is $\Lambda(\mG_4)$. Hence, by Lemma \ref{lemma: andrews},
$$
\sup_Q  \log N(\epsilon \|\widetilde F_{1,1}\|_{Q,2}, \mF_{1,1}, \| \cdot \|_{Q,2}) \leq Cs_n\log (a_n/\epsilon), \quad \text{for all } 0<\epsilon\leq 1,
$$
where
$$
\widetilde F_{1,1}(W) = \sup_{u\in\UU, j\in [\pp]} \sup_{\gamma\in \mathbb{R}^{\pp - 1 +p}\colon \|\gamma - \gamma_u^j\|_1 \leq C\sqrt{s_n}\tau_n}\Big(2|D_j - X^j\gamma|\Big)
$$
is an envelope of $\mF_{1,1}$. Observe that $|D_j - X^j\gamma|\leq |D_j - X^j\gamma_u^j|+|X^j(\gamma-\gamma_u^j)|\leq \sup_{u\in\UU, j\in [\pp]} |Z_u^j| + \|X\|_\infty C\sqrt{s_n}\tau_n$, and so
$$
\|\widetilde F_{1,1}\|_{P,q}\lesssim M_{n,1}+\sqrt{s_n}\tau_n M_{n,2}\lesssim M_{n,1}
$$
by Assumption \ref{ass: covariates} (observe that $\sqrt{s_n}\tau_n M_{n,2} \lesssim 1$ and $M_{n,1}\geq 1$).

Next we turn to $\mF_{1,2}$. Bounds for the uniform entropy numbers of $\mG_2$ and $\cup_{j\in[\pp]}\mG_{3,j}$ are established above. Consider $\mG_{5,j}$ for $j\in[\pp]$. Note that for all $u_1,u_2\in\mathcal U$,
\begin{align*}
&\|\gamma_{u_2}^j - \gamma_{u_1}^j\|_1\lesssim \sqrt{p + \pp}|u_2 - u_1|,\quad \|\gamma_{u_1}^j\|_1\lesssim \sqrt{s_n},\\
&\|\beta_{u_2}^j - \beta_{u_1}^j\|_1
\lesssim \Big(\|\theta_{u_2} - \theta_{u_1}\|_1+\|\beta_{u_2} - \beta_{u_1}\|_1 \Big)\lesssim \sqrt{p+\pp}|u_2 - u_1|
\end{align*}
by \eqref{eq: gamma-lipshitz}, \eqref{eq: thm 3.1 gamma bound}, and  Assumption \ref{ass: parameters}. Therefore, for all $\xi_1,\xi_2\in\mathbb R$ such that $|\xi_1|\leq C$ and $|\xi_2|\leq C$, and all $u_1,u_2\in\mathcal U$,
\begin{align*}
\Big\|(\xi_2, \beta_{u_2}^j ) - (\xi_1, \beta_{u_1}^j)\Big\|_1
&\leq |\xi_2 - \xi_1| + \|\beta_{u_2}^j - \beta_{u_1}^j\|_1 \\
&\lesssim |\xi_2 - \xi_1| + \sqrt{p + \pp}|u_2 - u_1|.
\end{align*}
Hence, Lemma \ref{lem: linear classes} implies that for all $j\in[\pp]$, the uniform entropy numbers of $\Lambda(\mG_{5,j})$ obey
$$
\sup_Q  \log N(\epsilon \|\widetilde F_5\|_{Q,2}, \Lambda(\mG_{5,j}), \| \cdot \|_{Q,2}) \leq C\log (a_n/\epsilon), \quad \text{for all } 0<\epsilon\leq 1,
$$
where $\widetilde F_5(Y,D,X) = 1 + M_{n,2}^{-1}(\|D\|_\infty\vee\|X\|_\infty)$ is its envelope. Hence, by Lemma \ref{lemma: andrews},
$$
\sup_Q  \log N(\epsilon \|\widetilde F_{1,2}\|_{Q,2}, \mF_{1,2}, \| \cdot \|_{Q,2}) \leq C\log (a_n/\epsilon), \quad \text{for all } 0<\epsilon\leq 1,
$$
where
\begin{align*}
\widetilde F_{1,2}(W)
&= C(1 + M_{n,2}^{-1}(\|D\|_\infty\vee\|X\|_\infty))\\
&\qquad \times\Big(\sup_{u\in\UU j\in[\pp]}|Z_u^j|+M_{n,2}^{-1}(\|D\|_\infty \vee \|X\|_\infty)\Big)
\end{align*}
is an envelope of $\mF_{1,2}$ that satisfies $\|\widetilde F_{1,2}\|_{P,q}\lesssim M_{n,1}$ by Assumption \ref{ass: covariates} and the Cauchy-Schwarz inequality.
Applying Lemma \ref{lemma: andrews} one more time finally shows that the uniform entropy numbers of $\mF_1$ obey (\ref{eq: F1 entropy bound}) with constants specified above and with an envelope $F_1 =  \widetilde F_{1,1} \vee \widetilde F_{1,2}$ satisfying $\|F_1\|_{P,q} \lesssim M_{n,1}$.

Next, we verify Assumption \ref{ass: AS}(v). Fix $u\in\UU$, $j\in[\pp]$, $\theta\in\Theta_{u j}$, and $\eta=(\eta^{(1)},\eta^{(2)},\eta^{(3)})\in\mathcal{T}_{u j}$. Then
\begin{align*}
&\Ep_P[\psi_{u j}(W,\theta,\eta)^2]\\
&\quad=\Ep_P\Big[\Ep_P\Big[\Big(Y_u - \Lambda\Big(D_j \theta +X^j\eta^{(2)}\Big) - \eta^{(1)}\Big)^2 \mid D,X\Big] (D_j - X^j\eta^{(3)})^2\Big]\\
&\quad\geq  \Ep_P[f_u^2(D_j - X^j\eta^{(3)})^2]\geq c,
\end{align*}
where the first inequality follows from the fact that for any random variable $\xi$, the function $x\mapsto \Ep[(\xi - x)^2]$ is minimized at $x = \Ep[\xi]$ and in this case $f_u^2 =  \text{Var}(Y_u\mid D,X)$, and the second inequality follows from Assumption \ref{ass: covariates}. In addition,
$$
\Ep_P[\psi_{u j}(W,\theta,\eta)^2]\leq \Ep_P[(D_j - X^j\eta^{(3)})^2] \lesssim 1
$$
by Assumptions \ref{ass: parameters} and \ref{ass: covariates}.
Therefore, Assumption \ref{ass: AS}(v) holds.

Finally, we verify Assumption \ref{ass: AS}(vi). The condition (a) holds by construction of $\tau_n$ and $v_n$. To verify the condition (b) observe that
\begin{align*}
&(B_{1n}\tau_n)^{\omega/2} (v_{n} \log a_n)^{1/2}  + n^{-1/2+1/q} v_{n} K_n\log a_n\\
&\qquad \lesssim n^{-1/2}s_n\log a_n + n^{-1/2+1/q}s_n M_{n,1}\log a_n \lesssim \delta_n
\end{align*}
by Assumption \ref{ass: covariates}. In addition,
$$
(\mS_n\log n/\sqrt n)^{\omega/2}(v_n\log a_n)^{1/2}\lesssim \sqrt{s_n}(\log a_n)\cdot(\log n)/\sqrt n\lesssim \delta_n,
$$
since $\mS_n\lesssim (\log a_n)^{1/2}$, which is established above, and $s_n\log a_n\leq \delta_n n^{1/2 - 1/q}$, which holds by Assumption \ref{ass: covariates}. The condition (b) follows.
The condition (c) holds because
$$
n^{1/2}B_{1 n}^2 B_{2 n}\tau_n^2 \lesssim n^{-1/2}s_n \log a_n \lesssim \delta_n
$$
as in the verification of the condition (b). This completes the verification of Assumptions \ref{ass: S1} and \ref{ass: AS} and thus completes the proof of the theorem.
\qed

\subsection*{Proof of Corollary \ref{cor: logistic clt}}
The asserted claim will follow from Corollary \ref{cor: general CLT} as long as we can verify its conditions. Assumptions \ref{ass: S1} and \ref{ass: AS} were verified in the proof of Theorem \ref{theorem:inferenceAlg1}. Therefore, it suffices to verify Assumption \ref{ass: OSR}(i,ii) and the growth conditions of Corollary \ref{cor: general CLT}. %In this proof, we use the same conventions about the constants $c,C$ and the sign $\lesssim$ as those used in the proof of Theorem \ref{theorem:inferenceAlg1}.

First, we verify Assumption \ref{ass: OSR}(i). Recall the function class
$$
\widetilde{\mathcal{F}} = \{\psi_{u j}(\cdot,\theta_{u j},\eta_{u j})\colon u\in\UU, j\in[\pp]\}
$$
defined in the proof of Theorem \ref{theorem:inferenceAlg1}, where it is also proven that its uniform entropy numbers obey \eqref{eq: uniform entropy numbers f tilde} with an envelope $\widetilde F$ satisfying $\|\widetilde F\|_{P,q}\lesssim M_{n,1}$. Also, note that Assumption \ref{ass: S1}(iv) gives $1\lesssim |J_{u j}|\lesssim 1$ for all $u\in\UU$ and $j\in[\pp]$, and that Assumption \ref{ass: AS}(v) gives $1\lesssim \Ep_P[\psi^2_{u j}(W,\theta_{u j},\eta_{u j})]\lesssim 1$ for all $u\in\UU$ and $j\in[\pp]$. Hence,
$$
\mF_0\subset \Big\{\xi\cdot f\colon f\in\widetilde\mF,\xi\in\mathbb R, c\leq |\xi|\leq C\Big\},
$$
and so Lemma \ref{lemma: andrews} implies that the uniform entropy numbers of $\mF_0$ obey
\begin{equation}\label{eq: entropy numbers F0}
\sup_Q  \log N(\epsilon \|F_0\|_{Q,2}, \mF_0, \| \cdot \|_{Q,2}) \leq C\log (a_n/\epsilon), \quad \text{for all } 0<\epsilon\leq 1,
\end{equation}
where its envelope $F_0$ satisfies $\|F_0\|_{P,q}\lesssim M_{n,1}$. Thus, Assumption \ref{ass: OSR}(i) holds with $\varrho_n = C$, $A_n = a_n = p\vee \pp \vee n$, and $L_n = CM_{n,1}$.

Next, we verify Assumption \ref{ass: OSR}(ii). For $k=3,4$, $u\in\UU$, and $j\in[\pp]$, we have
$$
\Ep_P\Big[ |\bar \psi_{uj}(W,\theta_{u j},\eta_{u j})|^k\Big]\lesssim \Ep_P\Big[|D_j - X^j\gamma_u^j|^k\Big]\lesssim 1
$$
by Assumptions \ref{ass: parameters} and \ref{ass: covariates} since $|Y_u|\leq 1$, $|\Lambda(t)|\leq 1$ for all $t\in\mathbb R$, and $|r_u|\leq 1$. Thus, Assumption \ref{ass: OSR}(ii) holds with the same $L_n = C M_{n,1}$ since $M_{n,1}\geq 1$.

Finally, with our choice of $A_n$ and $\varrho_n$, the growth conditions of Corollary \ref{cor: general CLT} hold by assumption. This completes the proof.
\qed

\subsection*{Proof of Corollary \ref{cor: logistic bands}}
The asserted claim will follow from Corollary \ref{theorem: general bs} as long as we can verify its conditions. Assumptions \ref{ass: S1} and \ref{ass: AS} were verified in the proof of Theorem \ref{theorem:inferenceAlg1}. Assumption \ref{ass: OSR}(i,ii) was verified in the proof of Corollary \ref{cor: logistic clt}. Therefore, it suffices to verify Assumptions \ref{ass: OSR}(iii) and \ref{ass: variance} and the growth conditions of Corollary \ref{theorem: general bs}. %In this proof, we use the same conventions about the constants $c,C$ and the sign $\lesssim$ as those used in the proof of Theorem \ref{theorem:inferenceAlg1}.

We split the proof into six steps. In Steps 1-3, we verify Assumption \ref{ass: variance}. In Steps 4 and 5, we verify Assumption \ref{ass: OSR}(iii). In Step 6, we verify the growth conditions of Corollary \ref{theorem: general bs}.

%\medskip
{\bf Step 1.} Here we show that
$$
\hat J_{u j} - J_{u j} = o_P(\log^{-1} a_n)
$$
uniformly over $u\in\UU$ and $j\in[\pp]$. For $\theta\in\mathbb R$, $\beta\in\mathbb R^{\pp-1+p}$, $\gamma\in\mathbb R^{\pp-1+p}$, and $j\in[\pp]$, define
\begin{align*}
&\widetilde\psi_j(W,\theta,\beta,\gamma) = -\Lambda'\Big(D_j\theta + X^j\beta\Big)D_j(D_j - X^j\gamma),\\
&\widetilde m_j(\theta,\beta,\gamma) = \Ep_P[\widetilde \psi_j(W,\theta,\beta,\gamma)].
\end{align*}
Then $\hat J_{u j} = \En[\widetilde \psi_j (W,\widetilde\theta_{u j},\hat\beta_u^j,\widetilde\gamma_u^j)]$ and $J_{u j} = \widetilde m_j(\theta_{u j},\beta_u^j,\gamma_u^j)$ for all $u\in\UU$ and $j\in[\pp]$. Therefore, by the triangle inequality,
$$
|\hat J_{u j} - J_{u j}|\leq \Big|\hat J_{u j} - \widetilde m_j(\widetilde\theta_{u j},\hat\beta_u^j,\widetilde\gamma_u^j)\Big| + \Big|\widetilde m_j(\widetilde\theta_{u j},\hat\beta_u^j,\widetilde\gamma_u^j) - \widetilde m_j(\theta_{u j},\beta_u^j,\gamma_u^j)\Big|.
$$
Define
\begin{multline*}
\mG_6 = \Big\{(Y,D,X)\mapsto -\Lambda'\Big(D_j \theta + X^j\eta^{(2)}\Big)D_j(D_j - X^j\eta^{(3)})\colon \\
u\in\UU, j\in[\pp],\theta\in\Theta_{u j},\eta = (\eta^{(1)},\eta^{(2)},\eta^{(3)})\in\mathcal{T}_{u j}\setminus \eta_{u j}\Big\}.
\end{multline*}
Then by Assumption \ref{ass: AS}(i), with probability $1-o(1)$,
$$
\Big|\hat J_{u j} - \widetilde m_j(\widetilde\theta_{u j},\hat\beta_u^j,\widetilde\gamma_u^j)\Big| \leq \sup_{f\in\mG_6}\Big|\En[f(W)] - \Ep_P[f(W)]\Big|
$$
for all $u\in\UU$ and $j\in[\pp]$. In addition, $\mG_6\subset -\Lambda'(\mG_4)\cdot\mG_4^2$, where the function class $\mG_4$ is introduced in the proof of Theorem \ref{theorem:inferenceAlg1}. Moreover,
{\small
\begin{align*}
\Big|D_j \theta + X^j\eta^{(2)}\Big|
&\lesssim \Big(\|D\|_\infty\vee \|X\|_\infty\Big)\|\eta^{(2)}\|_1\\
&\lesssim \sqrt{p+\pp}\Big(\|D\|_\infty\vee \|X\|_\infty\Big)\|\eta^{(2)}\|\lesssim \sqrt{p+\pp}\Big(\|D\|_\infty\vee \|X\|_\infty\Big)
\end{align*}}\!uniformly over $u\in\UU$, $j\in[\pp]$, $\theta\in\Theta_{u j}$, and $\eta = (\eta^{(1)},\eta^{(2)},\eta^{(3)}) = \mathcal{T}_{u j}\setminus\eta_{u j}$ by Assumptions \ref{ass: parameters}. Hence, applying Lemmas \ref{lemma: andrews} and \ref{lem: bounded lipschitz classes} with $K = M_{n,2} n^{2/q} (p+\pp)^{1/2}$ shows that the uniform entropy numbers of $\mG_6$ obey
$$
\sup_Q\log N(\epsilon \|F_6\|_{Q,2},\mG_6,\|\cdot\|_{Q,2})\leq Cs_n \log(a_n/\epsilon),\quad\text{for all }0<\epsilon\leq 1
$$
where $F_6(W) = (1 + (\|D\|_\infty\vee\|X\|_\infty)/(M_{n,2} n^{2/q}))\|D\|_{\infty}\widetilde F_{1,1}(W)$ is its envelope, and $\widetilde F_{1,1}$ is defined in the proof of Theorem \ref{theorem:inferenceAlg1}.
Also, recall that $\|\widetilde F_{1,1}\|_{P,q}\lesssim M_{n,1}$, which is established in the proof of Theorem \ref{theorem:inferenceAlg1}, and so
{\small
\begin{align*}
&\Big(\Ep_P\Big[\max_{1\leq i\leq n}|F_6(W_i)|^{q/4}\Big]\Big)^{4/q}\\
&\leq \Big(\Ep_P\Big[\max_{1\leq i\leq n}|\|D_i\|_{\infty}^{q/4}\widetilde F_{1,1}(W_i)|^{q/4}\Big]\Big)^{4/q} \\
&\qquad + \Big(\Ep_P\Big[\max_{1\leq i\leq n}\frac{(\|D_i\|_\infty\vee\|X_i\|_\infty)^{q/2}}{(M_{n,2} n^{2/q})^{q/4}}|\widetilde F_{1,1}(W_i)|^{q/4}\Big]\Big)^{4/q}\\
& \lesssim n^{2/q}M_{n,1}^2 + n^{4/q}\Big(\Ep_P\Big[\frac{(\|D\|_\infty\vee \|X\|_\infty)^q}{(M_{n,2} n^{2/q})^{q/2}}\Big]\cdot \Ep_P\Big[|\widetilde F_{1,1}(W)|^{q/2}\Big]\Big)^{2/q}\lesssim n^{2/q}M_{n,1}^2,
\end{align*}}
where the first line follows from the triangle inequality, and the second from the Cauchy-Schwarz inequality and Assumption \ref{ass: covariates}. In addition, $\|f\|_{P,2}\lesssim 1$ for all $f\in\mG_6$. Hence, Lemma \ref{lemma:CCK} implies that
$$
\sup_{f\in\mG_6}\Big|\En[f(W)] - \Ep_P[f(W)]\Big|\lesssim \sqrt{\frac{s_n\log a_n}{n}} + \frac{M_{n,1}^2s_n\log a_n}{n^{1-2/q}} = o(\log^{-1}a_n)
$$
with probability $1-o(1)$ by Assumption \ref{ass: covariates} and the growth condition $s_n\log^3 a_n/n=o(1)$.

Next, using the same arguments as those used to verify Assumption \ref{ass: S1}(v-a) in Theorem \ref{theorem:inferenceAlg1} shows that
$$
\Big|\widetilde m_j(\widetilde\theta_{u j},\hat\beta_u^j,\widetilde\gamma_u^j) - \widetilde m_j(\theta_{u j},\beta_u^j,\gamma_u^j)\Big|\lesssim \|\widetilde\theta_{u j} - \theta_{u j}\| + \|\hat\beta_u^j - \beta_u^j\| + \|\widetilde\gamma_u^j - \gamma_u^j\|
$$
and the right-hand of this inequality is bounded above by $(Cs_n\log a_n/n)^{1/2}$ uniformly over $u\in\UU$ and $j\in[\pp]$ with probability $1-o(1)$, as demostrated in the proof of Theorem \ref{theorem:inferenceAlg1}. In turns, $(Cs_n\log a_n/n)^{1/2} = o(\log^{-1}a_n)$ by assumption. Combining presented bounds gives the claim of this step.

%\medskip
{\bf Step 2.} Here we show that
$$
\En[\psi_{u j}^2(W,\widetilde\theta_{u j},\hat\eta_{u j})] - \Ep_P[\psi_{u j}^2(W,\theta_{u j},\eta_{u j})] = o_P(\log^{-1} a_n)
$$
uniformly over $u\in\UU$ and $j\in[\pp]$. The proof of this claim is similar to that in Step 1, where the main difference is that instead of $-\Lambda'(\cdot)$ in the function class $\mG_6$, we set $Y_u^2 - 2Y_u\Lambda(\cdot) + \Lambda^2(\cdot)$, with the resulting function class having the same envelope and its uniform entropy numbers obeying the same bounds as those derived fo $\mG_6$ (up to possibly different constants).

%\medskip
{\bf Step 3.} Here we finish the verification of Assumption \ref{ass: variance}. Observe that $1\lesssim J_{u j}\lesssim 1$ and $1\lesssim \Ep_P[\psi_{u j}^2(W,\theta_{u j},\eta_{u j})]\lesssim 1$ for all $u\in\UU$ and $j\in[\pp]$ by Assumptions \ref{ass: S1}(iv) and \ref{ass: AS}(v). Hence, $1\lesssim \sigma_{u j}^2\lesssim 1$, and so
{\small
\begin{align*}
&\Big|\frac{\hat\sigma_{u j}}{\sigma_{u j}} - 1\Big|
\leq \Big|\frac{\hat\sigma^2_{u j}}{\sigma^2_{u j}} - 1\Big|\lesssim \Big|\hat\sigma_{u j}^2 - \sigma_{u j}^2\Big|\\
&\quad \leq \Big|\hat J^{-2}_{u j} - J^{-2}_{u j}\Big|\En[\psi_{u j}^2(W,\widetilde\theta_{u j},\hat\eta_{u j})]\\
&\quad \quad +J_{u j}^{-2}\Big|\En[\psi_{u j}^2(W,\widetilde\theta_{u j},\hat\eta_{u j})] - \Ep_P[\psi_{u j}^2(W,\theta_{u j},\eta_{u j})]\Big|\\
&\quad\lesssim \Big|\hat J_{u j} - J_{u j}\Big| + \Big|\En[\psi_{u j}^2(W,\widetilde\theta_{u j},\hat\eta_{u j})] - \Ep_P[\psi_{u j}^2(W,\theta_{u j},\eta_{u j})]\Big| = o_P(\log^{-1}a_n)
\end{align*}}
uniformly over $u\in\UU$ and $j\in[\pp]$ by Steps 1 and 2. Therefore, Assumption \ref{ass: variance} holds for some $\varepsilon_n$ and $\Delta_n$ satisfying $\varepsilon_n\log a_n = o(1)$ and $\Delta_n=o(1)$.

%\medskip
{\bf Step 4.} Here we show that the inequality concerning the entropy numbers of $\widehat\mF_0$ in Assumption \ref{ass: OSR}(iii) holds with $\bar\varrho_n =C$, $\bar A_n = a_n = p\vee \pp\vee n$, and $\Delta_n = o(1)$. By construction of $\hat\psi_{u j}$, the function class $\{\hat\psi_{u j}(\cdot)\colon u\in\UU, j\in[\pp]\}$ contains at most $n\pp$ functions (as $u$ varies, new functions appear only as $u$ crosses one of the observations $(Y_i)_{i=1}^n$). Also, it follows from (\ref{eq: entropy numbers F0}) in the proof of Corollary \ref{cor: logistic clt} that the entropy numbers of $\mF_0 = \{\bar\psi_{u j}(\cdot)\colon u\in\UU,j\in[\pp]\}$ obey
$$
N(\epsilon,\mF_0,\|\cdot\|_{\mathbb P_{n,2}})\leq C\log(a_n\|F_0\|_{\mathbb P_{n,2}}/\epsilon)\leq C\log(a_n/\epsilon),\quad\text{for all }0<\epsilon\leq 1,
$$
with probability $1-o(1)$. Hence,
$$
\log N(\epsilon,\hat\mF_0,\|\cdot\|_{\mathbb P_{n,2}})\leq C\log(a_n/\epsilon),\quad\text{for all }0<\epsilon\leq 1
$$
with probability $1-o(1)$. The claim of this step follows.

%\medskip
{\bf Step 5.} Here we show that the second part of Assumption \ref{ass: OSR}(iii), that is, that with probability $1-\Delta_n$, we have $\|f\|_{\mathbb P_{n,2}}\leq \bar\delta_n$ for all $f\in\widehat\mF_0$, holds for some $\bar\delta_n$ and $\Delta_n$ satisfying $\bar\delta_n = o(\log^{-1}a_n)$ and $\Delta_n = o(1)$. By the triangle inequality,
\begin{align*}
&\Big\|\hat\sigma_{u j}^{-1} \hat J_{u j}^{-1}\psi_{u j}(W,\check\theta_{u j},\hat\eta_{u j}) - \sigma_{u j}^{-1}J_{u j}^{-1}\psi_{u j}(W,\theta_{u j},\eta_{u j})\Big\|_{\mathbb P_{n,2}}\\
&\quad \leq |\hat\sigma_{u j}^{-1}\hat J_{u j}^{-1} - \sigma_{u j}^{-1} J_{u j}^{-1}|\cdot\Big\|\psi_{u j}(W,\theta_{u j},\eta_{u j})\Big\|_{\mathbb P_{n,2}} \\
&\qquad + \hat\sigma_{u j}^{-1}\hat J_{u j}^{-1}\Big\|\psi_{u j}(W,\check\theta_{u j},\hat\eta_{u j}) - \psi_{u j}(W,\theta_{u j},\eta_{u j})\Big\|_{\mathbb P_{n,2}}.
\end{align*}
We bound two terms on the right-hand side of this inequality in turn. To bound the first term, observe that
$$
|\hat\sigma_{u j}^{-1}\hat J_{u j}^{-1} - \sigma_{u j}^{-1} J_{u j}^{-1}|=o_P(\log^{-1}a_n)
$$
uniformly over $u\in\UU$ and $j\in[\pp]$ by Steps 1 and 3 and since $1\lesssim J_{u j}\lesssim 1$ and $1\lesssim \sigma_{u j}\lesssim 1$, which is discussed in Step 3. Also, as established in the proof of Theorem \ref{theorem:inferenceAlg1}, the uniform entropy numbers of the function class $\widetilde\mF = \{\psi_{u j}(\cdot,\theta_{u j},\eta_{u j})\colon u\in\UU,j\in[\pp]\}$ obey \eqref{eq: uniform entropy numbers f tilde} with an envelope $\widetilde F$ satisfying $\|\widetilde F\|_{P,q}\lesssim M_{n,1}$. Moveover, $\Ep_P[f^2(W)]\lesssim 1$ uniformly over $f\in\widetilde \mF$ by Assumption \ref{ass: AS}(v). Therefore, Lemma \ref{lem: maximal inequality 2} shows that
\begin{align*}
&\Ep_P\Big[\sup_{u\in\UU,j\in[\pp]}\En[\psi_{u j}^2(W,\theta_{u j},\eta_{u j})]\Big]\\
&\quad\lesssim 1 + n^{-1/2+1/q}M_{n,1}\Big(\sqrt{\log a_n} + n^{-1/2+1/q}M_{n,1} \log a_n\Big)\\
&\quad\lesssim 1 + n^{-1+2/q}M_{n,1}^2\log a_n\lesssim 1,
\end{align*}
where the second inequality follows from Assumption \ref{ass: covariates}. Hence,
$$
|\hat\sigma_{u j}^{-1}\hat J_{u j}^{-1} - \sigma_{u j}^{-1} J_{u j}^{-1}|\cdot\|\psi_{u j}(W,\theta_{u j},\eta_{u j})\|_{\mathbb P_{n,2}} = o_P(\log^{-1}a_n)
$$
uniformly over $u\in\UU$ and $j\in[\pp]$.

To bound the second term, define
\begin{multline*}
\mG_7 = \Big\{\psi_{u j}(\cdot,\theta,\eta) - \psi_{u j}(\cdot,\theta_{u j},\eta_{u j})\colon u\in\UU,j\in[\pp], \theta\in\Theta_{u j},\\
|\theta - \theta_{u j}|\leq \sqrt{Cs_n\log a_n/n}, \eta\in\mathcal T_{u j}\Big\}
\end{multline*}
for sufficiently large constant $C>0$ and $\mathcal T_{u j}$ appearing in Assumption \ref{ass: AS}. Then $\mG_7\subset \mF_1 - \mF_1$, and so Lemma \ref{lemma: andrews} together with the bound for the uniform entropy numbers of $\mF_1$ established in the proof of Theorem \ref{theorem:inferenceAlg1} imply that the uniform entropy numbers of $\mG_7$ obey
$$
\sup_Q \log N(\epsilon\|F\|_{Q,2},\mG_7,\|\cdot\|_{Q,2})\leq Cs_n\log (a_n/\epsilon),\quad\text{for all }0<\epsilon\leq 1,
$$
where $F_7$ is its envelope satisfying $\|F_7\|_{P,q}\lesssim M_{n,1}$. In addition, Assumption \ref{ass: AS}(i) together with Step 1 in the proof of Theorem \ref{theorem:semiparametric} imply that with probability $1-o(1)$,
$$
\psi_{u j}(\cdot,\check \theta_{u j},\hat\eta_{u j}) - \psi_{u j}(\cdot,\theta_{u j},\eta_{u j})\in\mG_7
$$
for all $u\in\UU$ and $j\in[\pp]$ (recall that in the proof of Theorem \ref{theorem:inferenceAlg1}, we set $B_{1 n} = C$ and $\tau_n = (Cs_n\log a_n/n)^{1/2}$). Also, Assumptions \ref{ass: S1}(v-a) and \ref{ass: AS}(ii) show that $\Ep_P[f^2(W)]\lesssim s_n\log a_n/n$ uniformly over $f\in\mG_7$. Hence, it follows from Lemma \ref{lem: maximal inequality 2} that
\begin{align*}
&\Ep_P\Big[\sup_{f\in\mG_7}\En[f^2(W)]\Big]\\
&\lesssim s_n\log a_n/n + n^{-1/2+1/q}M_{n,1}\Big(\sqrt{s_n\log a_n} + n^{-1/2+1/q}M_{n,1}s_n\log a_n\Big)\\
&\lesssim n^{-1+2/q}M_{n,1}^2s_n\log a_n = o(\log^{-1} a_n).
\end{align*}
Hence,
$$
\hat\sigma_{u j}^{-1}\hat J_{u j}^{-1}\|\psi_{u j}(W,\check\theta_{u j},\hat\eta_{u j}) - \psi_{u j}(W,\theta_{u j},\eta_{u j})\|_{\mathbb P_{n,2}} = o_P(\log^{-1}a_n)
$$
uniformly over $u\in\UU$ and $j\in[\pp]$ by Assumption \ref{ass: covariates}. Combining presented bounds gives the asserted claim and completes the verification of Assumption \ref{ass: OSR}.

%\medskip
{\bf Step 6.} Recall that the growth conditions of Corollary \ref{cor: general CLT} were verified in the proof of Corollary \ref{cor: logistic clt}, where we set $\varrho_n = C$ and $A_n = a_n$. The other growth conditions of Corollary \ref{theorem: general bs}, $\varepsilon_n\varrho_n\log A_n = o(1)$ and $\bar\delta_n^2 \bar\varrho_n \varrho_n(\log \bar A_n)\cdot(\log A_n)=o(1)$ hold because we have $\varepsilon_n = o(\log^{-1} a_n)$, $\bar\varrho_n = C$, $\bar A_n = a_n$, and $\bar\delta_n = o(\log^{-1} a_n)$. This completes the proof of the corollary.
\qed

\subsection*{Proof of Corollary \ref{cor: simple example}}
To prove the asserted claim, we will apply Corollary \ref{cor: logistic bands}. Below we will verify Assumptions \ref{ass: covariates}(iii,v,vi,vii,viii,ix), \ref{ass: approximation error}(iii), and the growth conditions of Corollary \ref{cor: logistic bands}.

Set
$$
\bar M_n = \Big(\Ep_P\Big[(\|D\|_\infty \vee \|X\|_\infty)^{2q}\Big]\Big)^{1/(2q)} \ \ \ \mbox{and} \ \ \ \bar C_n := 1+\sup_{u\in \UU, j\in [\pp]} \|\gamma_u^j\|_1
$$
so that $\bar M_n \leq C_1$ and $\bar C_n \leq 1+C_1$ by assumption. Observe that for all $u\in\UU$ and $j\in[\pp]$,
$$
|Z_u^j| = |D_j - X^j\gamma_u^j|\leq (\|D\|_\infty\vee \|X\|_\infty)\cdot (1+\|\gamma_u^j\|_1)\lesssim \bar C_n(\|D\|_\infty\vee \|X\|_\infty)
$$
%where we used the inequality $\|\gamma_u^j\|\lesssim \sqrt{s_n}$ established in the proof of Theorem \ref{theorem:inferenceAlg1}.
Then
$$
\max_{j,k}\Big(\Ep_P[|Z_u^j X_k^j|^3]\Big)^{1/3} \lesssim \bar C_n\Big(\Ep_P\Big[(\|D\|_\infty\vee \|X\|_\infty)^6\Big]\Big)^{1/3}\lesssim \bar C_n\bar M_n^2 \lesssim 1.
$$
Therefore, given that $\log^6 a_n = o(n)$ by assumption, it follows that Assumption \ref{ass: covariates}(iii) holds for some $\delta_n$ satisfying $\delta_n^2\log a_n = o(1)$.

Also,
$$
\left(\Ep_P\left[\sup_{u\in\UU,j\in[\pp]}|Z_u^j|^{2q}\right]\right)^{1/(2q)}\lesssim \bar C_n\bar M_n \lesssim 1,
$$
and so Assumption \ref{ass: covariates}(v) holds with $M_{n,1} = C$ for sufficiently large constant $C$. In addition, since $s_n^2\log^3 a_n = o(n^{1 - 2/q})$, Assumption \ref{ass: covariates}(vi) holds for some $\delta_n$ satisfying $\delta_n^2\log a_n = o(1)$.

Further, Assumption \ref{ass: covariates}(vii) holds with $M_{n,2} = \bar M_n$ by definition of $\bar M_n$. In addition, since $s_n^2\log^2 a_n = o(n^{1-2/q})$, Assumption \ref{ass: covariates}(viii) holds for some $\delta_n$ satisfying $\delta_n^2\log a_n = o(1)$. Also, Assumption \ref{ass: covariates}(ix) holds for some $\delta_n$ satisfying $\delta_n^2\log a_n = o(1)$ since $M_{n,2} = \bar M_n\leq C_1$, $M_{n,1} \leq C$, $s_n \leq \delta_n n^{1/2 - 1/q}$, and $q>4$.

Moreover, since $\sup_{u\in\UU,j\in[\pp]}|\Ep_P[r_u Z_u^j]| = o((n\log a_n)^{-1/2})$, Assumption \ref{ass: approximation error}(iii) holds for some $\delta_n$ satisfying $\delta_n^2\log a_n = o(1)$. Finally, the growth conditions $M_{n,1}^{2/7}\log a_n = o(n^{1/7})$ and $M_{n,1}^{2/3}\log a_n = o(n^{1/3 - 2/(3q)})$ hold because $M_{n,1}\lesssim \bar C_n\bar M_n\lesssim C$, $\log^7 a_n = o(n)$, and $\log^3 a_n = o(n^{1-2/q})$.

Thus, there exists $\delta_n$ such that Assumptions \ref{ass: covariates}(iii,v,vi,vii,viii,ix) and \ref{ass: approximation error}(iii) as well as all growth conditions of Corollary \ref{cor: logistic bands} are satisfied. Since all other conditions of Corollary \ref{cor: logistic bands} are assumed, the asserted claim follows from that in Corollary \ref{cor: logistic bands}.
\qed

\section{Proofs for Section 4 }\label{sec: proofs for section 4}
In this appendix, we use $C$ to denote a strictly positive constant that is independent of $n$ and $P\in\mathcal P_n$. The value of $C$ may change at each appearance. Also, the notation $a_n \lesssim b_n$ means that $a_n\leq C b_n$ for all $n$ and some $C$. The notation $a_n\gtrsim b_n$ means that $b_n\lesssim a_n$. Moreover, the notation $a_n = o(1)$ means that there exists a sequence $(b_n)_{n\geq 1}$ of positive numbers such that (i) $|a_n|\leq b_n$ for all $n$, (ii) $b_n$ is independent of $P\in\mathcal P_n$ for all $n$, and (iii) $b_n\to 0$ as $n\to\infty$. Finally, the notation $a_n \lesssim_P b_n$ means that for any $\epsilon>0$, there exists $C$ such that $\Pr_P(a_n > C b_n) \leq \epsilon$ for all $n$, and the notation $a_n\gtrsim_P b_n$ means that $b_n\lesssim_P a_n$. Using this notation allows us to avoid repeating ``uniformly over $P\in\mathcal P_n$'' many times in the proofs of Theorems \ref{Thm:RateEstimatedLassoLogistic} and \ref{Thm:RateEstimatedLassoLinear}.

\subsection*{Proof of Theorem \ref{Thm:RateEstimatedLassoLogistic}}
%In this proof, we use $C$ to denote a constant that can be chosen to be independent of $P\in\mathcal P_n$ and $n$. Also, the notation $a_n \lesssim b$ and $a_n \gtrsim b_n$ mean that $a_n\leq C b_n$ and $a_n\geq C b_n$, respectively, for all $n$ and some $C$. Similarly, the notation $a_n\lesssim_P b_n$ and $a_n\gtrsim_P b_n$ mean that for any $\epsilon>0$, there exists $C$ such that $\Pr_P(a_n\leq C b_n)\geq 1-\epsilon$ and $\Pr_P(a_n\geq C b_n)\geq 1-\epsilon$, respectively, for all $n$ and some $C$.

In this proof, we will rely upon results in Appendix \ref{sec: generic results}. In particular, the asserted claims will follow from an application of Lemmas \ref{Lemma:LassoMRateRaw}, \ref{Lemma:LassoMSparsity}, and \ref{Lemma:PostLassoMRateRaw} (with some extra work). To follow the notation in Appendix \ref{sec: generic results}, define $X_u = (D',X')'$ and $w_u = f_u^2$ and redefine $\theta_u = (\theta_u',\beta_u')'$, $\widehat\theta_u = (\widehat\theta_u',\widehat\beta_u')'$, and $p = \pp + p$. Also, define $a_u = a_u(X_u)$ as a solution to the following equation:
\begin{equation}\label{eq: au def thm 4.1}
\Lambda(X_u'\theta_u) + r_u = \Lambda(X_u'\theta_u + a_u).
\end{equation}
Since $\Lambda$ is increasing, for each value of $X_u$, $a_u$ is uniquely defined. Then $\theta_u$ satisfies \eqref{A:EqMainFunc} with
\begin{align*}
M_u(Y_u,X_u,\theta,a) & =-\Big(1\{Y_u=1\}\log \Big(\G( X_u'\theta + a(X_u))\Big) \\
&\qquad + 1\{Y_u=0\}\log\Big(1-\G( X_u'\theta + a(X_u))\Big)\Big)
\end{align*}
for $\theta$ being a vector in $\mathbb R^p$ and $a$ being a function of $X_u$. Similarly, $\widehat\theta_u$ satisfies \eqref{Adef:LassoFunc} where $M_u(Y_u,X_u,\theta) = M_u(Y_u,X_u,\theta,\mathcal O)$ and $\mathcal O = \mathcal O(X_u)$, the identically zero function of $X_u$.

To apply Lemmas \ref{Lemma:LassoMRateRaw}, \ref{Lemma:LassoMSparsity}, and \ref{Lemma:PostLassoMRateRaw}, we need to verify Assumption \ref{ass: M}. In addition, one of the conditions in these lemmas is that \eqref{Eq:reg} holds with  probability $1-o(1)$. Verification of this condition will be done with the help of Lemma \ref{Thm:ChoiceLambda}, which in turn relies upon Condition WL. Therefore, below we also verify this condition.

We first verify Condition WL with $\epsilon_n = 1/n$ and $N_n = n$. Observe that since $\UU = [0,1]$, we have for any $\epsilon\in(0,1]$ that $N(\epsilon,\UU,d_{\UU}) \leq 1/\epsilon$, and so $\epsilon_n$ and $N_n$ satisfy the inequality $N_n\geq N(\epsilon_n,\UU,d_{\UU})$, which is the first requirement of Condition WL. Further, as in front of Condition WL in Appendix \ref{sec: generic results}, let
{\small
\begin{align*}
S_u &= \partial_\theta M_u(Y_u,X_u,\theta,a_u)|_{\theta = \theta_u} \\
&= -\Big(1\{Y_u = 1\}\frac{\Lambda'(X_u'\theta_u + a_u(X_u))}{\Lambda(X_u'\theta_u + a_u(X_u))} - 1\{Y_u = 0\}\frac{\Lambda'(X_u'\theta_u + a_u(X_u))}{1 - \Lambda(X_u'\theta_u + a_u(X_u))}\Big)\cdot X_u\\
& = -\Big(1\{Y_u = 1\}(1 - \Lambda(X_u'\theta_u + a_u(X_u)) - 1\{Y_u = 0\}(\Lambda(X_u'\theta_u + a_u(X_u))\Big)\cdot X_u\\
& = -\Big(Y_u - \Ep_P[Y_u\mid X_u]\Big)\cdot X_u.
\end{align*}}
Then $|S_{u k}| \leq |X_{u k}|$. In addition, since $\gamma\geq 1/n$, we have
$$
\Phi^{-1}(1 - \gamma/(2p N_n))\lesssim \sqrt{\log(p n)}\lesssim \sqrt{\log a_n}.
$$
Also, since $f_u\leq 1$, Assumption \ref{ass: covariates}(ii,iii) yields $\log^{1/2}a_n\lesssim \delta_n n^{1/6}$. Moreover, Assumption \ref{ass: covariates}(iv) yields $(\Ep_P[|S_{u k}|^3])^{1/3}\lesssim1$ uniformly over $u\in\UU$ and $k\in[p]$. Therefore, Condition WL(i) holds for some $\varphi_n$ satisfying $\varphi_n \lesssim \delta_n$. Assumption \ref{ass: covariates}(i,iv) also implies that uniformly over $u\in\UU$ and $k\in[p]$,
\begin{equation}\label{ULboundS2}
\Ep_P[|S_{u k}|^2] \leq \Ep_P[|X_{u k}|^2]\lesssim 1\text{ and }\Ep_P[|S_{u k}|^2] = \Ep_P[|f_u X_{u k}|^2]\gtrsim 1,
\end{equation}
and so Condition WL(ii) holds for some $\underline C$ and $\bar C$ depending only on the constants in Assumption \ref{ass: covariates}.

To verify Condition WL(iii), we apply Lemma \ref{PrimitiveWL}. Observe that $Y_u = 1\{Y\leq (1-u)\underline y + u\bar y\}$ and the class of functions $\{H(\cdot,u)\colon u\in\UU\}$ with $H(y,u) = 1\{y\leq (1-u)\underline y + u\bar y\}$ is VC-subgraph with index bounded by some $C$. Also, $X_u$ does not depend on $u$, and by Assumption \ref{ass: covariates}(iv,vii,viii), $\Ep_P[|X_{u k}|^4]\lesssim 1$ uniformly over $k\in[p]$ and
$$
(\Ep_P[\|X_u\|_\infty^{2q}])^{1/(2q)}\lesssim (\delta_n n^{1/2 - 1/q})^{1/2}.
$$
Moreover, by Assumption \ref{ass: density},
$
\Ep_P[|Y_u-Y_{u'}|^4] = |u - u'|
$
uniformly over $u,u'\in\UU$. Therefore, Lemma \ref{PrimitiveWL} with $2q$ replacing $q$ implies that Condition WL(iii) holds with $\Delta_n =(\log n)^{-1}$ and some $\varphi_n$ satisfying
$$
\varphi_n \lesssim \frac{\delta_n^{1/2}\log a_n}{n^{1/4}}\vee \frac{\log^{1/2}a_n}{n^{1/4}} = o(1),
$$
where the last assertion follows from $\log^{1/2}a_n \lesssim \delta_n n^{1/6}$, established above, and $\delta_n^2\log^4a_n = o(n)$, which holds by $\delta_n^2\log a_n = o(1)$.

%Condition WL(i) holds by Assumption \ref{ass: covariates}. To verify Condition WL(ii) we will apply Lemma \ref{PrimitiveWL}. Recall that $|Y_u|\leq 1$, $|Y_u-\Ep_P[Y_u\mid X_u]|\leq 1$ and $\Ep_P[X_{u j}^4]\leq C$. We have that $\Ep_P[ S_{uj}^2] = \Ep_P[f_u^2X_{u j}^2] \leq \frac{1}{4}\Ep_P[X_{u j}^2] \leq C$ and $\Ep_P[ S_{uj}^2]=\Ep_P[f_u^2X_{u j}^2]\geq \inf_{u\in\UU} \inf_{\|\delta\|=1}\Ep_P[f_u^2\{X_u'\delta\}^2] \geq c$. Also, $\Ep_P[\sup_{u\in\UU}|Y_u|^q\|X_u\|_\infty^q]^{1/q} \leq \Ep_P[\|X_u\|_\infty^q]^{1/q} \leq M_{n,2}$. Moreover, we have
%\begin{align*}
%\Ep_P[|Y_u-Y_{u'}|^4] & = \Ep_P[|1\{Y \leq (1-u)\underline{y}+u\bar y\} - 1\{ Y \leq (1-u')\underline{y}+u'\bar y\}|^4] \\
%& \leq |u-u'||\bar y-\underline{y}| \sup_{y\in [\underline{y},\bar y]}f_Y(y).
%\end{align*}
%Thus assumptions of Lemma \ref{PrimitiveWL} hold and we have with probability $1-o(1)$
%$$ \begin{array}{c}
%\sup_{u\in \UU} \max_{j\leq p}|(\En-\Ep_P)[S_{uj}^2]| \leq \delta_n, \ \ \ \max_{j\leq p}|(\En-\Ep_P)[X_{u j}^2]|\leq \delta_n, \\
%\sup_{|u-u'|\leq 1/n} \|\En[S_u-S_{u'}]\|_\infty \leq \delta_n n^{-1/2}, \ \ \sup_{|u-u'|\leq 1/n}\max_{j\leq[p]}|\Ep_P[S_{uj}^2-S_{u'j}^2]| \leq \delta_n, \end{array}$$
%which implies Condition WL(ii) (the extra result $\max_{j\leq p}|(\En-\Ep_P)[X_{u j}^2]|\leq \delta_n$ follows by similar arguments).

Next we verify Assumption \ref{ass: M}. It is well-known that the function $\theta\mapsto M_u(Y_u,X_u,\theta)$ is convex almost surely, which is the first requirement of Assumption \ref{ass: M}. Further, let us verify Assumption \ref{ass: M}(b). By Condition WL(iii), which was verified above, we have with probability $1-o(1)$ that $|(\En-\Ep_P)[S_{uk}^2]| =  o(1)$ uniformly over $u\in \UU$ and $k\in[p]$. So, it follows from (\ref{ULboundS2}) that with the same probability we have $\En[S_{uk}^2] = (1-o(1))\Ep_P[S_{uk}^2]$ uniformly over $u\in \UU$ and $k\in[p]$, and so Assumption \ref{ass: M}(b) holds for some $\Delta_n$, $\ell$, and $L$ satisfying $\Delta_n = o(1)$, $\ell = 1-o(1)$, and $L\lesssim 1$ for any $\hat \Psi_u$ such that
\begin{equation}\label{Eq:LoadReq}
\begin{array}{cc}
&(1-o(1))\Ep_P[S_{uk}^2] \leq \hat \Psi_{u kk}^2 \lesssim 1\\
&\text{ with probability $1-o(1)$ uniformly over $u\in\UU$ and $k\in[p]$}.
\end{array}
\end{equation}
Thus, it suffices to verify \eqref{Eq:LoadReq}. In the case $\bar m = 0$, we have by Lemma \ref{lemma:CCK} and Assumption \ref{ass: covariates}(i,iv,vii,viii) that $\En[X_{u k}^2] = (1-o(1))\Ep[X_{u k}^2]$ with probability $1-o(1)$ uniformly over $u\in\UU$ and $k\in[p]$. Thus, \eqref{Eq:LoadReq} holds since in this case,
$$
\hat \Psi_{ukk}^2=\frac{1}{4}\En[X_{uk}^2]=\frac{1-o(1)}{4}\Ep_P[X_{uk}^2]\lesssim 1
$$
and $4^{-1}\Ep[X_{u k}^2] \geq \Ep[f_u^2 X_{u k}^2] = \Ep[S_{u k}^2]$ (recall that $f_u^2\leq 1/4$) with probability $1-o(1)$ uniformly over $u\in\UU$ and $k\in[p]$.

To establish (\ref{Eq:LoadReq}) for $\bar m\geq 1$, we proceed by induction. Assuming that \eqref{Eq:LoadReq} holds when the number of loops in Algorithm \ref{AlgFunc} is $\bar m - 1$, we can complete the proof of the theorem to show that $\|X_u'(\widetilde \theta_u-\theta_u)\|_{\Pn,2} \lesssim (s_n\log a_n/n)^{1/2}$ with probability $1-o(1)$ uniformly over $u\in\UU$ for $m = \bar m - 1$. Then for $m = \bar m$, we have by the triangle inequality that
{\small
\begin{align*}
 |\hat l_{uk,m} - l_{u0k}| & \leq \Big( \En[ X_{u k}^2 \{ \Lambda(X_u'\theta_u)+r_u - \Lambda(X_u'\widetilde \theta_u)\}^2   ]  \Big)^{1/2} \\
 & \leq \Big( \|\Lambda(X_u'\theta_u)-\Lambda(X_u'\widetilde \theta_u)\|_{\Pn,2} +  \|r_u\|_{\Pn,2}\Big)\cdot \max_{1\leq i\leq n}\| X_{ui}\|_\infty  \lesssim_P \delta_n
\end{align*}}\!uniformly over $u\in\UU$ and $k\in[p]$ since  $\max_{1\leq i\leq n}\|X_{ui}\|_\infty \lesssim_P n^{1/(2q)}M_{n,2}$ by  Assumption \ref{ass: covariates}(vii), $\|r_u\|_{\Pn,2}\lesssim_P (s_n\log a_n/n)^{1/2}$ by Assumption \ref{ass: approximation error}(v), $n^{1/(2q)}M_{n,2} (s_n\log a_n/n)^{1/2} \leq \delta_n$ by Assumption \ref{ass: covariates}(viii), and the fact that $\Lambda$ is $1$-Lipschitz (observe that $M_{n,2}\geq 1$, and so $M_{n,2}\leq M_{n,2}^2$). Thus, (\ref{Eq:LoadReq}) holds with the number of loops in Algorithm \ref{AlgFunc} being $\bar m$. This completes verification of Assumption \ref{ass: M}(b).

To verify Assumption \ref{ass: M}(a), note that for any $\delta\in\mathbb R^p$,
\begin{align*}
&\{\partial_\theta M_u(Y_u,X_u,\theta_u)-\partial_\theta M_u(Y_u,X_u,\theta_u,a_u)\}'\delta \\
&\qquad = \{\Lambda(X_u'\theta_u) - \Lambda(X_u'\theta_u + a_u(X_u))\} X_u'\delta = -r_u X_u'\delta,
\end{align*}
and so
\begin{align*}
& \En[\partial_\theta M_u(Y_u,X_u,\theta_u)-\partial_\theta M_u(Y_u,X_u,\theta_u,a_u)]'\delta \\
&\qquad  \leq \|r_u/\sqrt{w_u}\|_{\Pn,2}\|\sqrt{w_u}X_u'\delta\|_{\Pn,2}  \leq C_n \|\sqrt{w_u}X_u'\delta\|_{\Pn,2}
\end{align*}
where the first inequality follows from the Cauchy-Schwarz inequality, and the second holds with probability $1-\bar\Delta_n$ for some $C_n$ satisfying $C_n \lesssim (s_n\log a_n/n)^{1/2}$ by Assumption \ref{ass: approximation error}. Thus, Assumption \ref{ass: M}(a) follows for given $C_n$ and $\Delta_n = \bar \Delta_n = o(1)$.

To verify Assumption \ref{ass: M}(c), note that Lemma O.2 in \cite{BCFH2013program} imply that for any $u\in\UU$, $A_u\subset\mathbb R^p$, and $\delta\in A_u$,
{\small
\begin{align*}
&\En[M_u(Y_u,X_u,\theta_u+\delta)] - \En[M_u(Y_u,X_u,\theta_u)] - \En[\partial_\theta M_u(Y_u,X_u,\theta_u)]'\delta \\
&\qquad + 2\|a_u/\sqrt{w_u}\|_{\Pn,2}\|\sqrt{w_u}X_u'\delta\|_{\Pn,2} \geq \Big(\|\sqrt{w_u}X_u'\delta\|_{\Pn,2}^2\Big)\wedge\Big(\bar q_{A_u}\|\sqrt{w_u}X_u'\delta\|_{\Pn,2}\Big),
\end{align*}}\!where
\begin{equation}\label{eq: thm 4.1 def q}
\bar q_{A_u} = \inf_{\delta\in A_u}\frac{(\En[w_u|X_u'\delta|^2])^{3/2}}{\En[w_u|X_u'\delta|^3]}.
\end{equation}
Next, we bound $\|a_u/\sqrt{w_u}\|_{\Pn,2}$. Fix some arbitrary value of $X_u$. Consider the case that $r_u = r_u(X_u)\geq 0$. Then $a_u = a_u(X_u)\geq 0$, and so combining the mean-value theorem and \eqref{eq: au def thm 4.1} shows that for some $t\in(0,a_u)$,
$$
r_u = a_u \Lambda'(X_u'\theta_u + t).
$$
Now, since the function $\Lambda'$ is unimodal,
$$
\Lambda'(X_u'\theta_u + t)\geq \Lambda'(X_u'\theta_u) \wedge \Lambda'(X_u'\theta_u + a_u).
$$
Further, observe that $\Lambda'(X_u'\theta_u + a_u) = f_u^2$ and
\begin{align*}
\Lambda'(X_u'\theta_u)
&= \Lambda(X_u'\theta_u)\cdot (1-\Lambda(X_u'\theta_u)) \\
&= (\Lambda(X_u'\theta_u) + r_u - r_u)\cdot(1-\Lambda(X_u'\theta_u) - r_u + r_u)\\
&= f_u^2 - r_u(1-\Lambda(X_u'\theta_u)) + r_u(\Lambda(X_u'\theta_u) + r_u)\\
&\geq f_u^2 - r_u + 2r_u\Lambda(X_u'\theta_u)\geq f_u^2 - r_u.
\end{align*}
In addition, by Assumption \ref{ass: approximation error}, $|r_u|\leq f_u^2/4$, so that $\Lambda'(X_u'\theta_u)\geq 3f_u^2/4$. Thus,
$$
|a_u|\leq 4r_u/(3f_u^2).
$$
Similarly, the same inequality can be obtained in the case that $r_u = r_u(X_u)<0$. Conclude that
$$
\|a_u/\sqrt{w_u}\|_{\Pn,2} \lesssim \|r_u/f_u^3\|_{\Pn,2}\lesssim \sqrt{s_n\log a_n/n}
$$
with probability at least $1-\bar\Delta_n$ uniformly over $u\in\UU$. Therefore, Assumption \ref{ass: M}(c) holds for any $A_u\subset \mathbb R^p$ with $\Delta_n = \bar\Delta_n$, $C_n \lesssim (s_n\log a_n/n)^{1/2}$ and $\bar q_{A_u}$ defined in \eqref{eq: thm 4.1 def q}.

%Assumption \ref{ass: M}(c) follows from Lemma \ref{Lemma:Minoration} with $C_n \geq \sup_{u\in\UU} \|r_u/f_u^3\|_{\Pn,2}$, where we have $\sup_{u\in\UU} \|r_u/f_u^3\|_{\Pn,2}\lesssim \{n^{-1}s\log (a_n)\}^{1/2}$ with probability $1-o(1)$ by Assumption \ref{ass: approximation error}, and $\bar q_{A_u}=\inf_{ \delta \in A_u} \En\[w_{u}| X_u'\delta|^2\]^{3/2}/\En\[w_{u}| X_u'\delta|^3\]$ where the define the sets $A_u = \Delta_{2\tilde c,u}\cup \{ \delta \in \RR^p :  \|\delta\|_1\leq \frac{3c\|\widehat\Psi_{u0}^{-1}\|_\infty}{\ell c - 1}\frac{n}{\lambda}C_n\|\sqrt{w_{u}} X_u'\delta\|_{\Pn,2}\}$.

Next, we apply Lemma \ref{Lemma:LassoMRateRaw}. We have to verify the condition on $\bar q_{A_u}$ required in the lemma. To do so, recall that $A_u = A_{u,1}\cup A_{u,2}$ where
\begin{align*}
&A_{u,1} = \{\delta\colon \|\delta^c_{T_u}\|_1\leq 2\tilde c\|\delta_{T_u}\|_1\},\\
&A_{u,2} =  \Big\{ \delta\colon \|\delta\|_1 \leq \frac{3n}{\lambda}\frac{c\|\widehat\Psi_{u0}^{-1}\|_\infty}{\ell c - 1}C_n\|\sqrt{w_{u}} X_u'\delta\|_{\Pn,2}\Big\}.
\end{align*}
Then $\bar q_{A_u}$ defined in \eqref{eq: thm 4.1 def q} equals $\bar q_{A_{u,1}}\wedge \bar q_{A_{u,2}}$ where $\bar q_{A_{u,1}}$ and $\bar q_{A_{u,2}}$ are defined similarly. To bound $\bar q_{A_{u,1}}$, we have
\begin{align*}
\bar q_{A_{u,1}}
&\geq \inf_{ \delta \in A_{u,1}} \frac{\En\[w_{u}| X_u'\delta|^2\]^{1/2}}{\max_{1\leq i\leq n}\| X_{u i}\|_\infty \|\delta\|_1 } \gtrsim_P \inf_{ \delta \in A_{u,1}} \frac{\En\[w_{u}| X_u'\delta|^2\]^{1/2}}{n^{1/(2q)}M_{n,2} \|\delta\|_1}\\
&\geq \inf_{ \delta \in A_{u,1}} \frac{\En\[w_{u}| X_u'\delta|^2\]^{1/2}}{n^{1/(2q)}M_{n,2} (1+2\tilde c)\sqrt {s_n} \|\delta_{T_u}\|}
\geq \frac{\bar \kappa_{2\tilde c}}{n^{1/(2q)}M_{n,2} (1+2\tilde c)\sqrt {s_n}}
\end{align*}
uniformly over $u\in\UU$ by Assumption \ref{ass: covariates}(viii) and definition of $\bar\kappa_{2\tilde c}$. By Lemma \ref{lemma:FSE}, sparse eigenvalues of order $\ell_n s_n$, for some sequence $\ell_n\to\infty$, are bounded away from zero and from above so that $\bar\kappa_{2\tilde c}$ is bounded away from zero with probability $1-o(1)$. Conclude that
$$
\bar q_{A_{u,1}} \gtrsim_P \frac{1}{n^{1/(2q)}M_{n,2} (1+2\tilde c)\sqrt {s_n}}\geq \frac{1}{\delta_n^{1/2}n^{1/4}}\gtrsim \left(\frac{s_n\log a_n}{\delta_n n}\right)^{1/2}
$$
uniformly over $u\in\mathcal U$ where the second inequality holds by Assumption \ref{ass: covariates}(viii) and the third by Assumption \ref{ass: covariates}(vi) (when we apply Assumption \ref{ass: covariates}(vi), we use the fact that $M_{n,1}\gtrsim 1$, which in turn follows from Assumption \ref{ass: covariates}(i)). Next, to bound $\bar q_{A_{u,2}}$, we have
\begin{align*}
\bar q_{A_{u,2}}
&\geq \inf_{ \delta \in A_{u,2}} \frac{\En\[w_{u}| X_u'\delta|^2\]^{1/2}}{\max_{1\leq i\leq n}\| X_{u i}\|_\infty \|\delta\|_1 } \gtrsim_P \inf_{ \delta \in A_{u,2}} \frac{\En\[w_{u}| X_u'\delta|^2\]^{1/2}}{n^{1/(2q)}M_{n,2} \|\delta\|_1}\\
&\geq \frac{\lambda}{3n C_n}\frac{\ell c - 1}{c}\frac{\|\widehat\Psi_{u0}^{-1}\|_\infty^{-1}}{n^{1/(2q)}M_{n,2}}\gtrsim_P \frac{\lambda}{C_n n^{1+1/(2q)}M_{n,2}}
\end{align*}
uniformly over $u\in\UU$ since $\sup_{u\in \UU}\|\widehat\Psi_{u0}^{-1}\|_\infty\lesssim 1$ with probability $1-o(1)$. Substituting $\lambda = c\sqrt n\Phi^{-1}(1 - \gamma/(2p N_n))$ and $C_n \lesssim (s_n \log a_n /n)^{1/2}$ gives
$$
\bar q_{A_{u,2}}\gtrsim_P \frac{1}{n^{1/(2q)}M_{n,2}\sqrt{s_n}}\geq \frac{1}{\delta_n^{1/2}n^{1/4}}\gtrsim\left(\frac{s_n\log a_n}{\delta_n n}\right)^{1/2}
$$
uniformly over $u\in\UU$. Moreover,
$$
\left(L+\frac{1}{c}\right)\|\widehat\Psi_{u0}\|_\infty\frac{\lambda\sqrt{s_n}}{n\bar\kappa_{2\tilde c}}+ 6\tilde c C_n \lesssim \left(\frac{s_n\log a_n}{n}\right)^{1/2}
$$
since $\sup_{u\in \UU}\|\widehat\Psi_{u0}\|_\infty\lesssim 1$ with probability $1-o(1)$. Hence, since $\delta_n = o(1)$, the condition on $\bar q_{A_u}$ required in Lemma \ref{Lemma:LassoMRateRaw} is satisfied with probability $1-o(1)$. In addition, note that \eqref{Eq:reg} holds with probability $1-o(1)$ by Lemma \ref{Thm:ChoiceLambda}. Therefore, applying Lemma \ref{Lemma:LassoMRateRaw} gives
\begin{equation}\label{eq: first claim thm 4.1}
\|\sqrt{w_u}X_u'(\hat\theta_u-\theta_u)\|_{\Pn,2} \lesssim (s_n\log a_n /n)^{1/2} \ \ \mbox{and} \ \ \|\hat\theta_u - \theta_u\|_1 \lesssim (s_n^2\log a_n /n)^{1/2}
\end{equation}
with probability $1-o(1)$ uniformly over $u\in\UU$.

The second inequality in \eqref{eq: first claim thm 4.1} gives the second inequality in the first asserted claim of the theorem. To transform the first inequality in \eqref{eq: first claim thm 4.1} into the first inequality in the first asserted claim of the theorem (and also to prove other claims), we apply Lemma \ref{Lemma:LassoMSparsity}. We have to verify \eqref{eq: def Ln}. To do so, note that
\begin{align*}
\sup_{u\in \UU}\max_{1\leq i\leq n}|X_{u i}'(\hat\theta_u-\theta_u)|
& \lesssim_P n^{1/(2q)}M_{n,2}\|\hat\theta_u-\theta_u\|_1 \\
& \lesssim \{n^{-1+1/q}M_{n,2}^2s_n^2\log a_n\}^{1/2} \lesssim \delta_n = o(1)
\end{align*}
by Assumption \ref{ass: covariates}(vii, viii) and since $M_{n,2}\geq 1$. Also, note that uniformly over $t$ and $\Delta t$ in $\RR$ with $|\Delta t|\leq 1$, we have
\begin{align}
|\Lambda(t + \Delta t) - \Lambda(t)|
& = \Big|\frac{e^{t+\Delta t}}{1 + e^{t + \Delta t}} - \frac{e^t}{1+e^t}\Big| = \frac{|e^{t+\Delta t} - e^t|}{(1+e^{t+\Delta t})(1+e^t)}\notag\\
&  = \frac{e^t|e^{\Delta t} - 1|}{(1+e^{t+\Delta t})(1+e^t)} \lesssim \frac{e^t|e^{\Delta t} - 1|}{(1+e^t)^2}\lesssim \Lambda'(t)\Delta t.\label{eq: l lipshitz derivation}
\end{align}
Thus, with probability $1-o(1)$,
\begin{align*}
&|[\partial_\theta M_u(Y_{u i},X_{u i},\hat\theta_u)-\partial_\theta M_u(Y_{u i},X_{u i},\theta_u)]'\delta|
\\
& \qquad = |\Lambda(X_{u i}'\widehat\theta_u) - \Lambda(X_{u i}'\theta_u)|\cdot|X_{u i}'\delta| \lesssim \Lambda'(X_{u i}'\theta_u)\cdot|X_{u i}'(\hat\theta_u-\theta_u)|\cdot |X_{u i}'\delta|
\end{align*}
uniformly over $i=1,\dots,n$ and $u\in\mathcal U$. Also, since $|r_u|\leq w_u/4$ by Assumption \ref{ass: approximation error} and $w_u = f_u^2\leq 1$,
$$
\Lambda'(X_{u i}'\theta_u) = f_{u i}^2 - r_{u i} + 2r_{u i}\Lambda(X_{u i}'\theta_u) + r_{u i}^2 \leq w_{u i} + 3|r_{u i}| + w_{u i}^2 \leq 3w_{u i} \leq 3\sqrt{w_{u i}};
$$
see the expression for $\Lambda'(X_{u i}'\theta_u)$ above in this proof. Therefore, for some constant $C$,
\begin{align*}
&|\En[\partial_\theta M_u(Y_u,X_u,\hat\theta_u)-\partial_\theta M_u(Y_u,X_u,\theta_u)]'\delta| \\
&\qquad \leq C\|\sqrt{w_u}X_u'(\hat\theta_u-\theta_u)\|_{\Pn,2}\|X_u'\delta\|_{\Pn,2} \leq L_n \|X_u'\delta\|_{\Pn,2}
\end{align*}
with probability $1 - o(1)$ uniformly over $u\in\UU$ for some $L_n$ satisfying $L_n \lesssim (s_n\log a_n/n)^{1/2}$. Hence, since $\sup_{u\in\UU}\semax{\ell_n s_n,u} \lesssim 1$ with probability $1-o(1)$ for some $\ell_n\to\infty$ sufficiently slowly by Lemma \ref{lemma:FSE}, Lemma \ref{Lemma:LassoMSparsity} implies that $\sup_{u\in\UU} \|\hat\theta_u\|_0 \lesssim s_n$ with probability $1-o(1)$, which is the second asserted claim of the theorem.

%with $\widetilde L_n = C\{n^{-1}s\log(a_n)\}^{1/2}$ for some constant $C$ where we used that and $0\leq w_{ui}\leq 1$. Thus, by Lemma  \ref{Lemma:LassoMSparsity}
%we have $\sup_{u\in\UU} \|\hat\theta_u\|_0 \lesssim s$ with probability $1-o(1)$ since $\sup_{u\in\UU}\semax{s\ell_n,u} \leq C$ with probability $1-o(1)$ Lemma \ref{lemma:FSE} under Assumptions \ref{ass: parameters} -- \ref{ass: approximation error} for $\ell\to\infty$ sufficiently slow.

In turn, since $\sup_{u\in\UU} \|\hat\theta_u\|_0 \lesssim s_n$ with probability $1-o(1)$, Lemma \ref{lemma:FSE} also establishes that
$$
\|\sqrt{w_u}X_u'(\hat\theta_u-\theta_u)\|_{\Pn,2}\gtrsim \|\sqrt{w_u}X_u'(\hat\theta_u-\theta_u)\|_{\Pr_P,2} \gtrsim \|\hat\theta_u-\theta_u\|
$$
with probability $1-o(1)$ uniformly over $u\in\UU$, where the inequality follows from Assumption \ref{ass: covariates}(i). Combining these inequalities with \eqref{eq: first claim thm 4.1} gives the first inequality in the first asserted claim of the theorem.

It remains to prove the claim about the estimators $\widetilde\theta_u$. We apply Lemma \ref{Lemma:PostLassoMRateRaw}. We have to verify the condition (\ref{qreq: PostSelection}) on $\bar q_{A_u}$ required in the lemma. To do so, we first bound $\bar q_{A_u}$ from below for $A_u=\{ \delta\in\RR^p \colon \|\delta\|_0\leq C s\}$ where $C$ is a constant such that $\widehat s_u +s_n \leq C s_n$ with probability $1-o(1)$ uniformly over $u \in \UU$. We have
\begin{align*}
\bar q_{A_u}&  = {\displaystyle \inf_{ \delta \in A_u}} \frac{\En\[w_{u}| X_u'\delta|^2\]^{3/2}}{\En\[w_{u}| X_u'\delta|^3\]} \geq {\displaystyle \inf_{ \delta \in A_u}} \frac{\En\[w_{u}| X_u'\delta|^2\]^{1/2}}{{\displaystyle \max_{1\leq i\leq n}}\| X_{u i}\|_\infty \|\delta\|_1 }\\
&\geq {\displaystyle \inf_{ \|\delta\|_0\leq Cs_n} }\frac{\En\[w_{u}| X_u'\delta|^2\]^{1/2}}{{\displaystyle \max_{1\leq i\leq n}}\| X_{u i}\|_\infty \sqrt{Cs_n}\|\delta\|_2 } \gtrsim_P\frac{\sqrt{\semin{Cs_n,u}}}{\sqrt{s_n}n^{1/(2q)}M_{n,2}}\gtrsim \frac{\log^{1/4}a_n}{\delta_n^{1/2} n^{1/4}}
\end{align*}
uniformly over $u\in\UU$, where the inequality preceding the last one follows from Assumption \ref{ass: covariates}(vii) and the definition of $\semin{Cs_n,u}$, and the last one follows from Assumption \ref{ass: covariates}(viii) and the observation that by Lemma \ref{lemma:FSE}, $\inf_{u\in\UU}\semin{Cs_n,u}$ is bounded away from zero with probability $1-o(1)$ uniformly over $u\in\UU$.

Next we bound from above the right-hand side of (\ref{qreq: PostSelection}). It follows by  (\ref{AuxMuUpper}) that uniformly over $u\in\UU$ with probability $1-o(1)$,
$$ \En[M_u(Y_u,X_u,\widetilde \theta_u)]-  \En[M_u(Y_u,X_u,\theta_u)] \lesssim s_n \log a_n/n$$
since $\lambda/n \lesssim (\log a_n/n)^{1/2}$, $\|\hat\theta_u-\theta_u\|_1\lesssim (s_n^2\log a_n/n)^{1/2}$, and  $\sup_{u\in\UU} \|\widehat\Psi_{u0}\|_\infty\lesssim 1$ with probability $1-o(1)$.
Furthermore, $C_n \lesssim (s_n\log a_n/n)^{1/2}$ and
$$
\sup_{u\in\UU}\|\En[S_u]\|_\infty \leq \sup_{u\in\UU}\|\widehat \Psi_{u0}\|_\infty\|\widehat \Psi_{u0}^{-1}\En[S_u]\|_\infty \lesssim \lambda/n
$$
with probability $1-o(1)$ by the choice of $\lambda$; see Lemma \ref{Thm:ChoiceLambda}. Hence, it follows that the right-hand side of (\ref{qreq: PostSelection}) is bounded up to a constant by $(s_n\log a_n/n)^{1/2}$ with probability $1-o(1)$ uniformly over $u\in\UU$. Also, since $s_n^2\log a_n/n\lesssim 1$ (see Assumption \ref{ass: covariates}(viii) and recall that $M_{n,2}\geq 1$) and $\delta_n = o(1)$, the condition (\ref{qreq: PostSelection})  on $\bar q_{A_u}$ required in Lemma \ref{Lemma:PostLassoMRateRaw} holds with probability $1-o(1)$ uniformly over $u\in\UU$.
Hence, Lemma \ref{Lemma:PostLassoMRateRaw} implies that
$$ \| \sqrt{w_u}X_u(\widetilde\theta_u-\theta_u)\|_{\Pn,2} \lesssim (s_n\log a_n/n)^{1/2}$$
with probability $1-o(1)$ uniformly over $u\in\UU$. Finally, as in the case of $\widehat\theta_u$'s, we also have $\| \sqrt{w_u}X_u(\widetilde\theta_u-\theta_u)\|_{\Pn,2}\gtrsim\|\widetilde\theta_u-\theta_u\|$ with probability $1-o(1)$ uniformly over $u\in\UU$, which gives the last asserted claim and completes the proof of the theorem. \qed
%\end{proof}

\subsection*{Proof of Theorem \ref{Thm:RateEstimatedLassoLinear}}
The strategy of this proof is similar to that of Theorem \ref{Thm:RateEstimatedLassoLogistic}. In particular, we will rely upon results in Appendix \ref{sec: generic results} with $\UU$ and $p$ replaced by $\UUU = \UU \times [\pp]$ and $\bar p = p + \pp - 1$, respectively, where for $\tu = (u,j)\in\UUU$, we set $Y_\tu = D_j$, $X_\tu=(X^j)'=(D_{[\pp]\setminus j}',X')'$, $\theta_\tu=\bar \gamma_u^j$, $\widehat\theta_\tu = \widehat\gamma_u^j$, $\widetilde \theta_\tu = \widetilde \gamma_u^j$, $a_\tu = (f_u,\bar r_{uj})$, $\bar r_\tu= \bar r_{u j}= X^j(\gamma_u^j-\bar\gamma_u^j)$,  and $w_\tu =\hat f_u^2$. Note that for all $\tu\in\UUU$, we have that $\theta_\tu$ satisfies \eqref{A:EqMainFunc} where
$$
M_\tu(Y_\tu,X_\tu,\theta,a) = 2^{-1}f^2(D_j - X^j\theta - r)^2
$$
for $\theta$ being a vector in $\mathbb R^{\bar p}$ and $a$ being a pair $(f,r)$ of functions of $D$ and $X$. Similarly, $\widehat\theta_\tu$ satisfies \eqref{Adef:LassoFunc} where $M_\tu(Y_\tu,X_\tu,\theta) = M_\tu(Y_\tu,X_\tu,\theta,(\widehat f_u,\mathcal O))$ and $\mathcal O = \mathcal O(D,X)$, the identically zero function of $D$ and $X$.

We first verify Condition WL with
$$
\epsilon_n=\frac{\delta_n^2}{n^{1/2+1/q}(p + \pp)^{1/2}(M_{n,1}^2\vee M_{n,2}^2)},
$$
$N_n = p\pp^2 n^2$, and the following semi-metric $d_{\UUU}$ on $\UUU$: for all $\tu = (u,j)$ and $\tu'=(u',j')$ in $\UUU$, $d_{\UUU}(\tu,\tu') = |u - u'|$ if $j = j'$ and $d_{\UUU}(\tu,\tu') = 1$ otherwise. Observe that $N(\epsilon,\UUU,d_{\UUU})\leq \pp/\epsilon$ for all $\epsilon>0$. Also, note that $M_{n,1}^2\vee M_{n,2}^2 \leq \delta_n n^{1/2-1/q}$ by Assumption \ref{ass: covariates}(vi,viii), and so
$$
1/\epsilon_n \leq n(p + \pp)^{1/2}/\delta_n\leq n^2(p + \pp)^{1/2}
$$
since $\delta_n\geq 1/n$ by Assumption \ref{ass: covariates}(i,vii,viii). Thus, $\epsilon_n$ and $N_n$ satisfy the inequality $N_n\geq N(\epsilon_n,\UUU,d_{\UUU})$, which is the first requirement of Condition WL.

Next, we verify Condition WL(i).
As in front of Condition WL in Appendix \ref{sec: generic results}, for $\tu = (u,j)$, let
\begin{align*}
S_{\tu} & = S_{uj} = \partial_\theta M_\tu(Y_\tu,X_\tu,\theta,a_\tu)|_{\theta = \theta_\tu}\\
& = - f_u^2(D_j-X^j\gamma_u^j)(X^j)'= - f_u^2Z_u^j(X^j)'.
\end{align*}
Then the inequality $\Phi^{-1}(1-t)\lesssim \sqrt{\log(1/t)}$, which holds uniformly over $t\in(0,1/2)$, implies that
$$
(\Ep_P[|S_{\tu k}|^3])^{1/3}\Phi^{-1}(1-\gamma/2pN_n) \lesssim (\Ep_P[ |Z_u^jX^j_k|^3])^{1/3}\log^{1/2} a_n\leq \delta_n n^{1/6}
$$
uniformly over $\tu\in\UUU$ and $k\in[\bar p]$, where the second inequality holds by Assumption \ref{ass: covariates}(iii). Hence, Condition WL(i) holds for some $\varphi_n$ satisfying $\varphi_n\lesssim \delta_n$.

To verify Condition WL(ii), note that by Assumption \ref{ass: covariates}(ii), we have $\Ep_P[S_{\tu k}^2]\gtrsim 1$ and $$
\Ep_P[S_{\tu k}^2] \leq \Ep_P[ |Z_u^jX^j_k|^2] \leq \Ep_P[ |Z_u^j|^4 + |X^j_k|^4] \lesssim 1
$$
uniformly over $\tu = (u,j) \in\UUU$ and $k\in[\bar p]$ by Assumption \ref{ass: covariates}(iv).

To verify Condition WL(iii), we use the decomposition
$$
S_{\tu k}-S_{\tu' k}= - (f^2_u-f^2_{u'})Z_u^jX_k^j + f_{u'}^2X^j(\gamma_u^j-\gamma_{u'}^j)X_k^j
$$
for $\tu=(u,j)$ and $\tu'=(u',j)$ in $\UUU$. By (\ref{eq: f-lipshitz}) and (\ref{eq: gamma-lipshitz}) we have
 \begin{equation}\label{Lipsf}
 |f_u^2-f_{u'}^2| \lesssim |u-u'| \ \ \ \mbox{and} \ \ \ \|\gamma_u^j-\gamma_{u'}^j\|_1\lesssim  \sqrt{p + \pp}|u-u'|
 \end{equation}
 uniformly over $u,u'\in\UU$ and $j\in[\pp]$. Therefore, uniformly over $\tu=(u,j)$ and $\tu'=(u',j)$ in $\UUU$ such that $|u-u'|\leq \epsilon_n$, we have
\begin{align}
|S_{\tu k}-S_{\tu' k}|& \lesssim \Big( |Z_u^j| \|X^j\|_\infty+\|X^j\|_\infty^2\sqrt{p + \pp}\Big)\cdot|u-u'|\nonumber\\
& \lesssim \Big( |Z_u^j|^2 +\|X^j\|_\infty^2\Big)\cdot \sqrt{p + \pp}\epsilon_n.\label{eq: ExtraS}
\end{align}
Since
$$
\Ep_P[\max_{1\leq i\leq n,j\in[\pp]}\sup_{u\in\UU} |Z_{ui}^j|^2] \leq n^{1/q}M_{n,1}^2
$$
and
$$
\Ep_P[\max_{1\leq i\leq n, j\in[\pp]}\|X^j_i\|_\infty^2] \leq n^{1/q}M_{n,2}^2
$$
by Assumption \ref{ass: covariates}(v,vii), we have by Markov's inequality that
with probability $1-o(1)$,
$$
\sup_{d_{\UUU}(\tu,\tu')\leq \epsilon_n}\max_{k\in[\bar p]} |\En[S_{\tu k}-S_{\tu' k}]| \lesssim  \delta_n n^{-1/2}.
$$
In addition, uniformly over $\tu,\tu'\in\UUU$ with $d_{\UUU}(\tu,\tu')\leq \epsilon_n$ and $k\in[\bar p]$,
\begin{align*}
|\Ep_P[S_{\tu k}^2-S_{\tu' k}^2]|
&\leq \Big(\Ep_P[(S_{\tu k}-S_{\tu' k})^2]\Big)^{1/2}\cdot \Big(\Ep_P[(S_{\tu k}+S_{\tu' k})^2]\Big)^{1/2}\\
& \lesssim \Big( (M_{n,1}^2 + M_{n,2}^2)(p+\pp)^{1/2}\epsilon_n\Big)^{1/2}\lesssim \delta_n.
\end{align*}
Further, let $\UU^\epsilon$ denote a minimal $\epsilon_n$-net for $\UU$. Using (\ref{Lipsf}) and (\ref{eq: ExtraS}), we obtain that with probability $1-o(1)$,
$$
\sup_{u\in \UU} \max_{j\in[\pp], k \in[\bar p]} |(\En-\Ep_P)[S_{u j k}^2]|  \lesssim  \sup_{u\in \UU^\epsilon} \max_{j\in [\pp], k \in[\bar p]} |(\En-\Ep_P)[S_{u j k}^2]| + \delta_n.
$$
To bound the first term on the right-hand side of this inequality, we apply Lemma \ref{thm:RV34} with $X_{u i}$ and $p$ replaced by $S_{\tu i}$ and $\bar p$, respectively, and $k=1$, where $S_{\tu i} = (S_{\tu 1 i},\dots,S_{\tu [\bar p] i})' = - f_{u i}^2 Z_{u i}^j (X_i^j)'$ for $\tu = (u,j)\in\UUU$ and $i=1,\dots,n$. With
\begin{align*}
K^2 & = \Ep_P\left[ \max_{1\leq i\leq n} \sup_{u\in \UU} \max_{j\in[\pp],k \in[\bar p]} S_{u j k i}^2 \right]\\
& \leq \Ep_P\left[ \max_{1\leq i\leq n} \sup_{u\in\UU,j\in[\pp]} |Z_{ui}^j|^2\|X_i^j\|_\infty^2\right] \lesssim n^{2/q}(M_{n,1}^4 + M_{n,2}^4),
\end{align*}
which holds by Assumption \ref{ass: covariates}(v,vii), the lemma yields
\begin{align*}
\displaystyle  \max_{u\in \UU^\epsilon} \max_{j\in [\pp], k \in[\bar p]} \Big|(\En-\Ep_P)[S_{u j k}^2]\Big|
& \lesssim_P n^{-1/2}n^{1/q}(M_{n,1}^2 + M_{n,2}^2) \log a_n \\
&\lesssim \delta_n \log^{1/2}a_n = o(1)
\end{align*}
by Assumption \ref{ass: covariates}(vi) and (viii). Thus, Condition WL(iii) holds with some $\Delta_n$ and $\varphi_n$ satisfying $\Delta_n = o(1)$ and $\varphi_n = o(1)$.

Next, we verify Assumption \ref{ass: M}. The function $\theta\mapsto M_\tu(Y_\tu,X_\tu,\theta)$ is convex almost surely, which is the first requirement of Assumption \ref{ass: M}. Further, to verify Assumption \ref{ass: M}(a), note that
{\small
$$
[\partial_\theta M_u(Y_u,X,\theta_u)-\partial_\theta M_u(Y_u,X,\theta_u,a_u)]'\delta=  - \Big(\hat f_u^2 \bar r_{uj} +(\hat f_u^2 - f_u^2)Z_u^j\Big)\cdot X^j\delta,
$$}\!so that by the Cauchy-Schwarz and triangle inequalities, since $w_u = \hat f_u^2$, we have
\begin{align*}
&|\En[\partial_\theta M_u(Y_u,X,\theta_u)-\partial_\theta M_u(Y_u,X,\theta_u,a_u)]'\delta| \\
&\qquad \leq \Big(\|\hat f_u\bar r_{uj}\|_{\Pn,2}+\|(\hat f_u^2 - f_u^2)Z_u^j/\hat f_u\|_{\Pn,2}\Big)\cdot \|\sqrt{w_u} X^j\delta\|_{\Pn,2}. %\\
% & \leq C_n \|\hat f_u X^j\delta\|_{\Pn,2}.
\end{align*}
To bound $\sup_{u\in \UU, j\in[\pp]}\|\hat f_u\bar r_{uj}\|_{\Pn,2}$, note that $\hat f_u \leq 1$, and so Lemma \ref{ControlRUJ} shows that with probability $1-o(1)$,
$$
\sup_{u\in\UU, j\in[\pp]}\|\hat f_u\bar r_{uj}\|_{\Pn,2} \leq \sup_{u\in\UU, j\in[\pp]}\|\bar r_{uj}\|_{\Pn,2}\lesssim (s_n\log a_n/n)^{1/2}.
$$
Also, by Lemma \ref{ControlExtra}, with probability $1-o(1)$,
\begin{equation}\label{eq: aaa thm 4.2}
\sup_{u\in \UU, j\in[\pp]} \|(\hat f_u^2 - f_u^2)Z_u^j/\hat f_u\|_{\Pn,2} \lesssim (s_n\log a_n/n)^{1/2}.
\end{equation}
Hence, Assumption \ref{ass: M}(a) holds for some $\Delta_n$ and $C_n$ satisfying $\Delta_n = o(1)$ and $C_n \lesssim (s_n\log a_n/n)^{1/2}$.

To prove Assumption \ref{ass: M}(b), as in Appendix \ref{sec: generic results}, for $\tu = (u,j) \in \UUU$, let $\widehat\Psi_{\tu 0} = \widehat\Psi_{u j 0} = \diag(\{l_{\tu 0 k}, k\in[\bar p]\})$ where $l_{\tu 0 k} = l_{u j 0 k} = (\En[S_{\tu k}^2])^{1/2}$, $k\in[\bar p]$. Note that by Condition WL(ii,iii), which is verified above,
$$
1\lesssim \widehat\Psi_{\tu 0 k}\lesssim 1
$$
with probability $1-o(1)$ uniformly over $\tu\in\UUU$ and $k\in[\bar p]$. Now, suppose that $\bar m = 0$ (even though Algorithm \ref{AlgFunc2} requires $\bar m\geq 1$). Then uniformly over $u\in\UU$, $j\in[\pp]$, and $k\in[\bar p]$ with probability $1-o(1)$,
\begin{align*}
\hat l_{u j k,0}\gtrsim \Big(\En[\hat f_u^4(D_j X^j_k)^2]\Big)^{1/2}\gtrsim \Big(\En[f_u^4(D_j X_k^j)^2]\Big)^{1/2}\gtrsim 1
\end{align*}
where the second inequality follows from the observation that $|\hat f_{u i}^2 - f_{u i}^2|\leq f_{u i}^2$ with probability $1- o(1)$ uniformly over $i=1,\dots,n$ and $u\in\UU$ (see \eqref{Multiplicative} in the proof of Lemma \ref{ControlExtra}), and the third from the same derivations as those used to obtain Condition WL(ii,iii). Also, uniformly over $u\in\UU$, $j\in[\pp]$, and $k\in[\bar p]$,
$$
\hat l_{u j k,0}\leq \max_{1\leq i\leq n}\|X^j_i\|_\infty(\En[D_j^2])^{1/2}\lesssim_P n^{1/(2q)}M_{n,2}
$$
by Assumption \ref{ass: covariates}(vii) since $\hat f_u\leq 1$. Therefore, Assumption \ref{ass: M}(b) holds with some $\Delta_n$, $L$, and $\ell$ satisfying $\Delta_n = o(1)$, $L\lesssim n^{1/(2q)}M_{n,2}\log^{1/2}a_n$, and $\ell\gtrsim 1$.

To establish Assumption \ref{ass: M}(b) for $\bar m\geq 1$, which is required by Algorithm \ref{AlgFunc2}, we proceed by induction. Assuming that Assumption \ref{ass: M}(b) holds with some $\Delta_n$, $\ell$, and $L$ satisfying $\Delta_n = o(1)$, $\ell\gtrsim 1$, and $L\lesssim n^{1/(2q)}M_{n,2}\log^{1/2}a_n$ when the number of loops in Algorithm \ref{AlgFunc2} is $\bar m - 1$, we can complete the proof of the theorem to show that $\|\hat f_u X^j(\widetilde \gamma_u^j-\bar\gamma_u^j)\|_{\Pn,2} \lesssim (L+1)(s_n\log a_n/n)^{1/2}$ with probability $1-o(1)$ uniformly over $u\in\UU$ and $j\in[\pp]$ for $m = \bar m - 1$. Thus, by the triangle inequality,
\begin{equation}\label{eq: thm 4.2 bbb}
\|\hat f_u X^j(\widetilde\gamma_u^j - \gamma_u^j)\|_{\Pn,2} \lesssim (L+1)(s_n^2\log a_n/n)^{1/2}
\end{equation}
with probability $1-o(1)$ uniformly over $u\in\UU$ and $j\in[\pp]$ since
\begin{align*}
\|\hat f_u X^j(\bar \gamma_u^j - \gamma_u^j)\|_{\Pn,2}
&\leq \|X^j(\bar\gamma_u^j - \gamma_u^j)\|_{\Pn,2}\\
&\leq \max_{1\leq i\leq n}\|X_i^j\|\cdot \|\bar\gamma_u^j - \gamma_u^j\|_1 \\
&\lesssim_P n^{1/(2q)}M_{n,2}(s_n^2\log a_n/n)^{1/2}
\end{align*}
uniformly over $u\in\UU$ and $j\in[\pp]$. Then for $m = \bar m$, we have uniformly over $u\in\UU$, $j\in[\pp]$, and $k\in[\bar p]$,
{\small
\begin{align*}
 |\hat l_{u j k,m} - l_{u j 0 k}| & = \Big| \Big(\En[ \hat f_u^4 (X_k^j)^2 (D^j-X^j\widetilde\gamma^j_u)^2   ]  \Big)^{1/2} \\
 &\quad - \Big( \En[  f_u^4 (X_k^j)^2 (D^j-X^j\gamma^j_u)^2   ]  \Big)^{1/2} \Big|\\
 & \leq \Big| \Big(\En[ \hat f_u^4 (X_k^j)^2 (X^j(\widetilde\gamma^j_u-\gamma^j_u))^2   ]  \Big)^{1/2} \\
 &\quad - \Big( \En[ (\hat f_u^2 - f_u^2)^2 (X_k^j)^2 (D^j-X^j\gamma^j_u)^2 ]  \Big)^{1/2} \Big|\\
 & \lesssim \Big(  \|\hat f_u X^j(\widetilde \gamma_u^j - \gamma_u^j)\|_{\Pn,2} +  \|(\hat f_u^2 - f_u^2)Z_u^j/\hat f_u\|_{\Pn,2} \Big)\cdot \max_{1\leq i\leq n}\| X^j\|_\infty  \\
 & \lesssim_P n^{1/(2q)}M_{n,2}(s_n^2\log a_n^2/n)^{1/2} n^{1/(2q)}M_{n,2}\\
 & = n^{-1/2 + 1/q} M_{n,2}^2 s_n\log a_n  \lesssim \delta_n \log^{1/2} a_n = o(1),
 \end{align*}}\!where the first inequality follows from the triangle inequality and the observation that $(\hat f_u^2 - f_u^2)^2\leq f_u^4 - \hat f_u^4$, the second from $\hat f_i\leq 1$ and $f_u\leq 1$, the third from \eqref{eq: aaa thm 4.2} and \eqref{eq: thm 4.2 bbb}, and the last from Assumption \ref{ass: covariates}(viii). Thus, for $\bar m\geq 1$, Assumption \ref{ass: M}(b) holds for some $\Delta_n$, $\ell$, and $L$ satisfying $\Delta_n = o(1)$, $\ell\gtrsim 1$, and $L\lesssim 1$.

Further, Assumption \ref{ass: M}(c) holds with $\Delta_n = 0$ and $\bar q_{A_\tu} = \infty$ for any $A_\tu$ since for any $\tu = (u,j)\in\UUU$ and $\delta\in\mathbb R^{\bar p}$, we have
\begin{align*}
&\En[M_{\tu}(Y_{\tu},X_\tu,\theta_\tu + \delta)] - \En[M_\tu(Y_\tu,X_\tu,\theta_\tu)] \\
&\qquad = - \En[\hat f_u^2(D_j - X^j\bar \gamma_u^j)X^j\delta] + 2^{-1}\En[(\hat f_u X^j\delta)^2]
\end{align*}
and
$$
2^{-1}\En[(\hat f_u X^j\delta)^2]\geq \En[\hat f_u (X^j\delta)^2] = \|\sqrt{w_u}X_\tu'\delta\|_{\Pn,2}
$$
where we used $\hat f_u \leq 1/2$.

We are now ready to apply Lemma \ref{Lemma:LassoMRateRaw}. Observe that by Lemma \ref{Thm:ChoiceLambda}, $\lambda$ satisfies \eqref{Eq:reg} with probability $1-o(1)$. Also, as established in \eqref{Multiplicative} in the proof of Lemma \ref{ControlExtra}, $|\hat f_{ui}^2 - f_{ui}^2|\leq f_{ui}^2/2$ with probability $1-o(1)$ uniformly over $i=1,\dots,n$ and $u\in\UU$. Therefore, since Lemma \ref{lemma:FSE} implies that for some $\ell_n$ satisfying $\ell_n\to\infty$,
$$
1\lesssim \min_{\|\delta\|_0\leq \ell_n s_n}\frac{\|f_u X_\tu'\delta\|_{\Pn,2}^2}{\|\delta\|^2}\leq \max_{\|\delta\|_0\leq \ell_n s_n}\frac{\|X_\tu'\delta\|_{\Pn,2}^2}{\|\delta\|^2}\lesssim 1
$$
with probability $1-o(1)$ uniformly over $\tu\in\UUU$, it follows that $\bar\kappa_{2\tilde c}\gtrsim 1$ with probability $1-o(1)$. In addition, $\sup_{\tu\in\UUU}\|\widehat\Psi_{\tu 0}\|_\infty\lesssim 1$ and $\sup_{\tu\in\UUU}\|\widehat\Psi_{\tu 0}^{-1}\|_\infty\lesssim 1$ with probability $1-o(1)$.
Therefore, applying Lemma \ref{Lemma:LassoMRateRaw} gives
{\small
\begin{equation}\label{eq: first claim thm 4.2}
\|\hat f_uX^j(\hat\gamma_u^j-\bar \gamma_u^j)\|_{\Pn,2} \lesssim (s_n\log a_n/n)^{1/2} \mbox{ and } \|\hat\gamma_u^j-\bar \gamma_u^j\|_1 \lesssim (s_n^2\log a_n/n)^{1/2}
\end{equation}}\!with probability $1-o(1)$ uniformly over $u\in\UU$ and $j\in[\pp]$.

The second inequality in \eqref{eq: first claim thm 4.2} gives the second inequality in the first asserted claim of the theorem. To transform the first inequality in \eqref{eq: first claim thm 4.2} into the first inequality in the first asserted claim of the theorem (and also to prove other claims), we apply Lemma \ref{Lemma:LassoMSparsity}. We have to verify \eqref{eq: def Ln}. To do so, note that for $\tu = (u,j)\in\UUU$, we have
$$
|[\partial_\theta M_\tu(Y_\tu,X_\tu,\hat\theta_\tu)-\partial_\theta M_\tu(Y_\tu,X_\tu,\theta_\tu)]'\delta|\leq |\hat f_u X^j(\hat\theta_\tu-\theta_\tu)|\cdot |\hat f_uX^j\delta|.
$$
Therefore, by the Cauchy-Schwarz inequality and since $\hat f_u \leq 1$,
\begin{align*}
&|\En[\partial_\theta M_\tu(Y_\tu,X_\tu,\hat\theta_\tu)-\partial_\theta M_\tu(Y_\tu,X_\tu,\theta_\tu)]'\delta|\\
&\qquad \leq \|\hat f_u X^j(\hat\gamma_u^j-\bar\gamma_u^j)\|_{\Pn,2}\|\hat f_uX^j\delta\|_{\Pn,2} \leq L_n \|X^j\delta\|_{\Pn,2}
\end{align*}
with probability $1-o(1)$ uniformly over $\tu = (u,j)\in\UUU$ for some $L_n$ satisfying $L_n \lesssim (s_n\log a_n/n)^{1/2}$. Thus, since $\sup_{\tu\in\UUU}\semax{\ell_n s_n,\tu} \lesssim 1$ for some $\ell_n\to\infty$ with probability $1-o(1)$ by Lemma \ref{lemma:FSE}, it follows from Lemma  \ref{Lemma:LassoMSparsity} that $\sup_{u\in\UU} \|\hat\gamma_u^j\|_0 \lesssim s_n$ with probability $1-o(1)$ uniformly over $u\in\UU$ and $j\in[\pp]$, which is the second asserted claim of the theorem.

In turn, with probability $1-o(1)$, uniformly over $u\in \UU$ and $j\in[\pp]$, we have
$$
\|\hat f_uX^j(\hat\gamma_u^j-\bar\gamma_u^j)\|_{\Pn,2} \gtrsim \| f_uX^j(\hat\gamma_u^j-\bar\gamma_u^j)\|_{\Pn,2} \gtrsim \|\hat\gamma_u^j-\bar\gamma_u^j\|.
$$
Combining these inequalities with \eqref{eq: first claim thm 4.2} gives the first inequality in the first asserted claim of the theorem.

It remains to prove the claim about the estimators $\widetilde\gamma_u^j$. We apply Lemma \ref{Lemma:PostLassoMRateRaw}. The condition (\ref{qreq: PostSelection}) on $\bar q_{A_\tu}$ required in the lemma holds almost surely since $\bar q_{A_\tu}=\infty$. Also, it follows from (\ref{AuxMuUpper}) that uniformly over $u\in\UU$ and $j\in[\pp]$ with probability $1-o(1)$,
$$
\En[M_\tu(Y_\tu,X_\tu,\widetilde \theta_\tu)] - \En[M_\tu(Y_\tu,X_\tu,\theta_\tu)]  \lesssim s_n \log a_n/n
$$
since $\lambda/n \lesssim (\log a_n/n)^{1/2}$, $\sup_{u\in\UU,j\in[\pp]}\|\hat\gamma_u^j-\bar\gamma_u^j\|_1\lesssim (s_n^2\log a_n/n)^{1/2}$, and $\sup_{u\in\UU,j\in[\pp]} \|\widehat\Psi_{uj0}\|_\infty\lesssim 1$ with probability $1-o(1)$.
Furthermore, $C_n \lesssim (s_n\log a_n/n)^{1/2}$ and
$$
\sup_{\tu\in\UUU}\|\En[S_{\tu}]\|_\infty \leq \sup_{\tu\in\UUU}\|\widehat \Psi_{\tu 0}\|_\infty\|\widehat \Psi_{\tu 0}^{-1}\En[S_{\tu}]\|_\infty \lesssim \lambda/n
$$
with probability $1-o(1)$ by the choice of $\lambda$; see Lemma \ref{Thm:ChoiceLambda}. In addition, uniformly over $\tu\in\UUU$ with probability $1-o(1)$, we have $\widehat s_\tu + s_n \lesssim s_n$ and $\semin{C s_n,\tu}\gtrsim 1$ for arbitrarily large $C$. Hence, by Lemma \ref{Lemma:PostLassoMRateRaw},
$$
\| \sqrt{w_\tu} X^j(\widetilde\gamma_u^j-\bar\gamma_u^j)\|_{\Pn,2} \lesssim (s_n\log a_n/n)^{1/2}
$$
with probability $1-o(1)$ uniformly over $u\in\UU$ and $j\in[\pp]$. Finally, as in the case of $\widehat\gamma_u^j$'s, we also have $ \| \sqrt{w_\tu} X^j(\widetilde\gamma_u^j-\bar\gamma_u^j)\|_{\Pn,2}\gtrsim \|\widetilde\gamma_u^j-\bar\gamma_u^j\|$ with probability $1-o(1)$ uniformly over $u\in\UU$ and $j\in[\pp]$, which gives the last asserted claim and completes the proof of the theorem. \qed

\section{Auxiliary Lemmas for Proofs of Theorems 4.1 and 4.2}

\begin{lemma}[Control of Approximation Error]\label{ControlRUJ}
Suppose that Assumptions \ref{ass: parameters} -- \ref{ass: approximation error} hold for all $P\in\mathcal P_n$. Then for $\bar r_{uj}=X^j(\gamma_u^j - \bar\gamma_u^j)$,  we have with probability $1-o(1)$ that
$$
\sup_{u\in\UU,j\in[\pp]}\En[\bar r_{uj}^2] \lesssim s_n\log a_n/n
$$
uniformly over $P\in\mathcal P_n$.
\end{lemma}

\begin{lemma}[Control of Estimated Weights and Score]\label{ControlExtra}
Suppose that Assumptions \ref{ass: parameters} -- \ref{ass: approximation error} hold for all $P\in\mathcal P_n$. Then with probability $1-o(1)$, we have
$$
\sup_{u\in \UU, j\in[\pp]} \|(\hat f_u^2 - f_u^2)Z_u^j/\hat f_u\|_{\Pn,2} \lesssim (s_n\log a_n/n)^{1/2}
$$
uniformly over $P\in\mathcal P_n$.
\end{lemma}

\begin{lemma}[Functional Sparse Eigenvalues]\label{lemma:FSE}
Suppose that Assumptions \ref{ass: parameters} -- \ref{ass: approximation error} hold for all $P\in\mathcal P_n$. Then for $\ell_n\to\infty$ slowly enough, we have
\begin{align*}
 \displaystyle \sup_{u\in\UU} \sup_{\|\delta\|_0\leq \ell_n s_n, \|\delta\|=1}| 1 - \| f_u (D',X')\delta\|_{\Pn,2}/\| f_u (D',X')\delta\|_{\Pr,2} | = o_P(1) \ \ \mbox{and} \\
  \displaystyle \sup_{\|\delta\|_0\leq \ell_n s_n, \|\delta\|=1}| 1 - \|  (D',X')\delta\|_{\Pn,2}/\|  (D',X')\delta\|_{\Pr,2} | = o_P(1)
\end{align*}
uniformly over $P\in\mathcal P_n$.
\end{lemma}

\subsection*{Proof of Lemma \ref{ControlRUJ}}
By Assumption \ref{ass: sparsity}, $\sup_{u\in\UU}\max_{j\in[\pp]}\|\bar\gamma_u^j\|_0 \leq s_n$ and
$$
\sup_{u\in\UU}\max_{j\in [\pp]}\Big(\|\bar \gamma_u^j-\gamma_u^j\| + s_n^{-1/2}\|\bar\gamma_u^j-\gamma_u^j\|_1\Big)\leq C_1(s_n\log a_n/n)^{1/2}.
$$
Also, by (\ref{eq: gamma-lipshitz2}), we have $\|\gamma_u^j-\gamma_{u'}^j\|\leq L_\gamma|u-u'|$ for some constant $L_\gamma$ uniformly over $u,u'\in\UU$. By the triangle inequality,
$$ \sup_{u\in\UU,j\in[\pp]} \En[\bar r_{uj}^2] \leq  \sup_{u\in\UU,j\in[\pp]}| (\En-\Ep_P)[\bar r_{uj}^2]| + \sup_{u\in\UU,j\in[\pp]}\Ep_P[\bar r_{uj}^2].$$

Consider the class of functions $\mG = \{ (D,X)\mapsto X^j(\gamma_u^j - \bar\gamma_u^j)\colon u\in \UU, j\in[\pp]\}$ and $\mG_{j,T} = \{(D,X)\mapsto X^j(\gamma_u^j - \gamma_{uT}^j)\colon u\in \UU\}$ for $j\in[\pp]$ and $T\subset [p + \pp -1]$ being a subset of the components of $X^j$ with $|T|\leq s_n$. Since $\bar\gamma_u^j = \gamma_{u T}^j$ for some $T = T_u^j$, it follows that $\mG \subset \cup_{j\in[\pp],|T|\leq s_n}\mG_{j,T}$. Also, we have $\|\gamma_{uT}^j-\gamma_{u'T}^j\|\leq \|\gamma_u^j-\gamma_{u'}^j\|$ for all $u,u'\in\UU$, $j\in[\pp]$, and $T\subset [p + \pp -1]$. Therefore, for fixed $j$ and $T$, we have
\begin{align*}
& \Big|(X^j(\gamma_{uT}^j-\gamma_u^j))^2-(X^j(\gamma_{u'T}^j-\gamma_{u'}^j))^2\Big|  \\
& \qquad = \Big|X^j(\gamma_{uT}^j-\gamma_u^j+
\gamma_{u'T}^j-\gamma_{u'}^j)X^j(\gamma_{uT}^j-\gamma_{u'T}^j+\gamma_{u'}^j-\gamma_u^j)\Big|\\
& \qquad \leq \|(D',X')'\|_\infty^2 \|\gamma_{uT}^j-\gamma_u^j+\gamma_{u'T}^j-\gamma_{u'}^j\|_1\|\gamma_{uT}^j-\gamma_{u'T}^j+\gamma_{u'}^j-\gamma_u^j\|_1\\
&\qquad  \leq  8\|(D',X')'\|_\infty^2\sup_{u\in\UU}\|\gamma_u^j\|_1 \{p+\pp\}^{1/2}\|\gamma_u^j-\gamma_{u'}^j\|\\
&\qquad \leq \Big(M^{-1} L'_\gamma \|(D',X')'\|_\infty^2\Big)M|u-u'|.
\end{align*}
where $L_\gamma'=8\sup_{u\in\UU,j\in[\pp]}\|\gamma_u^j\|_1\{p+\pp\}^{1/2}L_\gamma \lesssim a_n$ and we will set $M=a_n^2$.
Therefore, we have for the envelope $G(D,X) = \|(D',X')'\|_\infty^2(M^{-1}L_\gamma'+\sup_{u\in\UU,j\in[\pp]}\|\bar\gamma_u^j-\gamma_u^j\|_1^2)$ that for all $0<\epsilon\leq 1$ and all finitely-discrete probability measures $Q$,
\begin{align*}
\log N(\epsilon \|G\|_{Q,2}, \mG^2, \|\cdot\|_{Q,2}) & \lesssim s_n\log a_n +  \max_{j\in [\pp],|T|\leq s} \log\left(\epsilon \|G\|_{Q,2}, \mG_{j,T}^2, \|\cdot\|_{Q,2}\right)\\
& \lesssim s_n\log a_n + \log N\left(\epsilon/M,  \UU, d_\UU\right)  \\
& \lesssim  s_n\log a_n +  \log (a_n/\epsilon) \lesssim s_n\log(a_n/\epsilon).
\end{align*}
By Lemma \ref{lemma:CCK}, since
\begin{align*}
\left\|\max_{1\leq i\leq n} G(D_i,X_i) \right\|_{P,2}
& \lesssim n^{1/q}M_{n,2}\left(a_n^{-2}L_\gamma' + \sup_{u\in\UU, j\in [\pp]}\|\bar\gamma_u^j-\gamma_u^j\|_1^2\right) \\
& \lesssim n^{-1+1/q}M_{n,2}s^2_n\log a_n,
\end{align*}
we have with probability $1-o(1)$ that
\begin{align*}
 \sup_{u\in \UU, j\in[\pp]} |(\En-\Ep_P)[\bar r_{uj}^2]| & \lesssim \sqrt{ \frac{s_n\log a_n \sup_{u\in\UU,j\in[\pp]}\Ep_P[\bar r_{uj}^4]}{n}} \\
 &\qquad + \frac{s_n n^{-1+1/q}M_{n,2}s_n^2\log^2 a_n}{n}\\
& \lesssim \sqrt{\frac{s_n\log a_n}{n}}\frac{s_n\log a_n }{n} +\delta_n\frac{s_n\log a_n}{n}\lesssim \frac{s_n\log a_n}{n},
\end{align*}
where we used that $s_n\geq 1$,  $\bar r_{uj}=X^j(\gamma_u^j-\bar \gamma_u^j)$, $\Ep_P[\{(D',X')\xi\}^4] \leq C_1\|\xi\|^4$ by Assumption \ref{ass: covariates}(iv), $\|\bar\gamma_u^j-\gamma_u^j\|^2 \leq C_1^2 s_n\log a_n/n$ by Assumption \ref{ass: sparsity}, $M_{n,2} s_n^2\log a_n\leq \delta_n n^{1-1/q}$ implied by Assumption \ref{ass: covariates}(viii). Finally, the result follows since  $\Ep_P[\bar r_{uj}^2] = \Ep_P[\{X^j(\gamma_u^j-\bar \gamma_u^j)\}^2] \lesssim \|\gamma_u^j-\bar\gamma_u^j\|^2 \lesssim s_n\log a_n/n$ by Assumption \ref{ass: sparsity}.
\qed

\subsection*{Proof of Lemma \ref{ControlExtra}}
Recall that $\widehat f_u^2 = \widehat f_u^2(D,X) = \Lambda'(D'\widetilde \theta_u + X'\widetilde\beta_u)$ and that by Theorem \ref{Thm:RateEstimatedLassoLogistic}, we have
$$
\sup_{u\in\UU} (\|\widetilde \theta_u - \theta_u\| + \|\widetilde \beta_u- \beta_u\|) \lesssim (s_n\log a_n/n)^{1/2}
$$
and $\sup_{u\in \UU}\|(\widetilde \theta'_u,\widetilde \beta'_u)'\|_0 \lesssim s_n$ with probability $1-o(1)$. Also,
\begin{align*}
&\max_{1\leq i\leq n}|(D_i',X'_i)\{ (\widetilde \theta_u',\widetilde \beta_u')' - (\theta_u', \beta_u')'\}| \\ & \qquad \leq \max_{1\leq i\leq n}\|(D_i',X'_i)'\|_\infty\{\|\widetilde\theta_u-\theta_u\|_1+\|\widetilde\beta_u-\beta_u\|_1\} \\
&\qquad \lesssim_P n^{1/(2q)}M_{n,2}(s_n^2\log a_n/n)^{1/2} \leq \delta_n
\end{align*}
by Assumption \ref{ass: covariates}(vii, viii) since $M_{n,2}\geq 1$. Thus, for $\tilde t_{ui}=D'_i\widetilde \theta_u+ X_i'\widetilde \beta_u$ and $t_{ui}=D_i'\theta_u+X_i'\beta_u$, we have with probability $1-o(1)$ that $\sup_{u\in\UU,i\in[n]}|\tilde t_{ui}-t_{ui}|\leq \delta_n^{1/2} = o(1)$, and so $|\Lambda(\tilde t_{u i}) - \Lambda(t_{u i})|\lesssim \Lambda'(t_{u i})|\tilde t_{u i} - t_{u i}|$ uniformly over $u\in\UU$ and $i=1,\dots,n$ as in \eqref{eq: l lipshitz derivation}. Hence, the inequality $|x(1-x) - y(1-y)|\leq |x-y|$, which holds for all $x,y\in[0,1]$, implies that with probability $1-o(1)$,
\begin{equation}\label{Multiplicative}
|\hat f_{ui}^2-f_{ui}^2| \leq  |\G(\tilde t_{ui})-\G(t_{ui})-r_{ui}| \lesssim \G'(t_{ui})|\tilde t_{ui}-t_{ui}|+ |r_{ui}| \leq f_{ui}^2/2
\end{equation}
since $|r_u|\leq f_u^2/4$ by Assumption \ref{ass: approximation error} and
\begin{align}
\G'(t_{ui}) & = f_{ui}^2+2\Lambda(t_{u i})r_{u i} - r_{u i} + r_{u i}^2 \notag \\
& \leq f_{u i}^2 + 3|r_{u i}| + r_{u i}^2 \leq f_{u i}^2 + 4|r_{u i}|\leq 2f_{u i}^2 \label{eq: l prime bound}
\end{align}
by the definition of $f_u^2$ and since $|r_u|\leq 1$.
Therefore, with probability $1-o(1)$,
$$
\|(\widehat f_u^2 - f_u^2)Z_u^j/\widehat f_u\|_{\Pn,2} \lesssim \|(\widehat f_u^2 - f_u^2)Z_u^j/f_u\|_{\Pn,2}
$$
uniformly over $u\in\UU$ and $j\in[\pp]$. Hence, it suffices to show that with probability $1-o(1)$,
$$
\sup_{u\in\UU,j\in[\pp]}\|(\widehat f_u^2 - f_u^2)Z_u^j/f_u\|_{\Pn,2}\lesssim (s_n\log a_n/n)^{1/2}.
$$

Next, as in (\ref{eq: gamma-lipshitz}), we have uniformly over $u,u'\in\UU$ that $\|\gamma_u^j - \gamma_{u'}^j\|_1\lesssim (p + \pp)^{1/2}|u - u'|$, and so, given that $Z_u^j-Z_{u'}^j = X^j(\gamma_{u'}^j-\gamma_{u}^j)$, we have by Assumption \ref{ass: covariates}(vii) that
\begin{align*}
\max_{1\leq i\leq n}|Z_{ui}^j-Z_{u'i}^j|
& \leq \max_{1\leq i\leq n}\|(D'_i,X'_i)'\|_\infty \|\gamma_u^j-\gamma_{u'}^j\|_1\\
& \lesssim_P n^{1/(2q)}M_{n,2}(p+\pp)^{1/2}|u-u'|.
\end{align*}
Moreover, as in \eqref{eq: f-lipshitz}, we have uniformly over $u,u'\in\UU$ that $|f_u^2 - f_{u'}^2|\lesssim |u - u'|$.

Further, observe that for $a>1$, the inequality $|x|\leq \log(\sqrt a - 1)$ implies that
$$
\Lambda'(x) = \frac{e^x}{(1 + e^x)^2} = \frac{1}{e^{-x} + 2 + e^{x}}\geq \frac{1}{2(1 + e^{|x|})} \geq \frac{1}{2\sqrt a}.
$$
Also, by Assumptions \ref{ass: parameters}, \ref{ass: sparsity}, and \ref{ass: covariates}(vii),
$$
|t_{u i}|\leq \|(D'_i,X'_i)'\|_\infty(\|\theta_u\|_1 + \|\beta\|_1)\leq n^{1/(2q)}M_{n,2}\sqrt{s_n\log n}
$$
with probability $1-o(1)$ uniformly over $u\in\UU$ and $i=1,\dots,n$. Thus, applying the inequality above with $\sqrt a - 1 = \exp(n^{1/(2q)}M_{n,2}\sqrt{s_n \log n})$ gives
$$
f_{u i}^2 \geq \Lambda'(t_{u i})/2\gtrsim \exp(-n^{1/(2q)}M_{n,2}\sqrt{s_n\log n})
$$
with probability $1-o(1)$ uniformly over $u\in\UU$ and $i=1,\dots,n$. So,
$$
\En[f_u^{-2}]\lesssim \exp(n^{1/(2q)}M_{n,2}\sqrt{s_n\log n})
$$
with probability $1-o(1)$ uniformly over $u\in\UU$.

In addition, let
$$
\epsilon = \epsilon_n = \Big(n^{1 + 1/q}(M^2_{n,1}\vee M^2_{n,2})(p+\pp)^{1/2}\exp(n^{1/(2q)}M_{n,2}\sqrt{s_n\log n})\Big)^{-1},
$$
and let $\UU^\epsilon$ be an $\epsilon$-net of $\UU$ with $|\UU^\epsilon|\leq 1/\epsilon$. For all $i=1,\dots,n$, let $U_i$ be a value of $u\in\RR$ such that $Y_i =  (1-u)\underline y + u\bar y$. Note that $\widehat f_u$ does not vary with $u$ on any interval $[\underline u,\bar u]\subset\UU$ as long as $U_i\notin [\underline u,\bar u]$ for all $i=1,\dots,n$. Also, since $\epsilon\lesssim n^{-3}$, with probability $1-o(1)$, each interval $[u-2\epsilon,u+2\epsilon]$ with $u\in\UU^\epsilon$ contains at most one value of $U_i$'s by Assumption \ref{ass: density}. Now,
 \begin{align*}
&\sup_{u\in\UU}\max_{j\in[\pp]}\En[(\hat f_u^2 - f_u^2)^2(Z_u^j/f_u)^2] \lesssim \sup_{u\in\UU^\epsilon}\max_{j\in[\pp]}\En[(\hat f_u^2 - f_u^2)^2(Z_u^j/f_u)^2] \\
&\qquad + \sup_{u\in\UU,j\in[\pp]}\inf_{u'\in\UU^\epsilon}\Big|\En[(\hat f_u^2 - f_u^2)^2(Z_u^j/f_u)^2-(\hat f_{u'}^2 - f_{u'}^2)^2(Z_{u'}^j/f_{u'})^2]\Big|,
 \end{align*}
and uniformly over $j\in[\pp]$ and $u, u'\in\UU$ such that $\hat f_{u} = \hat f_{u'}$, we have with probability $1-o(1)$ that
\begin{align*}
& \Big|\En[(\hat f_u^2 - f_u^2)^2(Z_u^j/f_u)^2-(\hat f_{u'}^2 - f_{u'}^2)^2(Z_{u'}^j/f_{u'})^2]\Big|  \\
&\quad \leq \Big|\En[\{(\hat f_u^2 - f_u^2)^2-(\hat f_{u'}^2 - f_{u'}^2)^2\}(Z_u^j/f_u)^2]\Big| \\
&\qquad +\Big|\En[(\hat f_{u'}^2 - f_{u'}^2)^2\{(Z_u^j/f_u)^2-(Z_{u'}^j/f_{u'})^2\}]\Big|\\
&\quad \lesssim \Big|\En[\{(f_{u'}^2 - f_u^2)(Z_u^j/f_u)^2]\Big|+\Big|\En[ f_{u'}^4\{(Z_u^j/f_u)^2-(Z_{u'}^j/f_{u'})^2\}]\Big|\\
& \quad\lesssim |u'-u|\En[(Z_u^j/f_u)^2]+\Big|\En[ (f_{u'}^4/f_u^2)|\{Z_u^j-Z_{u'}^j\}\{Z_u^j+Z_{u'}^j\}|]\Big| \\
&\qquad +\Big|\En[ (f_{u'}^2/f_u^2)(Z_{u'}^j)^2(f_{u'}^2-f_u^2)]\Big|\\
& \quad\lesssim_P |u'-u|\cdot\Big(n^{1/q}M_{n,1}^2+n^{1/(2q)}M_{n,2}(p+\pp)^{1/2}n^{1/(2q)}M_{n,1}\Big)\cdot\En[f_u^{-2}].
\end{align*}
Thus, by the choice of $\epsilon$, and since with probability $1-o(1)$ each interval $[u-2\epsilon,u+2\epsilon]$ with $u\in\UU^\epsilon$ contains at most one value of $U_i$'s, we have with probability $1 - o(1)$ that
 $$
 \sup_{u\in\UU,j\in[\pp]}\inf_{u'\in\UU^\epsilon}\Big|\En[(\hat f_u^2 - f_u^2)^2(Z_u^j/f_u)^2-(\hat f_{u'}^2 - f_{u'}^2)^2(Z_{u'}^j/f_{u'})^2]\Big| \lesssim s_n\log a_n/n.
$$
Further by (\ref{Multiplicative}) and \eqref{eq: l prime bound}, with probability $1-o(1)$,
  \begin{align*}
&\sup_{u\in\UU^\epsilon,j\in[\pp]}\En[(\hat f_u^2 - f_u^2)^2(Z_u^j/f_u)^2]  \\
&\quad \lesssim \sup_{u\in\UU^\epsilon, j\in[\pp]}\En[ \Lambda'(t_{ui})^2|\tilde t_{ui}-t_{ui}|^2(Z_u^j/f_u)^2]  + \sup_{u\in\UU^\epsilon, j\in[\pp]}\En[r_u^2(Z_u^j/f_u)^2]\\
& \quad \lesssim \max_{j\in [\pp]}\sup_{u\in\UU^\epsilon,\|\delta\|_0\leq Cs_n, \|\delta\|=1} \En[\{(D',X')\delta\}^2(Z^j_u)^2] s_n\log a_n/n +  s_n\log a_n/n
 \end{align*}
 for $C$ large enough, where we used that $\sup_{u\in\UU}\En[r_u^2(Z_u^j/f_u)^2] \lesssim s_n\log a_n/n$ with probability $1-o(1)$ by Assumption \ref{ass: approximation error}, and $\sup_{u\in\UU} (\|\widetilde \theta_u - \theta_u\| + \|\widetilde \beta_u- \beta_u\|) \lesssim (s_n\log a_n/n)^{1/2}$ and $\sup_{u\in \UU}(\|(\widetilde \theta'_u,\widetilde \beta'_u)'\|_0 + \|(\theta'_u,\beta'_u)'\|_0) \lesssim s_n$ with probability $1-o(1)$.
Therefore, since
$$
\Ep_P[\{(D',X')\delta\}^2(Z^j_u)^2]\leq \Ep_P[\{(D',X')\delta\}^4]^{1/2}\Ep_P[(Z^j_u)^4]^{1/2} \lesssim 1,
$$
to establish the statement of the lemma it suffices to show that with probability $1-o(1)$,
$$
\max_{j\in [\pp]}\sup_{u\in\UU^\epsilon,\|\delta\|_0\leq Cs_n, \|\delta\|=1} |(\En-\Ep_P)[\{(D',X')\delta\}^2(Z^j_u)^2]| \lesssim 1.
$$
To do so, we will apply Lemma \ref{thm:RV34} with $\UU$ replaced by $\UU^\epsilon \times [\pp]$ and $X_{u}$ replaced by $Z_u^j (D',X')'$. We have
\begin{align*}
K& =\Big(\Ep_P\Big[\max_{1\leq i\leq n,u\in\UU^\epsilon}\|Z_{u i}^j(D',X')'\|_\infty^2\Big]\Big)^{1/2} \\
& \leq n^{1/q} \Big(\Ep_P\Big[\max_{u\in\UU^\epsilon}\|Z_u^j (D',X')'\|_\infty^{q}\Big]\Big)^{1/q}\\
%& \leq n^{1/q}\Big(\Ep_P[\|(D',X')'\|_\infty^{q}\|Z_u^j\|_\infty^{q}]\Big)^{1/q}\\
 &\leq n^{1/q} \Big(\Ep_P[\|(D',X')'\|_\infty^{2q}]\Ep_P[\|Z_u^j\|_\infty^{2q}]\Big)^{1/(2q)} \leq n^{1/q}M_{n,2}M_{n,1}
\end{align*}
by Assumption \ref{ass: covariates}(v,vii). Also,
$$
\sup_{\|\delta\|_0\leq Cs_n, \|\delta\|=1}\max_{u\in\UU^\epsilon,j\in[\pp]} \Ep_P[ (Z_u^j(D',X')\delta)^2] \lesssim 1
$$
 by Assumption \ref{ass: covariates}(iv). Then, by Lemma \ref{thm:RV34} we have for
$$
 \tilde \delta_n = n^{-1/2} K  s_n^{1/2} \Big( \log^{1/2} (\pp |\UU_0^\epsilon|) + (\log s_n)(\log^{1/2} n) (\log^{1/2} a_n)\Big)
 $$
 that
$$
\sup_{\|\delta\|_0\leq Cs_n, \|\delta\|=1}\max_{u\in\UU^\epsilon,j\in[\pp]}|(\En-\Ep_P)[(Z_u^j(D',X')\delta)^2]| \lesssim_P   \tilde \delta_n^2 + \tilde\delta_n.
$$
Now,
$$
\pp |\UU^\epsilon| \leq  \pp /\epsilon \leq n^{1 + 1/q}(M^2_{n,1}\vee M^2_{n,2})(p+\pp)^{3/2}\exp(n^{1/(2q)}M_{n,2}\sqrt{s_n\log n}),
$$
so that
$$
\log(\pp|\UU^\epsilon|)\lesssim \log a_n+n^{1/(2q)}M_{n,2}\sqrt{s_n\log n}.
$$
Using Assumption \ref{ass: covariates}(vi,viii,ix) and since $\delta_n^2\log a_n = o(1)$, we have
\begin{align*}
&\frac{(M_{n,1}\vee M_{n,2})^2s_n \log a_n}{n^{1/2-1/q}}\leq \delta_n \log^{1/2} a_n = o(1)\text{ and }\frac{M_{n,1}^2M_{n,2}^4 s_n}{n^{1-3/q}} = o(1),
\end{align*}
and so
\begin{align*}
\tilde \delta_n & \lesssim \frac{s_n^{1/2} n^{1/q}M_{n,2}M_{n,1}(\log s_n)(\log^{1/2} n)(\log^{1/2} a_n)}{n^{1/2}} \\
& \quad +\frac{s_n^{1/2} n^{1/q}M_{n,2}M_{n,1}n^{1/(4q)}M_{n,2}^{1/2}s_n^{1/4}\log^{1/4}n}{n^{1/2}}\\
& \lesssim \Big(\frac{M_{n,1}^2s_n\log a_n}{n^{1/2 - 1/q}}\Big)^{1/2}
\Big(\frac{M_{n,2}^2 s_n \log a_n}{n^{1/2-1/q}}\Big)^{1/2} \\
&\quad + \Big(\frac{M_{n,1}^2s_n\log n}{n^{1/2 - 1/q}}\Big)^{1/4}
\Big(\frac{M_{n,2}^2s_n}{n^{1/2 - 1/q}}\Big)^{1/4}\Big(\frac{M_{n,1}^{2}M_{n,2}^{4} s_n}{n^{1-3/q}}\Big)^{1/4} = o(1).
\end{align*}
This completes the proof.
\qed

\subsection*{Proof of Lemma \ref{lemma:FSE}}
Both results follow from Lemma \ref{thm:RV34}. We provide a proof only for the first result (the second result is simpler and  follows similarly).

Recall that by Assumption \ref{ass: covariates}(i),
$$
\inf_{\|\delta\|=1} \| f_u (D',X')\delta\|_{\Pr,2}\geq c_1.
$$
Also, observe that for any $x,y\in[0,1]$, we have
$$
\Big|\sqrt{x(1-x)} - \sqrt{y(1-y)}\Big|\leq \sqrt{|x-y|}.
$$
Therefore, since $f_u^2 = \Ep[Y_u\mid D,X](1 - \Ep[Y_u\mid D,X])$, by Assumption \ref{ass: density}, for any $u,u'\in\UU$, we have
$$
|f_{u'} - f_u| \leq \Big(|\Ep[Y_{u'} - Y_u\mid D,X]|\Big)^{1/2} \leq \Big(C_1|u' - u|\Big)^{1/2}.
$$
Hence, since $\ell_n\to\infty$, with probability $1 - o(1)$ uniformly over $u,u'\in\UU$ and $\delta\in\mathbb R^{p+\pp}$ with $\|\delta\| = 1$ and $\|\delta\|_0\leq \ell_n s_n$, we have
\begin{align*}
&\Big| \| f_{u'} (D',X')\delta\|_{\Pn,2}- \| f_u (D',X')\delta\|_{\Pn,2}  \Big| \\
&\qquad \leq \| (f_{u'}-f_{u}) (D',X')\delta\|_{\Pn,2} \leq \| f_{u'}-f_{u} \|_{\Pn,2} \max_{1\leq i\leq n} \|(D_i',X_i')'\|_\infty \|\delta\|_1\\
& \qquad \leq \| f_{u'}-f_{u} \|_{\Pn,2} n^{1/(2q)} M_{n,2} \ell_n \sqrt{s_n} \leq \Big(C_1|u' - u|\Big)^{1/2}n^{1/(2q)} M_{n,2} \ell_n \sqrt{s_n}
\end{align*}
by Assumption \ref{ass: covariates}(vii). Thus, for
$$
\epsilon = \epsilon_n = \frac{c_1^2}{C_1 n^{1/(2q)}M_{n,2}^2\ell_n^4 s_n},
$$
we have with probability $1-o(1)$ that
 $$
\sup_{|u-u'|\leq \epsilon, \|\delta\|_0\leq \ell_n s_n, \|\delta\|=1} \Big| \| f_{u'} (D',X')\delta\|_{\Pn,2}- \| f_u (D',X')\delta\|_{\Pn,2} \Big| \leq c_1/\ell_n.
$$

Now, let $\UU^\epsilon$ be an $\epsilon$-net of $\UU$ such that $|\UU^\epsilon|\leq 3/\epsilon$. We will apply Lemma \ref{thm:RV34} with $\UU$ replaced by $\UU^\epsilon$, $k=\ell_n s_n$, and $X_u$ replaced by $f_u(D',X')'$. Since $0\leq f_u \leq 1$, we have
\begin{align*}
K& =\Big(\Ep_P\Big[\max_{1\leq i\leq n}\max_{u\in \UU^\epsilon}f_{u i}^2\|(D_i',X_i')'\|_\infty^2\Big]\Big)^{1/2} \\
& \leq \Big(\Ep_P\Big[\max_{1\leq i\leq n}\|(D_i',X_i')'\|_\infty^2\Big]\Big)^{1/2}\leq  n^{1/(2q)}M_{n,2}
\end{align*}
by Assumption \ref{ass: covariates}(vii). Also,
\begin{align*}
&\sup_{\|\delta\|_0\leq \ell_n s_n, \|\delta\|=1} \max_{u\in \UU^\epsilon} \Ep_P[ f_u^2((D',X')\delta)^2] \\
&\qquad \leq \sup_{\|\delta\|_0\leq \ell_n s_n, \|\delta\|=1} \Ep_P[\{(D',X')\delta\}^2] \leq \sqrt{C_1}
\end{align*}
by Assumption \ref{ass: covariates}(iv). Thus, applying Lemma \ref{thm:RV34} gives
$$
\displaystyle \sup_{\|\delta\|_0\leq \ell_n s_n, \|\delta\|=1}\max_{u\in \UU^\epsilon}|(\En-\Ep_P)[f_u^2((D',X')\delta)^2]|  \lesssim_P   \tilde \delta_n^2 + \tilde \delta_n
$$
where
$$
\tilde \delta_n = n^{-1/2+1/(2q)}M_{n,2}\sqrt{\ell_n s_n}(\log^{1/2}a_n)(\log^{3/2}n).
$$
Finally, by Assumption \ref{ass: covariates}(viii),
\begin{align*}
\tilde \delta_n^2 &= n^{-1+1/q}M_{n,2}^2\ell_n s_n(\log a_n)(\log^3 n) \\
&\leq n^{-1/2}\ell_n \delta_n (\log^{1/2}a_n)(\log^3 n) = o(1)
\end{align*}
since $\ell_n \to\infty$ slowly enough and $\log^{1/2}a_n \lesssim \delta_n n^{1/6}$ by Assumption \ref{ass: covariates}(ii,iii). Combining presented bounds gives the asserted claim.
\qed

\section{Double Selection method for Logistic Regression with Functional Response Data}

In this section we discuss in details and provide formal results for the double selection estimator for logistic regression with functional response data.

\begin{algorithm}
(Based on double selection) For each $u\in \UU$ and $j\in[\pp]$: \\
\enspace \emph{Step 1'}. Run post-$\ell_1$-penalized logistic estimator (\ref{PostL1Logistic})  of $Y_{u}$ on  $D$ and $X$ to compute $(\widetilde\theta_u,\widetilde\beta_u)$. \\
\enspace \emph{Step 2'}. Define the weights $\hat f_{u}^2 = \hat f_u^2(D,X) = \G'(D_i'\widetilde\theta_u+X_i'\widetilde\beta_u)$.\\
\enspace \emph{Step 3'}. Run  the lasso estimator (\ref{EstLasso}) of $\hat f_u D_j$ on $\hat f_u X$ to compute $\hat\gamma_u^j$. \\
\enspace \emph{Step 4'}. Run logistic regression of $Y_{u}$ on  $D_j$ and all the selected variables in Steps 1' and 2' \\
\indent\indent\indent\indent to compute $\check \theta_{u j}$.
\end{algorithm}

The following result establishes the Bahadur representation for the double selection estimator (analog to Theorem 3.1 for score functions). % for the optimal score).

\begin{theorem}[Uniform Bahadur representation, double selection]\label{theorem:inferenceAlg2}  Suppose that Assumptions \ref{ass: parameters} -- \ref{ass: approximation error} hold for all $P\in\mathcal P_n$. Then, the estimator  $(\check \theta_{uj})_{u\in\UU,j\in[\pp]}$, based on the double selection, obeys as $n \to \infty$
$$
\Sigma_{uj}^{-1} \sqrt{n} (\check \theta_{uj} -  \theta_{uj}) = \Gn \bar\psi_{uj} + O_\Pr(\delta_n)  \mbox{ in }\ell^\infty(\UU\times[\pp])%  \mbox{uniformly over} \ u\in \UU, \ j\in[\pp],
$$
uniformly over $u\in\UU$, where $\Sigma^2_{uj}:=  \Ep_P[f_u^2(D-X^j\gamma_u^j)^2]^{-1}$.
\end{theorem}
\subsection*{Proof of Theorem \ref{theorem:inferenceAlg2}}
The analysis is reduced to the proof of Theorem \ref{theorem:inferenceAlg1}.
Let $\hat T_{uj} = \supp(\hat\theta_u)\cup\supp(\hat \beta_u)\cup\supp(\hat\gamma_u^j)$ for which by Theorems \ref{Thm:RateEstimatedLassoLogistic} and \ref{Thm:RateEstimatedLassoLinear} satisfies $\sup_{u\in\UU,j\in[\pp]}|\hat T_{uj}|\lesssim s_n$ with probability $1-o(1)$. Therefore Step 3 is a post-selection logistic regression which yields an initial rate of convergence $|\check\theta_{uj}-\theta_{uj}|+\|\bar\theta_{u[\pp]\setminus j}-\theta_{u[\pp]\setminus j} \| +\|\bar\beta_u-\beta_u\| \lesssim (s_n\log a_n/n)^{1/2}$. Moreover, by the first order condition of Step 3 we have
\begin{equation}\label{OPTdouble}
 \En[ \{Y_{ui}-\G(D_j\check\theta_{uj} + D_{[\pp]\setminus j}'\bar\theta_{u[\pp]\setminus j}+X'\bar \beta_u)\} (D_j, X^j_{\hat T_{uj}})']=0\end{equation}
so that any linear combination yields zero. Setting the parameters $(\widetilde \theta_u',\widetilde\beta_u')=(\bar\theta_{u[\pp]\setminus j},\bar\theta_{u[\pp]\setminus j}',\bar\beta_u')$, and $\hat z_u^j = (D_j, X^j_{\hat T_{uj}})(1,-\widetilde \gamma_u^j )=D_j-X^j \widetilde \gamma_u^j$, we recover the setting in the proof of Theorem \ref{theorem:inferenceAlg1}. The rest of the proof follows similarly. \qed

The double selection procedure benefits from additional variables selected in Step 2. The (estimated) weights used in the equation ensure that selection will ensure a near-orthogonality condition that is required to remove first order bias. In contrast, (naive) Post-$\ell_1$-logistic regression does not select such variables which in turn translates in to first order bias in the estimation of $\theta_{uj}$. We stress that Step 2 is tailored to the estimation of each coefficient $\theta_{uj}$ which enables the additional adaptivity. %In terms of rates of convergence in the prediction norm,  $\ell_1$-logistic regression is optimal when $p\gg n$.

The double selection achieves orthogonality conditions relative to all selected variables in finite samples. Although first-order equivalent to other estimator discussed here, this additional orthogonality could potentially lead to a better finite sample performance. To provide intuition why, consider the logistic regression case with $\pp = {\rm dim}(D) = 1$ and $\UU=\{0\}$ for simplicity. In this case we have
$$ E[ Y | D , X ] = \Lambda(D\theta_0 + X'\beta_0) \ \ \mbox{and let} \ \ f = \Lambda'(D\theta_0 + X'\beta_0). $$
Letting $\hat\gamma$ be the Lasso estimate of $fD$ on $fX$ we have that the one step estimator
$$ \bar \theta = \hat\theta - \En[(D-X'\hat \gamma)^2 ]^{-1}\En[\{Y - \Lambda(D\hat\theta + X'\hat\beta))\}(D-X'\hat \gamma)]$$
is an approximate solution for the moment condition
\begin{equation}\label{NewMoment}\En[\{Y - \Lambda(D\theta + X'\hat\beta)\}(D-X'\hat \gamma)] = 0.\end{equation}
Indeed, it is one Newton step from $\hat\theta$. Our proposed estimator based on estimated score functions %optimal score
defines $\check \theta$ as an exact solution for (\ref{NewMoment}), namely
$$\En[\{Y - \Lambda(D\check\theta + X'\hat\beta)\}(D-X'\hat \gamma)] = 0.$$
The double selection achieves that implicitly. Indeed, letting $\hat T = {\rm support}(\check\beta)\cup {\rm support}(\hat \gamma)$, the first order condition of running a logistic regression of $Y$ on $D$ and $X_{\hat T}$ yields
\begin{equation}\label{FOC}\En\left[\{Y - \Lambda(D\check\theta + X'\check\beta)\}\left(\begin{array}{c}D \\ X_{\hat T}\end{array}\right)\right] = 0\end{equation}
where $(\check\theta,\check\beta)$ is the solution of the logistic regression. By multiplying the vector $(1,-\hat\gamma')'$, the relation above implies
$$\En[\{Y - \Lambda(D\check\theta + X'\check\beta)\}(D-X'\hat \gamma)] = 0.$$
since ${\rm support}(\hat\gamma) \subset \hat T$ (the condition ${\rm support}(\hat\beta) \subset \hat T$ ensures that $\check\beta$ is a good approximation of $\beta_0$). Note that (\ref{FOC}) provides a more robust orthogonality condition and does not need to explicitly create the new score functions %optimal score
as the other two methods.

\section{Generic Finite Sample Bounds for $\ell_1$-Penalized M-Estimators: Nuisance Functions and Functional Data}\label{sec: generic results}
In this section, we establish a set of results for $\ell_1$-penalized M-estimators with functional data and high-dimensional parameters. These results are used in the proofs of Theorems \ref{Thm:RateEstimatedLassoLogistic} and \ref{Thm:RateEstimatedLassoLinear} and may be of independent interest.

We start with specifying the setting. Consider a data generating process with a functional response variable $(Y_{u})_{u\in \UU}$ and observable covariates $(X_u)_{u\in\UU}$ satisfying for each $u\in \mathcal{U}\subset \mathbb R^{d_u}$,
\begin{equation}\label{A:EqMainFunc}
\theta_u \in \arg\min_{\theta \in \RR^p}\Ep_P[M_u(Y_{u},X_u,\theta,a_u)],
\end{equation}
where $\theta_u$ is a $p$-dimensional vector of parameters, $a_u$ is a nuisance parameter that captures potential misspecification of the model, and $M_u$ is a known function. Here for all $u\in\UU$, $Y_u$ is a scalar random variable and $X_u$ is a $p_u$-dimensional random vector with $p_u\leq p$ for some $p$. %Throughout this section, we assume that for all $u\in\UU$, the function $\theta\mapsto M_u(Y_u,X_u,\cdot,a)$ is convex almost surely for all possible values of the parameter $a_u$.
We assume that the solution $\theta_u$ is sparse in the sense that the process $(\theta_u)_{u\in\UU}$ satisfies
$$\|\theta_u\|_0\leq s,  \quad \mbox{for all} \ u\in\UU.$$
Because the model (\ref{A:EqMainFunc}) allows for the nuisance parameter $a_u$, such sparsity assumption is very mild and formulation (\ref{A:EqMainFunc}) encompasses many cases of interest including approximately sparse models.

Throughout this section, we assume that we have $n$ i.i.d. observations, $\{ (Y_{ui},X_{ui})_{u\in\mathcal{U}}\}_{i=1}^n$, from the distribution of $(Y_u,X_u)_{u\in\UU}$ to estimate  $(\theta_u)_{u\in\mathcal{U}}$. In addition, we assume that an estimate $\widehat a_u$ of the nuisance parameter $a_u$ is available for all $u\in\UU$. Using the estimate $\widehat a_u$, we use the criterion function
$$
M_u(Y_{u},X_u,\theta):=M_u(Y_{u},X_u,\theta,\hat a_u)
$$
as a proxy for $M_u(Y_{u},X_u,\theta_u,a_u)$. We allow for the case where $p$ is much larger than $n$. %In cases where $a_u$ is a small approximation error we typically set $\hat a_u = 0$.

Since $p$ is potentially larger $n$, and the parameters $\theta_u$ are assumed to be sparse, we consider an $\ell_1$-penalized $M_u$-estimator (Lasso)  of $\theta_u$:
 \begin{equation}\label{Adef:LassoFunc}
 \hat\theta_u \in \arg\min_{\theta} \left(\En[M_u(Y_{u},X_u,\theta)] + \frac{\lambda}{n}\|\hat\Psi_u\theta\|_1\right) \end{equation}
where $\lambda$ is a penalty level and $\hat\Psi_u$ a diagonal matrix of penalty loadings. Further, for each $u\in\mathcal{U}$, we also consider a post-regularized (Post-Lasso) estimator of $\theta_u$:
\begin{equation}\label{Adef:PostFunc}
\widetilde\theta_u \in \arg\min_{\theta} \En[M_u(Y_{u},X_u,\theta)] \ \ : \ \ \supp(\theta)\subseteq \widehat T_u
\end{equation}
where $\widehat T_u = \supp(\hat\theta_u)$.

We assume that for each $u\in\UU$, the matrix of penalty loadings $\hat\Psi_u$ is chosen as an appropriate estimator of the following ``ideal'' matrix of penalty loadings: $\widehat\Psi_{u0} = \diag(\{ l_{u0k}, k=1,\ldots,p\})$, where
\begin{equation}\label{IdealLoading}
l_{u0k} = \Big(\En\Big[(\partial_{\theta_k} M_u(Y_{u},X_u,\theta_u,a_u))^2\Big]\Big)^{1/2},
\end{equation}
where $\partial_{\theta_k}M_u(Y_u,X_u,\theta_u,a_u)$ denotes a sub-gradient of the function $\theta\mapsto M_u(Y_u,X_u,\theta,a_u)$ with respect to the $k$th coordinate of $\theta$ and evaluated at $\theta = \theta_u$. The properties of $\hat\Psi_u$ will be specified below in lemmas. Also, we assume that the penalty level $\lambda$ is chosen such that with high probability,
\begin{equation}\label{Eq:reg} \frac{\lambda}{n} \geq c \sup_{u\in \mathcal{U}}\left\|\hat \Psi^{-1}_{u0} \En\left[\partial_\theta M_u(Y_{u},X_u,\theta_u, a_u) \right] \right\|_\infty,
\end{equation}where $c>1$ is a fixed constant. When $\mathcal{U}$ is a singleton, the condition \eqref{Eq:reg}  is similar to that in \cite{BickelRitovTsybakov2009}, \cite{BC-PostLASSO}, and \cite{BCW-SqLASSO}. When $\UU$ is a continuum of indices, a similar condition was previously used in \cite{BC-SparseQR} in the context of $\ell_1$-penalized quantile regression.% In the functional outcome case guaranteeing that the ``regularization event'' (\ref{Eq:reg}) holds with high probability also plays a key role in establishing desirable properties of Lasso and Post-Lasso estimators uniformly over $u\in\mathcal{U}$.

For $u\in\UU$, denote $T_u=\supp(\theta_u)$. Let $\ell$ and $L$ be some constants satisfying $L\geq \ell>1/c$. Also, let
$$
\tilde c = \frac{Lc+1}{\ell c-1}\sup_{u\in \UU}\|\widehat \Psi_{u0}\|_\infty\|\widehat \Psi_{u0}^{-1}\|_\infty,
$$
where for any diagonal matrix $A = \diag(\{ a_k, k=1,\ldots,p\})$, we denote $\|A\|_\infty = \max_{1\leq k\leq p}|a_k|$.
Let $(\Delta_n)_{n\geq 1}$ be a sequence of positive constants converging to zero, and let $(C_n)_{n\geq 1}$ be a sequence of random variables. Also, let $w_u=w_u(X_u)$ be some weights satisfying $0\leq w_u\leq 1$ almost surely. Finally, let $A_u$ be some random subset of $\mathbb R^p$ and $\bar q_{A_u}$ be a random variable possibly depending on $A_u$, where both $A_u$ and $\bar q_{A_u}$ are specified in the lemmas below.
To state our results in this section, we need the following assumption:%In this section we derive finite sample bounds based on Assumption \ref{ass: M} below. This assumption provides sufficient conditions that are implied by a variety of settings including generalized linear models.

\begin{assumption}[M-Estimation]\label{ass: M}
The function $\theta\mapsto M_u(Y_{u},X_{u},\theta)$ is convex almost surely, and with probability at least $1-\Delta_n$, the following inequalities hold for all $u\in\UU$:
\begin{itemize}
\item[(a)] $|\En[\partial_{\theta} M_u(Y_u,X_u,\theta_u)-\partial_{\theta} M_u(Y_u,X_u,\theta_u,a_u)]'\delta|\leq  C_n\|\sqrt{w_{u}} X_u'\delta\|_{\Pn,2}$ for all $\delta\in\mathbb R^p$;
\item[(b)] $\ell \widehat\Psi_{u0} \leq \widehat\Psi_u \leq L\widehat\Psi_{u0}$;
\item[(c)] for all $\delta \in A_u$,
{\small
\begin{align*}
&\En[M_u(Y_u,X_u,\theta_u + \delta)] - \En[M_u(Y_u,X_u,\theta_u)] -\En[\partial_{\theta} M_u(Y_u,X_u,\theta_u)]'\delta \\
&\qquad +2C_n\|\sqrt{w_{u}} X_u'\delta\|_{\Pn,2} \geq \left\{\|\sqrt{w_{u}} X_u'\delta\|_{\Pn,2}^2\right\} \wedge \left\{ \bar q_{A_u}\|\sqrt{w_{u}} X_u'\delta\|_{\Pn,2}\right\}.
\end{align*}}
\end{itemize}
\end{assumption}
In many applications one can take the weights to be $w_u = w_u(X_u) = 1$ but we allow for more general weights since it is useful for our results on the weighted Lasso with estimated weights. Also, in applications, we typically have $C_n \lesssim \{n^{-1}s \log (pn)\}^{1/2}$. Assumption \ref{ass: M}(a) bounds the impact of estimating the nuisance functions uniformly over $u\in\UU$. The loadings $\hat\Psi_u$ are assumed larger (but not too much larger) than the ideal choice $\hat\Psi_{u0}$ defined in (\ref{IdealLoading}). This is formalized in Assumption \ref{ass: M}(b). Assumption \ref{ass: M}(c) is an identification condition that will be imposed for particular choices of $A_u$ and $\bar q_{A_u}$. It relates to conditions in the literature derived for the case of a singleton $\UU$ and no nuisance functions, see  the restricted strong convexity\footnote{Assumption \ref{ass: M} (a) and (c) could have been stated with $\{C_n/\sqrt{s}\}\|\delta\|_1$ instead of $C_n\|\sqrt{w_{u}} X_u'\delta\|_{\Pn,2}$.} used in \cite{negahban2012unified} and the non-linear impact coefficients used  in \cite{BC-SparseQR} and \cite{BCK-SparseQRinference}.

Define the restricted eigenvalue
$$
\bar\kappa_{2\tilde c} = \inf_{u\in\UU} \inf_{\delta\in \Delta_{2\tilde c,u}} \|\sqrt{w_{u}}  X_u'\delta\|_{\Pn,2}/\|\delta_{T_u}\|
$$
where $\Delta_{2\tilde c,u} = \{ \delta : \|\delta_{T_u}^c\|_1\leq 2\tilde c \|\delta_{T_u}\|_1\}$. Also, define minimum and maximum spare eigenvalues
{\small
$$
\semin{m,u} = \min_{1\leq \|\delta\|_0\leq m}\frac{\|\sqrt{w_{u}}X_u'\delta\|_{\Pn,2}^2}{\|\delta\|^2} \ \ \mbox{and} \ \ \semax{m,u} = \max_{1\leq \|\delta\|_0\leq m}\frac{\|X_u'\delta\|_{\Pn,2}^2}{\|\delta\|^2}.
$$}\!The following results establish the rate of convergence and a sparsity bound for the $\ell_1$-penalized estimator $\widehat \theta_u$ defined in (\ref{Adef:LassoFunc}) as well as the rate of convergence for the post-regularized estimator $\widetilde\theta_u$ defined in (\ref{Adef:PostFunc}).

\begin{lemma}\label{Lemma:LassoMRateRaw}
Suppose that Assumption \ref{ass: M} holds with
{\small
$$
A_u = \{ \delta\colon \|\delta^c_{T_u}\|_1\leq 2\tilde c \|\delta_{T_u}\|_1\} \cup \{ \delta\colon \|\delta\|_1 \leq \frac{3n}{\lambda}\frac{c\|\widehat\Psi_{u0}^{-1}\|_\infty}{\ell c - 1}C_n\|\sqrt{w_{u}} X_u'\delta\|_{\Pn,2}\}
$$}\!and $\bar q_{A_u}>(L+\frac{1}{c})\|\widehat\Psi_{u0}\|_\infty\frac{\lambda\sqrt{s}}{n\bar\kappa_{2\tilde c}}+ 6\tilde c C_n$. In addition, suppose that $\lambda$ satisfies condition (\ref{Eq:reg}) with probability $1-\Delta_n$. Then, with probability at least $1-2\Delta_n$, we have for all $u\in\UU$ that
\begin{align*}
\|\sqrt{w_{u}}  X_u'(\hat \theta_u - \theta_u)\|_{\Pn,2} & \leq \Big(L+\frac{1}{c}\Big)\|\widehat\Psi_{u0}\|_\infty\frac{\lambda\sqrt{s}}{n\bar\kappa_{2\tilde c}}+ 6\tilde c C_n,\\
\|\hat \theta_u - \theta_u\|_1 & \leq  \Big( \frac{(1+2\tilde c)\sqrt{s}}{\bar\kappa_{2\tilde c}}+ \frac{3n}{\lambda}\frac{c\|\widehat\Psi_{u0}^{-1}\|_\infty}{\ell c - 1}C_n\Big) \\
&\quad \times\Big(\Big(L+\frac{1}{c}\Big)\|\widehat\Psi_{u0}\|_\infty\frac{\lambda\sqrt{s}}{n\bar\kappa_{2\tilde c}}+ 6\tilde c C_n\Big).
\end{align*}
\end{lemma}

\begin{lemma}\label{Lemma:LassoMSparsity}
In addition to conditions of Lemma \ref{Lemma:LassoMRateRaw}, suppose that with probability $1-\Delta_n$, we have for some random variable $L_n$ and all $u\in\UU$ and $\delta \in \RR^p$ that
\begin{equation}\label{eq: def Ln}
\Big|\{\En[\partial_\theta M_u(Y_u,X_u,\hat\theta_u)-\partial_\theta M_u(Y_u,X_u,\theta_u)]\}'\delta\Big| \leq L_n\|X_u'\delta\|_{\Pn,2}.
\end{equation}
Further, for all $u\in\UU$, let $\widehat s_u = \supp(\widehat T_u)$. Then with probability at least $1-3\Delta_n$, we have for all $u\in\UU$ that
$$
\hat s_u \leq \min_{m\in \mathcal{M}_u} \semax{m,u} L_u^2,
$$
where $\mathcal{M}_u=\{ m \in \mathbb{N} \colon  m \geq 2 \semax{m,u} L_u^2 \}$ and $L_u =\frac{c\|\widehat\Psi_{u0}^{-1}\|_\infty }{c\ell-1} \frac{n}{\lambda}\left\{ C_n + L_n \right\}$.
\end{lemma}

\begin{lemma}\label{Lemma:PostLassoMRateRaw}
Suppose that Assumption \ref{ass: M} holds with $A_u=\{ \delta : \|\delta\|_0 \leq \widehat s_u + s_u\}$ and
\begin{align}
\bar q_{A_u}
>2\max&\bigg\{\Big(\En[M_u(Y_u,X_u,\widetilde \theta_u)] - \En[M_u(Y_u,X_u,\theta_u)]\Big)_+^{1/2},\notag\\
&\quad \Big( \frac{\sqrt{\widehat s_u+s_u}\|\En[\partial_\theta M_u(Y_u,X_u,\theta_u,a_u)]\|_\infty}{\sqrt{\semin{\widehat s_u+s_u,u}}} + 3C_n\Big)\bigg\}.\label{qreq: PostSelection}
\end{align}
Then with probability at least $1 - \Delta_n$, we have for all $u\in\UU$ that
\begin{align*} \|\sqrt{w_{u}} X_u'(\widetilde \theta_u-\theta_u)\|_{\Pn,2} & \leq \Big(\En[M_u(Y_u,X_u,\widetilde \theta_u)] - \En[M_u(Y_u,X_u,\theta_u)]\Big)_+^{1/2} \\
&\quad +  \frac{\sqrt{\widehat s_u+s_u}\|\En[\partial_\theta M_u(Y_u,X_u,\theta_u,a_u)]\|_\infty}{\sqrt{\semin{\widehat s_u+s_u,u}}} + 3C_n.
\end{align*}
In addition, with probability at least $1-\Delta_n$, we have for all $u\in\UU$ that
\begin{equation}
\En[M_u(Y_u,X_u,\widetilde \theta_u)] - \En[M_u(Y_u,X_u,\theta_u)]  \leq \frac{\lambda L}{n}  \|\hat\theta_u - \theta_u\|_1\sup_{u\in\UU}\|\widehat \Psi_{u0}\|_\infty\label{AuxMuUpper}.
\end{equation}
\end{lemma}

%In Lemma \ref{Lemma:PostLassoMRateRaw}, if  $\supp(\hat \theta_u)\subseteq \widetilde T_u$, we have that
%and $\sup_{u\in\UU}\|\En[\partial_\theta M_u(Y_u,X_u,\theta_u,a_u)]\|_\infty \leq C' (\lambda/n) \sup_{u\in\UU}\|\widehat \Psi_{u0}\|_\infty$ with high probability for some constant $C'$ (a consequence of Lemma \ref{Thm:ChoiceLambda} below).

%These results generalize important results of the $\ell_1$-penalized estimators to the case of functional response data and estimated of nuisance functions. They play a key role in our analysis as these results are used to analyze Step 1 (Example 1) and Step 2 (Example 2) of the algorithms discussed in Section \ref{Sec:Application}. However their applicability seems much broader.

A key requirement in Lemmas \ref{Lemma:LassoMRateRaw} and \ref{Lemma:LassoMSparsity} is that $\lambda$ satisfies (\ref{Eq:reg}) with high probability. Therefore, below we provide a choice of $\lambda$ and a set of conditions under which the proposed choice of $\lambda$ satisfies this requirement. Let $d_{\mathcal U}\colon\mathcal U\times\mathcal U\to \mathbb R_+$ denote a metric on $\mathcal U$. Also, let
$$
S_{u} = \partial_\theta M_u(Y_u,X_u,\theta_u,a_u),\quad u\in\UU.
$$
Moreover, let $\underline C$ and $\bar C$ be some strictly positive constants.  Finally, $(\epsilon_n)_{n\geq 1}$, $(\varphi_n)_{n\geq 1}$, and $(N_n)_{n\geq1}$ be some sequences of positive constants, where $\varphi_n = o(1)$.

\medskip
\noindent
{\bf Condition WL.} {\em The constants $\epsilon_n$ and $N_n$ satisfy the inequality $N_n \geq N(\epsilon_n,\mathcal{U},d_\mathcal{U})$ and the following conditions hold:\\
 (i)  ${\displaystyle\sup_{u\in\UU} \max_{k\in [p]} } (\Ep_P[|S_{uk}|^3])^{1/3}\Phi^{-1}(1-\gamma/\{2p N_n\}) \leq \varphi_n n^{1/6}$;\\
 (ii) $\underline C \leq \Ep_P[|S_{uk}|^2] \leq \bar C$, for all $u\in \UU$ and $k\in[p]$;\\
(iii)  with probability at least $1-\Delta_n$,
\begin{align*}
&{\displaystyle\sup_{d_\mathcal{U}(u,u')\leq \epsilon_n}} \| \En[ S_{u}-S_{u'} ]\|_\infty\leq \varphi_n n^{-1/2},\text{ and }\\
&{\displaystyle \sup_{d_\mathcal{U}(u,u')\leq \epsilon_n}  \max_{k\in [p]}} \Big||\Ep_P[S_{uk}^2-S_{u'k}^2]| +|(\En-\Ep_P)[S_{uk}^2]|\Big|\leq \varphi_n.
\end{align*}}
%$$\begin{array}{c}
%{\displaystyle \sup_{u\in\mathcal{U}}\max_{j\leq p}} |(\En-\barEp_P)[f_j(X_u)^2\zeta_{u}^2]| \leq \delta_n, \ \ \mbox{and} \\ % {\displaystyle \sup_{u\in\mathcal{U}}} \En[ r_{u}^2] \leq  c_r^2, \\
%{\displaystyle \sup_{u,u'\in\mathcal{U},d_\mathcal{U}(u,u')\leq \epsilon} } \{(\En+\barEp_P)[(\zeta_{u}-\zeta_{u'})^2]\}^{1/2} \leq C\{\epsilon^\nu + n^{-1/2}\}.
%\end{array}$$
%(iv) Uniformly in $u\in \mathcal{U}$, $\ell \widehat \Psi_{u 0} \leq \widehat \Psi_{u} \leq L \widehat \Psi_{u 0}$,  where $\widehat \Gamma_{u 0jj}=\{\En[f_j(X_u)^2\zeta_{u}^2]\}^{1/2}$,  $1 - \delta_n\leq \ell \leq L  \leq C$ with probability $1-\Delta_n$.

Let
\begin{equation}\label{Eq:Def-lambda}
\lambda =   c' \sqrt{n} \Phi^{-1}(1-\gamma/\{2p N_n \}),
 \end{equation}
where $1-\gamma$ (with $\gamma = \gamma_n = o(1)$) is a confidence level associated with the probability of event (\ref{Eq:reg}), and $c'>c$ is a slack constant. The following lemma shows that this choice of $\lambda$ satisfies \eqref{Eq:reg} with high probability under Condition WL.

\begin{lemma}\label{Thm:ChoiceLambda}
Suppose that Condition WL holds. In addition, suppose that $\lambda$ satisfies \eqref{Eq:Def-lambda} for some $c'>c$ and $\gamma = \gamma_n \in [1/n,1/\log n]$. Then
$$
\Pr_P\left ( \lambda/n \geq c \sup_{u\in\mathcal{U}}\|\hat  \Psi^{-1}_{u0}\En[S_u]\|_\infty \right ) \geq 1-\gamma-o(\gamma)-\Delta_n.
$$
\end{lemma}

Condition WL(iii) is of high level. Therefore, to conclude this section, we present a lemma that gives easy to verify conditions that imply Condition WL(iii).

\begin{lemma}\label{PrimitiveWL}
Suppose that for all $u\in\UU$, $X_u = X$ and $Y_u = H(Y,u)$ where $Y$ is a random variable and $\{H(\cdot,u)\colon u\in\UU\}$ is a VC-subgraph class of functions bounded by one with index $C_Y$ for some constant $C_Y\geq 1$. In addition, suppose that for all $u\in\UU$, we have $S_u = (Y_u - \Ep_P[Y_u\mid X])\cdot X$. Moreover, suppose that $\max_{k\in [p]}\Ep_P[X_k^4] \leq \bar C$, $\underline C\leq \sup_{u\in \mathcal{U},k\in[p]} \Ep_P[S_{uk}^2] \leq \bar C$, and $\Ep_P[\|X\|_\infty^q]^{1/q} \leq K_n$, for some constants $\underline C,\bar C>0$ and $q\geq 4$ and a sequence of constants $(K_n)_{n\geq 1}$. Finally, suppose that $\Ep_P[|Y_{u}-Y_{u'}|^4]\leq C_u|u-u'|^\nu$ for any $u, u' \in \mathcal{U}$ and some constants $\nu$ and $C$. Then we have with probability at least $1-(\log n)^{-1}$ that
\begin{align}
&  \sup_{d_\UU(u,u')\leq 1/n} \|\En[ S_u-S_{u'}]\|_\infty \lesssim \Big(\frac{\log (n p K_n)}{n^{1 + \nu/2}}\Big)^{1/2} + \frac{K_n\log (n p K_n)}{n^{1-1/q}}\label{eq: prim ver 1}\\
&   \sup_{u\in\mathcal{U}}\max_{k\in p}|(\En-\Ep_P)[S_{u k}^2]|
 \lesssim \Big(\frac{ \log (n p K_n)}{n}\Big)^{1/2}  + \frac{K_n^2\log(n p K_n)}{n^{1-2/q}}\label{eq: prim ver 2}\\
&  \max_{k\in[p]}|\Ep_P[S_{u k}^2-S_{u' k}^2]|  \lesssim d_\UU(u,u')^{\nu/4}\label{eq: prim ver 3}
\end{align}
up-to constants that depend only on $\underline C,\bar C,C_Y,C_u$, $q$, and $\nu$.
\end{lemma}

%Lemma \ref{PrimitiveWL} covers several different cases including cases where $Y_u$ is generated by a non-smooth transformation of a random variable $Y$. For example, if $Y_u = 1\{ Y \leq \underline{y} (1-u)+u \bar{y} \}$ where $Y$ has bounded probability density function, we have $d_u = 1$, $K_n = \Ep_P[\|X\|_\infty^q]^{1/q}$, and $\nu = 1$.

%The following corollary summarizes these results for many applications of interest in well behaved designs.

%\begin{corollary}[Rates under Simple Conditions]
%Suppose that with probability $1-o(1)$ we have that $C_n \vee L_n \leq C\{n^{-1}s\log(pn)\}^{1/2}$, $(Lc+1)/(\ell c-1)\leq C$, $w_u = 1$, and Condition WL holds with $\log N_n \leq C\log(pn)$. Further suppose that with probability $1-o(1)$ the sparse minimal and maximal eigenvalues are well behaved, $c\leq \semin{s \ell_n,u } \leq \semax{s \ell_n,u} \leq C$ for some $\ell_n\to \infty$ uniformly over $u\in\UU$. Then with probability $1-o(1)$ we have
%$$ \sup_{u\in\UU}\|X_u'(\hat \theta_u - \theta_u)\|_{\Pn,2} \lesssim \sqrt{\frac{s\log(p n)}{n}}, \ \ \sup_{u\in\UU}\|\hat \theta_u - \theta_u\|_1 \lesssim \sqrt{\frac{s^2\log(p n)}{n}}, \ \ \mbox{and} \ \  \sup_{u\in\UU}\|\hat \theta_u \|_0 \lesssim s.$$
%Moreover, if  $\widetilde T_u=\supp(\hat \theta_u)$, we have that
%$$
%\sup_{u\in\UU}\|X_u'(\tilde \theta_u - \theta_u)\|_{\Pn,2} \lesssim \sqrt{\frac{s\log(p n)}{n}}.
%$$
%\end{corollary}

\section{Proofs for Appendix L}

\subsection*{Proof of Lemma \ref{Lemma:LassoMRateRaw}}
%\begin{proof}[Proof of Lemma \ref{Lemma:LassoMRateRaw}]
For $u\in\UU$, let $\delta_u = \hat\theta_u - \theta_u$ and $S_{u,n} = \En[\partial_\theta M_u(Y_u,X_u,\theta_u,a_u)]$. Throughout the proof, we will assume that the events (a), (b), and (c) in Assumption \ref{ass: M} as well as the event \eqref{Eq:reg} hold. These events hold with probability at least $1-2\Delta_n$. We will show that the inequalities in the statement of Lemma \ref{Lemma:LassoMRateRaw} hold under these events.

By definition of $\hat \theta_u$, we have
$$
\En[M_u(Y_u,X_u,\hat \theta_u)] + \frac{\lambda}{n}\|\widehat\Psi_u\hat\theta_u\|_1  \leq \En[M_u(Y_u,X_u, \theta_u)] + \frac{\lambda}{n}\|\widehat\Psi_u\theta_u\|_1.$$
Thus,
\begin{align}
&\En[M_u(Y_u,X_u,\hat \theta_u)] - \En[M_u(Y_u,X_u,\theta_u)] \notag \\
& \qquad \leq \frac{\lambda}{n}\|\widehat\Psi_u\theta_u\|_1 - \frac{\lambda}{n}\|\widehat\Psi_u\hat\theta_u\|_1 \leq \frac{\lambda}{n}\|\widehat\Psi_u\delta_{u,T_u}\|_1 - \frac{\lambda}{n} \|\widehat\Psi_u\delta_{u,T_u^c}\|_1 \notag\\
&\qquad \leq \frac{\lambda L}{n}\|\widehat\Psi_{u0}\delta_{u,T_u}\|_1 - \frac{\lambda \ell}{n} \|\widehat\Psi_{u0}\delta_{u,T_u^c}\|_1.\label{QQ:UpperBound}
\end{align}
Moreover, by convexity of $\theta \mapsto M_u(Y_{u},X_{u},\theta)$, we have
\begin{align}
& \En[M_u(Y_u,X_u,\hat \theta_u)] - \En[M_u(Y_u,X_u,\theta_u)]  \notag\\
& \geq \En[\partial_{\theta} M_u(Y_u,X_u,\theta_u)]'\delta_u \geq -\frac{\lambda}{n}\frac{1}{c}\|\widehat\Psi_{u0}\delta_{u}\|_1 -C_n\|\sqrt{w_{u}} X_u'\delta_u\|_{\Pn,2}\label{QQ:LowerBound}
\end{align}
where the second inequality holds by Assumption \ref{ass: M}(a) and
$$
\lambda/n\geq c\sup_{u\in\mathcal{U}}\|\widehat\Psi_{u0}^{-1}S_{u,n} \|_\infty
$$
since
\begin{align}
| \En[\partial_{\theta} M_u(Y_u,X_u,\theta_u)]'\delta_u|
& =  |S_{u,n}'\delta_u + \{\En[\partial_{\theta} M_u(Y_u,X_u,\theta_u)]-S_{u,n}\}'\delta_u| \nonumber \\
 & \leq    |S_{u,n}'\delta_u| + |\{\En[\partial_{\theta} M_u(Y_u,X_u,\theta_u)]-S_{u,n}\}'\delta_u|  \nonumber \\
& \leq   \|\widehat\Psi_{u0}^{-1}S_{u,n} \|_\infty \|\widehat\Psi_{u0}\delta_u\|_1+ C_n\|\sqrt{w_{u}}  X_u'\delta_u\|_{\Pn,2} \nonumber \\
%& \leq & \frac{\lambda}{n}\frac{1}{c}\|\widehat\Psi_{u0}\delta_{u,T_u^c}\|_1 + \|r_{u}/\sqrt{w_{u}}\|_{\Pn,2}\|\sqrt{w_{u}} X_u'\delta_u\|_{\Pn,2},
& \leq  \frac{\lambda}{n}\frac{1}{c}\|\widehat\Psi_{u0}\delta_{u}\|_1 + C_n\|\sqrt{w_{u}} X_u'\delta_u\|_{\Pn,2}.
\label{useful}
\end{align}

Combining (\ref{QQ:UpperBound}) and (\ref{QQ:LowerBound}), we have
\begin{equation}\label{Eq:92b}\frac{\lambda}{n}\frac{c\ell-1}{c}\| \widehat\Psi_{u0} \delta_{u,T_u^c}\|_1 \leq \frac{\lambda}{n}\frac{Lc+1}{c}\| \widehat\Psi_{u0}\delta_{u,T_u}\|_1+C_n\|\sqrt{w_{u}} X_u'\delta_u\|_{\Pn,2}, \end{equation}
and for $\tilde c = \frac{Lc+1}{\ell c -1}\sup_{u\in\mathcal{U}}\|\widehat\Psi_{u0}\|_\infty\|\widehat\Psi_{u0}^{-1}\|_\infty$, we have
$$
\|\delta_{u,T_u^c}\|_1 \leq \tilde c \|\delta_{u,T_u}\|_1 + \frac{n}{\lambda}\frac{c\|\widehat\Psi_{u0}^{-1}\|_\infty}{\ell c - 1}C_n\|\sqrt{w_{u}} X_u'\delta_u\|_{\Pn,2}.
$$
Now, suppose that $\delta_u \not\in \Delta_{2\tilde c,u}$, namely $\|\delta_{u,T_u^c}\|_1> 2\tilde c\|  \delta_{u,T_u}\|_1$. Then
$$
2\tilde c \|\delta_{u,T_u}\|_1\leq  \tilde c \|\delta_{u,T_u}\|_1 + \frac{n}{\lambda}\frac{c\|\widehat\Psi_{u0}^{-1}\|_\infty}{\ell c - 1}C_n\|\sqrt{w_{u}} X_u'\delta_u\|_{\Pn,2},
$$
and so
$$
\|\delta_{u,T_u}\|_1\leq \frac{n}{\lambda}\frac{c\|\widehat\Psi_{u0}^{-1}\|_\infty}{\ell c - 1}C_n\|\sqrt{w_{u}} X_u'\delta_u\|_{\Pn,2}
$$
since $\tilde c\geq 1$. Also,
$$
\|\delta_{u,T_u^c}\|_1 \leq \frac{1}{2}\|\delta_{u,T_u^c}\|_1 +  \frac{n}{\lambda}\frac{c\|\widehat\Psi_{u0}^{-1}\|_\infty}{\ell c - 1}C_n\|\sqrt{w_{u}} X_u'\delta_u\|_{\Pn,2},
$$
and so
$$
\|\delta_{u,T_u^c}\|_1 \leq  \frac{2n}{\lambda}\frac{c\|\widehat\Psi_{u0}^{-1}\|_\infty}{\ell c - 1}C_n\|\sqrt{w_{u}} X_u'\delta_u\|_{\Pn,2}.
$$
Therefore,
\begin{equation}\label{BoundL1Zero}
\|\delta_u\|_1 \leq  \frac{3n}{\lambda}\frac{c\|\widehat\Psi_{u0}^{-1}\|_\infty}{\ell c - 1}C_n\|\sqrt{w_{u}} X_u'\delta_u\|_{\Pn,2} =: I_u,
\end{equation}
as long as $\delta_u \not\in \Delta_{2\tilde c,u}$.
%\begin{align*}
%\| \delta_u\|_1 & \leq ( 1 +
%\{2\tilde c\}^{-1}) \| \delta_{u,T_u^c}\|_1\\
%& \leq ( 1 + \{2\tilde c\}^{-1})\tilde c\| \delta_{u,T_u}\|_1+( 1 + \{2\tilde c\}^{-1})\frac{n}{\lambda}\frac{c\|%\widehat\Psi_{u0}^{-1}\|_\infty}{\ell c - 1} C_n\|\sqrt{w_{u}} X_u'\delta_u\|_{\Pn,2}\\
%& \leq ( 1 +
%\{2\tilde c\}^{-1})\frac{1}{2}\| \delta_{u,T_u^c}\|_1+( 1 + \{2\tilde c\}^{-1})\frac{n}{\lambda}\frac{c\|\widehat\Psi_{u0}^{-1}\|_\infty}{\ell c - 1} C_n\|\sqrt{w_{u}} X_u'\delta_u\|_{\Pn,2}.
%\end{align*}
%The relation above implies that if $\delta_u \not\in \Delta_{2\tilde c,u}$
%\begin{align}\label{BoundL1Zero}
%\| \delta_u\|_1 & \leq \frac{4\tilde c}{2\tilde c-1}( 1 + \{2\tilde c\}^{-1})\frac{n}{\lambda}\frac{c\|%\widehat\Psi_{u0}^{-1}\|_\infty}{\ell c - 1} C_n\|\sqrt{w_{u}} X_u'\delta_u\|_{\Pn,2}\\
%& \leq \frac{6c\|\widehat\Psi_{u0}^{-1}\|_\infty}{\ell c - 1}\frac{n}{\lambda}C_n\|\sqrt{w_{u}} X_u'\delta_u\|%_{\Pn,2}=: I_u, \notag
% \end{align}
%where we used that $\frac{4\tilde c}{2\tilde c-1}( 1 + \{2\tilde c\}^{-1})\leq 6$ since $\tilde c \geq 1$.
In addition, if $\delta_u\in\Delta_{2\tilde c,u}$, then
$$
\|\delta_{u,T_u}\|_1 \leq \sqrt s \|\delta_{u,T_u}\|\leq \frac{\sqrt{s}}{\bar\kappa_{2\tilde c}} \|\sqrt{w_{u}} X_u'\delta_u\|_{\Pn,2} =: II_u.
$$
Hence, in both cases, we have
\begin{equation}\label{BoundFirstL1}
\|\delta_{u,T_u}\|_1 \leq I_u + II_u.
\end{equation}

Next, for every $u\in\mathcal{U}$, since
$$
A_u = \Delta_{2\tilde c,u} \cup \{ \delta : \|\delta\|_1 \leq \frac{3 n}{\lambda} \frac{c\|\widehat\Psi_{u0}^{-1}\|_\infty}{\ell c - 1}C_n\|\sqrt{w_{u}} X_u'\delta\|_{\Pn,2}\},
$$
it follows that $\delta_u\in A_u$, and we have
\begin{align*}
&\|\sqrt{w_{u}} X_u'\delta_u\|_{\Pn,2}^2 \wedge \left\{ \bar q_{A_u}\|\sqrt{w_{u}} X_u'\delta_u\|_{\Pn,2}\right\} \\
&\quad \leq \En[M_u(Y_u,X_u,\hat \theta_u)] - \En[M_u(Y_u,X_u,\theta_u)] \\
&\qquad -\En[\partial_{\theta} M_u(Y_u,X_u,\theta_u)]'\delta_u + 2C_n\|\sqrt{w_{u}} X_u'\delta_u\|_{\Pn,2} \\
& \quad \leq \Big(L+\frac{1}{c}\Big)\frac{\lambda}{n}\|\widehat\Psi_{u0}\delta_{u,T_u}\|_1  +3C_n\|\sqrt{w_{u}} X_u'\delta_u\|_{\Pn,2} \\
&\quad  \leq \Big(L+\frac{1}{c}\Big)\frac{\lambda}{n}\|\widehat\Psi_{u0}\|_\infty\left\{I_u + II_u \right\}+3C_n\|\sqrt{w_{u}} X_u'\delta_u\|_{\Pn,2} \\
&\quad \leq \Big\{\Big(L+\frac{1}{c}\Big)\|\widehat\Psi_{u0}\|_\infty\frac{\lambda\sqrt{s}}{n\bar\kappa_{2\tilde c}}+ 6\tilde c C_n\Big\}\|\sqrt{w_{u}} X_u'\delta_u\|_{\Pn,2},
\end{align*}
where the first inequality follows from Assumption \ref{ass: M}(c), the second from  (\ref{QQ:UpperBound}),  (\ref{useful}), and $\ell\geq 1/c$, the third from $\|\widehat\Psi_{u0}\delta_{u,T_u}\|_1\leq \|\widehat\Psi_{u0}\|_\infty \|\delta_{u,T_u}\|_1$ and  (\ref{BoundFirstL1}),  and the fourth from definitions of $I_u$, $II_u$, and $\tilde c$. Thus, as long as
$$
\bar q_{A_u} > \Big(L+\frac{1}{c}\Big)\|\widehat\Psi_{u0}\|_\infty\frac{\lambda\sqrt{s}}{n\bar\kappa_{2\tilde c}}+ 6\tilde c C_n,\quad\text{for all }u\in\UU,
$$
which is assumed, we have
$$
 \|\sqrt{w_{u}} X_u'\delta_u\|_{\Pn,2} \leq \Big(L+\frac{1}{c}\Big)\|\widehat\Psi_{u0}\|_\infty\frac{\lambda\sqrt{s}}{n\bar\kappa_{2\tilde c}}+ 6\tilde c C_n =: III_u, \quad \text{for all } u\in \mathcal{U}.
$$
This gives the first asserted claim.
%Since  the inequality $(x^2 \wedge a x) \leq b x$ holding for $x>0$ and $b< a <0$ implies $x \leq b$, the above system of the inequalities, provided that for every $u\in\mathcal{U}$ $$\bar q_{A_u}>3\left\{\Big(L+\frac{1}{c}\Big)\|\widehat\Psi_{u0}\|_\infty\frac{\lambda\sqrt{s}}{n\bar\kappa_{2\tilde c}}+ 6\tilde c C_n\right\},$$
% implies that
%$$ \|\sqrt{w_{u}} X_u'\delta_u\|_{\Pn,2} \leq 3\left\{\Big(L+\frac{1}{c}\Big)\|\widehat\Psi_{u0}\|_\infty\frac{\lambda\sqrt{s}}{n\bar\kappa_{2\tilde c}}+ 6\tilde c C_n\right\} =: III_u \ \ \ \mbox{for every} \ \ u\in \mathcal{U}.$$
The second asserted claim follows from
\begin{align*}
\|\delta_u\|_1 & \leq 1\{ \delta_u \in \Delta_{2\tilde c,u} \}\|\delta_u\|_1 + 1\{ \delta_u \notin \Delta_{2\tilde c,u} \}\|\delta_u\|_1 \\
 & \leq (1+2\tilde c) II_u + I_u \leq   \Big( \frac{(1+2\tilde c)\sqrt{s}}{\bar\kappa_{2\tilde c}}+ \frac{3n}{\lambda}\frac{c\|\widehat\Psi_{u0}^{-1}\|_\infty}{\ell c - 1}C_n\Big) III_u.
\end{align*}
This completes the proof. \qed
%\end{proof}

\subsection*{Proof of Lemma \ref{Lemma:LassoMSparsity}}
%\begin{proof}[Proof of Lemma \ref{Lemma:LassoMSparsity}]
For $u\in\UU$, let $S_{u,n} = \En[\partial_\theta M_u(Y_u,X_u,\theta_u,a_u)]$. Throughout the proof, we will assume that the events (a), (b), and (c) in Assumption \ref{ass: M} as well as the events \eqref{Eq:reg} and \eqref{eq: def Ln} hold. These events hold with probability at least $1-3\Delta_n$. We will show that the inequalities in the statement of Lemma \ref{Lemma:LassoMSparsity} hold under these events.

By definition of $\widehat \theta_u$, there exists a subgradient $\partial_\theta \En[M_u(Y_u,X_u,\hat\theta_u)]$ of $\En[M_u(Y_u,X_u,\hat\theta_u)]$, such that for every $j$ with $|\hat\theta_{uj}|>0$,
$$ |(\widehat\Psi_{u}^{-1}\partial_\theta \En[M_u(Y_u,X_u,\hat\theta_u)])_j| = \lambda/n.$$
Therefore, we have
\begin{align*}
&\frac{\lambda}{n}\sqrt{\hat s_u}  = \|(\widehat\Psi_{u}^{-1}\partial_\theta \En[M_u(Y_u,X_u,\hat\theta_u)])_{\hat T_u}\|\\
 & \leq  \|(\widehat\Psi_{u}^{-1}S_{u,n})_{\hat T_u}\| + \|(\widehat\Psi_{u}^{-1}\{ \En[\partial_\theta M_u(Y_u,X_u,\theta_u)]-S_{u,n}\})_{\hat T_u}\|\\
& \quad+\|(\widehat\Psi_{u}^{-1}\{ \En[\partial_\theta M_u(Y_u,X_u,\hat\theta_u)-\partial_\theta M_u(Y_u,X_u,\theta_u)]\})_{\hat T_u}\| \\
& \leq \|\widehat\Psi_{u}^{-1}\widehat\Psi_{u0}\|_\infty \|\widehat\Psi_{u0}^{-1} \En[S_{u,n}]\|_\infty \sqrt{\hat s_u} \\
&\quad +\|\widehat\Psi_u^{-1}\|_\infty C_n \sup_{\|\delta\|=1,\|\delta\|_0\leq \hat s_u}\|\sqrt{w_u}X_u'\delta\|_{\Pn,2} \\
& \quad+ \|\widehat\Psi_{u}^{-1}\|_\infty \sup_{\|\delta\|=1,\|\delta\|_0\leq \hat s_u}|\{ \En[\partial_\theta M_u(Y_u,X_u,\hat\theta_u)-\partial_\theta M_u(Y_u,X_u,\theta_u)]\}'\delta|\\
& \leq  \frac{\lambda}{c\ell n}\sqrt{\hat s_u} + \frac{\|\widehat\Psi_{u0}^{-1}\|_\infty}{\ell} \{ C_n  + L_n\}\sup_{\|\delta\|=1,\|\delta\|_0\leq \hat s_u}\|X_u'\delta\|_{\Pn,2}
\end{align*}
where the first inequality follows from the triangle inequality, the second from Assumption \ref{ass: M}(a), and the third from Assumption \ref{ass: M}(b) and inequalities \eqref{Eq:reg} and \eqref{eq: def Ln}.

Now, recall that $L_u = \frac{n}{\lambda}\frac{c\|\widehat\Psi_{u0}^{-1}\|_\infty }{c\ell-1} \left\{ C_n + L_n \right\}$. In addition, note that
$$
\sup_{\|\delta\|=1,\|\delta\|_0\leq \hat s_u}\|X_u'\delta\|_{\Pn,2}^2 = \semax{\hat s_u,u}.
$$
Thus,
 we have \begin{equation}\label{Eq:SparsityL}\hat s_u \leq \semax{\hat s_u,u}L_u^2.\end{equation}
 Consider any $M \in \mathcal{M}_u=\{ m \in \mathbb{N}: m >  2\semax{m,u}L_u^2\} $, and suppose that $\widehat s_u > M$. By the sublinearity of the maximum sparse eigenvalue (Lemma 3 in \cite{BC-PostLASSO}),  for any integer $k \geq 0$ and constant $\ell \geq 1$, we have $\semax{\ell k,u}  \leq  \lceil \ell \rceil \semax{k,u}, $
where $\lceil \ell \rceil$ denotes the  ceiling of $\ell$. Therefore,
\begin{align*}
\hat s_u  &\leq \semax{\hat s_u,u}L_u^2= \semax{M\hat s_u/M,u}L_u^2 \\
& \leq  \ceil{\frac{\hat s_u}{M}}\semax{M,u}L_u^2\leq \frac{2\widehat s_n}{M}\phi_{\max}(M,u)L_u^2
\end{align*}
since $\ceil{k}\leq 2k$ for any $k\geq 1$. Therefore, we have
$ M \leq   2\semax{M,u}L_u^2$
which violates the condition that $M \in \mathcal{M}_u$. Therefore, we have $\widehat s_u \leq M$.
In turn, applying (\ref{Eq:SparsityL}) once more with $\widehat s_u \leq M$ we obtain
 $ \hat s_u \leq   \semax{M,u}L_u^2.$ The result follows by minimizing the bound over $M \in \mathcal{M}_u$. \qed
%\end{proof}

\subsection*{Proof of Lemma \ref{Lemma:PostLassoMRateRaw}}
%\begin{proof}[Proof of Lemma \ref{Lemma:PostLassoMRateRaw}]
The second asserted claim, inequality \eqref{AuxMuUpper}, follows from the observation that with probability at least $1-\Delta_n$, for all $u\in\UU$, we have that
\begin{align*}
&\En[M_u(Y_u,X_u,\widetilde \theta_u)] - \En[M_u(Y_u,X_u,\theta_u)] \\
&\qquad \leq \En[M_u(Y_u,X_u,\hat \theta_u)] - \En[M_u(Y_u,X_u,\theta_u)] \\
&\qquad \leq \frac{\lambda}{n}\|\widehat\Psi_u \theta_u\|_1 - \frac{\lambda}{n}\|\widehat\Psi_u\widehat\theta_u\|_1
\leq \frac{\lambda}{n}  \|\hat\theta_u - \theta_u\|_1 \cdot \sup_{u\in\UU}\|\widehat \Psi_{u}\|_\infty
\\
& \qquad \leq \frac{\lambda L}{n}  \|\hat\theta_u - \theta_u\|_1 \cdot \sup_{u\in\UU}\|\widehat \Psi_{u0}\|_\infty,
\end{align*}
where the first inequality holds by the definition of $\widetilde \theta_u$, the second by the definition of $\widehat\theta_u$, the third by the triangle inequality, and the fourth by Assumption \ref{ass: M}(ii).

To prove the first asserted claim, assume that the events (a), (b), and (c) in Assumption \ref{ass: M} hold. These events hold with probability at least $1-\Delta_n$. We will show that the asserted claim holds under these events.

For $u\in\UU$, let $\tilde \delta_u=\widetilde \theta_u - \theta_u$, $S_{u,n} =  \En[\partial_\theta M_u(Y_u,X_u,\theta_u,a_u)]$, and  $\tilde t_{u} = \|\sqrt{w_{u}} X_u'\tilde \delta_u\|_{\Pn,2}$. By the inequality in Assumption \ref{ass: M}(c), we have
\begin{align*}
\tilde t_u^2 \wedge \left\{ \bar q_{A_u}\tilde t_u\right\} & \leq \En[M_u(Y_u,X_u,\widetilde \theta_u)] - \En[M_u(Y_u,X_u,\theta_u)] \\
&\quad -\En[\partial_{\theta} M_u(Y_u,X_u,\theta_u)]'\tilde\delta_u + 2C_n\tilde t_u \\
& \leq \En[M_u(Y_u,X_u,\widetilde \theta_u)] - \En[M_u(Y_u,X_u,\theta_u)] \\
&\quad + \|S_{u,n}\|_\infty \|\tilde\delta_u\|_1 + 3C_n\tilde t_u\\
& \leq \En[M_u(Y_u,X_u,\widetilde \theta_u)] - \En[M_u(Y_u,X_u,\theta_u)]  \\
&\quad + \tilde t_u\Big( \frac{\sqrt{\widehat s_u+s_u}\|S_{u,n}\|_\infty}{\sqrt{\semin{\widehat s_u+s_u,u}}} + 3C_n\Big)
\end{align*}
where the second inequality holds by  calculations as in (\ref{useful}), and the third inequality follows from
$$
\|\tilde\delta_u\|_1\leq \sqrt{\widehat s_u+s_u}\|\tilde\delta_u\|_2\leq \frac{\sqrt{\widehat s_u+s_u}}{\sqrt{\semin{\widehat s_u+s_u,u}}}\|\sqrt{w_u}X_u'\tilde\delta_u\|_{\Pn,2}.
$$
Next, if $\tilde t_u^2 > \bar q_{A_u}\tilde t_u$, then
$$ \bar q_{A_u}\tilde t_u \leq \frac{\bar q_{A_u}}{2} \{\En[M_u(Y_u,X_u,\widetilde \theta_u)] - \En[M_u(Y_u,X_u,\theta_u)]\}_+^{1/2} +\frac{\bar q_{A_u}}{2}\tilde t_u, $$
so that $\tilde t_u\leq  \{\En[M_u(Y_u,X_u,\widetilde \theta_u)] - \En[M_u(Y_u,X_u,\theta_u)]\}_+^{1/2}$. On the other hand, if $\tilde t_u^2\leq \bar q_{A_u}\tilde t_u$, then
{\small
$$
\tilde t_u^2 \leq \{\En[M_u(Y_u,X_u,\widetilde \theta_u)] - \En[M_u(Y_u,X_u,\theta_u)]\} + \tilde t_u\left( \frac{\sqrt{\widehat s_u+s_u}\|S_{u,n}\|_\infty}{\sqrt{\semin{\widehat s_u+s_u}}} + 3C_n\right).
$$}\!Since for positive numbers  $a$, $b$, $c$, inequality $a^2 \leq b + ac$ implies $a\leq \sqrt{b} + c$, we have
{\small
$$
\tilde t_u \leq \{\En[M_u(Y_u,X_u,\widetilde \theta_u)] - \En[M_u(Y_u,X_u,\theta_u)]\}_+^{1/2} + \left( \frac{\sqrt{\widehat s_u+s_u}\|S_{u,n}\|_\infty}{\sqrt{\semin{\widehat s_u+s_u}}} + 3C_n\right).
$$}\!In both cases, the inequality in the asserted claim holds. This completes the proof. \qed
%\end{proof}

\subsection*{Proof of Lemma \ref{Thm:ChoiceLambda}}
For brevity of notation, denote $\epsilon = \epsilon_n$. Also, let $\mathcal A_n$ denote the event that the inequalities in Condition WL(iii) hold. Then $\Pr_P(\mathcal A_n)\geq 1 - \Delta_n$. Further, by the triangle inequality,
\begin{align}
\sup_{u\in\mathcal{U}} \| \hat \Psi^{-1}_{u 0} \En[ S_u ]\|_\infty & \leq \max_{u\in\mathcal{U}^\epsilon} \| \hat \Psi^{-1}_{u 0} \En[S_u ]\|_\infty\label{eq: lemma lambda choice first}\\
&  \quad +\sup_{u\in\mathcal{U}^\epsilon,u'\in\mathcal{U},d_\mathcal{U}(u,u')\leq \epsilon} \| \hat \Psi^{-1}_{u 0} \En[S_u ] - \hat \Psi^{-1}_{u'0} \En[ S_{u'} ]\|_\infty\notag
\end{align}
where $\mathcal{U}^\epsilon$ is a minimal $\epsilon$-net of $\mathcal{U}$ so that $|\mathcal{U}^\epsilon|\leq N_n$.

For each $k=1,\ldots,p$ and $u\in \mathcal{U}^\epsilon$, we apply Lemma \ref{Lemma: MDSN} with $Z_i := S_{u k i}$, $\mu = 1$, and $\ell_n = c''\varphi_n^{-1}$, where $c''$ is a small enough constant that can be chosen to depend only on $\underline C$ and $\bar C$. Then Condition WL(i,ii) implies that
$$
0\leq \Phi^{-1}\left(1 - \frac{\gamma}{2p N_n}\right)\leq \frac{n^{1/6}M_n}{\ell_n} - 1
$$
where $M_n = (\Ep_P[Z_1^2])^{1/2}/(\Ep_P[Z_1^3])^{1/3}$, and so applying Lemma \ref{Lemma: MDSN}, the union bound, and the inequality $|\UU^{\epsilon}|\leq N_n$ gives
\begin{align}
&\displaystyle\Pr_P \left( {\displaystyle\sup_{u\in\mathcal{U}^\epsilon}\max_{k\in[p]}}
\frac{|\sqrt{n}\En[ S_{uk} ]|}{\sqrt{\En[S_{uk}^2]}} > \Phi^{-1}\Big(1-\mbox{$\frac{\gamma}{2p N_n}$}\Big) \right) \notag \\
&\qquad  \leq 2pN_n \cdot \frac{\gamma}{2p N_n}\cdot\Big( 1 + O(\varphi_n^{1/3})\Big) \leq \gamma+ o(\gamma) \label{SNcontrol}
\end{align}
since $\varphi_n = o(1)$. Also, observe that
$$
\max_{u\in\mathcal{U}^\epsilon} \| \hat \Psi^{-1}_{u 0} \En[S_u ]\|_\infty
 =  \sup_{u\in\mathcal{U}^\epsilon}\max_{k\in[p]} \frac{|\En[ S_{uk} ]|}{\sqrt{\En[S_{uk}^2]}}.
$$
Therefore, \eqref{SNcontrol} implies that with probability at least $1 - \gamma - o(\gamma)$,
\begin{equation}\label{eq: main event SNMD}
\max_{u\in\mathcal{U}^\epsilon} \| \hat \Psi^{-1}_{u 0} \En[S_u ]\|_\infty \leq n^{-1/2}\Phi^{-1}\left(1 - \frac{\gamma}{2p N_n}\right)
\end{equation}

Further, by the triangle inequality,
 \begin{align}
&{\displaystyle \sup_{u\in\mathcal{U}^\epsilon,u'\in\mathcal{U},d_\mathcal{U}(u,u')\leq \epsilon}} \| \hat \Psi^{-1}_{u 0} \En[ S_{u} ] - \hat \Psi^{-1}_{u'0} \En[  S_{u'} ]\|_\infty  \notag\\
&\qquad \qquad \leq {\displaystyle  \sup_{u\in\mathcal{U}^\epsilon,u'\in\mathcal{U},d_\mathcal{U}(u,u')\leq \epsilon}} \|( \hat \Psi^{-1}_{u 0} - \hat \Psi^{-1}_{u'0}) \hat \Psi_{u 0} \|_\infty \| \hat \Psi^{-1}_{u 0}\En[  S_{u}]\|_\infty \label{EqTri2}\\
&\qquad \qquad \qquad +  {\displaystyle \sup_{u,u'\in\mathcal{U},d_\mathcal{U}(u,u')\leq \epsilon}} \| \En[  S_{u}-  S_{u'} ]\|_\infty\|\hat \Psi^{-1}_{u'0} \|_\infty.\label{EqTri8}
\end{align}
To control the expression in (\ref{EqTri2}), note that by Condition WL(ii), on the event $\mathcal A_n$, $\hat \Psi_{u 0kk}$ is bounded away from zero  uniformly over $u\in\mathcal{U}$ and $k\in[p]$. Thus, we have uniformly over $u\in\mathcal{U}$ and $k\in[p]$ that
{\small
\begin{equation}\label{eqaux37a}|(\hat \Psi^{-1}_{u 0kk}-\hat \Psi^{-1}_{u'0kk})\hat \Psi_{u 0kk}|=|\hat \Psi_{u 0kk}-\hat \Psi_{u' 0kk}|\hat \Psi_{u' 0kk}^{-1} \lesssim |\hat \Psi_{u 0kk}-\hat \Psi_{u' 0kk}|
\end{equation}}\!on the event $\mathcal A_n$. Moreover, we have
{\small
\begin{align*}
& \sup_{u,u'\in\mathcal{U},d_\mathcal{U}(u,u')\leq \epsilon} \max_{k\in[p]}\Big| \{\En[S_{uk}^2]\}^{1/2}-\{\En[S_{u'k}^2]\}^{1/2}\Big| \\
&\qquad\qquad\leq \sup_{u,u'\in\mathcal{U},d_\mathcal{U}(u,u')\leq \epsilon}\max_{k\in[p]}
\Big(|\En[S_{uk}^2]-\En[S_{u'k}^2]|\Big)^{1/2} \\
&\qquad\qquad  \leq\sup_{u,u'\in\mathcal{U},d_\mathcal{U}(u,u')\leq \epsilon}\max_{k\in[p]}
\Big(2|(\En-\Ep_P)[S_{uk}^2]|+|\Ep_P[S_{uk}^2-S_{u'k}^2]|\Big)^{1/2} \lesssim \varphi_n^{1/2}
\end{align*}}\!on the event $\mathcal A_n$. Combining this bound with (\ref{eqaux37a}) implies that
$$
{\displaystyle  \sup_{u,u'\in\mathcal{U},d_\mathcal{U}(u,u')\leq \epsilon}} \| (\hat \Psi^{-1}_{u 0} - \hat \Psi^{-1}_{u'0})\hat\Psi_{u0} \|_\infty \lesssim \varphi_n^{1/2}
$$
on the event $\mathcal A_n$. Also, using standard bounds for the tails of Gaussian random variables gives
$$
\Phi^{-1}\left(1 - \frac{\gamma}{2p N_n}\right)\lesssim \sqrt{\log(2 p N_n/\gamma)}.
$$
Thus, on the intersection of events $\mathcal A_n$ and \eqref{eq: main event SNMD}, we have
\begin{align*}
&\sup_{u\in \mathcal{U}^\epsilon,u'\in\mathcal{U},d_\mathcal{U}(u,u')\leq \epsilon} \|  (\hat\Psi^{-1}_{u 0} - \hat \Psi^{-1}_{u'0})\hat \Psi_{u0} \|_\infty \|\widehat\Psi^{-1}_{u0}\En[ S_u ]\|_\infty \\
&\qquad  \qquad \lesssim (\varphi_n/n)^{1/2}\sqrt{\log(p N_n/\gamma)}.
\end{align*}
Finally, on the event $\mathcal A_n$, we have that the expression in \eqref{EqTri8} satisfies
$$
\sup_{u,u'\in\mathcal{U},d_\mathcal{U}(u,u')\leq \epsilon} \| \En[  S_{u}-  S_{u'} ]\|_\infty\|\hat \Psi^{-1}_{u'0} \|_\infty \leq \varphi_n n^{-1/2}.
$$

It follows that on the intersection of events $\mathcal A_n$ and \eqref{eq: main event SNMD}, for $n$ large enough, we have
$$
\sup_{u\in\mathcal{U}^\epsilon,u'\in\mathcal{U},d_\mathcal{U}(u,u')\leq \epsilon} \| \hat \Psi^{-1}_{u 0} \En[ S_{u} ] - \hat \Psi^{-1}_{u'0} \En[  S_{u'} ]\|_\infty\leq \frac{c' - c}{c}\cdot \Phi^{-1}\left(1 - \frac{\gamma}{2p N_n}\right),
$$
where we again used standard tail bounds for the tails of Gaussian random variables. The asserted claim now follows by recalling the inequality \eqref{eq: lemma lambda choice first} and noting that $\Pr_P(\mathcal A_n)\geq 1 - \Delta_n$ and that \eqref{eq: main event SNMD} holds with probability at least $1 - \gamma - o(\gamma)$. \qed

\subsection*{Proof of Lemma \ref{PrimitiveWL}}
%\begin{proof}[Proof of Lemma \ref{PrimitiveWL}]
For $j\in[p]$, let
\begin{align*}
&\mathcal{F}_j=\Big\{(Y,X)\mapsto Y_u X_j\colon u \in \mathcal{U}\Big\},\ \mathcal{F}'_j=\Big\{(Y,X)\mapsto X_j\Ep_P[Y_u\mid X]\colon u \in \mathcal{U}\Big\},\\
&\mathcal{G}_j = \Big\{(Y,X)\mapsto X_j^2\zeta_u^2\colon u \in \mathcal{U} \Big\}
\end{align*}
where $\zeta_u=Y_u - \Ep_P[Y_u\mid X]$. Note that the function $F(Y,X) = \|X\|_\infty$ is an envelope both for $\mathcal F_j$ and for $\mathcal F_j'$ for all $j\in[\pp]$. By assumption, $F$ can be chosen to satisfy $\|F\|_{P,q} \leq K_n$.

Because $\mathcal{F}_j$ is a product of a VC-subgraph class of functions with index bounded by $C_Y$ and a single function, Lemma \ref{lemma: andrews}(1) implies that its uniform entropy numbers obey
\begin{equation}\label{CNvc}
\log N( \epsilon \|F\|_{Q,2}, \mathcal{F}_j, \|\cdot\|_{Q,2}) \lesssim \log (e/\epsilon),\quad 0<\epsilon\leq 1.
\end{equation}
Also, Lemma \ref{Lemma:PartialOutCovering} implies that the uniform entropy numbers of $\mathcal{F}'_j$ obey
{\small
\begin{equation*}
\log \sup_Q N( \epsilon \|F\|_{Q,2}, \mathcal{F}'_j, \|\cdot\|_{Q,2})  \leq \log \sup_Q N\Big( \frac{\epsilon}{2} \|F\|_{Q,2}, \mathcal{F}_j, \|\cdot\|_{Q,2}\Big),\quad 0<\epsilon\leq 1.
\end{equation*}}\!Further, since $\mathcal{G}_j \subset (\mathcal{F}_j-\mathcal{F}'_j)^2$, $G=4F^2$ is an envelope for $\mathcal{G}_j$, and the uniform entropy numbers of $\mathcal G_j$ obey for all $\epsilon\in(0,1]$,
\begin{align}\label{Chain1G}
\log  N( \epsilon \|G\|_{Q,2}, \mathcal{G}_j, \|\cdot\|_{Q,2}) & \leq 2 \log N\Big( \frac{\epsilon}{2} \|2F\|_{Q,2}, \mathcal{F}_j-\mathcal{F}'_j, \|\cdot\|_{Q,2}\Big)\\
& \leq 2 \log  N\Big( \frac{\epsilon}{4} \|F\|_{Q,2}, \mathcal{F}_j, \|\cdot\|_{Q,2}\Big) \notag \\
&\quad + 2 \log  N\Big( \frac{\epsilon}{4} \|F\|_{Q,2}, \mathcal{F}'_j, \|\cdot\|_{Q,2}\Big)\notag \\
&  \leq 4 \log \sup_Q N\Big( \frac{\epsilon}{8} \|F\|_{Q,2}, \mathcal{F}_j, \|\cdot\|_{Q,2}\Big),\notag
\end{align}
where the first and the second inequalities follow from Lemma \ref{lemma: andrews}(2), and the third from the bound on uniform entropy numbers of $\mathcal F_j'$ above.
Hence, Lemma \ref{lemma: andrews}(2) implies that the uniform entropy numbers of $\mathcal{G}=\cup_{j\in[p]}\mathcal{G}_j$ obey
$$
\log N( \epsilon \|G\|_{Q,2}, \mathcal{G}, \|\cdot\|_{Q,2})  \lesssim \log(p/\epsilon),\quad 0<\epsilon\leq 1,
$$
where $G = 4F^2$ is its envelope. Therefore, since $|S_{u j}|\leq 2|X_j|$ and also $\max_{j\leq p}\Ep_P[X_{j}^4]\leq \bar C$ by assumption, Lemma \ref{lemma:CCK} implies that with probability at least $1-(\log n)^{-1}$,
$$
\sup_{u\in\mathcal{U}}\max_{j\leq p}|(\En-\Ep_P)[S_{uj}^2]|  \lesssim \sqrt{\frac{\log (n p K_n)}{n}} + \frac{n^{2/q}K_n^2}{n}\log(n p K_n),
$$
which gives \eqref{eq: prim ver 2}.

To verify \eqref{eq: prim ver 1}, note that
\begin{align*}
&\sup_{d_{\mathcal{U}}(u,u')\leq 1/n} \max_{j\in [p]}\Ep_P[ X_j^2(\zeta_u-\zeta_{u'})^2] \\
&\qquad \displaystyle \leq \sup_{d_{\mathcal{U}}(u,u')\leq 1/n} \max_{j\in [p]}\Ep_P[ X_j^2(Y_u-Y_{u'})^2] \\
& \qquad \leq \sup_{d_{\mathcal{U}}(u,u')\leq 1/n} \max_{j\leq p}\{\Ep_P[ X_j^4]\}^{1/2}\{\Ep_P[(Y_u-Y_{u'})^4]\}^{1/2} \\
 & \qquad \lesssim  \sup_{d_{\mathcal{U}}(u,u')\leq 1/n}|u-u'|^{\nu/2} \lesssim n^{-\nu/2}.
\end{align*}
Therefore, Lemma \ref{lemma:CCK} implies that we have with probability at least $1-(\log n)^{-1}$,
\begin{align*}
\sup_{d_{\mathcal{U}}(u,u')\leq 1/n}\|\En[S_{u}-S_{u'}]\|_\infty
& =\frac{1}{\sqrt{n}}  \sup_{d_{\mathcal{U}}(u,u')\leq 1/n} \max_{j\in p}|\Gn( X_j(\zeta_u-\zeta_{u'}))|\\
 & \lesssim \sqrt{\frac{\log(n p K_n)}{n^{1 + \nu/2}}} + \frac{n^{1/q}K_n\log(n p K_n)}{n},
\end{align*}
which gives \eqref{eq: prim ver 1}.

Finally, to verify \eqref{eq: prim ver 3} note that uniformly over $u,u'\in\UU$ and $j\in[p]$, we have
\begin{align*}
|\Ep_P[ S_{uj}^2 - S_{u'j}^2]| & = |\Ep_P[ (S_{uj} - S_{u'j})(S_{uj}+S_{u'j})]| \\
& \lesssim  \Big(\Ep_P[ (S_{uj} - S_{u'j})^2] \Big)^{1/2}\cdot\Big(\Ep[S_{u j}^2]+\Ep_P[S_{u' j}^2]\Big)^{1/2}\\
& \lesssim  \Big(\Ep_P[ X_j^2(Y_u - Y_{u'})^2]\Big)^{1/2} \\
& \lesssim \Big(\Ep_P[ X_j^4]\Big)^{1/4}\cdot\Big(\Ep_P[(Y_u - Y_{u'})^4]\Big)^{1/4} \lesssim d_\UU(u,u')^{\nu/4}.
\end{align*}
This completes the proof. \qed
%\end{proof}

\section{Bounds on Covering entropy}

Let $ ( W_i)_{i=1}^n$ be a sequence of independent copies of a random element $ W$  taking values in a measurable space $({\mathcal{W}}, \mathcal{A}_{{\mathcal{W}}})$ according to a probability law $P$. Let $\mathcal{F}$ be a set  of suitably measurable functions $f\colon {\mathcal{W}} \to \mathbb{R}$, equipped with a measurable envelope $F\colon \mathcal{W} \to \mathbb{R}$.

\begin{lemma}[Algebra for Covering Entropies] \text{  } \label{lemma: andrews}\\
(1) Let $\F$ be a VC subgraph class with a finite VC index $k$ or any
other class whose entropy is bounded above by that of such a VC subgraph class, then
the uniform entropy numbers of $\mF$ obey
\begin{equation*}
 \sup_{Q} \log  N(\epsilon \|F\|_{Q,2}, \F,  \| \cdot \|_{Q,2}) \lesssim 1+ k \log (1/\epsilon)\vee 0
\newline
\end{equation*}
(2) For any measurable classes of functions $\F$ and $\F^{\prime
} $ mapping $\mathcal{W}$ to $\Bbb{R}$,
\begin{align*}
&\log N(\epsilon \Vert F+F^{\prime }\Vert _{Q,2},\F+\F^{\prime
}, \| \cdot \|_{Q,2})\\
&\qquad \leq \log   N\left(\mbox{$ \frac{\epsilon }{2}$}\Vert F\Vert _{Q,2},\F, \| \cdot \|_{Q,2}\right)
+ \log N\left( \mbox{$ \frac{\epsilon }{2}$}\Vert F^{\prime }\Vert _{Q,2},\F^{\prime
}, \| \cdot \|_{Q,2}\right), \\
&\log  N(\epsilon \Vert F\cdot F^{\prime }\Vert _{Q,2},\F\cdot \F^{\prime
}, \| \cdot \|_{Q,2})\\
&\qquad \leq \log   N\left( \mbox{$ \frac{\epsilon }{2}$}\Vert F\Vert _{Q,2},\F, \| \cdot \|_{Q,2}\right)
+ \log N\left( \mbox{$ \frac{\epsilon }{2}$}\Vert F^{\prime }\Vert _{Q,2},\F^{\prime
}, \| \cdot \|_{Q,2}\right), \\
& N(\epsilon \Vert F\vee F^{\prime }\Vert _{Q,2},\F\cup \F^{\prime
}, \| \cdot \|_{Q,2})\\
&\qquad \leq   N\left(\epsilon\Vert F\Vert _{Q,2},\F, \| \cdot \|_{Q,2}\right)
+ N\left( \epsilon\Vert F^{\prime }\Vert _{Q,2},\F^{\prime
}, \| \cdot \|_{Q,2}\right).
\end{align*}
(3)  For any measurable class of functions $%
\mathcal{F}$ and a fixed function $f$ mapping $\mathcal{W}$ to $\Bbb{R}$,
\begin{equation*}
 \log \sup_{Q} N(\epsilon \Vert |f|\cdot F\Vert _{Q,2},f\cdot\F, \| \cdot \|_{Q,2})\leq \log \sup_{Q} N\left(
\epsilon /2\Vert F\Vert _{Q,2},\F, \| \cdot \|_{Q,2}\right)
\end{equation*}
(4)  Given measurable classes $\F_j$ and envelopes $F_j$, $j=1,\ldots,k$, mapping $\mathcal{W}$ to $\Bbb{R}$, a function $\phi\colon\Bbb{R}^k\to\Bbb{R}$ such that for $f_j,g_j\in\F_j$,
$ |\phi(f_1,\ldots,f_k) - \phi(g_1,\ldots,g_k) | \leq \sum_{j=1}^k L_j(x)|f_j(x)-g_j(x)|$, $L_j(x)\geq 0$, and fixed functions $\bar f_j \in \F_j$,  the class of functions $\mathcal{L}=\{\phi(f_1,\ldots,f_k)-\phi(\bar f_1,\ldots,\bar f_k)\colon f_j \in\mathcal{F}_j, j=1,\ldots,k\}$ satisfies
\begin{align*}
&\log \sup_Q N\left(\epsilon\Big\|\sum_{j=1}^kL_jF_j\Big\|_{Q,2},\mathcal{L}, \| \cdot \|_{Q,2}\right)\\
&\qquad \leq \sum_{j=1}^k\log  \sup_Q  N\left(
\mbox{$\frac{\epsilon}{k}$}\|F_j\|_{Q,2},\F_j, \| \cdot \|_{Q,2}\right).
\end{align*}
\end{lemma}

%\subsection*{Proof of Lemma \ref{lemma: andrews}}
\begin{proof}
See Lemma L.1 in \cite{BCFH2013program}.
\end{proof}

\begin{lemma}[Covering Entropy for Classes obtained as Conditional Expectations]\label{Lemma:PartialOutCovering}
 Let $\mathcal F$ denote a class of measurable functions $f\colon \mathcal{W}\times \mathcal{Y} \to \Bbb{R}$ with a measurable envelope $F$. For a given $f \in \mathcal{F}$, let $\bar f\colon \mathcal{W} \to \Bbb{R}$ be the function $\bar f (w) := \int f(w,y) d\mu_{w}(y)$ where $\mu_{w}$ is a regular conditional probability distribution over $y \in \mathcal{Y}$ conditional on $w\in\mathcal{W}$. Set $\bar{\mathcal{F}} = \{ \bar f \colon f \in \mathcal{F}\}$ and let $\bar F(w):=\int F(w,y) d\mu_w(y)$ be an envelope for $\bar{\mathcal{F}}$. Then, for $r, s \geq 1$,
     $$
    \log \sup_{Q} N(\epsilon \| \bar F\Vert _{Q,r}, \bar{\mathcal{F}}, \| \cdot \|_{Q,r}) \leq \log \sup_{\widetilde Q} N((\epsilon/4)^r \| F\Vert _{\widetilde Q,s},  \mathcal \F , \| \cdot \|_{\widetilde Q,s}),
    $$ where $Q$ belongs to the set of finitely-discrete probability measures over $\mathcal{W}$ such that  $0<\| \bar F\Vert _{Q,r}< \infty$, and $\widetilde Q$ belongs to the set of finitely-discrete probability measures over $\mathcal{W}\times \mathcal{Y}$ such that $0<\|  F\Vert _{\widetilde Q,s}< \infty$. In particular, for every $\epsilon > 0$ and any $k\geq 1$,
        $$
    \log \sup_{Q} N(\epsilon, \bar{\mathcal{F}}, \| \cdot \|_{Q,k}) \leq \log \sup_{\widetilde Q} N(\epsilon/2,  \mathcal \F , \| \cdot \|_{\widetilde Q,k} ).
    $$
\end{lemma}
%\subsection*{Proof of Lemma \ref{Lemma:PartialOutCovering}}
\begin{proof}
See Lemma L.2 in \cite{BCFH2013program}.
\end{proof}

\begin{lemma}\label{lem: linear classes}
Consider a mapping $\tilde u\mapsto \xi_{\tilde u}$ from $\widetilde{\mathcal{U}}=[0,1]^k$ into $\mathbb{R}^p$ and the class of functions $\mathcal{F} = \{x\mapsto \mathcal{M}(x'\xi_{\tilde u})\colon \tilde u\in\widetilde{\mathcal{U}}\}$ mapping $\mathbb{R}^p$ into $\mathbb R$ where $\mathcal{M}\colon\mathbb R\to \mathbb R$ is $L$-Lipschitz. Assume that $\|\xi_{\tilde u_2} - \xi_{\tilde u_1}\|_1\leq C\|\tilde u_2 - \tilde u_1\|$ for all $\tilde u_1,\tilde u_2\in\widetilde{\mathcal{U}}$ for some constant $C>0$. Then, for any $M>0$ the uniform entropy numbers of $\mathcal{F}$ satisfy
$$
\sup_Q\log N(\epsilon\|F\|_{Q,2},\mathcal{F},\|\cdot\|_{Q,2})\leq k\log(3LCMk/\varepsilon),\quad \text{for all }0<\epsilon\leq 1,
$$
where $F(x) = \sup_{\tilde u\in\widetilde{\UU}}|\mathcal{M}(x'\xi_{\tilde u})| + M^{-1}\|x\|_\infty$, $x\in\mathbb R^p$, is its envelope.
\end{lemma}

%\subsection*{Proof of Lemma \ref{lem: linear classes}}
\begin{proof}
Consider any $f_1,f_2\in\mathcal{F}$. There exist $\tilde u_1,\tilde u_2\in\widetilde{\mathcal{U}}$ such that $f_1(x) = \mathcal{M}(x'\xi_{\tilde u_1})$ and $f_2(x) = \mathcal{M}(x'\xi_{\tilde u_2})$ for all $x\in\mathbb R^p$. Therefore, since $\mathcal{M}$ is $L$-Lipschitz, we have
\begin{align*}
|\mathcal{M}(x'\xi_{\tilde u_2}) -\mathcal{M}(x'\xi_{\tilde u_1})| & \leq L\|x\|_\infty \|\xi_{u_2} - \xi_{u_1}\|_1 \leq L \|x\|_\infty C\|\tilde u_2 - \tilde u_1\|  \\
& \leq LC M \|\tilde u_2 - \tilde u_1\| \{ M^{-1}\|x\|_\infty+\sup\nolimits_{\tilde u\in\widetilde{\mathcal{U}}}|\mathcal{M}(x'\xi_{\tilde u})| \}  \\
& \leq  L CM\|\tilde u_2 - \tilde u_1\|F(x)
\end{align*}
by definition of the envelope $F(x) = M^{-1}\|x\|_\infty+\sup_{\tilde u\in\widetilde{\mathcal{U}}}|\mathcal{M}(x'\xi_{\tilde u})|$. Thus, for any finitely discrete probability measure $Q$ on $\mathbb R^p$,
$$
\|f_2 - f_1\|_{Q,2} \leq LCM\|\tilde u_2 - \tilde u_1\|\cdot\|F\|_{Q,2}.
$$
Recall that  since $B_\infty \subset B_2\sqrt{k}$ we have $N(B_\infty,\|\cdot\|,\epsilon) \leq N(B_2\sqrt{k},\|\cdot\|,\epsilon) \leq (1 + 2\sqrt{k}/\epsilon)^k$ where the last inequality follows from standard volume arguments. Furthermore, for any $\epsilon \leq \sqrt{k}$ we have $1 + 2\sqrt{k}/\epsilon \leq 3 k /\epsilon$. Therefore $\log N(\epsilon\|F\|_{Q,2},\mathcal{F},\|\cdot\|_{Q,2}) \leq k\log(3LCMk/\epsilon)$.
\end{proof}

\begin{lemma}\label{lem: bounded lipschitz classes}
Let $\mF$ be a class of functions with an envelope $F$. Also, let $\mathcal{M}\colon \mathbb R \to \mathbb R$ be an $L$-Lipschitz function bounded in absolute value by a constant $M$. Assume that for some positive constants $C_1$ and $C_2$, the uniform entropy numbers of $\mF$ obey
$$
\sup_Q\log N(\epsilon\|F\|_{Q,2},\mF,\|\cdot\|_{Q,2})\leq C_1\log(C_2/\epsilon),\quad \text{for all }0<\epsilon\leq 1.
$$
Then for any constant $K>0$, the uniform entropy numbers of the class of functions $\mathcal{M}(\mF) = \{\mathcal{M}(f)\colon f\in\mF\}$ obey
$$
\sup_Q \log N(\epsilon\|F_{\mathcal{M}(\mF)}\|_{Q,2},\mathcal{M}(\mF),\|\cdot\|_{Q,2})\leq C_1 \log(C_2 K/\epsilon),\quad \text{for all }0<\epsilon\leq 1,
$$
where $F_{\mathcal{M}(\mF)} = M+LF/K$ is its envelope.
\end{lemma}

%\subsection*{Proof of Lemma \ref{lem: bounded lipschitz classes}}
\begin{proof}
The result follows from the observation that for any $f',f''\in\mF$ and any finitely-discrete probability measure $Q$,
$$
\|\mathcal{M}(f') - \mathcal{M}(f'')\|_{Q,2}\leq L\|f_1 - f_2\|_{Q,2},
$$
so that if $\mF$ can be covered by $k$ balls of radius $\epsilon\|F\|_{Q,2}$ (in the $\|\cdot\|_{Q,2}$ norm), then $\mathcal{M}(\mF)$ can be covered by $k$ balls of radius $\epsilon L\|F\|_{Q,2}$ (in the same norm).
\end{proof}

\section{Some Probabilistic Inequalities}\label{sec: probab inequalities}

\begin{lemma}[Moderate deviations for self-normalized sums, \cite{jing:etal}]\label{Lemma: MDSN} Let $Z_{1}$,$\ldots$, $Z_{n}$ be independent, zero-mean random variables and $\mu \in (0,1]$. Let
$S_{n,n} = \sum_{i=1}^n Z_i,  \ \ V^2_{n,n} = \sum_{i=1}^nZ^2_i,$ $$M_n= \left \{\frac{1}{n} \sum_{i=1}^n \Ep [ Z_i^2 ] \right \}^{1/2} \Big / \left \{\frac{1}{n} \sum_{i=1}^n \Ep[|Z_i|^{2+\mu}] \right\}^{1/\{2+\mu\}}>0$$
and $0< \ell_n \leq n^{\frac{\mu}{2(2+\mu)}}M_n$. Then for some absolute constant $A$,
$$
\left |\frac{\Pr(|S_{n,n}/V_{n,n}|  \geq x) }{ 2 (1-\Phi(x))} - 1 \right |  \leq
\frac{A}{ \ell_n^{2+\mu}},  \ \ 0 \leq  x \leq  n^{\frac{\mu}{2(2+\mu)}}\frac{M_n}{\ell_n}-1.
$$
\end{lemma}

Let $ ( W_i)_{i=1}^n$ be a sequence of independent copies of a random element $ W$  taking values in a measurable space $({\mathcal{W}}, \mathcal{A}_{{\mathcal{W}}})$ according to a probability law $P$. Let $\mathcal{F}$ be a set  of suitably measurable functions $f\colon {\mathcal{W}} \to \mathbb{R}$, equipped with a measurable envelope $F\colon \mathcal{W} \to \mathbb{R}$.
%Let $M = \max_{1 \leq i \leq n} F(X_{i})$.

  \begin{lemma}[Maximal Inequality I, \cite{chernozhukov2012gaussian}]
\label{lemma:CCK}  Work with the setup above.  Suppose that $F\geq \sup_{f \in \mathcal{F}}|f|$ is a measurable envelope for $\mF$
with $\| F\|_{P,q} < \infty$ for some $q \geq 2$.  Let $M = \max_{i\leq n} F(W_i)$ and $\sigma^{2} > 0$ be any positive constant such that $\sup_{f \in \mF}  \| f \|_{P,2}^{2} \leq \sigma^{2} \leq \| F \|_{P,2}^{2}$. Suppose that there exist constants $a \geq e$ and $v \geq 1$ such that
\begin{equation*}
\log \sup_{Q} N(\epsilon \| F \|_{Q,2}, \mF,  \| \cdot \|_{Q,2}) \leq  v \log (a/\epsilon), \ 0 <  \epsilon \leq 1.
\end{equation*}
Then
\begin{equation*}
\Ep_P [ \| \bG_{n} \|_{\mF} ] \leq K  \left( \sqrt{v\sigma^{2} \log \left ( \frac{a \| F \|_{P,2}}{\sigma} \right ) } + \frac{v\| M \|_{P, 2}}{\sqrt{n}} \log \left ( \frac{a \| F \|_{P,2}}{\sigma} \right ) \right),
\end{equation*}
where $K$ is an absolute constant.  Moreover, for every $t \geq 1$, with probability $> 1-t^{-q/2}$,
\begin{multline*}
\| \bG_{n} \|_{\mF} \leq (1+\alpha) \Ep_P [ \| \bG_{n} \|_{\mF} ] \\
+ K(q) \Big [ (\sigma + n^{-1/2} \| M \|_{P,q}) \sqrt{t}
+  \alpha^{-1}  n^{-1/2} \| M \|_{P,2}t \Big ], \ \forall \alpha > 0,
\end{multline*}
where $K(q) > 0$ is a constant depending only on $q$.  In particular, setting $a \geq n$ and $t = \log n$,
with probability $> 1- c(\log n)^{-1}$,
\begin{equation} \label{simple bound}
\| \bG_{n} \|_{\mF} \leq K(q,c) \left ( \sigma \sqrt{v \log \left ( \frac{a \| F \|_{P,2}}{\sigma} \right ) } + \frac{v
 \| M \|_{P,q} } {\sqrt{n}}\log \left ( \frac{a \| F \|_{P,2}}{\sigma} \right ) \right),
\end{equation}
where $  \| M \|_{P,q}  \leq n^{1/q} \| F\|_{P,q}$ and  $K(q,c) > 0$ is a constant depending only on $q$ and $c$.

\end{lemma}

\begin{lemma}[Maximal Inequality II, \cite{chernozhukov2012gaussian}]\label{lem: maximal inequality 2}
Work with the setup above. Suppose that the conditions of Lemma \ref{lemma:CCK} are satisfied. Then
\begin{align*}
&\Ep_P\Big[\|\En[f^2(W)]\|_\mF\Big] - \sup_{f\in\mF}\Ep_P[f^2(W)]\\
&\qquad \leq \frac{K \|M\|_{P,2}}{\sqrt n}\Big(\sigma\sqrt{v\log\left(\frac{a\|F\|_{P,2}}{\sigma}\right)} + \frac{v\|M\|_{P,2}\log\left(\frac{a\|F\|_{P,2}}{\sigma}\right)}{\sqrt{n}} \Big),
\end{align*}
where $K$ is an absolute constant.
\end{lemma}

%\subsection*{Proof of Lemma \ref{lem: maximal inequality 2}}
\begin{proof}
The proof of the asserted claim coincides one-by-one with that given for the corresponding inequality in Lemma 2.2 of \cite{chernozhukov2012gaussian}, with the constant 3 replaced everywhere by the constant 2. At the end of the proof, the entropy integral
$$
J(\delta) = \int_0^\delta \sup_Q\sqrt{1+\log N(\epsilon \|F\|_{Q,2},\mF,\|\cdot\|_{Q,2})}d\epsilon
$$
is bounded by $\delta(v\log(a/\delta))^{1/2}$ under our condition on the uniform entropy numbers of $\mF$.
\end{proof}

\section{A Bound on Sparse Eigenvalues for Many Random Matrices}\label{sec: many random matrices}
The following lemma is a generalization of the main result in \cite{RudelsonVershynin2008} to many matrices.
\begin{lemma}\label{thm:RV34}
Let $\UU$ denote a finite set and $(X_{ui})_{u\in \UU}$, $i=1,\ldots, n$, be independent (across i) random vectors such that $X_{ui} \in \RR^p$ with $p\geq 2$ and $(\Ep[ \max_{1\leq i\leq n}\max_{u\in \UU}\|X_{ui}\|_\infty^2])^{1/2} \leq K$. Furthermore, for $k\geq 1$, define
$$
\delta_n:= \frac{K \sqrt{k}}{\sqrt n}\left(\log^{1/2} |\UU| + \log^{1/2} p + (\log k) (\log^{1/2}p) (\log^{1/2} n) \right),
$$
%where $A$ is a universal constant.
Then
\begin{align*}
&\Ep\left[ \sup_{\|\theta\|_0\leq k, \|\theta\| =1}\max_{u\in\UU} \left| \En\[ (\theta'X_{u})^2 - \Ep[(\theta'X_{u})^2] \]\right|\right]\\
&\qquad  \lesssim \delta_n^2 + \delta_n \sup_{\|\theta\|_0\leq k, \|\theta\| =1, u\in\UU} \sqrt{\En\Ep[(\theta'X_{u})^2]}
\end{align*}
up to a universal constant.
\end{lemma}
\begin{proof}
For $T\subset\{1,\dots,p\}$, let $B_T = \{\theta\in\mathbb R^p\colon \|\theta\| = 1, \supp(\theta) \subseteq T\}$.
Also, for $\mathcal{T} = \cup_{|T| = k} B_T \times \UU$, let $R:= \sup_{(\theta,u)\in \mathcal{T}}(\sum_{i=1}^n(\theta'X_{ui})^2)^{1/2}$ and $M:=\max_{1\leq i\leq n, u\in \UU}\|X_{ui}\|_\infty$. By symmetrization inequality, Lemma 6.3 in \cite{LedouxTalagrandBook}, we have
\begin{align*}
&n\Ep\left[ \sup_{\|\theta\|_0\leq k, \|\theta\| =1}\max_{u\in\UU} \left| \En\[ (\theta'X_{u})^2 - \Ep[(\theta'X_{u})^2] \]\right|\right] \\
&\qquad \leq 2 \Ep \left[ \Ep\left[ \sup_{(\theta,u)\in \mathcal{T}} \left| \sum_{i=1}^n \varepsilon_i(\theta'X_{ui})^2\right| \mid X\right]\right],
\end{align*}
where $X = (X_{u i})_{u\in\UU,1\leq i\leq n}$ and $(\varepsilon_i)_{i=1}^n$ is a sequence of independent Rademacher random variables that are independent of $X$. A consequence of Lemma 4.5 in \cite{LedouxTalagrandBook} (see equation (4.8)) gives
 $$
\Ep\left[ \sup_{(\theta,u)\in \mathcal{T}} \left| \sum_{i=1}^n \varepsilon_i(\theta'X_{ui})^2\right| \mid X\right] \leq (\pi/2)^{1/2} \Ep\left[ \sup_{(\theta,u)\in \mathcal{T}} \left| \sum_{i=1}^n g_i(\theta'X_{ui})^2\right| \mid X\right]
$$
where $(g_i)_{i=1}^n$ is a sequence of independent standard normal random variables that are independent of $X$. In turn, an application of Dudley's integral gives
$$
I_1:= \Ep\left[ \sup_{(\theta,u)\in \mathcal{T}} \left| \sum_{i=1}^n g_i(\theta'X_{ui})^2\right| \mid X\right] \leq 8 \int_0^{\rm diam(\mathcal{T})}\log^{1/2}N(\mathcal{T},d,\epsilon)d\epsilon
$$
where
$$
{\rm diam}(\mathcal{T})\leq 2 \sup_{(\theta,u)\in \mathcal{T}}(\sum_{i=1}^n (\theta' X_{ui})^4)^{1/2} \leq 2\sqrt{k} M R
$$
using that $|\theta' X_{ui}|\leq \|X_{ui}\|_\infty\|\theta\|_1\leq M\sqrt{k}$, and $d$ is the corresponding Gaussian semi-metric. Furthermore, we have
$$
\log N(\mathcal T,d,\epsilon) \leq \log |\UU| + \max_{u\in \UU} \log N(\cup_{|T| = k}B_T \times \{u\},d,\epsilon),
$$
so that
$$
I_1 \leq 16\sqrt{k} M R\log^{1/2}|\UU| +8 \int_0^{\rm diam(\mathcal{T})}\max_{u\in\UU}\log^{1/2}N(\cup_{|T|= k}B_T \times \{u\},d,\epsilon)d\epsilon.
$$
Now, for any $(\theta,u)$ and $(\bar \theta,u)$ in $D_u^k:=\cup_{|T| = k}B_T \times \{u\}$, we have
\begin{align*}
d((\theta,u),(\bar\theta,u))& = \Big( \sum_{i=1}^n\{(\theta' X_{ui})^2-(\bar\theta' X_{ui})^2\}^2\Big)^{1/2}\\
& \leq \Big( \sum_{i=1}^n\{(\theta' X_{ui})+(\bar\theta' X_{ui})\}^2\Big)^{1/2}\max_{1\leq i\leq n}|(\theta-\bar\theta)'X_{u i}|\\
& \leq 2 \sup_{(\theta,u)\in \mathcal{T}}\Big(\sum_{i=1}^n(\theta' X_{ui})^2\Big)^{1/2} \max_{1\leq i\leq n}|(\theta-\bar\theta)'X_{u i}|= 2R \|\theta-\bar\theta\|_{X_u}
\end{align*}
where we let $\|\delta\|_{X_u} := \max_{1\leq i\leq n}|\delta' X_{ui}|$. This implies that
$$
N(D_u^k,d,\epsilon) \leq N\left(D_u^k/\sqrt{k},\|\cdot\|_{X_u}, \epsilon/\{2\sqrt{k}R\}\right).
$$
Therefore, since ${\rm diam(\mathcal{T})}\leq 2\sqrt{k} M R$, we have
\begin{align*}
& \int_0^{\rm diam(\mathcal{T})}\max_{u\in\UU}\log^{1/2}N(D_u^k,d,\epsilon)d\epsilon \\
&\qquad \leq \int_0^{2\sqrt{k}M R}\max_{u\in\UU}\log^{1/2}N(D_u^k/\sqrt{k},\|\cdot\|_{X_u},\epsilon/\{2\sqrt{k}R\})d\epsilon\\
& \qquad =2\sqrt{k}R\int_0^{M}\max_{u\in\UU}\log^{1/2} N(D_u^k/\sqrt{k},\|\cdot\|_{X_u},\epsilon)d\epsilon.
\end{align*}
Note that $B_T/\sqrt{k} \subset B^1_T$ and $D_u^k/\sqrt{k}\subset B^1\times \{u\} $ where $B^1:= \{ \theta \in \RR^p\colon \|\theta\|_1 \leq 1\}$ and $B^1_T=\{ \theta \in B^1\colon \supp(\theta)\subseteq T\}$. It follows from Lemma 3.9 in \cite{RudelsonVershynin2008} that $N(B^1,\|\cdot\|_{X_u},\epsilon) \leq (2p)^{A\epsilon^{-2}M^2\log n}$ for all $\epsilon>0$ and some universal constant $A$. Moreover, as in the discussion after Lemma 3.9 in \cite{RudelsonVershynin2008}, we have $N(B^1_T,\|\cdot\|_{X_u},\epsilon) \leq (1+2M/\epsilon)^k$ for all $\epsilon>0$ and all $T\subset\{1,\dots,p\}$ with $|T| = k$, so that $N(D_u^k/\sqrt{k},\|\cdot\|_{X_u},\epsilon) \leq \binom{p}{k}(1+2M/\epsilon)^k$ for all $\epsilon>0$. Therefore,
\begin{align*}
& \int_0^{M}\max_{u\in\UU}\log^{1/2} N(D_u^k/\sqrt{k},\|\cdot\|_{X_u},\epsilon)d\epsilon  \\
 &\quad\leq \int_0^{M/\sqrt{k}} \log^{1/2}\left( \binom{p}{k}(1+2M/\epsilon)^k \right) d\epsilon  \\
 &\qquad+ \int_{M/\sqrt{k}}^{M}  \log^{1/2}\left( (2p)^{A\epsilon^{-2}M^2\log n} \right)  d\epsilon\\
& \quad\leq \frac{M}{\sqrt{k}}\log^{1/2}\binom{p}{k} + \sqrt{k}\int_0^{M/\sqrt{k}}\log^{1/2}(1+2M/\epsilon)d\epsilon  \\
&\qquad+ A^{1/2}M(\log^{1/2}n) (\log^{1/2}(2p)) \int_{M/\sqrt{k}}^{M} \frac{d\epsilon}{\epsilon}\\
& \quad\leq M\log^{1/2}p + M\Big(1/2 + \log^{1/2}(1 + 2\sqrt{k})\Big) \\
&\qquad + A^{1/2}M(\log^{1/2}n) (\log^{1/2}(2p)) \Big(\log M-\log(M/\sqrt{k})\Big)\\
& \quad\lesssim  M \Big(\log^{1/2}p + (\log k)(\log^{1/2}n)(\log^{1/2}p)\Big)
\end{align*}
up to a universal constant where in the third inequality, we used the fact that integrating by parts gives
$$
\sqrt k\int_0^{M/\sqrt k}\log^{1/2}(1 + 2M/\epsilon)d\epsilon \leq M(1/2 + \log^{1/2}(1+ 2\sqrt k)).
$$
Collecting the terms, we obtain
$$
I_1\lesssim \sqrt k M R\Big(\log^{1/2}|\UU|+\log^{1/2}p+(\log k)(\log^{1/2} n)(\log^{1/2} p)\Big).
$$

Therefore, since $K \geq (\Ep[M^2])^{1/2}$, setting
$$
\delta_n = \frac{K\sqrt k}{\sqrt n}\Big(\log^{1/2}|\UU|+\log^{1/2}p+(\log k)(\log^{1/2} n)(\log^{1/2} p)\Big)
$$
gives
\begin{align*}
 I_2 & = \Ep\left[ \sup_{\|\theta\|_0\leq k, \|\theta\| =1}\max_{u\in\UU} \left| \En\[ (\theta'X_{u})^2 - \Ep[(\theta'X_{u})^2] \]\right|\right]\\
  & \lesssim \frac{\delta_n\Ep[M R]}{K \sqrt n} \leq  (\delta_n/K) (\Ep[M^2])^{1/2} (\Ep[R^2/n])^{1/2}\leq \delta_n (E[R^2/n])^{1/2}\\
& \lesssim  \delta_n \Big(I_2 + \sup_{\|\theta\|_0\leq k, \|\theta\| =1, u\in \UU} \En\Ep[(\theta'X_{u})^2]\Big)^{1/2}.
\end{align*}
Thus, because $a \leq \delta_n(a+b)^{1/2}$ implies $a \leq \delta_n^2 + \delta_n b^{1/2}$, we have
$$
I_2 \lesssim \delta_n^2 + \delta_n \sup_{\|\theta\|_0\leq k, \|\theta\| =1, u\in \UU} \sqrt{\En\Ep[(\theta'X_{u})^2]}
$$
up to an absolute constant. This completes the proof.
\end{proof}

\bibliographystyle{plain}
%\bibliography{mybibVOLUME}

\begin{thebibliography}{10}

\bibitem{andrews94}
Andrews, D.W.K. (1994).
\newblock Asymptotics for semiparametric econometric models via stochastic
  equicontinuity.
\newblock {\em Econometrica}, 62(1):43--72.

\bibitem{BellChenChernHans:nonGauss}
Belloni, A., Chen, D, Chernozhukov, V., and Hansen, C. (2012).
\newblock Sparse models and methods for optimal instruments with an application
  to eminent domain.
\newblock {\em Econometrica}, 80:2369--2429.
\newblock arXiv:1010.4345, 2010.

\bibitem{BC-SparseQR}
Belloni, A. and Chernozhukov, V. (2011).
\newblock $\ell_1$-penalized quantile regression for high dimensional sparse
  models.
\newblock {\em Annals of Statistics}, 39(1):82--130.
\newblock arXiv:0904.2931, 2009.

\bibitem{BC-PostLASSO}
Belloni, A. and Chernozhukov, V. (2013).
\newblock Least squares after model selection in high-dimensional sparse
  models.
\newblock {\em Bernoulli}, 19(2):521--547.
\newblock arXiv:1001.0188, 2009.

\bibitem{BellChernHans:Gauss}
Belloni, A., Chernozhukov, V., and Hansen, C. (2010).
\newblock Lasso methods for gaussian instrumental variables models.
%\newblock {\em ArXiv:1012.1297}.
\newblock arXiv:1012.1297, 2010.

\bibitem{BCH2011:InferenceGauss}
Belloni, A., Chernozhukov, V., and Hansen, C. (2013).
\newblock Inference for high-dimensional sparse econometric models.
\newblock {\em Advances in Economics and Econometrics. 10th World Congress of
  Econometric Society. August 2010}, III:245--295.
  \newblock arXiv:1201.0220, 2011.

\bibitem{BelloniChernozhukovHansen2011}
Belloni, A., Chernozhukov, V., and Hansen, C. (2014).
\newblock Inference on treatment effects after selection amongst
  high-dimensional controls.
\newblock {\em Review of Economic Studies}, 81:608--650.
\newblock arXiv:1201.0224, 2011.

\bibitem{BCK-SparseQRinference}
Belloni, A., Chernozhukov, V., and Kato, K. (2013).
\newblock Valid Post-Selection Inference in High-Dimensional Approximately Sparse Quantile Regression Models.
\newblock arXiv:1312.7186, 2013. %{\em ArXiv:1312.7186}.

\bibitem{BCK-LAD}
Belloni, A., Chernozhukov, V., and Kato, K. (2015).
\newblock Uniform post selection inference for LAD regression models and other Z-estimators.
\newblock {\em Biometrika}, (102):77--94.
\newblock arXiv:1304.0282, 2013.

\bibitem{BCW-SqLASSO}
Belloni, A., Chernozhukov, V., and Wang, L. (2011).
\newblock Square-root-lasso: Pivotal recovery of sparse signals via conic
  programming.
\newblock {\em Biometrika}, 98(4):791--806.
\newblock  arXiv:1009.5689, 2010.

\bibitem{BCFH2013program}
Belloni, A., Chernozhukov, V., Fern{\'a}ndez-Val, I., and Hansen, C. (2013).
\newblock Program evaluation with high-dimensional data.
\newblock arXiv:1311.2645, 2013. %{\em ArXiv:1311.2645}.

\bibitem{BCW-rootLasso}
Belloni, A., Chernozhukov, V., and Wang, L. (2014).
\newblock Pivotal estimation via square-root lasso in nonparametric regression.
\newblock {\em The Annals of Statistics}, 42(2):757--788.
\newblock arXiv:1105.1475, 2013.


\bibitem{C92}
Chamberlain, G. (1992).
\newblock Efficiency bounds for semiparametric regression.
\newblock {\em Econometrica}, 60:567--596.

\bibitem{CCDDHNR16}
Chernozhukov, V., Chetverikov, D., Demirer, M., Duflo, E., Hansen, C., Newey, W., and Robins, J. (2016).
\newblock Double/debiased machine learning for treatment and structural parameters.
\newblock {\em Econometrics Journal, forthcoming}.
\newblock arXiv:1608.00060, 2016.

\bibitem{chernozhukov2013gaussian}
Chernozhukov, V., Chetverikov, D., and Kato, K. (2013).
\newblock Gaussian approximations and multiplier bootstrap for maxima of sums
  of high-dimensional random vectors.
\newblock {\em The Annals of Statistics}, 41(6):2786--2819.
\newblock arXiv:1212.6906, 2012.

\bibitem{chernozhukov2014honestbands}
Chernozhukov, V., Chetverikov, D., and Kato, K. (2014).
\newblock Anti-concentration and honest, adaptive confidence bands.
\newblock {\em The Annals of Statistics}, 42(5):1787--1818.
\newblock arXiv:1303.7152, 2013.

\bibitem{chernozhukov2014clt}
Chernozhukov, V., Chetverikov, D., and Kato, K. (2014).
\newblock Central limit theorems and bootstrap in high dimensions.
\newblock arxiv:1412.3661, 2014. %{\em ArXiv:1412.3661}.

\bibitem{chernozhukov2012gaussian}
Chernozhukov, V., Chetverikov, D., and Kato, K. (2014).
\newblock Gaussian approximation of suprema of empirical processes.
\newblock {\em The Annals of Statistics}, 42(4):1564--1597.
\newblock arXiv:1212.6885, 2012.


\bibitem{chernozhukov2012comparison}
Chernozhukov, V., Chetverikov, D., and Kato, K. (2015).
\newblock Comparison and anti-concentration bounds for maxima of gaussian
  random vectors.
\newblock {\em Probability Theory and Related Fields}, 162:47--70.
\newblock arXiv:1301.4807, 2013.

\bibitem{chernozhukov2015noncenteredprocesses}
Chernozhukov, V., Chetverikov, D., and Kato, K. (2015).
\newblock Empirical and multiplier bootstraps for suprema of empirical
  processes of increasing complexity, and related gaussian couplings.
\newblock arXiv:1502.00352, 2015. %{\em ArXiv:1502.00352}.

\bibitem{CFM13}
Chernozhukov, V., Fern\'{a}ndez-Val, I., and Melly, B. (2013).
\newblock Inference on counterfactual distributions.
\newblock {\em Econometrica}, 81:2205--2268.
\newblock  arXiv:0904.0951, 2009.


\bibitem{CHS15}
Chernozhukov, V., Hansen, C., and Spindler, M. (2015).
\newblock Post-selection and post-regularization inference in linear models with very many controls and instruments.
\newblock {\em Americal Economic Review: Papers and Proceedings}, 105:486--490.
\newblock  arXiv:1501.03185, 2015.

\bibitem{DZ17}
Deng, H. and Zhang, C.-H. (2017).
\newblock Beyond Gaussian approximation: bootstrap for maxima of sums of independent random vectors.
\newblock  arXiv:1705.09528, 2017.

\bibitem{Dudley99}
Dudley, R. (1999).
\newblock {\em Uniform central limit theorems}, volume~63 of {\em Cambridge
  Studies in Advanced Mathematics}.
\newblock Cambridge University Press, Cambridge.

\bibitem{HKB14}
Hothorn, T., Kneib, T., and B\"{u}hlmann, P. (2014).
\newblock Conditional transformation models.
\newblock {\em J. R. Statist. Soc. B}, 76:3--27.
\newblock arXiv:1201.5786, 2012.


\bibitem{javanmard2013confidence}
Javanmard, A. and Montanari, A. (2014).
\newblock Confidence intervals and hypothesis testing for high-dimensional
  regression.
\newblock {\em J. Mach. Learn. Res.}, 15:2869--2909.
\newblock arXiv:1306.3171, 2013.

\bibitem{javanmard2014confidence}
Javanmard, A. and Montanari, A. (2014).
\newblock Hypothesis testing in high-dimensional regression under the gaussian
  random design model: asymptotic theory.
\newblock {\em IEEE Transactions on Information Theory}, 60:6522--6554.
\newblock arXiv:1301.4240, 2013.


\bibitem{kosorok:book}
Kosorok, M. (2008).
\newblock {\em Introduction to Empirical Processes and Semiparametric
  Inference}.
\newblock Series in Statistics. Springer, Berlin.

\bibitem{leeb:potscher:pms}
Leeb, H. and P{\"o}tscher, B. (2008).
\newblock Can one estimate the unconditional distribution of
  post-model-selection estimators?
\newblock {\em Econometric Theory}, 24(2):338--376.
\newblock arXiv:0704.1584, 2007.


\bibitem{leeb:potscher:review}
Leeb, H. and P{\"o}tscher, B. (2008).
\newblock Recent developments in model selection and related areas.
\newblock {\em Econometric Theory}, 24(2):319--322.

\bibitem{leeb:potscher:hodges}
Leeb, H. and P{\"o}tscher, B. (2008).
\newblock Sparse estimators and the oracle property, or the return of {H}odges'
  estimator.
\newblock {\em J. Econometrics}, 142(1):201--211.
\newblock arXiv:0704.1466, 2007.

\bibitem{linton96}
Linton, O. (1996).
\newblock Edgeworth approximation for {MINPIN} estimators in semiparametric
  regression models.
\newblock {\em Econometric Theory}, 12(1):30--60.
\newblock  Cowles Foundation Discussion Papers (1086), 1994.

\bibitem{M93}
Mammen, E. (1993).
\newblock Bootstrap and wild bootstrap for high dimensional linear models.
\newblock {\em The Annals of Statistics}, 21:255--285.

\bibitem{newey90}
Newey, W. (1990).
\newblock Semiparametric efficiency bounds.
\newblock {\em Journal of Applied Econometrics}, 5(2):99--135.

\bibitem{newey94}
Newey, W. (1994).
\newblock The asymptotic variance of semiparametric estimators.
\newblock {\em Econometrica}, 62(6):1349--1382.

\bibitem{N59}
Neyman, J. (1959).
\newblock Optimal asymptotic tests of composite statistical hypotheses.
\newblock {\em Probability and Statistics, the Harold Cramer Volume}.

\bibitem{Neyman1979}
Neyman, J. (1979).
\newblock $c(\alpha)$ tests and their use.
\newblock {\em Sankhya}, 41:1--21.

\bibitem{NingLiuGT2014}
Ning, Y. and Liu, H. (2014).
\newblock A general theory of hypothesis tests and confidence regions for
  sparse high dimensional models.
\newblock arXiv:1412.8765, 2014. %{\em ArXiv:1412.8765}.


\bibitem{potscher:leeb:dpe}
P{\"o}tscher, B. and Leeb, H. (2009).
\newblock On the distribution of penalized maximum likelihood estimators: the
  {LASSO}, {SCAD}, and thresholding.
\newblock {\em J. Multivariate Anal.}, 100(9):2065--2082.
\newblock arXiv:0711.0660, 2009.


\bibitem{robins:dr}
Robins, J. and Rotnitzky, A. (1995).
\newblock Semiparametric efficiency in multivariate regression models with
  missing data.
\newblock {\em J. Amer. Statist. Assoc.}, 90(429):122--129.

\bibitem{S56}
Stein, C. (1956).
\newblock Efficient nonparametric testing and estimation.
\newblock {\em In Proc. 3rd Berkeley Symp. Math. Statist. and Probab.}.

\bibitem{vandeGeerBuhlmannRitov2013}
van~de Geer, S., B{\"u}hlmann, P., Ritov, Y., and Dezeure, R. (2014).
\newblock On asymptotically optimal confidence regions and tests for
  high-dimensional models.
\newblock {\em Annals of Statistics}, 42:1166--1202.
\newblock arXiv:1303.0518, 2013.

\bibitem{vdV}
van~der Vaart, A. (1998).
\newblock {\em Asymptotic Statistics}.
\newblock Cambridge University Press.

\bibitem{vdV-W}
van~der Vaart, A. and Wellner, J. (1996).
\newblock {\em Weak Convergence and Empirical Processes}.
\newblock Springer Series in Statistics.

\bibitem{c.h.zhang:s.zhang}
Zhang, C.-H, and Zhang, S. (2014).
\newblock Confidence intervals for low-dimensional parameters with
  high-dimensional data.
\newblock {\em J. R. Statist. Soc. B}, 76:217--242.
\newblock arXiv:1110.2563, 2011.


%\bibitem{ZC17}
%Zhang, X. and Cheng, G. (2017).
%\newblock Simultaneous inference for high-dimensional linear models.
%\newblock {\em Journal of the American Statistical Association}, 0:1--12.
%\%newblock arXiv:1603.01295, 2016.


\bibitem{ZhaoKolarLiuQR2014}
Zhao, T., Kolar, M., and Liu, H. (2014).
\newblock A general framework for robust testing and confidence regions in
  high-dimensional quantile regression.
\newblock arXiv:1412.8724, 2014. %{\em ArXiv:1412.8724}.

\end{thebibliography}

\begin{thebibliography}{10}


\bibitem{BC-SparseQR}
Belloni, A. and Chernozhukov, V. (2011).
\newblock $\ell_1$-penalized quantile regression for high dimensional sparse
  models.
\newblock {\em Annals of Statistics}, 39(1):82--130.
\newblock ArXiv, 2009.

\bibitem{BC-PostLASSO}
Belloni, A. and Chernozhukov, V. (2013).
\newblock Least squares after model selection in high-dimensional sparse
  models.
\newblock {\em Bernoulli}, 19(2):521--547.
\newblock ArXiv, 2009.


\bibitem{BCK-SparseQRinference}
Belloni, A., Chernozhukov, V., and Kato, K. (2013).
\newblock Robust inference in high-dimensional approximately sparse quantile
  regression models.
\newblock ArXiv, 2013. %{\em ArXiv:1312.7186}.

\bibitem{BCW-SqLASSO}
Belloni, A., Chernozhukov, V., and Wang, L. (2011).
\newblock Square-root-lasso: Pivotal recovery of sparse signals via conic
  programming.
\newblock {\em Biometrika}, 98(4):791--806.
\newblock Arxiv, 2010.

\bibitem{BCFH2013program}
Belloni, A., Chernozhukov, V., Fern{\'a}ndez-Val, I., and Hansen, C. (2013).
\newblock Program evaluation with high-dimensional data.
\newblock ArXiv, 2013. %{\em ArXiv:1311.2645}.


\bibitem{berlemann2014unraveling}
Berlemann, M., Enkelmann, S., and Kuhlenkasper, T. (2014).
\newblock Unraveling the relationship between presidential approval and the
  economy: a multidimensional semiparametric approach.
\newblock {\em Journal of Applied Econometrics}.

\bibitem{BickelRitovTsybakov2009}
Bickel, P., Ritov, Y., and Tsybakov, A. (2009).
\newblock Simultaneous analysis of {L}asso and {D}antzig selector.
\newblock {\em Annals of Statistics}, 37(4):1705--1732.
\newblock ArXiv, 2008.


\bibitem{chernozhukov2014honestbands}
Chernozhukov, V., Chetverikov, D., and Kato, K. (2014).
\newblock Anti-concentration and honest, adaptive confidence bands.
\newblock {\em The Annals of Statistics}, 42(5):1787--1818.
\newblock ArXiv, 2013.


\bibitem{chernozhukov2012gaussian}
Chernozhukov, V., Chetverikov, D., and Kato, K. (2014).
\newblock Gaussian approximation of suprema of empirical processes.
\newblock {\em The Annals of Statistics}, 42(4):1564--1597.
\newblock ArXiv, 2012.



\bibitem{chernozhukov2015noncenteredprocesses}
Chernozhukov, V., Chetverikov, D., and Kato, K. (2015).
\newblock Empirical and multiplier bootstraps for suprema of empirical
  processes of increasing complexity, and related gaussian couplings.
\newblock ArXiv, 2015. %{\em ArXiv:1502.00352}.

\bibitem{Dudley99}
Dudley, R. (1999).
\newblock {\em Uniform central limit theorems}, volume~63 of {\em Cambridge
  Studies in Advanced Mathematics}.
\newblock Cambridge University Press, Cambridge.


\bibitem{jing:etal}
Jing, B.-Y., Shao, Q.-M., and Wang, Q. (2003).
\newblock Self-normalized Cramer-type large deviations for independent random
  variables.
\newblock {\em Ann. Probab.}, 31(4):2167--2215.


\bibitem{LedouxTalagrandBook}
Ledoux, M. and Talagrand, M. (1991).
\newblock {\em Probability in Banach Spaces (Isoperimetry and processes)}.
\newblock Ergebnisse der Mathematik undihrer Grenzgebiete, Springer-Verlag.


\bibitem{negahban2012unified}
Negahban, S., Ravikumar, P., Wainwright, P., and Yu, B. (2012).
\newblock A unified framework for high-dimensional analysis of m-estimators
  with decomposable regularizers.
\newblock {\em Statistical Science}, 27(4):538--557.
\newblock ArXiv, 2010.




\bibitem{RudelsonVershynin2008}
Rudelson, M. and Vershynin, R. (2008).
\newblock On sparse reconstruction from fourier and gaussian measurements.
\newblock {\em Communications on Pure and Applied Mathematics}, 61:1025--1045.


\bibitem{vdV-W}
van~der Vaart, A. and Wellner, J. (1996).
\newblock {\em Weak Convergence and Empirical Processes}.
\newblock Springer Series in Statistics.








\end{thebibliography}

\end{document}